\documentclass[preprint]{ptephy}

\usepackage{etex}

\preprintnumber{KEK-TH-1916} 

\usepackage{amsmath, amssymb}
\usepackage{amsthm}

\usepackage{graphicx}
\usepackage{floatrow}

\makeatletter
\DeclareFontFamily{OMX}{MnSymbolE}{}
\DeclareSymbolFont{MnLargeSymbols}{OMX}{MnSymbolE}{m}{n}
\SetSymbolFont{MnLargeSymbols}{bold}{OMX}{MnSymbolE}{b}{n}
\DeclareFontShape{OMX}{MnSymbolE}{m}{n}{
    <-6>  MnSymbolE5
   <6-7>  MnSymbolE6
   <7-8>  MnSymbolE7
   <8-9>  MnSymbolE8
   <9-10> MnSymbolE9
  <10-12> MnSymbolE10
  <12->   MnSymbolE12
}{}
\DeclareFontShape{OMX}{MnSymbolE}{b}{n}{
    <-6>  MnSymbolE-Bold5
   <6-7>  MnSymbolE-Bold6
   <7-8>  MnSymbolE-Bold7
   <8-9>  MnSymbolE-Bold8
   <9-10> MnSymbolE-Bold9
  <10-12> MnSymbolE-Bold10
  <12->   MnSymbolE-Bold12
}{}

\let\llangle\@undefined
\let\rrangle\@undefined
\DeclareMathDelimiter{\llangle}{\mathopen}%
                     {MnLargeSymbols}{'164}{MnLargeSymbols}{'164}
\DeclareMathDelimiter{\rrangle}{\mathclose}%
                     {MnLargeSymbols}{'171}{MnLargeSymbols}{'171}
\makeatother

\usepackage{tikz-cd}
\usepackage[all,cmtip]{xy}

\usepackage{mathrsfs}

\usepackage{bbm}

\usepackage{mathtools}
\usepackage{tensor}

\usepackage{comment}

\renewcommand{\labelenumi}{\textup{(\roman{enumi})}}

\newtheorem{theorem}{Theorem}[section]
\newtheorem*{theorem*}{Theorem}
\newtheorem{proposition}[theorem]{Proposition}
\newtheorem*{proposition*}{Proposition}
\newtheorem{lemma}[theorem]{Lemma}
\newtheorem*{lemma*}{Lemma}
\newtheorem{corollary}[theorem]{Corollary}
\newtheorem*{corollary*}{Corollary}

\newtheorem*{definition*}{Definition}

\newtheorem*{fact*}{Fact}

\makeatletter
   
   \@addtoreset{equation}{section}
\makeatother

\begin{document}


\title{Quasi-probabilities in Conditioned Quantum Measurement and a Geometric/Statistical Interpretation of Aharonov's Weak Value}

\author{\name{Jaeha Lee}{1}, \name{Izumi Tsutsui}{2}}

\address{
\affil{}{Theory Center, Institute of Particle and Nuclear Studies, High Energy Accelerator Research Organization (KEK), 1-1 Oho, Tsukuba, Ibaraki 305-0801, Japan}
\affil{1}{jlee@post.kek.jp}
\affil{2}{izumi.tsutsui@kek.jp}
}

\begin{abstract}%
We show that the joint behaviour of an arbitrary pair of (generally non-commuting) quantum observables can be described by quasi-probabilities, which are an extended version of the standard probabilities used for describing the outcome of measurement for a single observable.  The physical situations that require these quasi-probabilities arise when one considers quantum measurement of an observable conditioned by some other variable, with the notable example being the weak measurement employed to obtain Aharonov's weak value.  Specifically, we present a general prescription for the construction of quasi-joint-probability (QJP) distributions associated with a given combination of observables.  These QJP distributions are introduced in two complementary approaches: one from a bottom-up, strictly operational construction realised by examining the mathematical framework of the conditioned measurement scheme, and the other from a top-down viewpoint realised by applying the results of spectral theorem for normal operators and its Fourier transforms.  It is then revealed that, for a pair of simultaneously measurable observables, the QJP distribution reduces to the unique standard joint-probability distribution of the pair, whereas for a non-commuting pair there exists an inherent indefiniteness in the choice of such QJP distributions, admitting a multitude of candidates that may equally be used for describing the joint behaviour of the pair.  In the course of our argument, we find that the QJP distributions furnish the space of operators in the underlying Hilbert space with their characteristic geometric structures such that the orthogonal projections and inner products of observables can, respectively, be given statistical interpretations as `conditionings' and `correlations'.  The weak value $A_{w}$ for an observable $A$ is then given a geometric/statistical interpretation as either the orthogonal projection of $A$ onto the subspace generated by another observable $B$, or equivalently, as the conditioning of $A$ given $B$ with respect to the QJP distribution under consideration.
\end{abstract}

\maketitle

%
%


\tableofcontents

\newpage

\section{Introduction}

Since the discovery of quantum mechanics in the beginning of the last century, our classical understanding of the concept of observables has undergone a drastic change.  It is by now widely accepted that, in the microscopic world, measured values of a physical quantity, termed \lq observable\rq\ in quantum mechanics, are intrinsically random, and that certain combinations of quantum observables do not admit coexistence, as exemplified typically by the pair of observables corresponding to the position and the momentum of a particle.

Such remarkable characteristics of quantum observables impose a strong limitation to the mathematical framework to be employed for describing their probabilistic behaviour; namely, it is no longer possible, in general, to assign probability spaces for the description of the joint behaviour of their arbitrary combinations in the classical sense.
Nonetheless, various attempts have been made to construct a proper mathematical framework for the probabilistic description of the combination of quantum observables that resembles the Kolmogorovian style of formulation of classical probability theory.
Extending the notion of probability has since been one of the major trends, which yielded the
extended notion of probability which goes generally by the name of `quasi-probability' or `pseudo-probability'.
Among the most celebrated proposal is the Wigner-Ville (WV) distribution \cite{Wigner_1932,Ville_1948}, commonly known as the Wigner function in the physics community, which is primarily considered for a canonically conjugate pair of quantum observables to describe their joint behaviour.  
Another, though less known, example is the Kirkwood-Dirac (KD) distribution \cite{Kirkwood_1933,Dirac_1945}, which is structured differently but is meant to serve a similar purpose for arbitrary pairs.

Historically, those proposals including the WV and KD distributions have been made more or less in a heuristic manner, and as such, the general mathematical framework for the study, including the prescription for the concrete construction of such distributions to a pair of arbitrary quantum observables, which may comprehensively be termed `quasi-joint-probability' (QJP) distributions, is still underdeveloped, not to mention a transparent overview of the relations among the QJPs.  We know, for instance, that 
both the WV and KD distributions retain similar properties to the standard joint-probability distributions defined for a pair of classical random variables,
but they exhibit their own outstanding queerness in that the former admits negative numbers to be assigned whereas the latter takes even complex numbers.  However, we still do not know whether the peculiar properties of joint-probability including those of the WV and KD distributions, which have occasionally been considered a serious impediment to their physical interpretation, are a norm of QJP distributions, or there can be other types of examples which share classical properties of joint-probability in different aspects.  The theme of this paper revolves around the concept of QJP distributions of quantum observables, with the first objective being to present a mathematically solid framework to address some of their problems in a more systematic and lucid manner.

Another motivation of this paper comes from the recent rise of interest in the novel quantum observable called the \emph{weak value},
which has been put forward by Aharonov and co-workers \cite{Aharonov_1988} based on their time-symmetric formulation of quantum mechanics \cite{Aharonov_1964} proposed more than a half century ago.  In simple terms, the weak value
\begin{equation}
A_{w} := \frac{\langle \psi^{\prime}, A \psi \rangle}{\langle \psi^{\prime}, \psi \rangle}
\end{equation}
is a physical quantity that supposedly characterises the value of the observable $A$ in the process specified by an initial state $\vert\psi \rangle$ and a final state $\vert\psi^{\prime} \rangle$ both specified in advance.  Unlike the standard physical value which is given by one of the eigenvalues of an observable $A$, 
the weak value admits a definite value for any $A$, and is envisaged to be meaningful even for a set of non-commutable observables simultaneously.

This inspired a new insight for analysing the quantum nature of the system as well as for 
understanding various counter-intuitive phenomena in quantum mechanics based on the weak value. For instance, the complex-valued nature of $A_{w}$ allows for a direct measurement of the wave function, offering a novel technique to rival the existing technology of quantum tomography.  This in turn alludes us to contemplate on the possible trajectory of a particle \cite{Lundeen_2011,Mori_2015}, a notion which has conventionally been deemed untenable due to the incompatibility of measuring the position and the momentum simultaneously.
The weak value also admits novel physical interpretations on such fundamental aspects of quantum mechanics as the wave-particle duality and the local existence of the physical quantity itself, offering us a possible resolution to some of the quantum paradoxes, including the three-box paradox \cite{Aharonov_1991}, Hardy's paradox \cite{Yokota_2009} and the Cheshire cat paradox \cite{Aharonov_2005}.  

Despite its growing attention, the status of the weak value in quantum mechanics is still not solid, and especially its physical interpretation is still open to debate.  One of the recent strategies in addressing this question has been to investigate its relations to quasi-probabilities, specifically those to the KD distribution \cite{Ozawa_2011, Hofmann_2014}.  In this paper, we shall follow this line of study and show, among others, that a novel geometric/statistical interpretation emerges from these distributions.  This necessitates a sound mathematical basis of QJP distributions, which we will provide in the course of our discussions.

The main theme of this paper is thus to obtain a more coherent understanding of the formalism of QJP distributions of quantum observables, and subsequently to apply the results in some areas of the foundational problems of quantum mechanics.  In view of this, the key problems regarding QJP distributions may be to\begin{enumerate}
\item 
provide a reasonably solid mathematical framework for the study of QJP distributions based on measure and integration theory, and possibly on the theory of generalised functions,
\item 
present a viable scheme to address the inherent indefiniteness/arbitrariness to the possible candidates for QJP distributions of non-commuting pairs of quantum observables, a methodical way for their constructions, and the relation between each of the candidates, and
\item 
devise a procedure for measuring such various candidates of QJP distributions in a systematic manner.
\end{enumerate}
We shall address these problems from two complementary approaches: one from a bottom-up, strictly operational construction realised by carefully reviewing the mathematical description of the conditioned measurement scheme, and the other from a top-down viewpoint realised by applying the results of spectral theorem for normal operators and its Fourier transforms.

The results of the study shall be subsequently applied to the analysis for the physical interpretation of the weak value.  To this end, we first concentrate on the $L^{2}$ structures which the QJP distributions naturally induce, and observe that they furnish a statistical interpretation of the geometric structures introduced on the space of observables in the underlying Hilbert space, analogously to those introduced in the space of random variables in classical probability theory.  Geometric concepts such as orthogonal projections and inner products are accordingly endowed with statistical interpretations as `conditionings' and `correlations', respectively, and in addition the representation of linear operators by functions provides us with a convenient tool for evaluating statistical quantities involved.  These observations form a basis to perform further study on the weak value in general.  As a result, 
the weak value $A_{w}$ is given a geometric/statistical interpretation: either as the orthogonal projection of an observable $A$ on the subspace generated by another observable $B$ which is determined by one of the predetermined states entering in the weak value, or equivalently, as the conditioning of $A$ given $B$ with respect to the QJP distribution under consideration. 
Although we shall not discuss it here, we mention that 
this interpretation also leads to a set of novel and remarkable inequalities of uncertainty relations for approximation/estimation which are capable of treating both the standard position-momentum inequality and the time-energy inequality \cite{Lee_2016}.

As for the practical outcomes of our argument laid out for QJP distributions, 
we mentioned earlier the systematic construction of QJP distributions and the geometric/statistical interpretation of the weak value, but each of these can be made more explicit as follows.  
First, for the systematic construction of QJP distributions, we furnish a general prescription which ensures that it can describe the joint behaviour of an arbitrary pair of quantum observables.
Specifically, inspired by the observations made on the Fourier transform of the product spectral measure of two simultaneously measurable observables $A$ and $B$, we introduce a mixture $\#(s,t)$ of the disintegrated components of $e^{-isA}$ and $e^{-itB}$ with real parameters $s, t$ 
for arbitrary pairs of (generally non-commuting) observables $A$ and $B$, and thereby define the QJP distribution of the pair by the inverse Fourier transform of the distribution $(s,t) \mapsto \langle \psi, \#(s,t) \psi \rangle / \|\psi\|^{2}$ to a given quantum state $|\psi\rangle$. 
Each of the QJP distributions is then found to possess reasonable properties to be qualified as what its name suggests to be, and one can confirm that both the WV distribution and the KD distribution do belong to this class.  The inherent arbitrariness observed to the candidates for QJP distributions is then understood as the possible variety of the way one could mix the disintegrated components of the unitary operators, which originates directly from the non-commutative nature of the pair of the observables $A$ and $B$.
A concrete measurement scheme for members of a specific subfamily of QJP distributions is further proposed.

For the geometric/statistical interpretation of the weak value, on the other hand,
we start by noting that, as distributions, each QJP distribution naturally induces an $L^{2}$ structure.  We will then find that the QJP distributions provide convenient methods of representing geometric structures in terms of the inner products of the form
\begin{equation}\label{def:intro_inn_prod}
\llangle B, A \rrangle_{\psi,\alpha} := \frac{1 + \alpha}{2} \cdot \frac{\langle B\psi, A\psi\rangle}{\|\psi\|^{2}} + \frac{1 - \alpha}{2} \cdot \frac{\langle A\psi, B\psi\rangle}{\|\psi\|^{2}}, \quad -1 \leq \alpha \leq 1,
\end{equation}
which can be introduced on the space of operators in the underlying Hilbert space by integration of functions.  With this inner product, we are allowed to consider orthogonal projections onto the subspaces $\mathfrak{E}_{\psi}(B)$ of operators generated by self-adjoint operators $B$, and find that the orthogonal projections can be interpreted as conditioning given $B$ with respect to the QJP distributions under consideration.  The projection
\begin{equation}
P_{\alpha}(A|B;\psi) = \int_{\mathbb{R}}  \left( \frac{1 + \alpha}{2} \cdot \frac{\langle b, A \psi \rangle}{\langle b, \psi \rangle} + \frac{1 - \alpha}{2} \cdot \frac{\langle \psi, A b \rangle}{\langle \psi, b \rangle} \right) dE_{B}(b)
\end{equation}
of the observable $A$ on the subspace $\mathfrak{E}_{\psi}(B)$ is further found to be described by the weak value\footnote{
Here, $E_{B}$ denotes the unique spectral measure associated to the self-adjoint operator $B$ (more on this in Section~\ref{sec:ups_II_pre}). Intuitively, $E_{B}(b) = |b\rangle\langle b|$ is the projection associated with each of the eigenvalues $b$ of $B$, and thus $dE_{B}(b) = |b\rangle\langle b| db$ in a laxer expression.},
providing us with its proper geometric/statistical interpretation (Proposition~\ref{prop:orth_proj_cond_qe}).

Having furnished a general introduction to the topic of QJP distributions of quantum observables and the weak value along with a brief summary of the content, we now
give the outline of the present paper.  
After this introductory section, we organize the main body, Section 2 to Section 7, into the following three logical groups of mutually interrelated topics:
\begin{enumerate}
\renewcommand{\labelenumi}{(\Alph{enumi})}
\item {\bf QJP: Heuristic Construction}\hspace{5pt}  Four sections starting from Section~\ref{sec:ups_I} to \ref{sec:ps_II} are devoted to a heuristic and bottom-up construction of QJP distributions of a pair of quantum observables. 
This is accomplished by a thorough analysis on the mathematical formalism of two measurement schemes.  One is the 
standard scheme, which we call the \lq unconditioned measurement  (UM) scheme\rq, in which we measure an observable $A$ under a given state as conventionally done (Section~\ref{sec:ups_I} to \ref{sec:ups_II}).  The other is what we call the \lq conditioned measurement (CM) scheme\rq, in which under a given  state we measure an observable $A$ along with another observable $B$ whose outcome is used for conditioning (Section~\ref{sec:ps_I} to \ref{sec:ps_II}).
Each of these analyses will be conducted on the level of (conditional) expectations and (conditional) probabilities.
\begin{enumerate}
\renewcommand{\labelenumii}{(\roman{enumii})}
\item {\bf UM I}\hspace{5pt}
We start by reviewing, in Section~\ref{sec:ups_I}, the UM scheme by a standard operator-centric approach, and investigate how one could reclaim the information of the target system by that means.
\item {\bf UM II}\hspace{5pt}
Subsequently, in Section~\ref{sec:ups_II}, we take a closer look on the UM scheme in the level of probabilities, where the quantity of interest is now not only the statistical average, but also the `raw' probability measure describing the probabilistic behaviour of the measurement outcomes of the meter observable,  and discuss how one could recover the probability measure describing the outcomes of the target observable.
\item {\bf CM I}\hspace{5pt}
From Section~\ref{sec:ps_I} onward, we turn our attention to the CM scheme.  
In Section~\ref{sec:ps_I}, we first conduct, in a parallel manner as we have done in the preceding Section~\ref{sec:ups_I}, an analysis in the operator level, where now the quantity of interest becomes the conditional expectation of the meter observable   given another conditioning observable $B$ of the target system.  
\item {\bf CM II}\hspace{5pt}
In Section~\ref{sec:ps_II}, the study of the CM scheme is given a probabilistic approach, where the quantity of interest is the Wigner-Ville distribution of a pair of canonically conjugate observables on the meter system conditioned by the outcome of the conditioning observable $B$ of the target system.  We then see that this implies the existence of the concept of QJP distributions of pairs of generally non-commuting observables.

\end{enumerate}

\item {\bf QJP: Formal Definition}\hspace{5pt}
Inspired by the heuristic arguments employed in the operational analyses over the preceding four sections, we devote Section~\ref{sec:qp_qo} to the top-down construction of QJP distributions for arbitrary pairs of generally non-commutating quantum observables.  We shall then summarise our findings obtained through Section~\ref{sec:ups_I} to Section~\ref{sec:ps_II} from a rather aerial viewpoint, discussing where the heuristic arguments and observations in the preceding sections find their places in this relatively general framework.

\item {\bf Application to the Interpretation of Weak Values}\hspace{5pt}
As an application of the mathematical formalism provided so far, in Section~\ref{sec:app} we conduct a study on the quantum analogue of correlations, which can be defined even for a pair of non-commuting observables.  This leads us to the aforementioned geometric/statistical interpretation of the weak value as conditional quasi-expectations.
\end{enumerate}
We shall finally summarise our results and give some concluding remarks in the last Section~\ref{sec:sc}.

Prior to our main discussions, however, we wish to say a few words about the mathematical preliminaries we supposed for the readers in preparing this paper.
The formalism that we intend to provide necessarily requires, on top of the mandatory functional analysis, moderate acquaintance to  measure and integration theory, preferably some familiarity with the basic terminologies in general topology, and ideally insight into the basic ideas of the theory of generalised functions.  The obvious difficulty is then to find a decent balance between rigour and generality on one side, and accessibility on the other.  To achieve this balance as much as possible, and assure our entire arguments to be fully accessible without any prior knowledge of advanced mathematics, we have included at the beginning of each section a subsection entitled Reference Materials containing a rather lengthy introduction of mathematical concepts that are used in the subsequent discussions.  While the authors took care in introducing these mathematical concepts and their results in a self-contained manner to respect their logical sequence, these Reference Materials are primarily intended to serve as a convenient place to summarise the basic concepts and results in a crash-course, and as such, the mathematical theories presented there are not intended to be learned from scratch.  For those who are interested in the mathematics itself are advised to be referred to standard textbooks on the respective topics, {\it e.g.}, for general topology \cite{Kelley_1975,Munkres,Querenburg_2011}, measure and integration theory \cite{Elstrodt_2011, Amann_1998a, Amann_1998b, Amann_2001, Rudin_1976, Rudin_1986}, functional analysis
\cite{Rudin_1991, Werner_2011},
and also those specifically targeting the audience from the physics community \cite{Reed_Simon_1975, Reed_Simon_1980, Goldhorn_2009, Goldhorn_2010}.
  Naturally, those who are already familiar with the preparatory materials may safely skip them and directly go to the main arguments that follow.

Admittedly, the style of discussion found in this paper is heavily oriented toward mathematical rigorousness and logical clarity rather than brevity and physical intuition, especially compared to those found in the majority of the literature in physics.
However, in spite of the possible initial hesitation that may be expected for the general readers due to the unfamiliarity of the style, the authors decided to adopt it in the belief 
that this way of presentation has its own merit, and that the costs will outweigh the rewards in the end.  In fact, several important concepts and results from the branches of mathematics mentioned above (specifically, measure and integration theory and functional analysis) are quite indispensable in understanding some of the interesting results obtained in this paper.  This is so, for instance, in defining the conditional quasi-expectations (to which Aharonov's weak value belongs as a special case) in terms of the Radon-Nikod{\'y}m derivative to understand their properties (Section~\ref{sec:psI_cond_quasi_exp}), in formulating the problem of the `limit of amplification' by conditioning in terms of essential suprema (Section~\ref{sec:psI_amplification}), in defining a family of QJP distributions of a combination of generally non-commuting quantum observables by the method of hashing (Section~\ref{sec:qp_qo}), and in providing geometric and `statistical' interpretation of conditional quasi-expectations (Section~\ref{sec:app}).  The authors hope that the readers will not be discouraged by these mathematical materials, but rather enjoy them to go through the discussions and reach the fruit of the physical results they finally brings forth.

\paragraph{Mathematical Notations Employed}
Throughout this paper, we denote by $\mathbb{K}$ either the real field $\mathbb{R}$ or the complex field $\mathbb{C}$, and define $\mathbb{K}^{\times} := \mathbb{K} \setminus \{0\}$. In order to avoid confusion, we denote the collection of all natural numbers including $0$ by $\mathbb{N}_{0}$, and $\mathbb{N}^{\times} := \mathbb{N}_{0} \setminus \{0\}$. Since our primary interest is on quantum mechanics, Hilbert spaces are always assumed to be complex. Conforming to the convention in physical literature, we denote the complex conjugate of a complex number $c \in \mathbb{C}$ by $c^{*}$, and an inner product $\langle \,\cdot\, , \,\cdot\, \rangle$ defined on a complex linear space is anti-linear in its first argument and linear in the second.
For simplicity, we adopt the natural units where we specifically have $\hbar = 1$, unless stated otherwise.

\newpage
\section{Unconditioned Measurement I: In Terms of Expectations}\label{sec:ups_I}

We start by providing a brief review on the archetype of the indirect measurement scheme widely known as the von Neumann measurement scheme. The scheme will be referred to as the \emph{unconditioned measurement} (UM) scheme in generic terms throughout this paper, primarily in order to contrast it with the \emph{conditioned measurement} (CM) scheme (which includes the \emph{post-selected measurement} scheme as a special case) discussed later.

\subsection{Reference Materials}\label{sec:ups_I_pre}

As a preamble to this section, we here include three introductory topics that form the basis of our study. We start by collecting some of the basic terminologies and results of measure and integration theory, based on which modern probability theory was established by Kolmogorov {\it et al.}
Subsequently, we provide a brief note on both the \emph{Schr{\"o}dinger representation} and the \emph{Weyl representation} of the canonical commutation relations (CCR), which will be extensively employed in describing the meter system in our measurement scheme. We finally close this subsection by providing a short summary on the precise definition of tensor products of Hilbert spaces and that of self-adjoint operators. Since these materials are included just to make our presentation self-contained, those who are already familiar with the subject may safely skip the contents and proceed directly to Section~\ref{sec:ups_I_ups}.

\subsubsection{A Crash-Course into Measure and Integration Theory}\label{sec:usp_I_MI}

We begin by presenting some of the most basic concepts and results of measure and integration theory, starting from the definition of measure spaces up to the construction of the Lebesgue integration, followed by the definition of $L^{p}$ spaces.

\paragraph{$\sigma$-algebras and Measurable Spaces}
Let $X$ be any set, and let $\mathfrak{P}(X)$ denote the power set%
\footnote{The symbol $\mathfrak{A}$ is the capital letter of the Fraktur typeface of `A' as in `Algebra', $\mathfrak{B}$ for `B' as in `Borel', $\mathfrak{E}$ for `E' as in `Erzeuger (generator)', $\mathfrak{O}$ for `O' as in `offen (open)' and $\mathfrak{P}$ for `P' as in `Potenz (power)' (some of them introduced shortly after).}
of $X$, {\it i.e.}, the collection of all subsets of $X$. A family $\mathfrak{A} \subset \mathfrak{P}(X)$ of subsets of $X$ is called a \emph{$\sigma$-algebra}
over $X$, if it satisfies the following conditions:
\begin{enumerate}
\item $X \in \mathfrak{A}$.
\item $A \in \mathfrak{A}$ implies $A^{c} := X \setminus A \in \mathfrak{A}$.
\item For any sequence $(A_{n})_{n \geq 1}$ of subsets of $X$, $\bigcup_{n=1}^{\infty} A_{n} \in \mathfrak{A}$ holds.
\end{enumerate}
Given a $\sigma$-algebra $\mathfrak{A}$ over $X$, each element $A \in \mathfrak{A}$ is called a \emph{measurable set}, and the ordered pair $(X, \mathfrak{A})$ is called a \emph{measurable space}.

\paragraph{Generator of a $\sigma$-algebra}
A trivial, but important property of $\sigma$-algebras is that, for any collection $(\mathfrak{A}_{i})_{i \in I}$ of $\sigma$-algebras over $X$ indexed by an index set $I$, the intersection $\bigcap_{i \in I} \mathfrak{A}_{i} = \{ A \in \mathfrak{P}(X) : A \in \mathfrak{A}_{i}, \forall i \in I\}$ is itself a $\sigma$-algebra over $X$.
This leads to the following basic fact: For any collection $\mathfrak{E} \subset \mathfrak{P}(X)$ of subsets of $X$, there exists a smallest (with respect to the set inclusion) $\sigma$-algebra encompassing $\mathfrak{E}$, namely, the intersection of all $\sigma$-algebras that encompass $\mathfrak{E}$. The intersection is called the \emph{$\sigma$-algebra generated by $\mathfrak{E}$}, denoted as $\sigma(\mathfrak{E})$, and $\mathfrak{E}$ is in turn called the \emph{generator} of $\sigma(\mathfrak{E})$.

\paragraph{Borel $\sigma$-algebras}
Let $X$ be a metric (or, in general, a topological) space, and let $\mathfrak{O}$ denote the collection of all open sets of $X$. We call the $\sigma$-algebra generated by $\mathfrak{O}$, the \emph{Borel $\sigma$-algebra of $X$}, and denote it by $\mathfrak{B}(X) := \sigma(\mathfrak{O})$. We prepare a special symbol for the special case $X = \mathbb{R}^{n}$ ($n \in \mathbb{N}^{\times}$), in which we denote the Borel $\sigma$-algebra of $\mathbb{K}^{n}$ by $\mathfrak{B}^{n} := \mathfrak{B}(\mathbb{R}^{n})$, which is among the most well-known examples of $\sigma$-algebras that, incidentally, also plays an important role in quantum theory. For simplicity, we occasionally denote $\mathfrak{B} := \mathfrak{B}^{1}$ whenever there is no risk of confusion.

\paragraph{Measures and Measure Spaces}

Let $(X, \mathfrak{A})$ be a measurable space. A map $\mu: \mathfrak{A} \to \overline{\mathbb{R}}$ from the $\sigma$-algebra $\mathfrak{A}$ to the \emph{extended real line} $\overline{\mathbb{R}} := \mathbb{R} \cup \{-\infty, \infty\}$ is called a \emph{measure}, if $\mu$ satisfies the following conditions:
\begin{enumerate}
\item $\mu(\emptyset) = 0$.
\item $\mu \geq 0$.
\item For any sequence $(A_{n})_{n \geq 1}$ of pairwise disjoint subsets of $X$, the \emph{countable additivity}
\begin{equation}\label{def:count_add}
\mu\left(\bigcup_{n=1}^{\infty} A_{n}\right) = \sum_{n=1}^{\infty} \mu(A_{n})
\end{equation}
holds.
\end{enumerate}
Given a measure $\mu$ over a measurable space $(X, \mathfrak{A})$, the ordered triple $(X, \mathfrak{A}, \mu)$ is called a \emph{measure space}.

\paragraph{Lebesgue-Borel Measure}
As a concrete example, we make notes on the $n$-dimensional \emph{Lebesgue-Borel measure} $\beta^{n}$ ($n \in \mathbb{N}^{\times}$) defined on the measurable space $(\mathbb{R}^{n}, \mathfrak{B}^{n})$, which is among the most well-known and important examples of measure spaces. To this end, we first recall that a measure $\mu$ on $(\mathbb{R}^{n}, \mathfrak{B}^{n})$ is called \emph{translation invariant}, if
\begin{equation}
\mu(B + a) = \mu(B), \quad B \in \mathfrak{B}^{n},
\end{equation}
holds for any $a \in \mathbb{R}^{n}$, where $B + a := \{ x + a: x \in B\}$. The Lebesgue-Borel 
measure $\beta^{n}$ is then specified as the unique translation invariant measure on $(\mathbb{R}^{n}, \mathfrak{B}^{n})$ that satisfies the normalisation condition $\beta^{n}(]0,1]^{n}) = 1$, where
\begin{equation}
]0,1]^{n} := \{ x \in \mathbb{R}^{n} : 0 < x_{i} \leq 1\,\,\, \text{for} \,\, 1 \leq i \leq n, \text{ $x_{i}$ is the $i$th coordinate of $x$}\}.
\end{equation}
This is the measure which is implicitly assumed for the most case in performing the usual integration by the symbol
\begin{equation}
\int_{-\infty}^{\infty} f(x)\, dx := \int_{\mathbb{R}} f(x)\ d\beta(x),
\end{equation}
which is a common practice in the physics community (the precise definition of the integral on the r.~h.~s. will be presented shortly after).  The proof of the existence and uniqueness of the Lebesgue-Borel measure will be found in most elementary textbooks on the topic.

\paragraph{Measurable Functions}

Let $(X, \mathfrak{A})$ and $(X^{\prime}, \mathfrak{A}^{\prime})$ be measurable spaces. A map $f: X \to X^{\prime}$ is called \emph{$\mathfrak{A}$-$\mathfrak{A}^{\prime}$ measurable} (or just \emph{measurable} for short, whenever the measure spaces concerned are obvious by context), if $f^{-1}(\mathfrak{A}^{\prime}) \subset \mathfrak{A}$ holds. In particular, we call a map $f:X \to X^{\prime}$ from a metric (or a topological) space $X$ to another metric (or a topological) space $X^{\prime}$ \emph{Borel-measurable} if it is $\mathfrak{B}(X)$-$\mathfrak{B}(Y)$ measurable. An important fact to note is that a continuous map $f: X \to X^{\prime}$ is necessarily Borel-measurable.

\paragraph{Numerical Functions}
In integration theory, it proves fruitful to consider not only real functions 
$f: X \to \mathbb{R}$, but also functions that take values in the extended real line $\overline{\mathbb{R}}$, which is called a \emph{numerical function}. One naturally equips $\overline{\mathbb{R}}$ with the ordering $-\infty < a < +\infty$, $a \in \mathbb{R}$, and may also define agreeable operations of addition, subtraction and multiplication,
where most of them should be self-evident, except for the following rather arbitrary definition
\begin{equation}
0 \cdot (\pm \infty) := (\pm \infty) \cdot 0 := 0, \quad \infty -\infty := -\infty + \infty := 0.
\end{equation}
We then define the $\sigma$-algebra on $\overline{\mathbb{R}}$ by
\begin{equation}
\overline{\mathfrak{B}} := \{ B \cup E : B \in \mathfrak{B}, \, E \subset \{-\infty, +\infty\}\},
\end{equation}
where, in particular, its restriction on the real line gives $\overline{\mathfrak{B}}|_{\mathbb{R}} = \mathfrak{B}$. We then say that a numerical function $f: (X, \mathfrak{A}) \to (\overline{\mathbb{R}}, \overline{\mathfrak{B}})$ is measurable, if it is $\mathfrak{A}$-$\overline{\mathfrak{B}}$ measurable.  Throughout this paper, we denote by $\mathcal{M}^{+}(\mathfrak{A})$  (or occasionally by $\mathcal{M}^{+}$, whenever the $\sigma$-algebra concerned is evident by context) the collection of all measurable non-negative numerical functions.

\paragraph{Lebesgue Integration}
In introducing the concept of integration, we proceed in three steps: We first define the integration for non-negative step functions, then extend the treatment to functions belonging to $\mathcal{M}^{+}$, and finally discuss the integrability of measurable numerical or complex functions.
\begin{enumerate}
\item {\it Integration of Step Functions. \hspace{15pt}}
Let $(X, \mathfrak{A}, \mu)$ be a measure space. A measurable function $f: (X, \mathfrak{A}) \to (\mathbb{R}, \mathfrak{B})$ is called a \emph{step function} (staircase function, simple function), if it takes only finite distinct values in $\mathbb{R}$. The collection of all measurable non-negative step functions will be denoted by $\mathcal{T}^{+}$. One readily sees that a non-negative step function $f \in \mathcal{T}^{+}$ admits an expression
\begin{equation}\label{eq:step_function_expression}
f = \sum_{k=1}^{m} a_{k}\chi_{A_{k}},
\end{equation}
where $a_{1}, \dots, a_{m} \geq 0$ are non-negative real numbers, $A_{1}, \dots, A_{m} \in \mathfrak{A}$ are measurable sets, and $\chi_{A}$ denotes the characteristic function
\begin{align}\label{def:characteristic_function}
\chi_{A}(x) =
\begin{cases}
    1, & x \in A, \\
    0, & x \notin A,
\end{cases}
\end{align}
of the subset $A \subset X$.
We then define the \emph{($\mu$-)integral of $f$ (over $X$)} as
\begin{equation}\label{def:lebesgue_integral_01}
\int_{X} f\ d\mu := \sum_{k=1}^{m} a_{k}\mu(A_{k}),
\end{equation}
whose value lies in $[0,\infty]$. Note that, although the expression (\ref{eq:step_function_expression}) is non-unique due to the possible choice of the
measurable sets used, 
the definition \eqref{def:lebesgue_integral_01} is well-defined since the outcome of the integral is independent of the choice.

\item
{\it Integration of Functions in $\mathcal{M}^{+}$. \hspace{15pt}}
Now that we have defined the Lebesgue integral of non-negative step functions, we next define the integral of non-negative measurable numerical functions. For $f \in \mathcal{M}^{+}$, the Lebesgue integral of $f$ is defined as
\begin{equation}\label{def:lebesgue_integral_02}
\int_{X} f\ d\mu := \sup \left\{ \int_{X} s\ d\mu : 0 \leq s \leq f, s \in \mathcal{T}^{+} \right\}.
\end{equation}
The above definition \eqref{def:lebesgue_integral_02} is consistent with that for step functions \eqref{def:lebesgue_integral_01} introduced earlier, for one readily checks that the integral coincides for $f \in \mathcal{T}^{+} \subset \mathcal{M}^{+}$.

\end{enumerate}

Before we move on to the final step, we introduce some useful notations. We let $\mathbb{K}$ denote either the real field $\mathbb{R}$ or the complex field $\mathbb{C}$, and we understand them to be respectively equipped with the Borel $\sigma$-algebra $\mathfrak{B}$ or $\mathfrak{B}^{2}$. Analogously, we let
\begin{equation*}
\hat{\mathbb{K}} := \overline{\mathbb{R}} \text{ or } \mathbb{C}, \text{ respectively equipped with the $\sigma$-algebra $\hat{\mathfrak{B}} := \overline{\mathfrak{B}}$ or $\mathfrak{B}^{2}$}
\end{equation*}
for later convenience. For a numerical function $f : X \to \overline{\mathbb{R}}$, we define its positive and negative parts as
\begin{equation}
f^{\pm}(x) := \max( \pm f(x), 0).
\end{equation}
One then sees that a function $f: X \to \hat{\mathbb{K}}$ is measurable if and only if all the positive and negative parts of both the real and imaginary parts $(\,\mathrm{Re} f)^{\pm}$, $(\,\mathrm{Im} f)^{\pm}$ of $f$ are measurable.
Given the necessary preparations, we finally obtain the following definition:
\begin{definition*}[Lebesgue Integral]
Under the assumptions above, a function $f: X \to \hat{\mathbb{K}}$ is called $\mu$-integrable (or simply integrable) over $X$ if $f$ is measurable, and all the four integrals
\begin{equation}
\int_{X} (\,\mathrm{Re} f)^{\pm}\ d\mu,\qquad \int_{X} (\,\mathrm{Im} f)^{\pm}\ d\mu
\end{equation}
are finite. The value
\begin{equation}
\int_{X} f\ d\mu := \int_{X} (\,\mathrm{Re} f)^{+}\ d\mu - \int_{X} (\,\mathrm{Re} f)^{-}\ d\mu + i \int_{X} (\,\mathrm{Im} f)^{+}\ d\mu - i \int_{X} (\,\mathrm{Im} f)^{-}\ d\mu
\end{equation}
is then called the ($\mu$-)integrable of $f$ (over $X$) or the Lebesgue integral of $f$ (over $X$ with respect to $\mu$).
\end{definition*}
\noindent
By definition, linearity
\begin{equation}
\int_{X} (af(x) + bg(x))\, d\mu(x) = a\int_{X} f(x)\, d\mu(x) + b\int_{X} g(x)\, d\mu(x),
\end{equation}
of the integration naturally follows as expected.
For a measurable set $A \in \mathfrak{A}$, the use of the shorthand
\begin{equation}
\int_{A}f\ d\mu := \int_{X} \chi_{A} \cdot f\ d\mu
\end{equation}
is common, where $\chi_{A}$ is the characteristic function of the measurable set.

\paragraph{Probability Spaces and Expectation Values}
A measure space $(X,\mathfrak{A},\mu)$ is called a \emph{probability space}, if the measure is normalised by unity $\mu(X) = 1$. Given a probability space $(X,\mathfrak{A},\mu)$ and a $\mu$-integrable function $f$, the total integration of $f$ is occasionally denoted by
\begin{equation}
\mathbb{E}[f;\mu] := \int_{X} f\ d\mu,
\end{equation}
and called the \emph{expectation value} of $f$ under $\mu$.

\paragraph{Dominated Convergence Theorem}
The advantage of the Lebesgue integration (over the familiar Riemann counterpart) especially manifests itself when dealing with convergence. For later use throughout this paper, we make a note of one of the most powerful and oft-used theorems regarding the interchange of limit and integration. To this end, we first furnish some terminologies. Let $(X, \mathfrak{A}, \mu)$ be a measure space, and let a statement $E$ be defined on each element $x \in X$. We say that the statement $E$ holds \emph{($\mu$-) almost everywhere} (abbreviation: ($\mu$)-a.e.), if there exists a measurable set $N \in \mathfrak{A}$ with $\mu(N) = 0$ such that the statement $E$ holds for $ x \in X \setminus N$.
\begin{theorem*}[Dominated Convergence Theorem]
Let $(X, \mathfrak{A}, \mu)$ be a measure space, and let $f, f_{n}: X \to \hat{\mathbb{K}}$ ($n \in \mathbb{N}^{\times}$) be measurable. If the sequence of the functions converge point-wise $\lim_{n \to \infty} f_{n} = f$ $\mu$-a.e., and if moreover there exists an $\mu$-integrable function $g \in \mathcal{M}^{+}$ such that $|f_{n}| \leq g$ holds $\mu$-a.e. for all $n \in \mathbb{N}^{\times}$, then
\begin{equation}
\lim_{n \to \infty} \int_{X} \left| f_{n} - f \right| \ d\mu = 0
\end{equation}
holds, which in particular implies
\begin{equation}
\lim_{n \to \infty} \int_{X} f_{n}\ d\mu = \int_{X} f\ d\mu.
\end{equation}
\end{theorem*}

\paragraph{$L^{p}$ Spaces}

Having provided the definition of the Lebesgue integration, we close this subsection by introducing an important class of function spaces: $L^{p}$. Let $\mathcal{L}^{p}(\mu)$, $1 \leq p <\infty$, denote the space of all measurable functions $f: X \to \mathbb{K}$ for which its $L^{p}$-norm
\begin{equation}\label{def:Lp_norm}
\|f\|_{p} := \left( \int_{X} |f|^{p}\ d\mu \right)^{1/p}
\end{equation}
is finite. For $p = \infty$, we let $\mathcal{L}^{\infty}(\mu)$ denote the space of all $f$ for which its \emph{essential supremum}
\begin{equation}\label{def:ess_sup}
\|f\|_{\infty} := \inf\{\lambda \in [0,\infty] : |f| \leq \lambda \text{ $\mu$-a.e.} \}
\end{equation}
is finite (such a function is called \emph{essentially bounded}).  The term essential supremum is justified by the fact that the evaluation $|f| \leq \|f\|_{\infty}$ $\mu$-a.e. universally holds (to see this, observe that if $\|f\|_{\infty} < \infty$ is given, $\{ x : |f| > \|f\|_{\infty}\} = \bigcup_{n=1}^{\infty}\{|f| > \|f\|_{\infty} + 1/n\}$ is a set of measure zero). Now, by identifying two functions $f, g \in \mathcal{L}^{p}(\mu)$ by the equivalence relation $f \sim g \Leftrightarrow f = g\ \text{$\mu$-a.e.}$, we obtain a quotient space $L^{p}(\mu) := \mathcal{L}^{p}(\mu)/\sim$. For simplicity, it is customary to denote an element of $L^{p}(\mu)$ by its representative $f \in  \mathcal{L}^{p}(\mu)$ whenever there is no risk of confusion. For $f \in L^{p}(\mu)$, one finds that the quantity $\|f\|_{p}$, $1 \leq p \leq \infty$ is well-defined (irrespective of the choice of the representative), and that this in fact provides a norm on $L^{p}(\mu)$, called the $L^{p}$-norm. The norm $\|\cdot\|_{p}$ is also known to be complete and hence makes $L^{p}(\mu)$ into a Banach space. The case $p=2$ is of particular interest in the context of quantum mechanics, where the integration,
\begin{equation}\label{def:L2-inner_product}
\langle g, f \rangle := \int g^{*}f\ d\mu,
\end{equation}
defines an inner product that satisfies $\langle f, f \rangle = \|f\|_{2}^{2}$, making $L^{2}(\mu)$ into a Hilbert space.
As a special case, we are mostly interested in the choice $(\mathbb{R}^{n}, \mathfrak{B}^{n}, \beta^{n})$ of the measure space. Conforming to convention in physical literature, we prepare a special symbol for the $L^{p}$ spaces of it and denote $L^{p}(\mathbb{R}^{n}) := L^{p}(\beta^{n})$.

\paragraph{H{\"o}lder's Inequality}

Among the most important inequality regarding $L^{p}$-spaces is the H{\"o}lder's inequality.
\begin{theorem}[H{\"o}lder's Inequality]
Let $1 \leq p,q \leq \infty$, $\frac{1}{p} + \frac{1}{q} = 1$, where we understand $1/\infty :=0$, and let $f, g : X \to \hat{\mathbb{K}}$ be measurable.  Then,
\begin{equation}
\|fg\|_{1} \leq \|f\|_{p}\|g\|_{q}
\end{equation}
holds.
\end{theorem}
\noindent
For the specific choice $p,q = 2$, the resulting inequality has its own name as the Cauchy-Schwarz Inequality.

\subsubsection{Rudimentary Techniques in handling the CCR}

While the contents of the following topics are widely known, we include this material mainly for reader's convenience, and also for self-consistency and reference.

\paragraph{Schr{\"o}dinger Representation of the CCR}

We start by recalling the definition of the Schwartz space.
A function $f :\mathbb{R}^{n} \to \mathbb{K}$ is called \emph{rapidly decreasing} when
\begin{equation}
\lim_{|x| \to \infty} x^{\gamma}f(x) = 0
\end{equation}
holds for any $\gamma := (\gamma_{1}, \dots, \gamma_{n})$ with $\gamma \in \mathbb{N}_{0}^{n}$.  Here, the multi-index symbol $\gamma\in \mathbb{N}^{n}_{0}$ is understood to be used as 
\begin{equation}\label{def:use_alpha}
x^{\gamma} := x_{1}^{\gamma_{1}} \cdots x_{n}^{\gamma_{n}}, \quad D^{\gamma} := (D_{1})^{\gamma_{1}} \cdots (D_{n})^{\gamma_{n}},
\end{equation} 
where $D_{i} := \partial/\partial x_{i}$ is the partial differentiation operator with respect to the variable $x_{i}$.
The space
\begin{equation}
\mathscr{S}(\mathbb{R}^{n}) := \{f \in C^{\infty}(\mathbb{R}^{n}): D^{\gamma}f\ \text{is rapidly decreasing},\ \gamma \in \mathbb{N}^{n}_{0} \},
\end{equation}
is then called the \emph{Schwartz space}, and its elements are in turn called \emph{Schwartz functions}. The Schwartz space is known to be a dense subspace $\mathscr{S}(\mathbb{R}^{n}) \subset L^{p}(\mathbb{R}^{n})$ for $1 \leq p < \infty$. A well-known example of Schwartz functions is provided by the form,
\begin{equation}
x^{\gamma}e^{-a|x|^{2}} \in \mathscr{S}(\mathbb{R}^{n}), \quad \gamma \in \mathbb{N}^{n}_{0},\ \, a>0.
\end{equation}
Specifically, the Gaussian wave-functions, which also appear later in our analysis, are among the most oft-used members of the Schwartz space belonging to this class.

Now that we have the necessary definitions, we return to the main topic of this subsection and, for simplicity, confine ourselves to the case $n=1$ without loss of generality.
We start by introducing a pair of important operators $\hat{x}$ and $\hat{p}$ on the Hilbert space $L^{2}(\mathbb{R})$.
Among these, $\hat{x}:  \mathrm{dom}(\hat{x}) \to L^{2}(\mathbb{R})$ is an operator on $L^{2}(\mathbb{R})$ defined by the multiplication of $x$ on a function $f$,
\begin{equation}\label{def:position_operator}
\hat{x} : f(x) \mapsto xf(x),
\end{equation}
with its domain,
\begin{equation}
\mathrm{dom}(\hat{x}) := \{ f \in L^{2}(\mathbb{R}) : xf \in L^{2}(\mathbb{R}) \}.
\end{equation}
The operator $\hat{x}$ is known to be self-adjoint and is called the (one-dimensional) \emph{position operator}.

Next, consider the operator $-i D$ defined on $\mathscr{S}(\mathbb{R})$ with $D := d/dx$ being the usual differential operator in our case $n=1$.
The operator $-iD :  \mathscr{S}(\mathbb{R}) \to L^{2}(\mathbb{R})$ is known to be essentially self-adjoint, 
which allows us to define the (one-dimensional) \emph{momentum operator} by its self-adjoint extension%
\footnote{While the explicit identification of the domain of the operator $\hat{p}$ is not quite straightforward, we mention that it is given by
\begin{equation}\label{eq:domain_p}
\mathrm{dom}(\hat{p}) = \left\{ f \in L^{2}(\mathbb{R}) : f|_{J} \in \mathrm{AC}(J) \text{ for all compact sub-intervals } J \subset \mathbb{R},\ \frac{df}{dx} \in L^{2}(\mathbb{R}) \right\},
\end{equation}
where $f|_{J}$ denotes the restriction of the function $f$ on the interval $J$, and $\mathrm{AC}(J)$ denotes the space of all absolutely continuous functions on $J$. Here, a function $f: [a,b] \to \mathbb{K}$ is called \emph{absolutely continuous}, if for every $\epsilon > 0$, there exists a $\delta > 0$ such that 
\begin{equation}
\sum_{k=1}^{n}(b_{k} - a_{k}) < \delta \quad \Rightarrow \quad \sum_{k=1}^{n}|f(b_{k}) - f(a_{k})| < \epsilon
\end{equation}
holds for arbitrary partitions $a \leq a_{1} < b_{1} \leq a_{2} < b_{2} \leq \dots \leq a_{n} < b_{n} \leq b$, $n \in \mathbb{N}^{\times}$ of the interval $[a,b]$. 
It is known that a function $f: [a,b] \to \mathbb{K}$ is absolutely continuous if and only if $f$ is differentiable almost everywhere (hence $df/dx$ in \eqref{eq:domain_p} is well-defined), its derivative is Lebesgue integrable $df/dx \in L^{1}(J)$, and that
\begin{equation}
f(t) - f(s) = \int_{s}^{t} \frac{df}{dx}\ dx, \quad s, t \in [a,b],\ s \leq t
\end{equation}
holds ({\it cf.} fundamental theorem of calculus).
},
\begin{equation}\label{def:momentum_operator}
\hat{p} := \overline{-i D}.
\end{equation}
Here, the overline on a closable operator denotes its closure, which in the case of an essentially self-adjoint operator is equivalent to its (unique) self-adjoint extension.
 
One then verifies that the pair $\{\hat{x}, \hat{p}\}$ satisfies the familiar (one-dimensional) \emph{canonical commutation relations (CCR)},
\begin{gather}
[\hat{x}, \hat{p}] = iI, 
\label{eq:CCR_01} \\
[\hat{x}, \hat{x}] = 0, \quad [\hat{p}, \hat{p}] = 0, 
\label{eq:CCR_02}
\end{gather}
on the subspace $\mathscr{S}(\mathbb{R}) \subset L^{2}(\mathbb{R})$, where $I$ denotes the identity operator and $[X, Y] := XY - YX$ denotes the commutator for operators $X, Y$, whose domain is understood to be $\mathrm{dom}([X,Y]) := \mathrm{dom}(XY) \cap \mathrm{dom}(YX)$.

In general, let $\{\mathcal{H}, \mathcal{D}, \{Q,P\}\}$ be a combination consisting of a Hilbert space $\mathcal{H}$, its dense subspace $\mathcal{D} \subset \mathcal{H}$, and a pair of self-adjoint operators $\{Q,P\}$ on $\mathcal{H}$. We say that $\{\mathcal{H}, \mathcal{D}, \{Q,P\}\}$ is a (one-dimensional) \emph{representation of the CCR}, if the CCR
\begin{gather}
\label{eq:ccr}
[Q, P] = iI, \\
\label{eq:ccr2}
[Q, Q] = 0, \quad [P, P] = 0
\end{gather}
hold on the domain $\mathcal{D}$ fulfilling
\begin{equation}
\label{eq:weyl_subspace}
\mathcal{D} \subset \mathrm{dom}(QQ) \cap \mathrm{dom}(QP) \cap \mathrm{dom}(PQ) \cap \mathrm{dom}(PP).
\end{equation}
One then concludes from the above argument that the combination,
\begin{equation}
\left\{ L^{2}(\mathbb{R}), \mathscr{S}(\mathbb{R}), \{\hat{x}, \hat{p}\}\right\},
\end{equation}
gives a concrete example for the representation of the CCR, called the (one-dimensional) \emph{Schr{\"o}dinger representation of the CCR}.

\paragraph{Weyl Representation of the CCR}
We call a combination $\{\mathcal{H}, \{Q, P\}\}$ consisting of a Hilbert space $\mathcal{H}$ and a pair of self-adjoint operators $\{Q, P\}$, a (one-dimensional) \emph{Weyl representation of the CCR}, if $\{Q, P\}$ satisfies the \emph{Weyl relations}:
\begin{equation}\label{eq:weyl01}
e^{isQ}e^{itP}=e^{-istI}e^{itP}e^{isQ}, 
\end{equation}
\begin{equation}\label{eq:weyl02}
e^{isQ}e^{itQ}=e^{itQ}e^{isQ}, \quad e^{isP}e^{itP}=e^{itP}e^{isP}, 
\end{equation}
for $s, t \in \mathbb{R}$.
One of the advantages of the Weyl relations, as compared to the CCR, is that they deal only with unitary operators, for which no particular consideration for the domain of the involved operators is necessary because of their boundedness.   Fortunately, in the present case one can actually prove that the pair $\{\hat{x}, \hat{p}\}$ of the position and momentum operators introduced earlier satisfy the Weyl relations \eqref{eq:weyl01} and \eqref{eq:weyl02} on $L^{2}(\mathbb{R})$.  This implies that 
the Schr{\"o}dinger representation of the CCR $\{L^{2}(\mathbb{R}), \mathscr{S}(\mathbb{R}), \hat{x}, \hat{p}\}$ furnishes an example of the Weyl representation of the CCR, at least in the case of the configuration space $\mathbb{R}$.  One also finds that this is true for the Euclidean configuration space $\mathbb{R}^n$.

One may naturally be interested in how the Weyl representation of the CCR relates to the standard representation of the CCR.  To this end, we first begin by collecting some of the necessary definitions and basic theorems.
Recall that a vector-valued map $F: U \to V$ from an open subset $U \subset \mathbb{R}$ to a normed space $V$ is called \emph{strongly continuous} at $t_{0} \in U$ if
\begin{equation}
\lim_{u \to 0} \|F(u + t_{0}) - F(t_{0})\| = 0
\end{equation}
with respect to the norm $\|\cdot \|$ on $V$, 
and in turn, strongly continuous on $U$ if it is strongly continuous at every point of $U$. 
The map $F$ is then called \emph{strongly differentiable} at $t_{0} \in U$ with strong derivative $F^{\prime}(t_{0}) \in V$ if
\begin{equation}
\lim_{u \to 0} \left\| \frac{F(u + t_{0}) - F(t_{0})}{u} - F^{\prime}(t_{0}) \right\| = 0
\end{equation}
holds, and accordingly strongly differentiable on $U$ if it is strongly differentiable at every point of $U$. 
We will occasionally write its strong derivative in either of the notations,
\begin{equation}
\frac{dF(t_{0})}{dt} = \frac{dF}{dt}(t_{0}) = \left. \frac{d}{dt} F(t) \right|_{t=t_{0}} := F^{\prime}(t_{0}).
\end{equation}

Now, let $A : \mathcal{H} \supset \mathrm{dom}(A) \to \mathcal{H}$ be a self-adjoint operator on a Hilbert space $\mathcal{H}$, and consider a one-parameter unitary group $\{e^{itA}\}_{t \in \mathbb{R}}$ (defined by means of functional calculus). 
Then, Stone's theorem on one-parameter unitary groups states that, on account of the boundedness of the unitary operator,
for a fixed $|\phi\rangle \in \mathcal{H}$
the unitary group yields a strongly continuous vector-valued map,
\begin{equation}
F : t \mapsto e^{itA}|\phi\rangle, \quad t \in \mathbb{R},
\end{equation}
for any self-adjoint operator $A$.
However, consideration of the domain $\mathrm{dom}(A)$ becomes necessary when 
differentiation of the map is considered.  In fact, the map is strongly differentiable on $\mathbb{R}$ if and only if $|\phi\rangle \in \mathrm{dom}(A)$, in which case the derivative reads
\begin{equation}
\frac{dF}{dt}(t) = ie^{itA}A|\phi\rangle = iAe^{itA}|\phi\rangle.
\end{equation}

Returning to our main topic, we rewrite the r.~h.~s. of \eqref{eq:weyl01} to obtain
\begin{equation}
e^{isQ}e^{itP} = e^{itP}e^{is(Q - tI)}, \quad s, t \in \mathbb{R}.
\end{equation}
Considering the vector-valued map,
\begin{equation}\label{eq:weyl_func}
s \mapsto e^{isQ}e^{itP}|\psi\rangle = e^{itP}e^{is(Q - tI)}|\psi\rangle,
\end{equation}
for a fixed $|\psi\rangle \in \mathrm{dom}(Q)= \mathrm{dom}(Q - tI)$, $t \in \mathbb{R}$, one concludes from the above argument that the r.~h.~s. of \eqref{eq:weyl_func} is strongly differentiable at all $s \in \mathbb{R}$ with the derivative
\begin{align}
\frac{d}{ds} \left( e^{itP}e^{is(Q - tI)}|\psi\rangle \right)
    &= e^{itP} \left( \frac{d}{ds}e^{is(Q - tI)}|\psi\rangle \right) \nonumber \\
    &= ie^{itP}e^{is(Q - tI)}(Q - tI)|\psi\rangle, \quad s, t \in \mathbb{R}.
\end{align}
Note here that the first equality follows from the linearity and boundedness (hence, continuity) of the unitary operator $e^{itP}$.
Turning to the l.~h.~s. of \eqref{eq:weyl_func}, differentiability implies that $e^{itP}|\psi\rangle \in \mathrm{dom}(Q)$, whereby 
one has
\begin{equation}
\frac{d}{ds} \left( e^{isQ}e^{itP}|\psi\rangle \right) = ie^{isQ}Qe^{itP}|\psi\rangle.
\end{equation}
Combining the two results, one duly obtains 
\begin{equation}\label{eq:weyl_diff}
e^{isQ}Qe^{itP}|\psi\rangle = e^{itP}e^{is(Q-tI)} (Q-tI) |\psi\rangle, \quad s, t \in \mathbb{R}.
\end{equation}
Taking $s=0$, one finds the validity of the operator identity $Q e^{itP} = e^{itP}(Q - tI)$, or equivalently
\begin{equation}\label{eq:weak_weyl}
e^{-itP}Q e^{itP} = Q - tI, \quad t \in \mathbb{R},
\end{equation}
on the subspace $\mathrm{dom}(Q)$. This shows how the unitary adjoint action generated by $P$ results in a parallel translation $Q \mapsto Q - tI$ on its conjugate operator $Q$%
\footnote{Note that what we are discussing here is something more than just proving the Campbell-Baker-Hausdorff formula.}.
Now, if one further considers the vector-valued map by rewriting \eqref{eq:weak_weyl}, 
\begin{equation}\label{eq:diff_2_start}
Q e^{itP}|\psi\rangle = e^{itP}Q|\psi\rangle - te^{itP}|\psi\rangle,\quad t \in \mathbb{R},
\end{equation}
one proves the differentiability of the r.~h.~s. for the choice of the initial state $|\psi\rangle \in \mathrm{dom}(PQ) \cap \mathrm{dom}(P)$, which yields
\begin{equation}
\frac{d}{dt} \left( e^{itP}Q|\psi\rangle - te^{itP}|\psi\rangle \right) = (ie^{itP}PQ - e^{itP} - ite^{itP}P ) |\psi\rangle, \quad t \in \mathbb{R}.
\end{equation}
Turning to the l.~h.~s., differentiability also leads to
\begin{align}
\frac{d}{dt} \left( Q e^{itP}|\psi\rangle \right)
    &= Q \left( \frac{d}{dt} e^{itP}|\psi\rangle \right) \nonumber \\
    &= iQ e^{itP}P|\psi\rangle, \quad t \in \mathbb{R},
\end{align}
where, in particular, $e^{itP}P|\psi\rangle \in \mathrm{dom}(Q)$ is implied, and the first equality is due to the closedness of the operator $Q$ (recall that a self-adjoint operator is necessarily closed). By combining the above two results, one has
\begin{equation}
\left(iQ e^{itP}P\right)|\psi\rangle = \left(ie^{itP}PQ - e^{itP} - ite^{itP}P\right)|\psi\rangle, \quad t \in \mathbb{R}.
\end{equation}
Taking $t=0$, we learn that this in particular leads to the operator identity,
\begin{equation}
QP = PQ + iI \quad \Leftrightarrow \quad [Q, P] = iI,
\end{equation}
on the subspace $\mathrm{dom}(PQ) \cap \mathrm{dom}(P)$. One also sees from this result that the choice $|\psi\rangle \in \mathrm{dom}(PQ) \cap \mathrm{dom}(P)$ automatically implies $|\psi\rangle \in \mathrm{dom}(QP)$.

Proceeding further from \eqref{eq:weyl02} by analogous reasoning, one eventually obtains the CCR
(\ref{eq:ccr}) and (\ref{eq:ccr2}) on the domain (\ref{eq:weyl_subspace}).
In the case where $\mathcal{D}$ is dense, one sees that a Weyl representation of the CCR $\{\mathcal{H}, \{Q,P\}\}$ together with the subspace $\mathcal{D}$ indeed gives a representation of the CCR. In fact, in the case where $\mathcal{H}$ is separable, $\mathcal{D}$ is known to be dense.

In passing, we mention that the importance of the Weyl relations becomes evident when one considers configuration spaces, 
other than the Euclidean space $\mathbb{R}^n$, where no reasonable counterpart of the CCR can be defined.  
For instance, when the configuration space is given by a coset space $G/H$ where $G$ is a Lie group and $H$ its subgroup (typical examples being the spheres $S^n \simeq O(n+1)/O(n)$), one can readily adopt the inherent group theoretic structure of the configuration space to define the Weyl relations extended to the space.  
Unlike the Euclidean case, such extended Weyl relations are known to admit a multiple of inequivalent representations.

\subsubsection{Tensor Product of Hilbert Spaces and Self-adjoint Operators}

We finally provide a brief review on tensor products of Hilbert spaces and those of self-adjoint operators. Although the topic is elementary, we find it beneficial to give a summary of its precise definition in consideration of its extensive use due to the nature of this paper focusing on indirect measurement schemes.

\paragraph{Algebraic Tensor Products}
Let $V, W$ be $\mathbb{K}$-vector spaces. We call an ordered pair
\begin{equation}
(V \otimes W,\, \otimes)
\end{equation}
consisting of a vector space $V \otimes W$ and a bilinear map $\otimes: V \times W \to V \otimes W$, an (algebraic) \emph{tensor product} of vector spaces $V$ and $W$, if for any $\mathbb{K}$-vector space $Z$ and a bilinear map $T: V \times W \to Z$, there exists a unique linear map $\widetilde{T}: V \otimes W \to Z$ for which the diagram
\begin{equation}
\xymatrix{
V \times W
	\ar[rd]_{T}
	\ar[r]^{\otimes}
& V \otimes W
	\ar[d]^{\widetilde{T}}\\
& Z
}
\end{equation}
commutes%
\footnote{
We say that a diagram is a \emph{commutative diagram}, or more casually, \emph{the diagram commutes}, if all directed paths in the diagram with the same start and endpoints lead to the same result by composition.
} (universal property of (algebraic) tensor products). 
Each element of $V \otimes W$ is called a \emph{tensor}, and the bilinear map $\otimes$ is called the \emph{tensor map}, the image of which 
is denoted by
\begin{equation}
v \otimes w := \otimes(v,w).
\end{equation}
The thus defined tensor products are in fact unique up to isomorphism. Indeed if $(V \otimes W,\, \otimes)$ and $(V^{\prime} \otimes^{\prime} W^{\prime},\, \otimes^{\prime})$ were two of such, then by first letting $Z = V^{\prime} \otimes W^{\prime}$ and $T = \otimes^{\prime}$ in the above diagram, and then subsequently by changing roles of $(V \otimes W,\, \otimes)$ and $(V^{\prime} \otimes^{\prime} W^{\prime},\, \otimes^{\prime})$, one concludes that $\widetilde{\otimes}$ and $\widetilde{\otimes^{\prime}}$ are linear bijections with $\widetilde{\otimes^{\prime}} \circ \widetilde{\otimes} = I$. In this sense, we may refer to $(V \otimes W,\, \otimes)$ as \emph{the} tensor product of $V$ and $W$, and forget about the way how it is constructed%
\footnote{One finds several concrete constructions of tensor products in various literatures. See, for example \cite{Roman_2008}.}. 
One of the basic facts worth of special note is that, given two bases $\{e_{i}\}_{i \in I}$ and $\{f_{j}\}_{j \in J}$ of $V$ and $W$, respectively, the tensors $\{e_{i} \otimes f_{j}\}_{i \in I, j \in J}$ form a basis of $V \otimes W$.

\paragraph{Tensor Product of Hilbert Spaces}
We are specifically interested in tensor products of Hilbert spaces. For a pair of Hilbert spaces $(\mathcal{H}_{1}, \langle \,\cdot\, , \,\cdot\, \rangle_{\mathcal{H}_{1}})$ and $(\mathcal{H}_{2}, \langle \,\cdot\, , \,\cdot\, \rangle_{\mathcal{H}_{2}})$, we denote by
\begin{equation}
(\mathcal{H}_{1} \,\widehat{\otimes}\, \mathcal{H}_{2}, \,\widehat{\otimes}\,)
\end{equation}
their algebraic tensor product defined from their purely algebraic structures described as above. We then introduce
\begin{equation}
\langle \phi_{1} \,\widehat{\otimes}\, \phi_{2},\, \psi_{1} \,\widehat{\otimes}\, \psi_{2} \rangle := \langle \phi_{1}, \psi_{1} \rangle_{\mathcal{H}_{1}} \langle \phi_{2}, \psi_{2} \rangle_{\mathcal{H}_{2}},\quad \phi_{i}, \psi_{i} \in \mathcal{H}_{i}
\end{equation}
defined for pairs of all tensors of the form $D := \{v \otimes w : v \in V, w \in W\}$,
and let it extend linearly on whole $\mathcal{H}_{1} \,\widehat{\otimes}\, \mathcal{H}_{2} = \mathrm{Span} (D)$. Here,
\begin{equation}\label{def:lin_span}
\mathrm{span} (S) := \{k_{1}v_{1} + \cdots + k_{n}v_{n} : k_{i} \in \mathbb{K}, v_{i} \in S, n \in \mathbb{N}\}
\end{equation}
denotes the subspace of a $\mathbb{K}$-vector space $V$ spanned by a nonempty set $S \subset V$, {\it i.e.} the set of all finite linear combinations of vectors belonging to $S$. It is routine to check that the thus defined extension $\langle \,\cdot\,, \,\cdot\, \rangle_{\mathcal{H}_{1} \,\widehat{\otimes}\, \mathcal{H}_{2}}$ is well-defined, and one moreover proves that the extension in fact makes itself an inner product on $\mathcal{H}_{1} \,\widehat{\otimes}\, \mathcal{H}_{2}$, making the pair $(\mathcal{H}_{1} \,\widehat{\otimes}\, \mathcal{H}_{2},\, \langle \,\cdot\,, \,\cdot\, \rangle_{\mathcal{H}_{1} \,\widehat{\otimes}\, \mathcal{H}_{2}})$ into a pre-Hilbert space ({\it i.e.}, an inner product space). The tensor map $\,\widehat{\otimes}\,$ can be also shown to be continuous with respect to the topology that the inner product generates. We then finally define the completion of the pre-Hilbert space, and denote it by
\begin{equation}
(\mathcal{H}_{1} \otimes \mathcal{H}_{2},\, \langle \,\cdot\,, \,\cdot\, \rangle_{\mathcal{H}_{1} \otimes \mathcal{H}_{2}}).
\end{equation}
The new space $(\mathcal{H}_{1} \otimes \mathcal{H}_{2},\, \langle \,\cdot\,, \,\cdot\, \rangle_{\mathcal{H}_{1} \otimes \mathcal{H}_{2}})$ is a Hilbert space by construction, and together with the continuous extension $\otimes$ of the bilinear map $\,\widehat{\otimes}\,$, is called the (topological) \emph{tensor product of the Hilbert spaces} $\mathcal{H}_{1}$ and $\mathcal{H}_{2}$.
The map $\otimes$ is called the \emph{tensor map} and the elements of $\mathcal{H}_{1} \otimes \mathcal{H}_{2}$ are called \emph{tensors}.

\paragraph{Tensor Product of Linear Operators}
A pair of linear operators $A_{i}: \mathcal{H}_{i} \supset \mathrm{dom}(A_{i}) \to \mathcal{H}_{i}$, $i= 1, 2$, defines
a natural bilinear map
\begin{equation}
A_{1} \times A_{2}: \mathrm{dom}(A_{1}) \times \mathrm{dom}(A_{1}) \to \mathcal{H}_{1} \times \mathcal{H}_{2}, \quad (|\phi_{1}\rangle, |\phi_{2}\rangle) \mapsto (A_{1}|\phi_{1}\rangle,\, A_{2}|\phi_{2}\rangle).
\end{equation}
From the universal property of the algebraic tensor product mentioned above, one readily sees the existence of a unique linear map
\begin{equation}\label{eq:pre_pre_tensor_product_of_operators}
A_{1} \,\widehat{\otimes}\, A_{2}: \mathrm{dom}(A_{1}) \,\widehat{\otimes}\, \mathrm{dom}(A_{1}) \to \mathcal{H}_{1} \,\widehat{\otimes}\, \mathcal{H}_{2}
\end{equation}
that makes the diagram
\begin{equation}
\xymatrix{
\mathrm{dom}(A_{1}) \times \mathrm{dom}(A_{2})
	\ar[d]_{A_{1} \times A_{2}}
	\ar[r]^{\,\widehat{\otimes}\,}
	\ar@{.>}[dr]|\circlearrowleft
& \mathrm{dom}(A_{1}) \,\widehat{\otimes}\, \mathrm{dom}(A_{2})
	\ar[d]^{A_{1} \,\widehat{\otimes}\, A_{2}}\\
\mathcal{H}_{1} \times \mathcal{H}_{2}
    \ar[r]^{\,\widehat{\otimes}\,}
& \mathcal{H}_{1} \,\widehat{\otimes}\, \mathcal{H}_{2}
}
\end{equation}
commute. Note in particular that the diagram implies
\begin{equation}
A_{1} \,\widehat{\otimes}\, A_{2}(|\phi_{1}\rangle \,\widehat{\otimes}\, |\phi_{2}\rangle) = |A_{1}\phi_{1}\rangle \,\widehat{\otimes}\, |A_{2}\phi_{2}\rangle, \quad |\phi_{i} \rangle \in \mathrm{dom}(A_{i}),\ i=1,2.
\end{equation}
Extending both the domain and the range of \eqref{eq:pre_pre_tensor_product_of_operators}, we can think of
\begin{equation}\label{eq:pre_tensor_product_of_operators}
A_{1} \,\widehat{\otimes}\, A_{2} : \mathcal{H}_{1} \otimes \mathcal{H}_{2} \supset \mathrm{dom}(A_{1}) \,\widehat{\otimes}\, \mathrm{dom}(A_{2}) \to \mathcal{H}_{1} \otimes \mathcal{H}_{2}
\end{equation}
as an operator on the Hilbert space $\mathcal{H}_{1} \otimes \mathcal{H}_{2}$.

\paragraph{Tensor Product of Self-Adjoint Operators}
Now, for a pair of densely defined closable operators $A_{i}: \mathcal{H}_{i} \supset \mathrm{dom}(A_{i}) \to \mathcal{H}_{i}$, $i= 1, 2$, the operator \eqref{eq:pre_tensor_product_of_operators} itself is known to be closable, whereby we define the \emph{tensor product}
\begin{equation}\label{def:tensor_operator}
A_{1} \otimes A_{2} := \overline{A_{1} \,\widehat{\otimes}\, A_{2}}
\end{equation}
of the pair by its closure. Specifically, since self-adjoint operators are densely defined and closed, the tensor product \eqref{def:tensor_operator} is always well-defined.
Although self-adjointness is not preserved in general by taking \eqref{eq:pre_tensor_product_of_operators}, its essential self-adjointness is at least known to be guaranteed.
As the closure of an essentially self-adjoint operator, this makes the tensor product \eqref{def:tensor_operator} itself self-adjoint, which is precisely the definition of the \emph{tensor product of self-adjoint operators}.

\subsection{Unconditioned Measurement}\label{sec:ups_I_ups}

Now that we have reviewed the necessary materials, we begin our study on the unconditioned measurement scheme.  Suppose that the experimenter wishes to extract information of the combination of a given but unknown observable $A$ and a state $|\phi\rangle \in \mathcal{H}$  of the target system, without direct access to it.  To accomplish this, 
one first arranges an auxiliary meter system $\mathcal{K}$ equipped with a pair of observables $\{Q, P\}$ for which $\{\mathcal{K}, \{Q, P\}\}$ gives a Weyl representation of the CCR. 
As we have seen above, the choice $\mathcal{K} = L^{2}(\mathbb{R})$, $Q = \hat{x}$ and $P = \hat{p}$ gives a concrete example. 
One then prepares the meter system in a certain initial state represented by the vector $|\psi\rangle \in \mathcal{K}$, and combines the two systems into the direct product state $|\phi \otimes \psi \rangle \in \mathcal{H} \otimes \mathcal{K}$. Choosing an observable $Y$ of the meter system $\mathcal{K}$ 
either by $Y = Q$ or $Y = P$, the composite system is subjected to a von Neumann type interaction, 
\begin{equation}\label{intro:von_Neumann_interaction}
|\Psi^{g}\rangle := e^{-igA \otimes Y} |\phi \otimes \psi\rangle, \quad g \in \mathbb{R},
\end{equation}
{\it i.e.}, a unitary evolution on the composite system parametrised by a real number $g$, which is often interpreted as the intensity, its time duration, or the combination thereof, of the interaction between the two systems.
Finally, the experimenter performs local measurement of an observable $X$ of the meter system $\mathcal{K}$ by choosing either by $X = Q$ or $X = P$ (chosen independently of $Y$), or equivalently $I \otimes X$ on the generally entangled composite state $|\Psi^{g}\rangle$ after the interaction (see figure~\ref{fig:ucm}).
\begin{figure}
\floatbox[{\capbeside\thisfloatsetup{capbesideposition={right,top},capbesidewidth=0.7\textwidth}}]{figure}[\FBwidth]
{\caption{A graphical illustration of the unconditioned measurement scheme.  The figure is to be read from top to bottom. The initial state preparation stage of both the target and the meter systems is depicted in the top part, and the manner in which the two quantum systems undergoes a von Neumann type interaction is illustrated in the middle part. The composite system after the interaction, which is depicted in the bottom part, generally becomes entangled. One finally performs a measurement of an observable $X$ on the meter system.}\label{fig:ucm}}
{\includegraphics[width=0.3\textwidth]{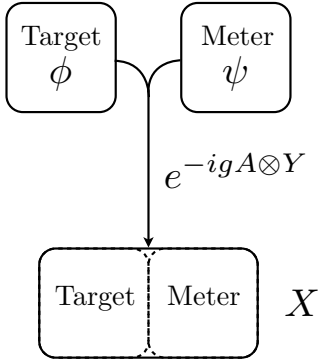}}
\end{figure}

As a preparation for further analysis, we first introduce the reduced density operator 
\begin{equation}\label{def:meterstate}
\psi^{g} := \mathrm{Tr}_\mathcal{H}[|\Psi^{g}\rangle\langle\Psi^{g}|]
\end{equation}
representing the state of the meter system $\mathcal{K}$ after the measurement%
\footnote{Here we are adopting, instead of the more common usage 
$\rho^{g}$, a slightly unusual notation $\psi^{g}$ to denote the generically mixed state of the meter.   This we do because
we wish to reserve the letter $\rho$ for the density of some absolutely continuous complex measures (see Section~\ref{sec:absolute_continuity}).  
However, our notation has an advantage on its own in that, if we also write the state as $|\psi^{g}\rangle$ when it is pure as we usually do, 
the correspondence between the two, $\psi^{g}$ and $|\psi^{g}\rangle$ (both represent the same state), becomes obvious.  
}
obtained by taking the partial trace of the composite state $|\Psi^{g}\rangle$ with respect to the target system $\mathcal{H}$. The quantity of interest for our measurement is thus the expectation value
\begin{align}\label{eq:exp_x_outcome}
\mathbb{E}[I \otimes X;\Psi^{g}]
    &:= \frac{\langle \Psi^{g}, (I \otimes X) \Psi^{g} \rangle}{\|\Psi^{g}\|^{2}} \nonumber \\
    &= \frac{\mathrm{Tr}_{\mathcal{K}}\left[ X \mathrm{Tr}_{\mathcal{H}}\left[ |\Psi^{g}\rangle\langle\Psi^{g}| \right] \right]}{\mathrm{Tr}_{\mathcal{K}}\left[ \mathrm{Tr}_{\mathcal{H}} \left[ |\Psi^{g}\rangle\langle\Psi^{g}| \right] \right]} \nonumber \\
    &= \frac{\mathrm{Tr}_{\mathcal{K}}\left[ X \psi^{g} \right] }{\mathrm{Tr}_{\mathcal{K}}\left[ \psi^{g} \right]} \nonumber \\
    &=: \mathbb{E}[X;\psi^{g}],
\end{align}
of the observable $I \otimes X$ on the composite state $|\Psi^{g}\rangle$ after the interaction, which can interchangeably be written in terms of that of the local observable $X$ on the density matrix $|\psi^{g}\rangle$ of the meter system.

\paragraph{Main Objective of this Subsection}
The main objective of this subsection is to demonstrate the following basic proposition, which provides the sufficient condition for its well-definedness and its explicit evaluations. For definiteness, we shall from now on fix $Y=P$ without loss of generality.
\begin{proposition}[Unconditioned Measurement I]\label{prop:UCM_I}
In the context of the UM scheme, let $Y=P$ for definiteness. Given the right choices
\begin{enumerate}
\item If $X = Q$: $|\phi\rangle \in \mathrm{dom}(A)$, $|\psi\rangle \in \mathrm{dom}(X)$,
\item If $X = P$: $|\phi\rangle \in \mathcal{H}$, $|\psi\rangle \in \mathrm{dom}(X)$,
\end{enumerate}
of the initial states of both the target and the meter systems,
depending on the choice of the observable $X$ on the meter system to be measured, the composite state after the interaction lies in $|\Psi^{g}\rangle \in \mathrm{dom}(I \otimes X)$, $g \in \mathbb{R}$. The expectation value \eqref{eq:exp_x_outcome} thus remains finite  for all range of the interaction parameter, which reads
\begin{equation}\label{prop:UCM_I_formula}
\mathbb{E}[X;\psi^{g}] = 
\begin{cases}
\mathbb{E}[Q; \psi] + g\, \mathbb{E}[A; \phi], &\quad (X = Q) \\
\mathbb{E}[P; \psi], &\quad (X = P)
\end{cases}
\end{equation}
for each of the choice of $X$.
\end{proposition}

\paragraph{Some Operator Identities}

Before we move on to the proof, we make notes on some important operator identities that will be extensively used throughout this paper.
Our analysis is based on the following operator identities on the composite Hilbert space $\mathcal{H} \otimes \mathcal{K}$, similar to those of \eqref{eq:weyl01} and \eqref{eq:weyl02}.
\begin{lemma}
Let $H$ and $K$ be Hilbert spaces, and let $A$ be a self-adjoint operator on $\mathcal{H}$, and $\{Q, P\}$ be a pair of self-adjoint operators on $\mathcal{K}$ for which $\{\mathcal{K}, \{Q, P\}\}$ defines a Weyl representation of the CCR. Then, the operator equalities
\begin{gather}
e^{isI \otimes Q}e^{itA \otimes P} = e^{-istA \otimes I}e^{itA \otimes P}e^{isI \otimes Q}, \quad s, t \in \mathbb{R}, \label{eq:weyl_analogue01} \\
e^{isI \otimes P}e^{itA \otimes P} = e^{itA \otimes P}e^{isI \otimes P}, \quad s, t \in \mathbb{R},\label{eq:weyl_analogue02}
\end{gather}
hold.
\end{lemma}
\begin{proof}
Since \eqref{eq:weyl_analogue02} is trivial, we only need to prove \eqref{eq:weyl_analogue01}. To this end, we first consider the special case where the self-adjoint operator $A$ on the target system $\mathcal{H}$ has a spectrum $\sigma(A)$ of finite cardinality. Letting $\sigma(A) = \{a_{1}, \dots, a_{N}\}$, $N \in \mathbb{N}^{\times}$ be any enumeration of its eigenvalues,
the spectral decomposition of $A$ reads
\begin{equation}\label{eq:spect_decomp_fin}
A = \sum_{n = 1}^{N} a_{n} \Pi_{a_{n}},
\end{equation}
where $\Pi_{a_{n}}$ is the projection on the eigenspace associated with the eigenvalue $a_{n}$.  In the case where the eigenspace is one-dimensional (or non-degenerate), one may write $\Pi_{a_{n}} = | a_{n} \rangle\langle a_{n}  |$ with the eigenstate $| a_{n} \rangle$ for which $A |\phi_{n}\rangle = a_{n} |a_{n}\rangle$ holds (more on the topic of spectral decomposition in Section~\ref{sec:spectral_theorem}). 
Now, for an arbitrary self-adjoint operator $Z$ on the meter system $\mathcal{K}$, one may expect from the defining property $\Pi_{a_{n}}^2 = \Pi_{a_{n}}$ of projections that the formal computation
\begin{align}\label{eq:int_op_fin}
e^{itA \otimes Z}
    &= \sum_{k = 0}^{\infty} \frac{(it)^{k}}{k!} \left( A \otimes Z \right)^{k} \nonumber \\
    &= \sum_{k = 0}^{\infty} \frac{(it)^{k}}{k!} \left( \left( \sum_{n = 1}^{N} a_{n}^{k} \Pi_{a_{n}} \right) \otimes Z^{k} \right) \nonumber \\
	&= \sum_{k = 0}^{\infty} \sum_{n = 1}^{N} \left( \Pi_{a_{n}} \otimes \frac{(ita_{n}Z)^{k}}{k!} \right) \nonumber \\
	&= \sum_{n = 1}^{N} \left( \Pi_{a_{n}} \otimes e^{ita_{n}Z} \right), \quad t \in \mathbb{R},
\end{align}
is legitimate. This in fact turns out to be correct as an operator identity on $\mathcal{H} \otimes \mathcal{K}$ with full rigour, which can be proven in a fairly straightforward manner by means of rudimentary techniques of functional calculus.
One then has
\begin{align}
e^{isI \otimes Q}e^{itA \otimes P}
    &= \left( I \otimes e^{isQ} \right) \left( \sum_{n = 1}^{N} \left( \Pi_{a_{n}} \otimes e^{ita_{n}P} \right) \right) \nonumber \\
    &= \sum_{n = 1}^{N} \left( \Pi_{a_{n}} \otimes \left( e^{isQ}e^{ita_{n}P} \right) \right) \nonumber \\
    &= \sum_{n = 1}^{N} \left( \Pi_{a_{n}} \otimes \left( e^{ista_{n}}e^{ita_{n}P}e^{isQ} \right) \right) \nonumber \\
    &= \sum_{n = 1}^{N} \left( \Pi_{a_{n}} \otimes e^{ita_{n}(P -sI)} \right) \left( I \otimes e^{isQ} \right) \nonumber \\
    &= e^{itA \otimes (P-sI)}e^{is I \otimes Q} \nonumber \\
    &= e^{-istA \otimes I}e^{itA \otimes P}e^{isI \otimes Q}, \quad s, t \in \mathbb{R},
\end{align}
which proves (\ref{eq:weyl_analogue01}) for our special case, where we have used \eqref{eq:weyl01} in the third step. Returning to the general case in which $A$ is now an arbitrary self-adjoint operator, one observes that the well-definedness of both the left-most and right-most hand sides of the above equality remains valid.  From this, one may expect that the same result also holds for the general case, which indeed turns out to be true (as usual, one may prove this without much difficulty through rudimentary techniques of functional calculus).
\end{proof}

\paragraph{Measurement Outcomes}
We now return to the main problem of this subsection. We are interested in finding the condition for which \eqref{eq:exp_x_outcome} is well-defined, and subsequently in obtaining an explicit formula in terms of the components of both the target and the meter system. Since most of the techniques employed here is the same as those introduced in Section~\ref{sec:ups_I_pre}, we shall proceed by sketching the proofs.

\begin{proof}[Proof of Proposition~\ref{prop:UCM_I}]
Let us begin by choosing the operator $X=Q$ for the measurement of the meter system, and thereby rewrite the r.~h.~s. of \eqref{eq:weyl_analogue01} to obtain
\begin{equation}
e^{isI \otimes Q}e^{itA \otimes P} = e^{itA \otimes P}e^{is
(\overline{I \otimes Q - tA \otimes I})},\quad s, t \in \mathbb{R}
\end{equation}
for better usability%
\footnote{Note here that the sum of two (possibly unbounded) self-adjoint operators is not necessarily self-adjoint. Fortunately, essential self-adjointness is at least assured for the sum of $I \otimes Q$ and $A \otimes I$ for our case. We may thus take the self-adjoint extension of their sum in order to ensure its self-adjointness (more to this in Section~\ref{sec:sim_meas_obs}).}.
By differentiating both sides of the above equality and taking $s=0$, an analogous argument given earlier for obtaining \eqref{eq:weak_weyl} leads to the operator identity,
\begin{equation}\label{eq:weak_weyl_analogue01}
e^{-itA \otimes P} (I \otimes Q) e^{itA \otimes P} =  \overline{I \otimes Q - tA \otimes I}, \quad t \in \mathbb{R},
\end{equation}
on 
the subspace $\mathrm{dom}(\overline{I \otimes Q - tA \otimes I})$.  This ensures that, if $|\Phi\rangle \in \mathrm{dom}(\overline{I \otimes Q - tA \otimes I})$, then one has 
\begin{equation}\label{eq:eqeta}
e^{itA \otimes P}|\Phi\rangle \in \mathrm{dom}(I \otimes Q).
\end{equation}
Here, it may be worthwhile to note the analogy between \eqref{eq:weak_weyl} and \eqref{eq:weak_weyl_analogue01}.
To put this in our context, let $t = -g$ above. If one chooses $|\psi\rangle \in \mathrm{dom}(Q)$ as the meter state, and likewise assumes $|\phi\rangle \in \mathrm{dom}(A)$ as the system state prepared prior to the interaction, one has in particular $|\phi \otimes \psi \rangle \in \mathrm{dom}(\overline{I \otimes Q + gA \otimes I})$.  Then, equating $|\Phi\rangle = |\phi \otimes \psi \rangle$ in \eqref{eq:eqeta}, we find
\begin{equation}\label{eq:int_dom_Q}
|\Psi^{g}\rangle = e^{-igA \otimes P} |\phi \otimes \psi\rangle \in \mathrm{dom}(I \otimes Q), \quad g \in \mathbb{R}.
\end{equation}
This guarantees that the expectation value \eqref{eq:exp_x_outcome} of the observable $I \otimes Q$ on the composite state $|\Psi^{g}\rangle$ remains finite and is given by
\begin{align}\label{eq:ups_exp_Q}
\mathbb{E}[I \otimes Q; \Psi^{g}]
    &:= \frac{\langle \Psi^{g}, (I\otimes Q) \Psi^{g}\rangle}{\|\Psi^{g}\|^{2}} \nonumber \\
    &= \frac{\langle \phi \otimes \psi, (e^{igA \otimes P} (I \otimes Q) e^{-igA \otimes P}) \phi \otimes \psi\rangle}{\|\phi\|^{2}\|\psi\|^{2}} \nonumber \\
    &= \frac{\langle \phi \otimes \psi, (\overline{I \otimes Q + gA \otimes I}) \phi \otimes \psi\rangle}{\|\phi\|^{2}\|\psi\|^{2}} \nonumber \\
    &= \mathbb{E}[Q; \psi] + g  \mathbb{E}[A; \phi], \quad g \in \mathbb{R},
\end{align}
for any such combination of the initial states.

Evidently, for the choice $X=P$,  one finds the validity of the operator identity
\begin{equation}\label{eq:weak_weyl_analogue02}
e^{itA \otimes P} (I \otimes P) e^{-itA \otimes P} = I \otimes P, \quad t \in \mathbb{R},
\end{equation}
on the subspace $\mathrm{dom}(I \otimes P)$ by analogous reasoning.
From this, one readily concludes that the expectation value of $I \otimes P$ reads
\begin{equation}\label{eq:ups_exp_P}
\mathbb{E}[I \otimes P; \Psi^{g}] = \mathbb{E}[P; \psi], \quad g \in \mathbb{R},
\end{equation}
which is well-defined for any choice of the state $|\psi\rangle \in \mathrm{dom}(P)$ of the meter system and $g \in \mathbb{R}$, irrespective of the initial choice of the state $|\phi\rangle \in \mathcal{H}$ of the target system.
\end{proof}

\subsection{Recovery of the Target Profile}

Now that we have revealed the explicit behaviour of the measurement outcomes of the meter, we are thus interested in recovering the information of the target system from it. As one may expect from the statement in Proposition~\ref{prop:UCM_I}, the information of the target system (which should essentially consist of the specification of the pair of $A$ and $|\phi\rangle$) manifests itself in the form of the expectation value $\mathbb{E}[A;\phi]$. In recovering the desired information, one subsequently recognises from \eqref{prop:UCM_I_formula} that it fully suffices to examine only the outcomes of the measurement of the observable $X$ conjugate to $Y$, and there is no use for that of the choice $X=Y$ (this is to be contrasted with the conditional measurement we discuss later). Specifically, one finds below that there are two typical techniques in obtaining the desired information: one is to investigate the behaviour of the measurement outcome \eqref{prop:UCM_I_formula} in the strong region $g \to \pm \infty$ of the interaction parameter, and the other is to examine the local behaviour of it around $g = 0$, which shall be respectively called the \emph{strong unconditioned measurement} and the \emph{weak unconditioned measurement} in this paper.

\subsubsection{Strong Unconditioned Measurement}\label{sec:ups_I_sups}

Our result \eqref{prop:UCM_I_formula} shows that the expectation value of the measurement of $X = Q$ behaves linearly with respect to $g$, and that its growth is proportional to the expectation value $\mathbb{E}[A; \phi]$ of the target observable. The experimenter would thus divide the measurement outcomes of $Q$ by $g$ and then take the limit of the strong coupling $g \to \pm \infty$ (or equivalently ${g^{-1} \to 0}$):
\begin{align}\label{eq:ups_op_recovery}
\lim_{g^{-1} \to 0} \frac{\mathbb{E}[Q; \psi^{g}]}{g}
    &= \mathbb{E}[A; \phi] + \lim_{g^{-1} \to 0} \frac{\mathbb{E}[Q; \psi]}{g} \nonumber \\
    &= \mathbb{E}[A; \phi].
\end{align}
allowing the recovery of the desired information of the target system $\mathbb{E}[A; \phi]$ in the form of expectation values%
\footnote{Alternatively, one may consider the shift of the expectation value,
\begin{align}
\Delta_{X}(g)
    := \mathbb{E}[X;\psi^{g}] - \mathbb{E}[X;\psi^{0}] 
    =
\begin{cases}
     g \mathbb{E}[A; \phi],   & (X=Q), \\
     0,    & (X=P),
\end{cases}
\end{align}
for $g \in \mathbb{R}$ as a quantity directly related to the observable $A$ of the system. 
For the choice $X=Q$, one then simply has
\begin{equation}
\frac{\Delta_{Q}(g)}{g} = \mathbb{E}[A; \phi], \quad g \in \mathbb{R}^{\times},
\end{equation}
which might be a more straight-forward way to be employed practically.}.

\subsubsection{Weak Unconditioned Measurement}\label{sec:ups_I_wups}

The same information may be obtained by examining the weak region ($g \to 0$) of the interaction. Indeed, one trivially finds from \eqref{prop:UCM_I_formula} that
\begin{align}\label{eq:ups_weak_I}
\left. \frac{d^{n}}{dg^{n}}\mathbb{E}[Q; \psi^{g}] \right|_{g=0} =
 \begin{cases}
 \mathbb{E}[Q; \psi], & n = 0, \\
 \mathbb{E}[A; \phi], & n = 1, \\
 0, & n \geq 2,
 \end{cases}
\end{align}
which implies that the expectation value $\mathbb{E}[A; \phi]$ of our interest may also be obtained as the first differential coefficient ($n=1$) of the measured outcome at $g=0$.

\subsubsection{Discussion}\label{sec:ups_I_discussion}

While this whole section consisted of rather trivial results, the line of arguments presented here serves as the baseline of our analysis throughout this paper. Namely, we first examine the full behaviour of the target of our measurement (for this section, is was the expectation value \eqref{eq:exp_x_outcome} of the observable $X$ of the meter) and intend to obtain an explicit description of how the profile of the initial configuration of the the target system gets mixed into that of the meter system through the interaction (which, for the current case, is the result \eqref{prop:UCM_I_formula}). We then intend to extract the information of the target system (for this section, it is the expectation value $\mathbb{E}[A;\phi]$) by separating it from the measurement outcomes. Specifically, we find that examining either the strong or the weak region of the interaction parameter $g$ reveals itself useful for this purpose, and this should be the strategy that we take in the subsequent sections.

In the next section, the UM scheme is analysed in depth in terms of probabilities, following the same line as described above. Specifically, while the distinction between the strong and the weak measurements looked rather vague at the operator level, we shall see shortly that these two strategies are recognised to be qualitatively different from the viewpoint of probabilities.

\newpage
\section{Unconditioned Measurement II: In Terms of Probabilities}\label{sec:ups_II}

We have so far conducted an analysis of the UM scheme on the operator level, where the quantity of interest is the \emph{expectation value} of an observable. 
However, one may be interested in the raw information that the measurement provides, {\it i.e.}, the \emph{probability} describing the behaviour of each measurement outcomes, which is the target of our study in this section.

\subsection{Reference Materials}\label{sec:ups_II_pre}

To prepare for our discussion, we here provide a concise summary on the topic of complex measures and integration with respect to them. We next make a brief review on the spectral theorem for self-adjoint operators and recall the general framework for describing the ideal measurement of a quantum observable.  Subsequently, we expound on density functions and see how this relates to the description by measures.

\subsubsection{The Space of Complex Measures}

As a preparation in dealing with the spectral theorem for self-adjoint operators, we collect below the basic definitions and results regarding complex measures and integration with respect to them.

\paragraph{Signed Measures, Jordan Decomposition and Total Variation}

Let $(X, \mathfrak{A})$ be a measurable space. A map $\nu : \mathfrak{A} \to \overline{\mathbb{R}}$ is called a \emph{signed measure}, if it satisfies the following properties:
\begin{enumerate}
\item $\nu(\emptyset) = 0$.
\item $\nu(\mathfrak{A}) \subset ]-\infty, + \infty]$ or $\nu(\mathfrak{A}) \subset [-\infty, + \infty[$.
\item Countable additivity \eqref{def:count_add} holds for any sequence $(A_{n})_{n \geq 1}$ of pairwise disjoint subsets of $X$.
\end{enumerate}
They are, in a sense, generalisations of the concept of the standard measures by allowing negative numbers to be assigned to each measurable sets. A signed measure $\nu$ is called \emph{finite} if $\nu(\mathfrak{A}) \subset \mathbb{R}$. One of the most important properties of a signed measure is described by the \emph{Jordan decomposition theorem}, which states that every singed measure $\nu$ has the \emph{Jordan decomposition}, {\it i.e.}, a unique decomposition of $\nu$ into a difference 
\begin{equation}
\nu = \nu^{+} - \nu^{-}
\end{equation}
of two measures $\nu^{+}$ and $\nu^{-}$, respectively called the \emph{positive} and \emph{negative variation} of $\nu$, and at least one of which being finite. Here, the positive and negative variations are \emph{singular} to one another, denoted as $\nu^{+}\perp\nu^{-}$, in the sense there exists a decomposition of $X = P \cup N$ into two measurable sets such that $\nu^{+}(N) = 0$ and $\nu^{-}(P) = 0$ holds. The Jordan decomposition is minimal in the following sense: Given any decomposition $\nu = \rho - \sigma$ of $\nu$ into two measures $\rho$, $\sigma$, at least one of which being finite, then $\nu^{+} \leq \rho$, $\nu^{-} \leq \sigma$ holds.

Let $\mathbf{M}(\mathfrak{A})$ denote the collection of all finite signed measures. One readily sees that $\mathbf{M}(\mathfrak{A})$ becomes an $\mathbb{R}$-linear space, equipped with the natural addition $(\mu + \nu)(A) := \mu(A) + \nu(A)$ and scalar multiplication $(c\mu)(A) := c\mu(A)$ for $\mu, \nu \in \mathbf{M}(\mathfrak{A})$ and $c \in \mathbb{R}$.
Now, let $\nu = \nu^{+} - \nu^{-}$ be the Jordan decomposition of $\nu \in \mathbf{M}(\mathfrak{A})$, and define a new measure by their sum
\begin{equation}
|\nu| := \nu^{+} + \nu^{-},
\end{equation}
called the \emph{variation} of $\nu$. We then define its \emph{total variation} by $\|\nu\| := |\nu|(X)$, which is nothing but the evaluation of the whole space $X$ by the non-negative measure $|\nu|$.
One proves that the total variation defines a norm on $\mathbf{M}(\mathfrak{A})$, and in fact makes $( \mathbf{M}(\mathfrak{A}), \| \cdot \|)$ into a real Banach space.

\paragraph{Complex Measures}

Let $(X, \mathfrak{A})$ be a measurable space. A map $\nu : \mathfrak{A} \to \mathbb{C}$ is called a \emph{complex measure}, when it is countably additive \eqref{def:count_add}.
One sees that $\nu$ is a complex measure if and only if both its real and imaginary parts $\,\mathrm{Re}\,[\nu]$, $\,\mathrm{Im}\,[\nu]$ are finite signed measures. Analogous to the case of signed measures, the collection $\mathbf{M}_{\mathbb{C}}(\mathfrak{A})$ of all complex measures on $(X, \mathfrak{A})$ becomes a $\mathbb{C}$-linear space, equipped with the natural addition and scalar multiplication.
For a complex measure $\nu \in \mathbf{M}_{\mathbb{C}}(\mathfrak{A})$, we define the \emph{variation} of a measurable set $A \in \mathfrak{A}$ by
\begin{equation}
|\nu|(A) := \sup \left\{ \sum_{j=1}^{\infty} |\nu(A_{j})| : A_{j} \in \mathfrak{A} \text{ disjoint for} \,\, j \geq 1, \, A = \bigcup_{j=1}^{\infty} A_{j} \right\},
\end{equation}
and also its \emph{total variation},
\begin{equation}\label{def:totvar}
\|\nu\| := |\nu|(X).
\end{equation}
The definition coincides with the previous definition when $\nu$ happens to be a signed measure. The total variation $\|\nu\|$ of $\nu$ is known to be the smallest positive measure $\mu$ on $(X, \mathfrak{A})$ satisfying $|\nu(A)| \leq \mu(A)$, $A \in \mathfrak{A}$. In parallel to the case of signed measures, one finds that the total variation defines a norm on the linear space $\mathbf{M}_{\mathbb{C}}(\mathfrak{A})$ and makes $(\mathbf{M}_{\mathbb{C}}(\mathfrak{A}), \| \cdot \|)$ into a complex Banach space.

\paragraph{Integration over Complex Measures}
It is now tempting to define integration with respect to complex measures, as a natural extension to that defined for (standard) measures. For a complex measure $\nu \in \mathbf{M}_{\mathbb{C}}(\mathfrak{A})$, we let $\rho := \,\mathrm{Re}\,[\nu]$, $\sigma := \,\mathrm{Im}\,[\nu]$ and consider the intersection of the spaces
\begin{equation}
\mathcal{L}^{1}(\nu) := \mathcal{L}^{1}(\rho^{+}) \cap \mathcal{L}^{1}(\rho^{-}) \cap \mathcal{L}^{1}(\sigma^{+}) \cap \mathcal{L}^{1}(\sigma^{-}),
\end{equation}
where $\rho^{\pm}$ and $\sigma^{\pm}$ respectively being the positive and negative variations of $\rho$ and $\sigma$.
We then define the Lebesgue integral of $f \in \mathcal{L}^{1}(\nu)$ with respect to $\nu$ by
\begin{equation}
\int_{X} f\ d\nu := \int_{X} f\ d\rho^{+} - \int_{X} f\ d\rho^{-} + i \int_{X} f\ d\sigma^{+} - i \int_{X} f\ d\sigma^{-}.
\end{equation}
Linearity of the Lebesgue integral with respect to the complex measure follows naturally as expected.

\paragraph{New Measure from Old}

There are several ways to construct a new (complex) measure from a given measure. We mention below two of the most important manners that are frequently employed throughout this paper.

\begin{enumerate}
\renewcommand{\labelenumi}{(\Alph{enumi})}
\item {\it Measure with Density. \hspace{12pt}}
Let $(X, \mathfrak{A}, \mu)$ be a measure space. Given a $\mu$-integrable function $f : X \to \mathbb{C}$, one may define a complex measure by
\begin{equation}\label{def:Mass_mit_Dichte}
\nu(A) := \int_{A} f\ d\mu, \quad A \in \mathfrak{A}.
\end{equation}
The complex measure constructed in this manner is occasionally called the \emph{complex measure with the density $f$ with respect to $\mu$}, and we write it as $\nu = f \odot \mu$. A measurable function $g : X \to \hat{\mathbb{K}}$ is known to be $(f \odot \mu)$-integrable, if and only if the product $g \cdot f$ is $\mu$-integrable, in which case the equality
\begin{equation}\label{eq:Transformationsformel_Mass_mit_Dichte}
\int_{X} g \ d(f \odot \mu) = \int_{X} g \cdot f \ d\mu
\end{equation}
holds.
\item {\it Image Measure. \hspace{12pt}}
Let $(X, \mathfrak{A}, \mu)$ be a measure space. Given another measurable space $(Y, \mathfrak{B})$ and a measurable map $f : X \to Y$, one may construct a new measure on $(Y, \mathfrak{B})$ by
\begin{equation}\label{def:Bildmass}
f(\mu)(B) := \mu(f^{-1}(B)), \quad B \in \mathfrak{B},
\end{equation}
called the \emph{image measure (push-forward measure) of $\mu$ with respect to $f$}. A measurable function $g : Y \to \hat{\mathbb{K}}$ is known to be $f(\mu)$-integrable, if and only if the composition $g \circ f$ is $\mu$-integrable, in which case the the \emph{change of variables formula}
\begin{equation}\label{eq:Transformationsformel_Bildmass}
\int_{Y} g \ df(\mu) = \int_{X} g \circ f \ d\mu
\end{equation}
holds.
\end{enumerate}

\paragraph{Measure Algebra}

The space of complex measures has an additional well-known structure regarding convolutions.
The \emph{convolution} of the two complex measures $\mu, \nu \in \mathbf{M}_{\mathbb{C}}(\mathfrak{B}^{n})$ is defined by
\begin{align}\label{def:convolution_measure}
(\mu \ast \nu)(B) := \int_{\mathbb{R}^{n}} \mu(B - x)\ d\nu(x), \quad B \in \mathfrak{B}^{n}.
\end{align}
One can easily confirm that the convolution is a bilinear operation, and is moreover shown to be associative $\mu \ast (\nu \ast \rho) = (\mu \ast \nu) \ast \rho$ and commutative $\mu \ast \nu = \nu \ast \mu$. Together with the evaluation $\|\mu \ast \nu\| \leq \|\mu\| \|\nu\|$ based on the total variation norm \eqref{def:totvar}, one sees that the convolution makes the complex Banach space $\mathbf{M}_{\mathbb{C}}(\mathfrak{B}^{n})$ into a complex commutative Banach algebra, called the \emph{measure algebra} of $\mathfrak{B}^{n}$. The measure algebra $\mathbf{M}_{\mathbb{C}}(\mathfrak{B}^{n})$ has a multiplicative identity $e$ given by the delta measure $e = \delta_{0}$ centred at the origin, that is,
\begin{equation}
\mu \ast \delta_{0} = \delta_{0} \ast \mu = \mu
\end{equation}
holds for all $\mu \in \mathbf{M}_{\mathbb{C}}(\mathfrak{B}^{n})$. 
Here, the \emph{delta measure} (or the \emph{Dirac measure}) $\delta_{a}$ is a finite measure centred at $a \in \mathbb{R}^{n}$ defined by
\begin{equation}\label{def:delta_measure}
\delta_{a}(B) =
    \begin{cases}
        1, & a \in B, \\
        0, & a \notin B,
    \end{cases}
    \qquad B \in \mathfrak{B}^{n},
\end{equation}
characterised by the integral
\begin{equation}
\int_{\mathbb{R}^{n}} f(x)\ d\delta_{a}(x) = f(a),
\end{equation}
whenever the integration is well-defined. It is essentially the same object as the delta distribution that appears in the theory of generalised functions.

\subsubsection{The Space of Density Functions}\label{sec:absolute_continuity}

For later use, we are particularly interested in the special subspace of the space $\mathbf{M}_{\mathbb{C}}(\mathfrak{B}^{n})$ of complex measures, namely, the space of \emph{absolutely continuous} complex measures with respect to the Lebesgue-Borel measure $\beta^{n}$. We shall provide a concise review on its definition, make comments on its relation to the space of complex \emph{density functions}, and sees that the subspace reveals itself to be a sub-algebra of the measure algebra.

\paragraph{Absolute Continuity and Density Functions}

Let $\mu$ and $\nu$ be signed (or complex) measures on a measurable space $(X, \mathfrak{A})$. We say that $\nu$ is \emph{$\mu$-continuous} or \emph{absolutely continuous with respect to $\mu$}, written as $\nu \ll \mu$, if $\mu(A) = 0$ implies $\nu(A) = 0$ for all $A \in \mathfrak{A}$. A signed measure $\mu$ is called \emph{$\sigma$-finite} if there exists a sequence $(A_{n})_{n \geq 1}$ of disjoint measurable sets $A_{n} \in \mathfrak{A}$ satisfying $X = \bigcup_{n=1}^{\infty} A_{n}$ and $|\mu(A_{n})| < \infty$ ($n \in \mathbb{N}^{\times}$). By definition, finite measures are always $\sigma$-finite. The Lebesgue-Borel measure $\beta^{n}$ is among the most important examples of $\sigma$-finite measures. The following theorem is of great importance.
\begin{theorem*}[Radon-Nikod{\'y}m Theorem for Complex Measures]
Let $\mu$ be a $\sigma$-finite measure and $\nu \ll \mu$ be a complex measure. Then, $\nu$ has a density with respect to $\mu$, that is, there exists a $\mu$-integrable function $\rho : X \to \mathbb{C}$ such that $\nu = \rho \odot \mu$, and $\rho$ is unique $\mu$-a.e. If $\nu$ happens to be positive, then one may choose $\rho \geq 0$.
\end{theorem*}
\noindent
In the above situation of the Radon-Nikod{\'y}m theorem, the function $\rho$ satisfying $\nu = \rho \odot \mu$ is called the \emph{Radon-Nikod{\'y}m derivative} (or more casually, the \emph{density}), and is denoted by
\begin{equation}
\rho =: \frac{d\nu}{d\mu}.
\end{equation}
This is nothing but to say that
\begin{equation}\label{def:abs_cont}
\nu(A) = \int_{A} \frac{d\nu}{d\mu}(x)\ d\mu(x), \quad A \in \mathfrak{A},
\end{equation}
holds, if explicitly written out. For a $\nu$-integrable function $f$, a direct application of \eqref{eq:Transformationsformel_Mass_mit_Dichte} leads to
\begin{equation}\label{prop:Mass_mit_Dichte}
\int_{\mathbb{X}} f(x)\ d\nu(x) = \int_{X} f(x) \frac{d\nu}{d\mu}(x)\ d\mu(x),
\end{equation}
in which the notation for the Radon-Nikod{\'y}m derivative (which might at first seems strange) reveals its advantage.

\paragraph{Absolute Continuity with respect to the Lebesgue-Borel Measure}

We are particularly interested in the sub-family $L^{1}(\mathfrak{B}^{n}) \subset \mathbf{M}_{\mathbb{C}}(\mathfrak{B}^{n})$ consisting of complex measures that are absolutely continuous with respect to the Lebesgue-Borel measure $\beta^{n}$ on $(\mathbb{R}^{n},\mathfrak{B}^{n})$. Whenever there is no risk of confusion, members of $L^{1}(\mathfrak{B}^{n})$ shall occasionally be referred to as \emph{absolutely continuous measures}, simply without reference to the base measure $\beta^{n}$. One readily finds that the collection $L^{1}(\mathfrak{B}^{n})$ forms a linear subspace of $\mathbf{M}_{\mathbb{C}}(\mathfrak{B}^{n})$.
Now, uniqueness $\beta^{n}$-a.e. of the Radon-Nikod{\'y}m derivative allows us to define a linear map
\begin{equation}\label{def:abs_cont_ident}
L^{1}(\mathfrak{B}^{n}) \to L^{1}(\mathbb{R}^{n}),\ \mu \mapsto \frac{d\mu}{d\beta^{n}},
\end{equation}
which maps an absolutely continuous complex measure to its density.
Conversely, one may construct a new complex measure given an integrable function $f \in L^{1}(\mathbb{R}^{n})$ by $\nu := f \odot \beta^{n}$. From this, one obtains a bijective linear map between the space of absolutely continuous complex measures $L^{1}(\mathfrak{B}^{n})$ and the space of integrable functions $L^{1}(\mathbb{R}^{n})$, associating an absolutely continuous complex measure $\nu \in L^{1}(\mathfrak{B}^{n})$ to its density $d\nu/d\beta^{n} \in L^{1}(\mathbb{R}^{n})$. In this manner, one may identify a specific subspace of the space of complex measures with that of integrable functions as
\begin{equation}\label{def:abs_measure_identification}
 L^{1}(\mathfrak{B}^{n}) \cong L^{1}(\mathbb{R}^{n}),
\end{equation}
and may translate and interpret various properties of complex measures in terms of density functions.
To discuss how this works, let $d\nu/d\beta^{n} \in L^{1}(\mathbb{R}^{n})$ be the density of $\nu \in L^{1}(\mathfrak{B}^{n})$ with respect to the Lebesgue-Borel measure. One confirms from \eqref{prop:Mass_mit_Dichte} that, for any measurable function $g$, the equality
\begin{equation}\label{eq:Mass_mit_Dichte_AC}
\int_{\mathbb{R}^{n}} g(x)\ d\nu(x) = \int_{\mathbb{R}^{n}} g(x)\frac{d\nu}{d\beta^{n}}(x)\ d\beta^{n}(x)
\end{equation}
holds whenever the integration exists.  In this manner, one may replace the Lebesgue integration of $g$ with respect to the complex measure $\nu$ (the l.~h.~s.) by that with respect to the Lebesgue-Borel measure with the help of the (possibly more familiar notion of) density function $d\nu/d\beta^{n}$ (the r.~h.~s.).

\paragraph{Convolution Algebra}

The space $L^{1}(\mathfrak{B}^{n})$ of absolutely continuous complex measures is readily shown to be a topologically closed subset  (with respect to the topology induced by the total variation norm $\|\cdot \|$ in \eqref{def:totvar}) of the Banach space $\mathbf{M}_{\mathbb{C}}(\mathfrak{B}^{n})$. This implies that the subspace $L^{1}(\mathfrak{B}^{n})$ is itself a Banach space.
One then finds that the linear bijection \eqref{def:abs_cont_ident} between the two Banach spaces actually defines an \emph{isometric (linear) isomorphism}, which is to say that
\begin{equation}\label{eq:isometry}
\|\nu\| = \left\|\frac{d\nu}{d\beta^{n}}\right\|_{1}
\end{equation}
holds for all $\nu \in L^{1}(\mathfrak{B}^{n})$, where the l.~h.~s. is the total variation norm \eqref{def:totvar} of the complex measure $\nu$ and the r.~h.~s. is the $L^{1}$-norm \eqref{def:Lp_norm} of its density function. 

We next see how this bijection plays with convolution. To this end, we first recall that a linear subspace $\mathfrak{I}$ of a commutative algebra $\mathfrak{A}$ is called an \emph{ideal} if it `absorbs' multiplication by elements of $\mathfrak{A}$, {\it i.e.},
\begin{equation}
i \in \mathfrak{I},\ a \in \mathfrak{A} \quad \Rightarrow \quad i \cdot a = a \cdot i \in \mathfrak{I}.
\end{equation}
In fact, it is known that the subspace $L^{1}(\mathfrak{B}^{n})$ forms an ideal of the measure algebra $\mathbf{M}_{\mathbb{C}}(\mathfrak{B}^{n})$, which is to say that
\begin{equation}
\mu \in L^{1}(\mathfrak{B}^{n}),\ \nu \in \mathbf{M}_{\mathbb{C}}(\mathfrak{B}^{n}) \quad \Rightarrow \quad \mu \ast \nu = \nu \ast \mu \in L^{1}(\mathfrak{B}^{n}).
\end{equation}
In passing, the density of the convolution $\mu \ast \nu$ above is given by the convolution of the density of $\mu$ and the complex measure $\nu$ as
\begin{equation}\label{eq:density_ac_cmeas}
\frac{d(\mu \ast \nu)}{d\beta^{n}} = \frac{d\mu}{d\beta^{n}} \ast \nu,
\end{equation}
in which we understand the convolution of an integrable function $f \in L^{1}(\mathbb{R}^{n})$ and a complex measure $\mu \in \mathbf{M}_{\mathbb{C}}(\mathfrak{B}^{n})$ to be
\begin{equation}
(f \ast \mu)(x) := \int_{\mathbb{R}^{n}} f(x-y)\ d\mu(y),
\end{equation}
where the integral is well-defined $\beta^{n}$-a.e. for $x \in \mathbb{R}^{n}$.

In particular, being an ideal trivially implies that the space $L^{1}(\mathfrak{B}^{n})$ of absolutely continuous complex measures is closed under the operation of convolution, {\it i.e.}, it forms a \emph{sub-algebra} of the measure algebra $\mathbf{M}_{\mathbb{C}}(\mathfrak{B}^{n})$. Applying \eqref{prop:Mass_mit_Dichte} to \eqref{eq:density_ac_cmeas}, one concludes that the density of the convolution of two absolutely continuous complex measures $\mu, \nu \in L^{1}(\mathfrak{B}^{n})$ is given by the convolution of their densities as
\begin{equation}\label{eq:density_ac_ac}
\frac{d(\mu \ast \nu)}{d\beta^{n}} = \frac{d\mu}{d\beta^{n}} \ast \frac{d\nu}{d\beta^{n}},
\end{equation}
in which we understand the familiar convolution of two integrable functions $f, g \in L^{1}(\mathbb{R}^{n})$ to be
\begin{equation}\label{def:convol_func_func}
(f \ast g)(x) := \int_{\mathbb{R}^{n}} f(x-y)g(y)\ d\beta^{n}(y),
\end{equation}
where the integral is well-defined $\beta^{n}$-a.e. for $x \in \mathbb{R}^{n}$. Equality \eqref{eq:density_ac_ac} implies that, equipped with the convolution \eqref{def:convol_func_func}, the space $L^{1}(\mathbb{R}^{n})$ of integrable functions becomes a Banach algebra that is isomorphically mapped to the sub-algebra $L^{1}(\mathfrak{B}^{n}) \subset \mathbf{M}_{\mathbb{C}}(\mathfrak{B}^{n})$ by the \emph{isometric algebra isomorphism} \eqref{def:abs_cont_ident}. Incidentally, the sub-algebra $L^{1}(\mathfrak{B}^{n}) \cong L^{1}(\mathbb{R}^{n})$ of the measure algebra is given its own name, and is occasionally called the \emph{convolution algebra}.

At this point, we note that the convolution algebra $L^{1}(\mathfrak{B}^{n})$ is a proper sub-algebra of the measure algebra $\mathbf{M}_{\mathbb{C}}(\mathfrak{B}^{n})$ in general, {\it i.e.}, not every complex measure may be represented by integrable functions.  This can be readily seen by observing that the delta measure $\delta_{a}$  centred at $a \in \mathbb{R}^{n}$ \eqref{def:delta_measure} does not admit a description by density functions.  Intuitively, such a density function, if existed, would be given by the `delta function'  centred at $a$, but it is actually a distribution and not a member of $L^{1}(\mathbb{R}^{n})$ as required. This leads to the basic fact that the convolution algebra $L^{1}(\mathbb{R}^{n})$ is non-unital, {\it i.e.}, it 
lacks a multiplicative identity in the sense that there is no element $e \in L^{1}(\mathbb{R}^{n})$ for which
\begin{equation}
e \ast f = f
\end{equation}
holds for all $f \in L^{1}(\mathbb{R}^{n})$.  This should be contrasted to the measure algebra $\mathbf{M}_{\mathbb{C}}(\mathfrak{B}^{n})$, which always possesses a multiplicative identity.

\subsubsection{Product Measures}

Given two measure spaces $(X, \mathfrak{A}, \mu)$ and $(Y, \mathfrak{B}, \nu)$, we intend to construct a `product measure' on the product space $X \times Y$ so that $\rho(A \times B) = \mu(A) \nu(B)$ holds for all $A \in \mathfrak{A}$, $B \in \mathfrak{B}$. As its domain of definition, we let
\begin{align}\label{def:product_set_of_sigma_algebras}
\mathfrak{A} \ast \mathfrak{B} := \{A \times B : A \in \mathfrak{A}, B \in \mathfrak{B} \}
\end{align}
and define
\begin{equation}
\mathfrak{A} \otimes \mathfrak{B} := \sigma(\mathfrak{A} \ast \mathfrak{B})
\end{equation}
to be the \emph{product-$\sigma$-algebra} of $\mathfrak{A}$ and $\mathfrak{B}$. The following fact and definition is of importance.
\begin{definition*}[Product Measure]
Given two measure spaces $(X, \mathfrak{A}, \mu)$ and $(Y, \mathfrak{B}, \nu)$, let both $\mu$ and $\nu$ be $\sigma$-finite. Then there exists a unique measure $\mu \otimes \nu : \mathfrak{A} \otimes \mathfrak{B} \to \overline{\mathbb{R}}$ such that
\begin{equation}\label{def:product_measure}
\mu \otimes \nu (A \times B) = \mu(A) \nu(B), \quad A \in \mathfrak{A},\ B \in \mathfrak{B}
\end{equation}
holds. The measure $\mu \otimes \nu$ is $\sigma$-finite and is called the product measure of $\mu$ and $\nu$.
\end{definition*}
\noindent
The integration with respect to the product measure $\mu \otimes \nu$ of two $\sigma$-finite measures $\mu$ and $\nu$ can be performed by iterated integration of each of the respective variables. This is the essence of the following \emph{Fubini's Theorem}, which belongs to one of the most oft-used theorems of integration theory.
\begin{theorem*}[Fubini's Theorem]
Let $\mu$ and $\nu$ be $\sigma$-finite. Then, the following statements hold:
\begin{enumerate}
\item
If $f: X \otimes Y \to \hat{\mathbb{K}}$ is $\mu \otimes \nu$-integrable, then $f(x, \cdot)$ is $\nu$-integrable for almost all $x \in X$. Moreover
\begin{equation}
A := \{x \in X : f(x,\cdot) \text{ is not $\nu$-integrable }\} \in \mathfrak{A};
\end{equation}
and likewise
\begin{equation}
B := \{y \in Y : f(\cdot,y) \text{ is not $\mu$-integrable }\} \in \mathfrak{B}.
\end{equation}
The functions
\begin{equation}
x \mapsto \int_{Y} f(x,y)\ d\nu(y) \quad  x \mapsto \int_{X} f(x,y)\ d\mu(x)
\end{equation}
are respectively $\mu$-integrable on $A^{c}$ and $\nu$-integrable on $B^{c}$, and the equalities
\begin{align}
\int_{X\times Y} f\ d\mu\otimes\nu
    &= \int_{X} \left( \int_{Y} f(x,y)\ d\nu(y) \right)\ d\mu(x) \nonumber \\
    &= \int_{Y} \left( \int_{X} f(x,y)\ d\mu(x) \right)\ d\nu(y)
\end{align}
hold.
\item
If $f: X \otimes Y \to \hat{\mathbb{K}}$ is $\mu \otimes \nu$-integrable, and one of the integrals
\begin{equation}
\int_{X \times Y} |f|\ d\mu\otimes\nu,\ \int_{X} \left( \int_{Y} |f(x,y)|\ d\nu(y) \right)\ d\mu(x),\ \int_{Y} \left( \int_{X} |f(x,y)|\ d\mu(x) \right)\ d\nu(y)
\end{equation}
is finite, then all three of them are finite and agree, $f$ is $\mu \otimes \nu$-integrable, and the statements under \textup{(i)} hold.
\end{enumerate}
\end{theorem*}

\subsubsection{Measure on Topological Spaces}

Let $X$ be a metric space (or a topological space). One may naturally be interested in how the topology relates to the complex measures defined on the Borel $\sigma$-algebra $\mathfrak{B} := \mathfrak{B}(X)$ generated by it. To this end, we briefly review one of the prominent results in the study of this realm, namely the famous \emph{Riesz-Markov-Kakutani Representation Theorem}.
In order to avoid complexity, we shall only deal with the case where the given measurable space is $(\mathbb{R}^{n}, \mathfrak{B}^{n})$. Observing now that a complex measure $\nu \in \mathbf{M}_{\mathbb{C}}(\mathfrak{B}^{n})$ generates an (algebraic) linear map $f \mapsto \int_{\mathbb{R}^{n}} f d\nu$ that maps a function to a complex number, the opposite question is then our interest, namely: what class of linear functionals admits representation by integration with respect to some complex measure?

\paragraph{Riesz-Markov-Kakutani Representation Theorem}

Let $C_{0}(\mathbb{R}^{n})$ be the space of all continuous functions $f: \mathbb{R}^{n} \to \mathbb{C}$ that vanish at infinity, in the sense for every $\epsilon > 0$ there exists a compact subset $K \subset \mathbb{R}^{n}$ for which $|f|K^{c}| < \epsilon$ holds. The space $C_{0}(\mathbb{R}^{n})$ equipped with the supremum norm $\|f\|_{\infty} := \sup \{|f(x)| : x \in \mathbb{R}^{n} \}$ is known to be a Banach space.
Now for each $\nu \in \mathbf{M}_{\mathbb{C}}(\mathfrak{B}^{n})$, the map
\begin{equation}
I_{\nu} : f \mapsto \int_{\mathbb{R}^{n}} f\ d\nu, \quad f \in C_{0}(\mathbb{R}^{n})
\end{equation}
gives rise to a continuous ({\it i.e.}, bounded) $\mathbb{C}$-linear functional from $C_{0}(\mathbb{R}^{n})$ to $\mathbb{C}$, for indeed the evaluation
\begin{equation}
|I_{\nu}(f)| \leq \|\nu\| \cdot \|f\|_{\infty}
\end{equation}
holds.
The Riesz-Markov-Kakutani representation theorem is a classical theorem in measure and integration theory stating that the converse is also true, which is to say that, for any continuous $\mathbb{C}$-linear functional $I \in C_{0}^{\prime}(\mathbb{R}^{n})$, there exists a unique complex measure $\nu \in \mathbf{M}_{\mathbb{C}}(\mathfrak{B}^{n})$ for which
\begin{equation}
I(f) = \int_{\mathbb{R}^{n}} f\ d\nu, \quad f \in C_{0}(\mathbb{R}^{n})
\end{equation}
holds. The precise statement is given as follows.
\begin{theorem*}[Riesz-Markov-Kakutani Representation Theorem for Euclidian Spaces]
The correspondence
\begin{align}\label{eq:RMK-thm_cor}
\begin{split}
&\Phi : \mathbf{M}_{\mathbb{C}}(\mathfrak{B}^{n}) \to C_{0}^{\prime}(\mathbb{R}^{n}), \\
&\Phi(\nu)(f) := \int_{\mathbb{R}^{n}} f\ d\nu, \quad \nu \in \mathbf{M}_{\mathbb{C}}(\mathfrak{B}^{n}),\, f \in C_{0}(\mathbb{R}^{n}),
\end{split}
\end{align}
that maps a complex measure to a continuous linear functional on $C_{0}$ is a bijection, which moreover satisfies
\begin{equation}\label{eq:Riesz_rep_isometry}
\|\Phi(\nu)\| = \|\nu\|.
\end{equation}
In other words, the space of complex measures $\mathbf{M}_{\mathbb{C}}(\mathfrak{B}^{n})$ is isomorphic to the topological dual of $C_{0}(\mathbb{R}^{n})$, and can be mapped to each other by an isometric isomorphism.
\end{theorem*}
\noindent
Here, the norm on $\mathbf{M}_{\mathbb{C}}(\mathfrak{B}^{n})$ on the r.~h.~s. of \eqref{eq:Riesz_rep_isometry} is naturally the total variation norm, and the norm on the topological dual $C_{0}^{\prime}(\mathbb{R}^{n})$ (the l.~h.~s.) is the operator norm defined by
\begin{equation}
\|I\| := \sup_{\|f\|_{\infty}\leq 1} |I(f)|, \quad I \in C_{0}^{\prime}(\mathbb{R}^{n}).
\end{equation}
In this sense we identify
\begin{equation}\label{eq:RMK-thm}
\mathbf{M}_{\mathbb{C}}(\mathfrak{B}^{n}) \cong C_{0}^{\prime}(\mathbb{R}^{n}),
\end{equation}
and may interchangeably interpret a continuous $\mathbb{C}$-linear functional on the space $C_{0}(\mathbb{R}^{n})$ as a complex measure on the measurable space $(\mathbb{R}^{n}, \mathfrak{B}^{n})$, and vice versa.

\subsubsection{Spectral Theorem and its Consequences}\label{sec:spectral_theorem}

We next provide a concise review on some of the basic facts regarding the spectral theorem for self-adjoint operators, which is just the generalisation of the familiar eigendecomposition theorem for Hermitian matrices on finite-dimensional vector spaces to the arbitrary dimensional case. In order to avoid confusion with operators, Borel sets on $\mathbb{R}^{n}$ shall occasionally be denoted by $\Delta \in \mathfrak{B}^{n}$ in place of $B$, especially when we are working in the context of quantum mechanics.

\paragraph{Spectral Measures}
Closely associated to the notion of complex measures is that of spectral measures on a Hilbert space $\mathcal{H}$. Let $L(\mathcal{H})$ denote the space of all bounded operators on $\mathcal{H}$, and recall that a map
\begin{equation}
E : \mathfrak{B}^{n} \to L(\mathcal{H}), \quad \Delta \mapsto E(\Delta)
\end{equation}
is called an $n$-dimensional \emph{spectral measure} (or \emph{projection-valued measure}), if each $E(\Delta)$, $\Delta \in \mathfrak{B}^{n}$ is an orthogonal projection on $\mathcal{H}$ and satisfies 
\begin{enumerate}
\item $E(\emptyset) = 0$, $E(\mathbb{R}^{n}) = I$,
\item for pairwise disjoint $\Delta_{1}, \Delta_{2}, \dots \in \mathfrak{B}^{n}$,
\begin{equation}
\sum_{i=1}^{\infty}E(\Delta_{i}) |\phi\rangle = E\left(\bigcup_{i=1}^{\infty}\Delta_{i}\right) |\phi\rangle, \qquad \forall |\phi\rangle \in \mathcal{H}.
\end{equation}
\end{enumerate}
The support of a spectral measure $E$ on $\mathfrak{B}^{n}$ is defined as the smallest Borel set $\Delta \in \mathfrak{B}^{n}$ that satisfies $E(\Delta) = I$.
An important point is that 
a spectral measure $E$ and a pair of vectors $|\phi\rangle, |\phi\rangle \in \mathcal{H}$ induce a complex measure on $\mathfrak{B}^{n}$ given by
\begin{equation}\label{def:complex_measure_A}
\Delta \mapsto \langle \phi^{\prime}, E(\Delta) \phi \rangle, \quad \Delta \in \mathfrak{B}^{n}.
\end{equation}

\paragraph{Spectral Theorem of Self-adjoint Operators}

Having recalled the necessary definitions, we now state the \emph{spectral theorem for self-adjoint operators}, which constitutes one of the most important mathematical ingredients in quantum mechanics.
\begin{theorem*}[Spectral decomposition of self-adjoint operators]
Let $A : \mathcal{H} \supset \mathrm{dom}(A) \to \mathcal{H}$ be self-adjoint. Then there exists a unique one-dimensional spectral measure $E_{A}$ supported on the spectrum $\sigma(A) \subset \mathbb{R}$ of $A$ satisfying
\begin{equation}\label{thm:sd}
\langle \phi^{\prime}, A \phi \rangle = \int_{\sigma(A)} a\ d\langle \phi^{\prime}, E_{A}(a) \phi \rangle, \quad \forall |\phi\rangle \in \mathrm{dom}(A),\ \forall |\phi^{\prime}\rangle \in \mathcal{H},
\end{equation}
where the r.~h.~s. of the equality is understood as the Lebesgue integral with respect to the complex measure $\Delta \mapsto \langle \phi^{\prime}, E_{A}(\Delta) \phi \rangle$
induced from $E_{A}$ and the pair of vectors $|\phi\rangle$ and $|\phi^{\prime}\rangle$.
\end{theorem*}
\noindent
Under the situation above, the self-adjoint operator $A$ is occasionally written symbolically as
\begin{equation}\label{thm:sd_op}
A = \int_{\sigma(A)} a\ dE_{A}(a),
\end{equation}
in terms of integration with respect to its spectral measure.

\paragraph{Finite-dimensional Case}
To see the meaning of the above formula, we make a brief note on how the familiar eigendecomposition theorem for Hermitian matrices appears as a special case of the general statement. Let $A$ be a Hermitian matrix on an $N$-dimensional complex Hilbert space $\mathcal{H} := \mathbb{C}^{N}$, $N \in \mathbb{N}^{\times}$. The eigendecomposition theorem states that, there exists an orthonormal basis $ \mathcal{B}_{A} := \{|a_{1}\rangle, \dots, |a_{N}\rangle\}$ of $\mathcal{H}$ with real numbers $a_{1}, \dots, a_{N} \in \mathbb{R}$ such that
\begin{equation}
A |a_{i}\rangle = a_{i} |a_{i}\rangle, \quad i = 1, \dots, N,
\end{equation}
hold.  
For each eigenvalue $a \in \sigma(A) = \{a_{1}, \dots, a_{N}\}$ of $A$,  we have the projection $\Pi_{a}$
onto the subspace, 
\begin{equation}
\mathcal{H}_{a} := \mathrm{span} \{ |a\rangle \in \mathcal{B}_{A} : A |a\rangle = a |a\rangle \}
\end{equation}
spanned by the collection of all eigenvectors associated with $a$. As we noted before, when the eigenstate $| a \rangle$ is non-degenerate for $a$, or the subspace $\mathcal{H}_{a}$ is one-dimensional, we may write $\Pi_a = | a \rangle \langle a |$.
With the projection $ \Pi_a $ in hand, the spectral measure of $A$ is defined by
\begin{equation}\label{def:spectrum_finite}
E_{A}(\Delta) := \sum_{a \in \sigma(A) \cap \Delta} \Pi_a, \quad \Delta \in \mathfrak{B},
\end{equation}
with the convention $\sum_{a \in \emptyset}\Pi_a := 0$. 
One readily verifies that $E_{A}$ is indeed a spectral measure supported on its spectrum $\sigma(A)$, and subsequently sees that the projection $\Pi_{a} = E_{A}(\{a\})$ is nothing but the image of the spectral measure $E_{A}$ on the Borel set $\{a\} \in \mathfrak{B}$ consisting of a single eigenvalue $a \in \sigma(A)$ of the observable $A$. One then finds
\begin{align}\label{fin_sd_op}
A = \sum_{a \in \sigma(A)} a\, \Pi_a  = \sum_{a \in \sigma(A)} a\, E_{A}(\{a\})
\end{align}
in accordance with \eqref{eq:spect_decomp_fin}, and subsequently proves
\begin{equation}\label{fin_sd}
\langle \phi^{\prime}, A \phi \rangle = \sum_{a \in \sigma(A)} a\, \langle \phi^{\prime}, E_{A}(\{a\}) \phi \rangle, \quad \forall |\phi^{\prime}\rangle, |\phi\rangle \in \mathcal{H}.
\end{equation}
The spectral decomposition formula \eqref{thm:sd} and the formal expression \eqref{thm:sd_op} are respectively just the generalisations of the finite dimensional versions \eqref{fin_sd} and \eqref{fin_sd_op}.

\paragraph{Functional Calculus}
By means of the spectral decomposition of a self-adjoint operator, one may create a new set of operators from it.  Let $A : \mathcal{H} \supset \mathrm{dom}(A) \to \mathcal{H}$ be a self-adjoint operator on a Hilbert space $\mathcal{H}$, and let $E_{A}$ the unique one-dimensional spectral measure associated with it.  Given a measurable complex function $f : \mathbb{R} \to \mathbb{C}$, the integral
\begin{equation}
\langle\phi^{\prime}, f(A) \phi\rangle = \int_{\sigma(A)} f(a)\ d\langle \phi^{\prime}, E_{A}(a)\phi\rangle
\end{equation}
defines a unique linear operator $f(A)$ on $\mathcal{H}$, where 
\begin{equation}
|\phi\rangle \in \mathrm{dom}(f(A)) := \left\{ |\phi\rangle \in \mathcal{H} : \int_{\sigma(A)} \left| f(a) \right|^{2}\ d\langle \phi, E_{A}(a)\phi\rangle < \infty \right\}
\end{equation}
is any vector belonging to its domain, and $|\phi^{\prime}\rangle \in \mathcal{H}$.  The operator $f(A)$ is occasionally written symbolically as
\begin{equation}\label{def:func_calc}
f(A) = \int_{\sigma(A)} f(a)\ dE_{A}(a),
\end{equation}
in terms of integration with respect to its spectral measure.

\paragraph{Born Rule and Quantum Measurement}

The axiom of quantum mechanics states that a quantum observable is represented by a self-adjoint operator $A : \mathcal{H} \supset \mathrm{dom}(A) \to \mathcal{H}$ on a Hilbert space $\mathcal{H}$, and that the probabilistic behaviour of the outcomes of an ideal measurement of $A$ on the state $|\phi\rangle \in \mathcal{H}$ is described by the probability measure,
\begin{equation}\label{def:prob_measrue_A}
\Delta \mapsto \mu_{A}^{\phi}(\Delta) := \frac{\langle\phi, E_{A}(\Delta) \phi\rangle}{\|\phi\|^{2}}, \quad \Delta \in \mathfrak{B}.
\end{equation}
Here, the spectral measure $E_{A}$ is induced from $A$ by the spectral theorem, 
and the Born rule proclaims that the measurement outcome be given by one of the elements in the spectrum $\sigma(A)$ and that 
$\mu_{A}^{\phi}(\Delta)$ provides the probability of finding the measurement in the measurable set $\Delta \in \mathfrak{B}$. 
Given $|\phi\rangle \in \mathrm{dom}(A)$, one then realises from the spectral theorem \eqref{thm:sd} that the statistical average of the measurement outcomes of $A$ gives the expectation value,
\begin{align}\label{def:expectation_A}
\int_{\mathbb{R}} a\ d\mu_{A}^{\phi}(a)
= \frac{\langle \phi, A \phi \rangle}{\|\phi\|^{2}} =: \mathbb{E}[A;\phi],
\end{align}
where the l.~h.~s of the first equality is understood to be the Lebesgue integral with respect to the probability measure \eqref{def:prob_measrue_A}.

\subsubsection{Observables admitting a Description by Density Functions}

While the analysis based on probability measures provides an adequately general framework to work with, we find it useful to prepare a terminology for a special class of observables for which probability density functions, not just probability measures, are available to fully describe the behaviour of the measurement outcomes.

\paragraph{Observable admitting a description by probability density functions}
In this paper, we simply say that 
an observable $A$ \emph{admits a description by probability density functions}, if the probability measure \eqref{def:prob_measrue_A} induced by the spectral measure of $A$ is absolutely continuous with respect to the Lebesgue-Borel measure for every choice of the quantum state $|\phi\rangle \in \mathcal{H}$, which is to say that, if for every $|\phi\rangle \in \mathcal{H}$, there exists an integrable function $\rho_{A}^{\phi} \in L^{1}(\mathbb{R})$ such that
\begin{equation}\label{def:absolutely_continuous}
\mu_{A}^{\phi}(\Delta)
= \int_{\Delta} \rho_{A}^{\phi}(a)\ d\beta(a), \quad \Delta \in \mathfrak{B},
\end{equation}
holds.  

A well-known example of it is provided by the one-dimensional position operator $\hat{x}$ on $L^{2}(\mathbb{R})$ defined in \eqref{def:position_operator}. Indeed, one proves that the spectral measure of $\hat{x}$ is given by the multiplication of the characteristic function \eqref{def:characteristic_function} as
\begin{equation}
E_{\hat{x}}(\Delta) : \psi(x) \mapsto \chi_{\Delta}(x)\psi(x), \quad \psi \in L^{2}(\mathbb{R}),
\end{equation}
for each $B \in \mathfrak{B}$, so that
\begin{align}
\langle \psi_{1}, E_{\hat{x}}(\Delta) \psi_{2} \rangle
    &= \int_{\mathbb{R}} \psi_{1}^{*}(x) \chi_{\Delta}(x) \psi_{2}(x)\ d\beta(x) \nonumber \\
    &= \int_{\Delta} \psi_{1}^{*}(x) \psi_{2}(x)\ d\beta(x), \quad \Delta \in \mathfrak{B},
\end{align}
holds. Specifically, this implies that
\begin{equation}
\mu_{\hat{x}}^{\psi}(\Delta) = \int_{\Delta} \frac{|\psi(x)|^{2}}{\| \psi \|^{2}_{2}}\ d\beta(x), \quad \Delta \in \mathfrak{B},
\end{equation}
where the denominator of the integrand of the r.~h.~s. denotes the square of the $L^{2}$-norm of $\psi \in L^{2}(\mathbb{R})$ (see \eqref{def:Lp_norm}). One thus concludes that the density of the probability measure $\mu_{\hat{x}}^{\psi}$ is provided by
\begin{equation}
\rho_{\hat{x}}^{\psi}(x) = \frac{|\psi(x)|^{2}}{\| \psi \|^{2}_{2}}.
\end{equation}
Incidentally, it is known that each member of the pair of observables $\{Q, P\}$ that satisfies the Weyl relations \eqref{eq:weyl01} and \eqref{eq:weyl02} admits descriptions in terms of density functions.

However, it should be noted that this is not always the case in general: an observable with the spectrum consisting of a finite number of discrete eigenvalues (such as spin) provides a simple counterexample. To see this, let $A$ be such an observable with $N \in \mathbb{N}^{\times}$ distinct eigenvalues, and let $\sigma(A) = \{a_{1}, \dots, a_{N} \}$ be any enumeration of its spectrum. A straightforward application of \eqref{def:spectrum_finite} leads to
\begin{equation}\label{eq:prob_meas_fin_spec_obs}
\mu_{A}^{\phi} = \sum^{N}_{n=1} \mu_{A}^{\phi}(\{a_{n}\}) \cdot \delta_{a_{n}},
\end{equation}
in which one sees that the probability measure $\mu_{A}^{\phi}$ is given by the weighted sum of delta measures centred at each eigenvalue.  Obviously, since each of the delta measures is not absolutely continuous, the resultant probability measure does not admit a description by density functions.

For later use, we also note that, once the observable $A$ admits a description in terms of probability density functions,
then the complex measure \eqref{def:complex_measure_A} is also absolutely continuous for an arbitrary pair of vectors $|\phi^{\prime}\rangle, |\phi\rangle \in \mathcal{H}$.  That this is the case can be seen by a straightforward application of the polarisation identity
\begin{align}\label{def:polarisation_identity}
\langle\phi^{\prime}, T\phi \rangle
    &= \frac{1}{4} \left\{ \langle\phi^{\prime} + \phi, T(\phi^{\prime} + \phi) \rangle
    - \langle\phi^{\prime} - \phi, T(\phi^{\prime} - \phi) \rangle \right. \nonumber \\
    & \qquad \left. + i\langle\phi^{\prime} + i\phi, T(\phi^{\prime} + i\phi) \rangle
    - i\langle\phi^{\prime} - i\phi, T(\phi^{\prime} - i\phi) \rangle \right\}
\end{align}
with respect to the operator $T$ valid for any pair of vectors $|\phi\rangle, |\phi^{\prime}\rangle \in \mathrm{dom}(T)$, where we simply replace $T=E_{A}(\Delta)$ for each $\Delta \in \mathfrak{B}$.

\subsubsection{Simultaneously measurable Observables}\label{sec:sim_meas_obs}

For reference, we briefly review the basic mathematical definitions and facts involved in describing measurements of simultaneously measurable observables, including the simultaneous measurement of local observables on the tensor product of Hilbert spaces.

\paragraph{Strong Commutativity of Self-adjoint Operators}

Let $A$ and $B$ be self-adjoint operators on a Hilbert space $\mathcal{H}$, and let $E_{A}$ and $E_{B}$ be their respective spectral measures. We say that the pair of operators $A$ and $B$ \emph{strongly commutes}, if
\begin{equation}
E_{A}(\Delta_{A})\, E_{B}(\Delta_{B}) = E_{B}(\Delta_{B})\,E_{A}(\Delta_{A}), \quad \Delta_{A}, \Delta_{B} \in \mathfrak{B}
\end{equation}
holds as an operator equality. Note that the strong commutativity of $A$ and $B$ implies its (familiar) commutativity $AB = BA$. On the other hand, it is known that the converse is in general not true in the case where either (or both) of the operators happens to be unbounded. The term strong commutativity is named after this fact, for it indicates a stronger condition than mere commutativity.

\paragraph{Product Spectral Measures}
It is a basic result of functional analysis that, given such a pair of $A$ and $B$ of strongly commuting self-adjoint operators, there exists a unique two-dimensional spectral measure
$E_{A, B}$ called the \emph{product spectral measure} of $A$ and $B$, for which
\begin{align}\label{eq:product_spectral_measure}
E_{A, B}(\Delta_{A} \times \Delta_{B}) = E_{A}(\Delta_{A})\, E_{B}(\Delta_{B}) = E_{B}(\Delta_{B})\, E_{A}(\Delta_{A}), \quad \Delta_{A}, \Delta_{B} \in \mathfrak{B}
\end{align}
holds.  This is a straightforward operator-valued analogue of product measures in measure theory.
With a pair of vectors $|\phi\rangle, |\phi^{\prime}\rangle \in \mathcal{H}$ being specified, this gives rise to a complex measure on $(\mathbb{R}^{2}, \mathfrak{B}^{2})$, defined by
\begin{align}
\Delta \mapsto \left\langle \phi^{\prime}, E_{A, B}(\Delta) \, \phi \right\rangle, \quad \Delta \in \mathfrak{B}^{2}.
\end{align}
In the context of quantum mechanics, for a given pair of simultaneously measurable quantum observables represented by strongly commuting self-adjoint operators $A$ and $B$, the probabilistic behaviour of the outcomes of an ideal simultaneous measurement of both the observables on the state $|\phi\rangle \in \mathcal{H}$ is described by the \emph{joint-probability distribution}
\begin{equation}\label{def:prob_measrue_A_simul}
\Delta \mapsto \mu_{A,B}^{\phi}(\Delta) :=  \frac{\langle\phi, E_{A, B}(\Delta)\,  \phi\rangle}{\|\phi\|^{2}}, \quad \Delta \in \mathfrak{B}^{2},
\end{equation}
of the pair of observables $A$ and $B$ on the state $|\phi\rangle$,
which is a two-dimensional probability measure on the measurable space $(\mathbb{R}^{2}, \mathfrak{B}^{2})$. Here, the r.~h.~s. of \eqref{def:prob_measrue_A_simul} is interpreted as the probability of finding the outcomes of a simultaneous measurement of both observables in the Borel set $\Delta \in \mathfrak{B}^{2}$.
Note that the measurement outcomes of $A$ and $B$ may not be independent, {\it i.e.}, the equality,
\begin{align}
\mu_{A,B}^{\phi}(\Delta_{A} \times \Delta_{B}) = \mu_{A}^{\phi}(\Delta_{A}) \cdot \mu_{B}^{\phi}(\Delta_{B}), \quad \Delta_{A}, \Delta_{B} \in \mathfrak{B},
\end{align}
may not necessarily hold, or in other words, the joint-probability distribution is not necessarily the product measure $\mu_{A,B}^{\phi} \neq \mu_{A}^{\phi} \otimes \mu_{B}^{\phi}$ of each of the respective measurements, in general.

\paragraph{Functional Calculus regarding simultaneously measurable Observables}
Given a pair of strongly commuting self-adjoint observables $A$ and $B$, one readily confirms
\begin{equation}
A = \int_{\mathbb{R}} a\ dE_{A,B}(a,b), \quad B = \int_{\mathbb{R}} b\ dE_{A,B}(a,b).
\end{equation}
As for the sum and product of the observables, we first note the following basic fact.
\begin{lemma}\label{lem:st_strong_commutative}
Let a pair of self-adjoint operators $A$ and $B$ strongly commute. Then,
\begin{enumerate}
\item The operators $A$ and $B$ commute with each other on $\mathrm{dom}(AB) \cap \mathrm{dom}(BA)$, and the anti-commutator%
\footnote{
Here, the domain of the anti-commutator $\{X,Y\} := XY + YX$ of the pair of operators $X$, $Y$ are understood to be $\mathrm{dom}(\{X,Y\}) := \mathrm{dom}(XY) \cap \mathrm{dom}(YX)$.
}
\begin{equation}
\{A,B\} := AB + BA
\end{equation}
is essentially self-adjoint.
\item $A+B$ is essentially self-adjoint.
\end{enumerate}
\end{lemma}
\noindent
As a direct consequence, we thus have the operator equalities
\begin{align}
\overline{A + B} &= \int_{\mathbb{R}^{2}} (a + b)\ dE_{A,B}(a,b) \\
\frac{\overline{AB + BA}}{2} &= \int_{\mathbb{R}^{2}} ab\ dE_{A,B}(a,b)
\end{align}
worth of special notice.
As above, overlines on closable operators denote their closures, and specifically for essentially self-adjoint operators, their self-adjoint extensions.

\paragraph{Composite Systems}
We comment on the special case of the above situation in which the Hilbert space of our interest is the tensor product $\mathcal{H} \otimes \mathcal{K}$ of the target system $\mathcal{H}$ and the meter system $\mathcal{K}$, and the operators involved are (local) self-adjoint operators $A_{1}$ and $A_{2}$ on the respective Hilbert spaces.
Observing that the operators
\begin{equation}
\tilde{A}_{1} := A_{1} \otimes I, \quad \tilde{A}_{2} := I \otimes A_{2}
\end{equation}
strongly commute with each other on the composite Hilbert space $\mathcal{H} \otimes \mathcal{K}$, and that their spectral measures respectively read
\begin{equation}
E_{\tilde{A}_{1}} := E_{A_{1}} \otimes I, \quad E_{\tilde{A}_{2}} := I \otimes E_{A_{2}},
\end{equation}
the previous argument leads to the existence of a unique two-dimensional product spectral measure $E_{A_{1}} \otimes E_{A_{2}} := E_{\tilde{A}_{1},\tilde{A}_{2}}$
satisfying the operator equality
\begin{align}
\left( E_{A_{1}} \otimes E_{A_{2}} \right)(\Delta_{1} \times \Delta_{2})
    &= E_{\tilde{A}_{1}}(\Delta_{1}) E_{\tilde{A}_{2}}(\Delta_{2})\nonumber \\
    &= (E_{A_{1}}(\Delta_{1}) \otimes I) (I \otimes E_{A_{2}}(\Delta_{2})) \nonumber \\
    &= E_{A_{1}}(\Delta_{1}) \otimes E_{A_{2}}(\Delta_{2}), \quad \Delta_{1}, \Delta_{2} \in \mathfrak{B}.
\end{align}
Here, the left-most hand side denotes the two-dimensional spectral measure defined as in \eqref{eq:product_spectral_measure}, while the right-most hand side denotes the tensor product of the self-adjoint operators $E_{A_{1}}(\Delta_{1})$ and $E_{A_{2}}(\Delta_{2})$ for each $\Delta_{1}, \Delta_{2} \in \mathfrak{B}$.
As we have seen in the previous argument, this gives rise to a complex measure,
\begin{align}
\Delta \mapsto \langle \Phi^{\prime}, \left(E_{A_{1}} \otimes E_{A_{2}}\right)(\Delta)\, \Phi \rangle, \quad \Delta \in \mathfrak{B}^{2},
\end{align}
for a given selection of a pair $|\Psi\rangle, |\Phi \rangle \in \mathcal{H} \otimes \mathcal{K}$
of vectors of the composite system, and the map,
\begin{equation}\label{def:prob_measrue_A_entangled}
\Delta \mapsto \mu_{A_{1},A_{2}}^{\Phi}(\Delta) := \frac{\langle\Phi, \left(E_{A_{1}} \otimes E_{A_{2}}\right)(\Delta)\, \Phi\rangle}{\|\Phi\|^{2}}, \quad \Delta \in \mathfrak{B}^{2},
\end{equation}
(here,  we have slightly abused the notation on the l.~h.~s. by writing $A_{n}$ in place of $\tilde{A}_{n}$ for each $n = 1, 2$) provides a probability measure describing the probabilistic behaviour of the outcomes of the ideal local measurements simultaneously performed on each system in the state $|\Phi\rangle \in \mathcal{H} \otimes \mathcal{K}$.

In passing, we note that in the case where the state $|\Phi\rangle$ happens to be a direct product state $|\Phi\rangle = |\phi_{1} \otimes \phi_{2} \rangle$, 
the induced joint-probability distribution of the two local observables \eqref{def:prob_measrue_A_entangled} becomes the product measure of the two probability measures associated with $A_{1}$ and $A_{2}$, 
\begin{align}\label{eq:sim_prob_measrue_dps}
\mu_{A_{1},A_{2}}^{\phi_{1} \otimes \phi_{2}}(\Delta_{1} \times \Delta_{2}) = \mu_{A_{1}}^{\phi_{1}}(\Delta_{1}) \cdot \mu_{A_{2}}^{\phi_{2}}(\Delta_{2}), \quad \Delta_{1}, \Delta_{2} \in \mathfrak{B},
\end{align}
indicating that the measurement outcomes of each local measurement $A_{1}$ and $A_{2}$ 
are statistically independent ({\it i.e.}, $\mu_{A_{1},A_{2}}^{\phi_{1} \otimes \phi_{2}} = \mu_{A_{1}}^{\phi_{1}} \otimes \mu_{A_{2}}^{\phi_{2}}$).   
On the other hand, if one chooses the state $|\Phi\rangle$ to be an entangled state ({\it i.e.}, those states in $\mathcal{H} \otimes \mathcal{K}$ that are not direct product states), the joint-probability distribution \eqref{def:prob_measrue_A_entangled} is no more a product measure of those associated to the local observables in general. In the language of physics, this implies that the local measurements performed on each remote system may have some correlation if the state of the composite system happens to be entangled, and this is widely considered to be one of the most intriguing properties of quantum mechanics.
Of course, statistical independence between the target and the meter systems is useless for the purpose of our measurement, and we naturally need an entangled state $|\Phi\rangle$ in order to retrieve any meaningful information of the former system out of the measurement of the latter.

\paragraph{Sum and Product of Local Observables}
As for the sum and product of a pair of local observables, we note that a direct application of Lemma~\ref{lem:st_strong_commutative} leads to
\begin{equation}
A_{1} \otimes I = \int_{\mathbb{R}} a_{1}\ d\left(E_{A_{1}} \otimes E_{A_{2}}\right)(a_{1},a_{2}), \quad I \otimes A_{2} = \int_{\mathbb{R}} a_{2}\ d\left(E_{A_{1}} \otimes E_{A_{2}}\right)(a_{1},a_{2}),
\end{equation}
and subsequently
\begin{align}
\overline{A_{1} \otimes I + I \otimes A_{2}} &= \int_{\mathbb{R}^{2}} (a_{1} + a_{2})\ d\left(E_{A_{1}} \otimes E_{A_{2}}\right)(a_{1}, a_{2}), \\
A_{1} \otimes A_{2} &= \int_{\mathbb{R}^{2}} a_{1}a_{2}\ d\left(E_{A_{1}} \otimes E_{A_{2}}\right)(a_{1}, a_{2}),
\end{align}
as expected.

\subsection{Unconditioned Measurement}\label{sec:ups_II_ups}

Now that we have recalled the necessary mathematical concepts and results, we shall embark on our main analysis. 
The target of our analysis is the probability measure describing the behaviour of the outcome of the composite observable $I \otimes X$ on $|\Psi^{g}\rangle$, which may be rewritten in terms of that of the local observable $X$ on the mixed state $\psi^{g}$ as
\begin{align}\label{eq:pdf_x_outcome}
\mu_{I \otimes X}^{\Psi^{g}}(\Delta)
    = \mu_{X}^{\psi^{g}}(\Delta), \quad \Delta \in \mathfrak{B},
\end{align}
where the last definition $\mu_{X}^{\psi^{g}}(\Delta) := \mathrm{Tr}[E_{X}(\Delta)\psi^{g}] / \mathrm{Tr}[\psi^{g}]$ is merely a straightforward extension of probability measures \eqref{def:prob_measrue_A} for density operators (for the proof of the equality \eqref{eq:pdf_x_outcome}, just replace $X$ with $E_{X}(\Delta)$ in \eqref{eq:exp_x_outcome}).

\paragraph{Main Objective of this Subsection}
The primary interest of our study is now to investigate how the information of the target system is encoded into the profile of the outcome of the meter system \eqref{eq:pdf_x_outcome} through the interaction.
As in the previous subsection, we assume without loss of generality that the meter observable $Y$ coupled with the target observable $A$ to yield the von Neumann interaction \eqref{intro:von_Neumann_interaction} is given by $Y=P$. The main objective of the passage is to demonstrate the following proposition as an answer to this question. The results, which shall be shortly demonstrated, form the bases we rely on in conducting our further study.
\begin{proposition}[Unconditioned Measurement II.a]\label{prop:UCM_II}
In the context of the UM scheme, let $Y=P$ be fixed for definiteness, and let $|\phi\rangle \in \mathcal{H}$ and $|\psi\rangle \in \mathcal{K}$ respectively be the initial states of the target and the meter systems. Then, the probability measure \eqref{eq:pdf_x_outcome} for both the choice $X = Q, P$ reads
\begin{equation}\label{eq:outcome_prob_mod_01}
\begin{split}
\mu_{Q}^{\psi^{g}}
   &= \mu_{Q}^{\psi} \ast \mu_{(gA)}^{\phi}, \\
\mu_{P}^{\psi^{g}}
   &= \mu_{P}^{\psi},
\end{split}
\qquad g \in \mathbb{R},
\end{equation}
in which the resultant profile of the measurement outcomes of $X$ after the interaction can be exclusively written by the convolution of the initial profiles of both the target and the meter systems.
\end{proposition}
\noindent
Specifically, the interaction causes the change only in the profile of the outcome of the observable $X$ conjugate to $Y$, in which the initial profile of the target system acts upon that of the meter system through convolution of measures. On the other hand, the profile of $X$ for the same choice as $Y$ is left untouched.  The proposition can be readily demonstrated by observing that the change of the spectral measure of the measuring observables $(I \otimes X)$ with respect to the unitary operator $U(g) := e^{-igA \otimes P}$ is provided by
\begin{align}
U(-g)\, E_{I \otimes Q}(\Delta)\, U(g) &= E_{\overline{(I \otimes Q + gA \otimes I)}}(\Delta), \quad \Delta \in \mathfrak{B}, \\
U(-g)\, E_{I \otimes P}(\Delta)\, U(g) &= E_{(I \otimes P)}(\Delta), \quad \quad \Delta \in \mathfrak{B},
\end{align}
in the Heisenberg picture (they are respectively direct consequences of \eqref{eq:weyl_analogue01} and \eqref{eq:weyl_analogue02}), and that the probability distribution dictating the probabilistic behaviour of the sum of two simultaneously measurable observables is described by the convolution of both the individual profiles of the observables involved (which is in parallel to the well-known result for random variables in classical probability theory).  However, in the main passages that follow, we intend to provide a more elementary and straightforward demonstration.  As a corollary to this, one equivalently has:
\begin{corollary}[Unconditioned Measurement II.b]\label{cor:UCM_II}
Under the same condition as above, the result \eqref{eq:outcome_prob_mod_01} can also be rewritten as
\begin{equation}\label{eq:outcome_prob_mod_02}
\begin{split}
\mu_{(g^{-1}Q)}^{\psi^{g}}
   &= \mu_{(g^{-1}Q)}^{\psi} \ast \mu_{A}^{\phi}, \\
\mu_{(gP)}^{\psi^{g}}
   &= \mu_{(gP)}^{\psi},
\end{split}
\qquad g \in \mathbb{R}^{\times},
\end{equation}
by rescaling the outcome by the interaction parameter.
\end{corollary}
\noindent
The two different manners \eqref{eq:outcome_prob_mod_01} and \eqref{eq:outcome_prob_mod_02} of describing the effect of the interaction correspond to the two possible ways of combining the interaction parameter $g$ in the unitary group as
\begin{equation}\label{eq:combining_of_the_interaction}
U(g) := e^{-i(gA)\otimes P} = e^{-i A\otimes (gP)}.
\end{equation}
Combining the interaction parameter $g$ and the target observable $A$ (the former) corresponds to the scaling of the target observable $A \to gA$, whereas combining $g$ and the meter observable $P$ (the latter) corresponds to the scaling of the pair of the meter observables $\{Q, P\} \to \{g^{-1}Q, gP\}$.
Note that the pair of scaled observables $\{g^{-1}Q, gP\}$ for $g \in \mathbb{R}^{\times}$ still satisfies the Weyl relations \eqref{eq:weyl01} and \eqref{eq:weyl02}.

Later on, we shall be investigating how one could recover the information of the target system $\mu_{A}^{\phi}$ based on the results that we obtained here. Incidentally, one finds that probing either the strong or the weak region of the interaction parameter proves itself useful for this purpose, and the equalities \eqref{eq:outcome_prob_mod_01} and \eqref{eq:outcome_prob_mod_02} shall serve as the respective starting points for analysing the weak and the strong UM schemes.

\paragraph{Preliminary Observation}
For our purpose, we first consider the case where the target observable $A$ has a finite point spectrum $\sigma(A) = \{ a_{1}, \dots, a_{N} \}$, $N \in \mathbb{N}^{\times}$. Writing the spectral decomposition of $A$ as \eqref{eq:spect_decomp_fin} and applying \eqref{eq:int_op_fin}, one finds that the composite state after the interaction reads
\begin{align}\label{eq:interaction_finite}
|\Psi^{g}\rangle
    &= \sum_{n = 1}^{N} \left( \Pi_{a_{n}} \otimes e^{-iga_{n}P} \right) |\phi \otimes \psi \rangle \nonumber \\
    &= \sum_{n = 1}^{N} \left( |\Pi_{a_{n}}\phi\rangle \otimes |e^{-iga_{n}P} \psi\rangle \right), \quad g \in \mathbb{R}.
\end{align}
It then follows that
\begin{align}\label{eq:importantidentity}
\mu_{Q}^{\psi^{g}}(\Delta)
    &= \frac{\| (I \otimes E_{Q}(\Delta)) \Psi^{g}\|^{2}}{\|\Psi^{g}\|^{2}} \nonumber \\
    &= \sum_{m = 1}^{N}\sum_{n = 1}^{N} \frac{\langle \phi, \Pi_{a_{m}}\Pi_{a_{n}}\phi \rangle}{\|\phi\|^{2}} \cdot \frac{\langle e^{-iga_{m}P}\psi, E_{Q}(\Delta) e^{-iga_{n}P} \psi\rangle}{\|\psi\|^{2}} \nonumber \\
    &= \sum_{n = 1}^{N} \frac{\| \Pi_{a_{n}}\phi\|^{2}}{\|\phi\|^{2}} \cdot \frac{\| E_{Q}(\Delta)e^{-iga_{n}P}\psi\|^{2}}{{\|\psi\|^{2}}} \nonumber \\
    &= \sum_{n = 1}^{N} \mu_{A}^{\phi}(\{a_{n}\}) \cdot \mu_{Q}^{\psi}(\Delta - ga_{n}) \nonumber \\
    &= \sum_{n = 1}^{N} \mu_{A}^{\phi}(\{a_{n}\}) \cdot  \int_{\mathbb{R}} \mu_{Q}^{\psi}(\Delta - ga)\ d\delta_{a_{n}}(a) \nonumber \\
    &= \int_{\mathbb{R}} \mu_{Q}^{\psi}(\Delta - ga) \ d\mu_{A}^{\phi}(a), \quad  g \in \mathbb{R}, \quad \Delta \in \mathfrak{B},
\end{align}
where we have used the operator equality%
\footnote{
This is a direct result of \eqref{eq:weak_weyl}.
}
$e^{itP}E_{Q}(\Delta)e^{-itP} = E_{Q}(\Delta - t)$, $\Delta \in \mathfrak{B}$ in the third to last equality, and have applied \eqref{eq:prob_meas_fin_spec_obs} to obtain the last equality.

\paragraph{Description of the Measurement Outcome}
Returning to the general case, where the target observable $A$ is now arbitrary, we may conjecture from \eqref{eq:importantidentity} that
\begin{equation}\label{eq:ups_prob_dens_func}
\mu_{Q}^{\psi^{g}}(\Delta) = \int_{\mathbb{R}} \mu_{Q}^{\psi}(\Delta - ga) \ d\mu_{A}^{\phi}(a), \quad  g \in \mathbb{R}, \quad \Delta \in \mathfrak{B}
\end{equation}
generally holds, which indeed turns out to be true; it can be shown straightforwardly in the general framework of functional analysis and measure and integration theory.
From \eqref{eq:ups_prob_dens_func}, we see that the probability measure describing the behaviour of the measurement outcome of $Q$ on the (mixed) state $\psi^{g}$ after the interaction can be explicitly given by those of the initial states of both the meter and the system. Speaking in an intuitive way, each value $a \in \sigma(A)$ of the spectrum of $A$ causes a translation $\mu_{Q}^{\psi}(\Delta) \mapsto \mu_{Q}^{\psi}(\Delta-ga)$, $\Delta \in \mathfrak{B}$ to the probability measure of the initial meter state while keeping its `shape' of the profile intact, and each of these effects is all added over, weighted by the original probability $\mu_{A}^{\phi}$ of the target observable $A$.

Parallel to this, we remark that the ideal measurement of the observable $X=P$ after the von Neumann interaction would result in
\begin{equation}
\mu_{P}^{\psi^{g}} = \mu_{P}^{\psi}, \quad g \in \mathbb{R},
\end{equation}
which states that the interaction does not alter the profile of the measurement of $X=P$ at all. This can be readily shown by changing $Q$ to $P$ in \eqref{eq:importantidentity}, and by applying the operator equality $e^{itP}E_{P}(\Delta)e^{-itP} = E_{P}(\Delta)$, $\Delta \in \mathfrak{B}$.

\paragraph{Scaling of Measures and Density Functions}
For later arguments, it proves convenient to rewrite our previous result \eqref{eq:ups_prob_dens_func} in terms of convolution of measures after introducing 
some notations.  
Let $\mu \in \mathbf{M}_{\mathbb{C}}(\mathfrak{B}^{n})$ be a complex measure, and define a parametrised family $\{\mu_{t}\}_{t \in \mathbb{R}}$ of complex measures by
\begin{equation}\label{def:measure_scaling}
\mu_{t}(B) :=
    \begin{cases}
    \mu(t^{-1}\Delta), & \quad t \in \mathbb{R}^{\times}, \\
    \mu(\mathbb{R}^{n}) \cdot \delta_{0}(\Delta), & \quad t = 0,
    \end{cases}
\qquad \Delta \in \mathfrak{B}^{n}.
\end{equation}
Note that this definition is well-defined, for the continuity of the map $x \mapsto tx$ implies its Borel-measurability, hence $t^{-1}\Delta \in \mathfrak{B}^{n}$ for $\Delta \in \mathfrak{B}^{n}$. The coefficient $\mu(\mathbb{R}^{n})$ multiplied to the delta measure for $t=0$ is to keep the total evaluation $\mu_{t}(\mathbb{R}^{n}) = \mu(\mathbb{R}^{n})$ constant for all $t \in \mathbb{R}$.
Intuitively speaking, this parametrisation allows us to narrow down the profile of a given complex measure $\mu$ while keeping its total evaluation $\mu_{t}(\mathbb{R}^{n}) = \mu(\mathbb{R}^{n})$ intact, so that it `tends' in an intuitive way to the delta measure (weighted by its total evaluation $\mu(\mathbb{R}^{n})$) as $t \to 0$. To help visualise this, suppose that $\mu$ is absolutely continuous and write $\rho := d\mu/d\beta^{n}$ for simplicity. One then finds
\begin{align}
\mu(t^{-1}\Delta)
    &= \int_{(t^{-1}\Delta)} \rho(x)\ d\beta^{n}(x) \nonumber \\
    &= \int_{\mathbb{R}^{n}} \chi_{\Delta}(tx) \rho(x)\ d\beta^{n}(x) \nonumber \\
    &= \int_{\mathbb{R}^{n}} \chi_{\Delta}(x) \cdot |t|^{-n}\rho\left(\frac{x}{t}\right)\ d\beta^{n}(x), \nonumber \\
    &= \int_{\Delta} \rho_{t}(x)\ d\beta^{n}(x), \quad \Delta \in \mathfrak{B}^{n},\ t \in \mathbb{R}^{\times},
\end{align}
where we have introduced the scaling
\begin{equation}\label{def:function_scaling}
f_{t}(x) := {|t|^{-n}}\, f\left(\frac{x}{t}\right), \qquad t \in \mathbb{R}^{\times},
\end{equation}
of any given integrable function $f \in L^{1}(\mathbb{R}^{n})$ by $t \in \mathbb{R}^{\times}$. This implies that $\mu_{t}$ is also absolutely continuous for each $t \in \mathbb{R}^{\times}$ by definition, and that its density is given by $\rho_{t}$, {\it i.e.},
\begin{equation}\label{eq:density_scaling}
\frac{d\mu_{t}}{d\beta^{n}} = \left(\frac{d\mu}{d\beta^{n}}\right)_{t}, \quad t \in \mathbb{R}^{\times},
\end{equation}
where the l.~h.~s. is the density of the scaled probability measure $\mu_{t}$, and the r.~h.~s. is the density of the original probability measure $\mu$ scaled by $t$ as in \eqref{def:function_scaling}. In the special case where $\mu$ is a probability measure, one may intuitively see that the parametrisation \eqref{def:function_scaling} takes any non-negative integrable function with the total integral of unity ({\it i.e.}, a probability density function) to the `delta function' in the limit $t \to 0$.

\paragraph{Von Neumann Interaction and Convolution}
Now, note here that for each $t \in \mathbb{R}^{\times}$, the probability measure $\mu_{t}$ is nothing but the image measure \eqref{def:Bildmass} of $\mu$ with respect to the map $x \mapsto tx$ ({\it i.e.}, multiplication by $t$). With the help of the change of variables formula for image measures \eqref{eq:Transformationsformel_Bildmass}, one confirms that the equality
\begin{equation}\label{eq:Transformationsformel_Skalierung}
\int_{\mathbb{R}^{n}}f(x)\ d\mu_{t}(x) = \int_{\mathbb{R}^{n}}f(tx)\ d\mu(x), \quad t \in \mathbb{R}
\end{equation}
holds for all $f$ that is integrable with respect to $\mu$. This allows us to rewrite \eqref{eq:ups_prob_dens_func} in terms of convolution as
\begin{align}\label{eq:outcome_prob01}
\mu_{Q}^{\psi^{g}}
    &= \mu_{Q}^{\psi} \ast \left(\mu_{A}^{\phi}\right)_{g}, \quad  g \in \mathbb{R}.
\end{align}
Alternatively, by scaling $\Delta \to g\Delta$ in \eqref{eq:ups_prob_dens_func},
one finds from the definition that
\begin{equation}\label{eq:outcome_prob02}
\left(\mu_{Q}^{\psi^{g}}\right)_{g^{-1}} = \left(\mu_{Q}^{\psi}\right)_{g^{-1}} \ast \mu_{A}^{\phi}, \quad  g \in \mathbb{R}^{\times},
\end{equation}
which is another way to describe how the von Neumann type interaction causes a change in the profile of the meter observable $X=Q$.

\paragraph{Scaling of Observables}

We make a short digression at this point to seek for the physical meaning of the two findings \eqref{eq:outcome_prob01} and \eqref{eq:outcome_prob02}, which we have just acquired.
To prepare for our argument, we first introduce some notations regarding scaling of spectral measures, in parallel to that of complex measures as we have done before. Let $E : \mathfrak{B}^{n} \to L(\mathcal{H})$ be an $n$-dimensional spectral measure on the Hilbert space $\mathcal{H}$, and define a parametrised family $\{E_{t}\}_{t \in \mathbb{R}}$ of spectral measures by
\begin{equation}
E_{t}(\Delta) :=
    \begin{cases}
        E(t^{-1}\Delta), & \quad t \in \mathbb{R}^{\times}, \\
        E_{0}(\Delta), & \quad t = 0,
    \end{cases}
    \qquad \Delta \in \mathfrak{B}^{n}.
\end{equation}
Here, we have introduced the `delta spectral measure' $E_{0}$ centred at $0 \in \mathbb{R}^{n}$, defined by
\begin{equation}\label{def:delta_spectral_meas_0}
E_{0}(\Delta) :=
    \begin{cases}
        I, & \quad 0 \in \Delta, \\
        0, & \quad 0 \notin \Delta,
    \end{cases}
\qquad \Delta \in \mathfrak{B}^{n}.
\end{equation}
Incidentally, for the one-dimensional case ($n=1$), the delta spectral measure $E_{0}$ centred at the origin is nothing but the spectral measure accompanying the zero operator $0$ on $\mathcal{H}$.

We next confirm some basic facts regarding scaling of observables and their accompanying spectral measures. Let $E_{A}$ be the spectral measure of a self-adjoint operator $A : \mathcal{H} \supset \mathrm{dom}(A) \to \mathcal{H}$. The goal is to specify the spectral measure of the scaled self-adjoint operator $tA$, ($t \in \mathbb{R}$) and to show that
\begin{equation}\label{eq:spectral_measure_scaled_obs}
E_{(tA)} = \left(E_{A}\right)_{t}, \quad t \in \mathbb{R},
\end{equation}
where the l.~h.~s. is the desired spectral measure accompanying the scaled operator $tA$, whereas the r.~h.~s. is the spectral measure accompanying the operator $A$ scaled by $t$. To see this, first observe the following equality
\begin{align}
\langle \phi, (tA) \phi \rangle
    &= \int_{\mathbb{R}} ta\ d\mu_{A}^{\phi}(a) \nonumber \\
    &= \int_{\mathbb{R}} a\ d(\mu_{A}^{\phi})_{t}(a) \nonumber \\
    &= \int_{\mathbb{R}} a\ d\langle \phi, (E_{A})_{t}(a) \phi \rangle, \quad t \in \mathbb{R}
\end{align}
for the choice $|\phi\rangle \in \mathrm{dom}(A)$, where we have used \eqref{eq:Transformationsformel_Skalierung} to obtain the second to last equality. Applying the polarisation identity \eqref{def:polarisation_identity} for $T = tA$, one then has
\begin{equation}
\langle \phi^{\prime}, (tA) \phi \rangle = \int_{\mathbb{R}} a\ d\langle \phi^{\prime}, (E_{A})_{t}(a) \phi \rangle, \quad t \in \mathbb{R}
\end{equation}
for any $|\phi\rangle, |\phi^{\prime}\rangle \in \mathrm{dom}(A)$. Observing that the domain of a self-adjoint operator is dense in $\mathcal{H}$ by definition, one may continuously extend the above equality on $|\phi^{\prime}\rangle \in \mathcal{H}$, based on which the uniqueness of the spectral measure leads to the desired result \eqref{eq:spectral_measure_scaled_obs}.

Returning to our main line of arguments, we first observe that the equality \eqref{eq:spectral_measure_scaled_obs} leads to
\begin{equation}
\mu_{(tA)}^{\phi} = \left(\mu_{A}^{\phi}\right)_{t}, \quad t \in \mathbb{R},
\end{equation}
which states that the probability measure describing the ideal measurement outcome of the scaled observable $tA$ on the state $|\phi\rangle \in \mathcal{H}$ coincides with that of the original observable $A$ scaled by $t$. Armed with this result, one may reformulate our previous findings \eqref{eq:outcome_prob01} and \eqref{eq:outcome_prob02} respectively as
\begin{equation}\label{eq:outcome_prob_mod_01_proof}
\left\{
\begin{split}
\mu_{Q}^{\psi^{g}}
   &= \mu_{Q}^{\psi} \ast \mu_{(gA)}^{\phi}, \\
\mu_{P}^{\psi^{g}}
   &= \mu_{P}^{\psi},
\end{split}
\qquad g \in \mathbb{R},\right.
\end{equation}
and
\begin{equation}\label{eq:outcome_prob_mod_02_proof}
\left\{
\begin{split}
\mu_{(g^{-1}Q)}^{\psi^{g}}
   &= \mu_{(g^{-1}Q)}^{\psi} \ast \mu_{A}^{\phi}, \\
\mu_{(gP)}^{\psi^{g}}
   &= \mu_{(gP)}^{\psi},
\end{split}
\qquad g \in \mathbb{R}^{\times},\right.
\end{equation}
where we have also explicitly written down the profile of the outcome of the measurement of $X=P$. This completes our proof for Proposition~\ref{prop:UCM_II} and Corollary~\ref{cor:UCM_II}.

\subsection{Recovery of the Target Profile}

We now consider the inverse problem of what we have discussed so far, that is, we argue how one can recover the probability measure $\mu_{A}^{\phi}$ of the target observable $A$ from the probability measure $\mu_{Q}^{\psi^{g}}$ obtained through the measurement of $Q$ on the meter system.  Following the same line in the previous section, one finds it useful to probe either the strong or the weak region of the interaction for this purpose, which we shall see below one by one.

\subsubsection{Strong Unconditioned Measurement}\label{sec:Strong_Unconditioned_Measurement}

We first concentrate on \eqref{eq:outcome_prob_mod_02} (or equivalently \eqref{eq:outcome_prob02}), and observe that the problem of recovering the desired probability measure reduces to the problem of `deconvolution', where one wishes to find the solution $\mu := \mu_{A}^{\phi}$ of the equation of the form
\begin{equation}\label{eq:invprob}
\nu_{\mathrm{out}} = \nu_{\mathrm{in}} \ast \mu,
\end{equation}
having knowledge and control over both the `input' $\nu_{\mathrm{in}} := \mu_{(g^{-1}Q)}^{\psi}$ and `output' $\nu_{\mathrm{out}} := \mu_{(g^{-1}Q)}^{\psi^{g}}$ on their respective sides. 
Whilst there is rich literature on the topic of deconvolution, we take a specific approach to the solution in order to make our arguments simple. 

\paragraph{Main Objective of this Passage}
A quick observation leads us to a na{\"i}ve expectation that, if one could attune the input so that $\nu_{\mathrm{in}}$ may become a multiplicative identity (in our case, it is the delta measure $\delta_{0}$ centred at the origin), or in the case where this is impossible, if one gradually approximates the input close enough to it, then, one may obtain the desired solution $\mu$ directly as the measured output $\nu_{\mathrm{out}} \to \delta_{0} \ast \mu = \mu$. One of the typical manners in which we attain such gradual approximation would be to fix the initial state $\psi$ and taking the strong limit $g^{-1} \to 0$ ($g \to \pm \infty$) of the interaction parameter, so that $\nu_{\mathrm{in}} = (\mu_{Q}^{\psi})_{g^{-1}}$ `tends' towards the desired identity $\delta_{0}$ in an intuitive manner (recall \eqref{def:measure_scaling} and \eqref{def:function_scaling}). The main objective of this passage is to confirm that this idea is indeed valid, and thus to state it in a mathematically rigorous way.

As it becomes apparent through the line of discussions below, there are some certain mathematical hurdles that must be overcome to achieve this objective. In order to avoid much intricacies, we shall impose certain condition to the choice of the target observable, and present our main result in the following way:
\begin{proposition}[Strong Unconditioned Measurement]
In the context of the UM scheme, suppose that
\begin{enumerate}
\item the target observable $A$ admits description by density functions,
\item the initial profile $\mu_{Q}^{\psi}$ of the meter observable $Q$ on the state $|\psi\rangle$ is compactly supported%
\footnote{We say that a complex measure $\nu$ has a compact support if there exists a compact subset $K \subset \mathbb{R}$ for which the restriction of the variation $|\nu|$ on the complement $|\nu||_{K^{c}} = 0$ is a zero measure.}.
\end{enumerate}
Then, the scaled profile of $Q$ after the interaction converges to the desired target in the strong limit of interaction
\begin{equation}
\lim_{g \to \pm \infty} \left\| \mu_{(g^{-1}Q)}^{\psi^{g}} - \mu_{A}^{\phi} \right\| = 0
\end{equation}
with respect to the total variation norm (or, equivalently the $L^{1}$-norm) for any choice of the initial states $|\phi\rangle \in \mathcal{H}$.
\end{proposition}
\noindent
The remainder of this passage is devoted to its demonstration.

\paragraph{Preliminary Observations}
Let us make a preliminary observation following the above idea. The first thing we realise is that, in general, we cannot prepare the input $\nu_{\mathrm{in}}$ so that its profile may exactly coincide with the multiplicative identity $\delta_{0}$. To see this quickly, first recall
that the realisable input probability measures $\nu_{\mathrm{in}}$ are exactly those that are absolutely continuous with respect to the Lebesgue-Borel measure. Since the delta measure $\delta_{0}$ does not belong to the space $L^{1}(\mathfrak{B})$, one concludes that it is impossible to prepare the input in such a way that $\nu_{\mathrm{in}} = \delta_{0}$ holds.
An alternative approach to this problem may be to consider a sequence of inputs $(\nu_{\mathrm{in}})_{n}$ that tends to the delta measure $\delta_{0}$ in hope that the resultant sequence of multiplicative products $(\nu_{\mathrm{out}})_{n} := (\nu_{\mathrm{in}})_{n} \ast \mu$ also converges towards the desired solution $\mu$ in the limit. Indeed, if one could only construct a sequence $(\nu_{\mathrm{in}})_{n}$ so that
\begin{equation}\label{eq:input_convergence}
\lim_{n \to \infty} \|(\nu_{\mathrm{in}})_{n} - \delta_{0}\| = 0,
\end{equation}
under the total variation norm, one concludes from the evaluation
\begin{align}\label{eq:continuity_convolution}
\|(\nu_{\mathrm{in}})_{n} \ast \mu - \mu\| = \|((\nu_{\mathrm{in}})_{n} - \delta_{0}) \ast \mu \| \leq \|(\nu_{\mathrm{in}})_{n} - \delta_{0}\| \cdot \| \mu \|
\end{align}
that the outcome tends to the desired solution
\begin{equation}\label{eq:output_convergence}
\lim_{n \to \infty} \|(\nu_{\mathrm{out}})_{n} - \mu\| = 0
\end{equation}
in the limit. Unfortunately, however, one immediately realises that this idea also fails, since in general there is no such sequence $(\nu_{\mathrm{in}})_{n}$ that meets the condition \eqref{eq:input_convergence} in the first place, for indeed, since the space $L^{1}(\mathfrak{B})$ of absolutely continuous complex measures is a topologically closed subset of the measure algebra $\mathbf{M}_{\mathbb{C}}(\mathfrak{B})$, a sequence in $L^{1}(\mathfrak{B})$ never converges to an element outside of $L^{1}(\mathfrak{B})$ with respect to the total variation norm.

\paragraph{Discussion on the possible Approaches}

From the quick overview of our current situation, we learn that the problem at hand is to do with the \emph{topology} we have given to the measure algebra $\mathbf{M}_{\mathbb{C}}(\mathfrak{B})$. 
Namely, the topology induced from the total variation norm is too strong (fine) for our convenience. 
A fundamental cure for this would thus be to equip the space with a weaker (coarser) topology on $\mathbf{M}_{\mathbb{C}}(\mathfrak{B})$ such that, at least, it may allow us to construct sufficiently abundant sequences (or nets, in general) of the `inputs' in $L^{1}(\mathfrak{B})$ that converges towards $\delta_{0}$, and that the sequence of the resulting `outputs' ({\it i.e.}, the multiplicative product \eqref{eq:invprob}) would subsequently converge towards the desired solution in the limit%
\footnote{
A straightforward candidate for such a topology would be the weak-$\ast$ topology based on the identification \eqref{eq:RMK-thm} by the Riesz-Markov-Kakutani representation theorem, namely, the initial topology with respect to the family of all algebraic linear functionals of the form $\mu \mapsto \int_{\mathbb{R}} f d\mu$, where $f \in C_{0}(\mathbb{R})$. 
One eventually finds that the norm topology of the total variation is nothing but the strong topology with respect to the identification, which implies that the weak-$\ast$ topology is strictly weaker than the topology we currently have at hand. Moreover, direct application of the dominated convergence theorem and Fubini's theorem reveals that the convergence of a sequence of probability measures $\nu_{n} \to \delta_{0}$ implies $\nu_{n} \ast \mu \to \mu$ (both the convergence is meant in weak-$\ast$), which is a much cleaner result than what we have seen in the main paragraphs. As an example of such a sequence (net) of probability measures converging towards $\delta_{0}$, one finds that the scaling $\nu_{t}$ \eqref{def:measure_scaling} of a given probability measure $\nu$ is typical. In fact, the scaling becomes a continuous parametrisation from $\mathbb{R}$ to $\mathbf{M}_{\mathbb{C}}(\mathfrak{B})$ under the topology, which is also a welcome property.
}.

However, since this strategy, while being desirable, presupposes moderate familiarity with the mathematical branch of general topology, which the authors have deemed to be beyond the scope of this paper, an alternative approach to the problem without explicit exposure to it would be favourable (possibly at the cost of generality, while hopefully having the merit of being mathematically less demanding). In this paper, this would be accomplished by introducing an auxiliary concept of `approximate identities', whose definition would be shortly presented.
In essence, we focus only on the convergence of the output in the total variation norm, based on the observation that, even though there is \emph{no} sequence of the input that converges to the delta measure \eqref{eq:input_convergence}, there are certain conditions in which the sequence of the output \emph{do} converge towards the desired solution \eqref{eq:output_convergence}. As a preliminary observation to this approach, observe that the output $\nu_{\mathrm{out}}$ also necessarily lies in $L^{1}(\mathfrak{B})$%
\footnote{To see this, recall that the output can be written as a multiplicative product of two probability measures with one of which being absolutely continuous, and that the space $L^{1}(\mathfrak{B})$ of absolutely continuous complex measures is an ideal in $\mathbf{M}_{\mathbb{C}}(\mathfrak{B})$.},
and by recalling that $L^{1}(\mathfrak{B})$ is closed under the topology induced by the total variation norm, one finds that the candidates of the solution $\mu$ towards which the sequence of outputs could ever converge are only those that also lie in $L^{1}(\mathfrak{B})$. Based on this inspection, in what follows, we shall only treat the case in which \emph{the target observable $A$ admits a description by density functions},
which is to say that the solutions $\mu = \mu_{A}^{\phi}$ are always guaranteed to lie in $L^{1}(\mathfrak{B})$, is assumed. 

\paragraph{Approximate Identities}

The convolution algebra $L^{1}(\mathfrak{B}) \cong L^{1}(\mathbb{R}^{n})$, contrasted to the measure algebra $\mathbf{M}_{\mathbb{C}}(\mathfrak{B})$,  is non-unital.
In order to compensate the inconvenience arising from the lack of a multiplicative identity, a weaker concept is often used in analysing problems involving algebras.
In this paper, we call a family $\{e_{t}\}_{t>0}$ of elements of $L^{1}(\mathbb{R}^{n})$ an \emph{approximate identity}, if for every element $f\in L^{1}(\mathbb{R}^{n})$, the convolution $e_{t} \ast f$ converges to $f$ in the topology induced by the $L^{1}$-norm, {\it i.e.},
\begin{equation}\label{eq:appident}
\lim_{t \to 0} \|e_{t} \ast f - f\|_{1} = 0, \quad f \in L^{1}(\mathbb{R}^{n}).
\end{equation}

Before we move on to the construction of an example, we collect some necessary terminologies. Recall that the support of a function $f: \mathbb{R}^{n} \to \mathbb{K}$ is a subset of $\mathbb{R}^{n}$ defined by
\begin{equation}
\mathrm{supp}(f) := \overline{\{x \in \mathbb{R}^{n} : f(x) \ne 0\}},
\end{equation}
where the overline on a set denotes its topological closure. A support of a function $f: \mathbb{R}^{n} \to \mathbb{K}$ is said to be compact if $\mathrm{supp}(f)$ is bounded.
Now, let $\eta \in L^{1}(\mathbb{R}^{n})$ be any integrable function possessing a compact support with the total integration of unity,
\begin{equation}
\int_{\mathbb{R}^{n}} \eta(x)\ dx = 1.
\end{equation} 
With this, consider a family $\{\eta_t\}_{t \in \mathbb{R}^{\times}}$ of scaled functions defined as in \eqref{def:function_scaling},
which preserve the total integration of unity for all $t \in \mathbb{R}^{\times}$.
One may then intuitively expect that $\eta_t$ tends to the `delta function' in the limit $t \to 0$ and can be used for an approximate identity,
\begin{equation}\label{ineq:approx_ident}
\lim_{t \to 0} \|\eta_t \ast f - f \|_{1} = 0,
\end{equation}
for all $f \in L^{1}(\mathbb{R}^{n})$. 
To confirm that this is indeed the case, observe the inequality
\begin{align}
\|\eta_t \ast f - f \|_{1} 
    &:= \int_{\mathbb{R}^{n}} \left| \left( \int_{\mathbb{R}^{n}} \eta_t(y) f(x-y) \ d\beta^{n}(y) \right) - f(x) \right|\ d\beta^{n}(x) \nonumber \\
    &= \int_{\mathbb{R}^{n}} \left| \int_{\mathbb{R}^{n}} \eta_t(y) (f(x-y) - f(x))\ d\beta^{n}(y) \right|\ d\beta^{n}(x) \nonumber \\
    &\leq \int_{\mathbb{R}^{n}} |\eta(y)| \left( \int_{\mathbb{R}^{n}}  |f(x-ty) - f(x)|\ d\beta^{n}(x) \right)\ d\beta^{n}(y) \nonumber \\
    &= \int_{\mathbb{R}^{n}} |\eta(y)| \cdot \|\tau_{(-ty)}f- f\|_{1}\ d\beta^{n}(y),
\end{align}
where $\tau_{a}$ is the translation operator defined by
\begin{equation}\label{eq:translation}
\tau_{a}f(x) := f(x + a).
\end{equation}
Recalling that $\lim_{a \to 0} \|\tau_{a}f- f\|_{1} = 0$ for any $f \in L^{1}(\mathbb{R}^{n})$, we see that for any $\epsilon > 0$, there exists a $\delta >0$ for which $a \in K_{\delta}(0) := \{x \in \mathbb{R}^{n} : |x| < \delta \}$ leads to $\|\tau_{a}f- f\|_{1} < \epsilon$. By taking $|t|$ small enough so that $\mathrm{supp}(\eta) \subset K_{t^{-1}\delta}(0)$, we find that the r.~h.~s of the above inequality is less than $\epsilon$. This shows that the family defined by
\begin{equation}
e_{t} := \eta_{(\pm t)}, \quad t > 0,
\end{equation}
makes a simple example of approximate identities (here, the meaning of the subscript on both sides of the equation is not to be confused, where the subscript on the l.~h.~s. indicates an index of the elements of the convolution algebra $L^{1}(\mathbb{R}^{n})$, whereas that on the r.~h.~s. indicates the scaling parameter of an integrable function $\eta$ defined in \eqref{def:function_scaling}).
Obviously, the construction of such approximate identities is highly non-unique, and one may attain it in various different ways.

\paragraph{Realisation of Approximate Identities}
Our observation so far revealed that, as long as the target profile $\mu \in L^{1}(\mathfrak{B})$ is absolutely continuous, by considering the family of inputs $\{(\nu_{\mathrm{in}})_{t}\}_{t > 0}$
in such a way that it makes an approximate identity in $L^{1}(\mathfrak{B})$, the resulting family of outputs $(\nu_{\mathrm{out}})_{t} := (\nu_{\mathrm{in}})_{t} \ast \mu$ would successfully converge to the desired solution
\begin{equation}
\lim_{t \to 0} \|(\nu_{\mathrm{out}})_{t} - \mu \| = 0
\end{equation}
in the $L^{1}$-norm (or equivalently, in the total variation norm)%
\footnote{We note again that the subscripts $t$ used here is meant to be an index, and not to be confused with that denoting scaling of complex measures.}.
We are now interested in the construction of such approximate identities for our current situation.
To this, we first observe that, since the profile of the input $\nu_{\mathrm{in}} = \mu_{(g^{-1}Q)}^{\psi}$ in our case is exclusively determined by the choice of the interaction parameter $g$ and the initial state $|\psi\rangle \in \mathcal{K}$ of the meter system, the problem reduces to finding a sequence of the pair
$(g, |\psi\rangle)_{t}$, $t>0$ that makes the input an approximate identity. As an example of such a construction, we first fix the initial state $|\psi\rangle$ and observe that the density of the input is given by
\begin{equation}
\frac{d\mu_{(g^{-1}Q)}^{\psi}}{d\beta} = \left(\frac{d\mu_{Q}^{\psi}}{d\beta}\right)_{g^{-1}}, \quad g \in \mathbb{R}^{\times},
\end{equation}
where we have used our previous result \eqref{eq:density_scaling}. Then, choosing $|\psi\rangle$ so that the density of $\mu_{Q}^{\psi}$ may be compactly supported, one realises that taking the strong limit of the interaction $g^{-1} \to 0$ (or equivalently $g \to \pm \infty$) yields the desired result. In turn, we fix the interaction parameter $g \in \mathbb{R}^{\times}$ and choose a sequence of initial states that makes the corresponding probability measures an approximate identity. Since the scaling of an approximate identity by $g^{-1}$ is still an approximate identity, one achieves another example of such a construction.

One thus finds a general guiding principle for the construction of an approximate identity to be the combination of the two manoeuvres, namely, either
\begin{itemize}
\item by taking the strong limit of the interaction $g^{-1} \to 0$,
\item by narrowing down the profile of the probability measures to the delta measure (symbolically $\mu_{Q}^{\psi} \to \delta_{0}$) by changing the meter state $|\psi\rangle \in \mathcal{K}$.
\end{itemize}
In order to explicitly see how these work together, choose a sequence of initial states $|\psi\rangle$, $|\psi_{h}\rangle \in \mathcal{K}$, $h >0$, such that the density of the initial profile $\mu_{Q}^{\psi}$ is compactly supported and that the parametrisation corresponds to its scaling
\begin{equation}\label{eq:approx_ident_state}
\frac{d\mu_{Q}^{\psi_{h}}}{d\beta} = \left( \frac{d\mu_{Q}^{\psi}}{d\beta} \right)_{h},
\end{equation}
which makes itself an approximate identity as $h \to 0$
(one may easily construct such a sequence in the special case in which the meter system is described in the Schr{\"o}dinger representation of the CCR%
\footnote{One may choose any wave-function $\psi \in L^{2}(\mathbb{R})$ with compact support,
and define
\begin{equation}\label{eq:approx_ident_state_Sch}
\psi_{(h)}(x) := |h|^{-1/2}\psi\left( \frac{x}{h} \right).
\end{equation}
Here, the braces among the subscript $h$ to denote the index is merely employed in order to avoid confusion with that denoting scaling of a function \eqref{def:function_scaling}. One then readily finds that this qualifies as an example of the desired family \eqref{eq:approx_ident_state}.}).
 Then, observing that the scaling of it by $g^{-1}$ is
\begin{equation}\label{eq:ups_g_and_h}
\left(\frac{d\mu_{Q}^{\psi_{h}}}{d\beta}\right)_{g^{-1}} = \left(\frac{d\mu_{Q}^{\psi}}{d\beta}\right)_{hg^{-1}},
\end{equation}
one finds that it is indeed an approximate identity that tends to the delta in the limit as $hg^{-1} \to 0$ together.

\paragraph{Concluding Remarks}
In conclusion, we see that the UM scheme allows us to recover the information of the target system and its observable $A$, not only in the form of expectation values described earlier, but also in the form of probability measures $\mu_{A}^{\phi}$. This is accomplished by taking the limit of either narrowing the profile of the probability measure $\mu_{Q}^{\psi}$ of the meter system, or intensifying the interaction parameter $g \to \pm \infty$, or otherwise by appropriately balancing both contributions and having $hg^{-1} \to 0$ as a whole. In this sense, we may say that intensifying the interaction parameter has an equivalent role to narrowing the profile of the probability measure of the meter.
It may thus appear reasonable that, also in this respect, the von Neumann measurement scheme is sometimes referred to as the `strong measurement' or the `sharp measurement'.

\subsubsection{Weak Unconditioned Measurement}\label{sec:ups_II_wups}

We shall see next how the measurement outcome of the UM scheme behaves locally around $g=0$ in terms of probability measures. Specifically, we are interested in the (higher-order) derivatives of the map
\begin{equation}\label{eq:param_to_outcome}
\mathbb{R} \to L^{1}(\mathfrak{B}), \ g \mapsto \mu_{Q}^{\psi^{g}},
\end{equation}
which is now a map from the real line $\mathbb{R}$ to the space of complex measures $L^{1}(\mathfrak{B}) \subset \mathbf{M}_{\mathbb{C}}(\mathfrak{B})$.

\paragraph{Main Objective of this Passage}

The main objective of this passage is to first compute the derivatives of the map \eqref{eq:param_to_outcome} at the origin $g=0$, and subsequently argue how one may reconstruct the profile of the probability measure $\mu_{A}^{\phi}$ of our interest from the information obtained. However, as one realises in the line of discussion that follows, this involves certain mathematical intricacies. In order to avoid any difficulties and complication that may arise, we impose some restrictions to the configuration of the target and meter systems, and thus obtain the following two propositions, the first of which shall be demonstrated in the main passages below.
\begin{proposition}[Outcome of the Weak Unconditioned Measurement]\label{prop:WUCM_II}
In the context of the UM scheme, suppose that
\begin{enumerate}
\item the target profile $\mu_{A}^{\phi}$ is compactly supported,
\item the density of $\mu_{Q}^{\psi}$ belongs to the Schwartz space $d\mu_{Q}^{\psi}/d\beta \in \mathscr{S}(\mathbb{R})$.
\end{enumerate}
Then, the map \eqref{eq:param_to_outcome} is arbitrarily many times strongly differentiable in the $L^{1}$-norm (or, equivalently, in the total variation norm), and its derivatives at $g=0$ reads
\begin{equation}\label{eq:main_result_wupsm}
\left. \frac{d^{n}}{dg^{n}} \mu_{Q}^{\psi^{g}} \right|_{g=0} = \mathbb{E}[A^{n};\phi] \cdot (-D)^{n} \mu_{Q}^{\psi}, \quad n \in \mathbb{N}_{0},
\end{equation}
where $D$ denotes the operation uniquely specified through the relation
\begin{equation}\label{def:diff_measure}
d(D \nu) / d\beta := D (d\nu / d\beta),
\end{equation}
by differentiating the density of absolutely continuous complex measures $\nu \in L^{1}(\mathfrak{B})$ whose density $d\nu/d\beta \in \mathscr{S}(\mathbb{R})$ lies in the Schwartz space.
\end{proposition}
\noindent
Note that compactness of the support of $\mu_{A}^{\phi}$ implies the existence of all the higher-order moments $|\, \mathbb{E}[A^{n};\phi]\,| < \infty$ of the observable $A$, and that the Schwartz space is closed under the operation of differentiation ({\it i.e.}, $D^{n} (d\nu / d\beta) \in \mathscr{S}(\mathbb{R})$), hence both sides of \eqref{eq:main_result_wupsm} is well-defined.
Operationally, the above proposition implies that one may obtain not only the expectation value ($n=1$) of $\mu_{A}^{\phi}$, as we have found by the operator level analysis \eqref{eq:ups_weak_I} conducted in the previous section, but also its higher-order moments
\begin{equation}\label{eq:higher-order_moments_A}
\mathbb{E}[A^{n};\phi] = \int_{\mathbb{R}} a^{n}\ d\mu_{A}^{\phi}(a), \quad n \in \mathbb{N}_{0},
\end{equation}
by probing the local behaviour of the interaction around $g=0$. Incidentally, one might expect that one could recover the full profile of the original probability measure $\mu_{A}^{\phi}$ by knowing enough numbers of its higher-order moments, which in fact turns out to be positive under our assumption.
\begin{proposition}[Weak Unconditioned Measurement]
Let $A$ be self-adjoint and $|\phi\rangle \in \mathcal{H}$ for which the probability measure $\mu_{A}^{\phi}$ is compactly supported. Given another compactly supported probability measure $\mu$ on $(\mathbb{R}, \mathfrak{B})$ such that all their higher moments
\begin{equation}
\mathbb{E}[A^{n};\phi] = \int_{\mathbb{R}} a^{n}\ d\mu(a), \quad n \in \mathbb{N}_{0},
\end{equation}
coincide with those of $\mu_{A}^{\phi}$, then the two probability measures agree $\mu = \mu_{A}^{\phi}$. In other words, one may uniquely reconstruct the probability measure $\mu_{A}^{\phi}$ of the target system by knowing all the higher moments of $A$ by means of the weak UM.
\end{proposition}
\begin{proof}
In fact, this is one instance of the famous problems collectively called the classical \emph{moment problem} \cite{Hausdorff_1921_a,Hausdorff_1921_b}. We provide a sketch of the proof for our specific case at hand, and to this, we first observe that knowing all the higher-order moments \eqref{eq:higher-order_moments_A} is equivalent to knowing the integral $\int p(a)\, d\mu_{A}^{\phi}(a)$ of all polynomials $p \in P(K)$ on some compact subset $K \subset \mathbb{R}$ on which $\mu_{A}^{\phi}$ is supported. Now, choose a compact subset $K \subset \mathbb{R}$ that contains the support of both $\mu_{A}^{\phi}$ and $\mu$, {\it i.e.}, $\mu|_{K^{c}} = \mu_{A}^{\phi}|_{K^{c}} = 0$, and observe that the space of continuous functions on $K$ trivially coincide with that of continuous functions on $K$ that vanishes at infinity $C(K) = C_{0}(K)$. We thus have $C(K)^{\prime} = C_{0}(K)^{\prime} \cong \mathbf{M}_{\mathbb{C}}(\mathfrak{B}|_{K})$ by the Riesz-Markov-Kakutani representation theorem. Since the space of polynomials $P(K)$ is dense in $C(K)$ with respect to the supremum norm ({\it cf.} Stone-Weierstra{\ss} approximation theorem), one concludes that $\int_{\mathbb{R}} p(a)\, d\mu(a) = \int p(a)\, d\mu_{A}^{\phi}(a)$, $p \in P(K)$ implies $\mu = \mu_{A}^{\phi}$.
\end{proof}

\paragraph{Preliminary Observation}
We now begin our analysis. To provide some preliminary observation to this problem, we start by observing that the target of our study would be the following formal expression
\begin{equation}\label{eq:wups_prob_formal}
\left. \frac{d}{dg} \mu_{Q}^{\psi^{g}} \right|_{g=0} := \lim_{g \to 0} \frac{\mu_{Q}^{\psi^{g}} - \mu_{Q}^{\psi}}{g},
\end{equation}
in which we leave aside, just for now, all the inherent subtleties that will shortly become apparent regarding the operation of taking the limit.
Now, since the numerator of the r.~h.~s. of the above formula can be written as
\begin{equation}
\mu_{Q}^{\psi^{g}} - \mu_{Q}^{\psi} = \mu_{Q}^{\psi} \ast \left( \left(\mu_{A}^{\phi}\right)_{g} - \delta_{0} \right),
\end{equation}
one finds that the analysis of \eqref{eq:wups_prob_formal} reduces to the study of the formal expression of the form
\begin{equation}\label{eq:wups_form}
\nu_{\mathrm{out}}^{\prime}(0) := \left. \frac{d}{dt}\nu_{\mathrm{out}}(t) \right|_{t=0} = \lim_{t \to 0}\ \nu_{\mathrm{in}} \ast \frac{\mu_{t} - \delta_{0}}{t},
\end{equation}
where $\mu, \nu_{\mathrm{in}} \in \mathbf{M}_{\mathbb{C}}(\mathfrak{B})$ are probability measures (the latter being absolutely continuous), $\nu_{\mathrm{out}}(t) := \nu_{\mathrm{in}} \ast \mu_{t}$, and the subscript on $\mu_{t}$ denotes the scaling defined in \eqref{def:measure_scaling}.
In studying \eqref{eq:wups_form}, one might find it a decent starting point to focus on the formal expression (the right component of the above convolution)
\begin{equation}\label{eq:diff_measure_scaling}
\lim_{t \to 0}\ \frac{\mu_{t} - \delta_{0}}{t} =: \mu_{0}^{\prime}.
\end{equation}
From this, one realises that our problem is nothing but the differentiability of the map $t \mapsto \mu_{t}$ at the origin $t=0$ (recall that we have defined $\mu_{0} := \delta_{0}$ for any probability measure $\mu$),
and thus have symbolically written the limit of the above expression by $\mu_{0}^{\prime}$, temporarily leaving aside the question of its existence and well-definedness just as before. It would then be tempting to expect
\begin{equation}\label{eq:diff_conv}
\nu_{\mathrm{out}}^{\prime}(0) = \nu_{\mathrm{in}} \ast \mu_{0}^{\prime},
\end{equation}
which should resolve our main problem fairly nicely.

\paragraph{A Formal Computation of the Derivative}
Guided by the above na{\"i}ve observation, we are naturally led to consider what the derivative of the map $t \to \mu_{t}$ at $t=0$ for a given probability measure $\mu \in \mathbf{M}_{\mathbb{C}}(\mathfrak{B})$ would look like. 
As a first step, suppose for simplicity that $\mu$ is absolutely continuous, and denote its density by $\eta := d\mu/d\beta$. Armed with our previous findings $\eta_{t} = d\mu_{t}/d\beta$, $t \in \mathbb{R}^{\times}$ regarding scaling of measures and that of its densities (see \eqref{eq:density_scaling}), we then intend to formally obtain
\begin{equation}
\mu_{0}^{\prime} = \lim_{t \to 0} \mu_{t}^{\prime}
\end{equation}
in view of density functions, by first computing its derivative at $t>0$ and then taking the limit $t \to 0$.
Now, assuming suitable differentiability and integrability conditions for the density $\eta$, one computes the derivative of the map $t \mapsto \eta_{t}$ at $t > 0$ as
\begin{align}
\lim_{h \to 0} \frac{\eta_{t + h} - \eta_{t}}{h}
    &= -\frac{1}{t^{2}} \eta\left(\frac{x}{t}\right) - \frac{x}{t^{3}}(D\eta)\left(\frac{x}{t}\right) \nonumber \\
    &= - D \left( \frac{1}{t}\frac{x}{t}\eta\left(\frac{x}{t}\right) \right) \nonumber \\
    &= - D\left(x\eta\right)_{t}, \quad t > 0,
\end{align}
where $D := d/dx$ was the usual operation of differentiation. Then, one might be tempted to formally proceed as
\begin{align}
\lim_{t \to 0} D\left(x\eta\right)_{t}
    &= D \left[ \lim_{t \to 0} \left(x\eta\right)_{t} \right] \nonumber \\
    &= D \left[ \left( \int_{\mathbb{R}} x\eta\ d\beta \right) \cdot \delta_{0} \right] \nonumber \\
    &= \mathbb{E}[x;\mu] \cdot D \delta_{0},
\end{align}
where we have used \eqref{def:measure_scaling} in the second equality.
The above argument implies that the derivative of the map $t \to \mu_{t}$ at the origin would appear as
\begin{equation}\label{eq:derivative_scaled_measure}
\mu_{0}^{\prime} = \mathbb{E}[x;\mu] \cdot (- D) \delta_{0},
\end{equation}
which is the `derivative of the delta measure' weighted by the expectation value of the original probability measure $\mu$. 
As for the general case in which the original probability measure $\mu$ is now not necessarily absolutely continuous, 
we may conjecture that, since the r.~h.~s. of \eqref{eq:derivative_scaled_measure} does not depend on the absolute continuity of the original probability measure $\mu$, the same result should hold even in the general case as well.

\paragraph{Discussion on the possible Approaches}
While we have conducted a very formal discussion above, the result in fact turns out to be true and can be made mathematically fully rigorous in the framework of the \emph{theory of generalised functions (distributions)}. In fact, it turns out that the derivative $\mu_{0}^{\prime}$ that appears in \eqref{eq:derivative_scaled_measure} is no longer a member of the space $\mathbf{M}_{\mathbb{C}}(\mathfrak{B})$ of complex measures%
\footnote{Incidentally, one may recall that the (higher-order) derivatives of the delta distribution appears in several branches of physics, one of the most familiar of which being presumably the theory of electromagnetism. The derivative of the delta distribution $D \delta_{0}$ is among the most well-known example of a distribution that cannot be expressed by a complex measure. In order to provide an intuitive reasoning with the tools at hand, let $\varphi$ be a smooth function with compact support ({\it i.e}, a test function) satisfying $(D\varphi)(0) = 1$. As a concrete example, one may take $\varphi(x) := x\varphi_{0}(x)$ with
\begin{equation}
\varphi_{0}(x) :=
    \begin{cases}
    e^{-\frac{1}{1-x^{2}}}  & (|x| < 1) \\
    0     &(|x| \geq 1).
    \end{cases}
\end{equation}
Defining a sequence of test functions by $\varphi_{n}(x) := n^{-1}\varphi(nx)$, $n \in \mathbb{N}^{\times}$, observe that the dominated convergence theorem necessarily implies $\lim_{n \to \infty} \int_{\mathbb{R}} \varphi_{n} d\mu = 0$ for any complex measure $\mu \in \mathbf{M}_{\mathbb{C}}(\mathfrak{B})$. On the other hand, with the help of an auxiliary smooth density function $\rho$ to symbolically express the delta distribution by the limit of its scaling $\delta_{0} = \lim_{t \to 0} \rho_{t}$, one may formally compute the integral of $\varphi_{n}$ weighted by the `density' $D \delta_{0}$ as
\begin{align}\label{eq:distributional_delta_not_a_measure}
\int_{\mathbb{R}} \varphi_{n}(x)\ (D\delta_{0})(x) d\beta(x)
    &= \lim_{t \to 0} \left( \int_{\mathbb{R}} \varphi_{n}(x)\ (D\rho_{t})(x) d\beta(x) \right) \nonumber \\
    &= \lim_{t \to 0} \left( -\int_{\mathbb{R}} (D\varphi_{n})(x)\ \rho_{t}(x) d\beta(x) \right) \nonumber \\
    &= -\int_{\mathbb{R}} (D\varphi_{n})(x)\ \delta_{0}(x) d\beta(x),
\end{align}
where we have used integration by parts to obtain the second equality. This implies $\lim_{n \to \infty} \int_{\mathbb{R}} \varphi_{n} (D\delta_{0})d\beta = \lim_{n \to \infty} -(D\varphi_{n})(0) = -1$, which would lead to a contradiction if $(D\delta_{0})$ were to be expressed by a complex measure.
},
and accordingly the framework in which we have been working so far ({\it i.e.}, the space of complex measures) is insufficient for our analysis.
For further study of the weak UM scheme, a preferable approach would thus be to expand our framework by introducing the space of distributions. While this method has a great merit in being able to conduct our analysis with decent generality (and in fact, distributions have their role, not just in this subsection, but also later in studying the quasi-joint-probability distributions in Section~\ref{sec:ps_II} and \ref{sec:qp_qo}), at the same time, it has a drawback in that it would be rather mathematically demanding, especially since the theory of distributions is build up on the results of general topology.

In view of this, an alternative approach to the problem without direct exposure to the theory of distributions would be favourable. To this end, recalling the idea employed in the previous subsection, we concentrate only on the differentiability of the multiplicative product \eqref{eq:wups_form}, setting aside the intricacies involving that of the map $t \mapsto \mu_{t}$ we have seen above. To see what we mean, we first expect, by combining \eqref{eq:diff_conv} and \eqref{eq:derivative_scaled_measure}, that the derivative of the map $t \mapsto \nu_{\mathrm{out}}(t)$ at the origin be written as
\begin{equation}
\nu_{\mathrm{out}}^{\prime}(0) = \mathbb{E}[x;\mu] \cdot \left( \nu_{\mathrm{in}} \ast (- D) \delta_{0} \right).
\end{equation}
Now, assuming suitable differentiability condition of the density $\rho_{\mathrm{in}} := d\nu_{\mathrm{in}}/d\beta$ of the imput $\nu_{\mathrm{in}}$ as a starting point, we employ an auxiliary smooth density function $\eta$ to symbolically express the delta distribution by the limit of its scaling $\delta_{0} = \lim_{t \to 0} \eta_{t}$ (a similar technique is used in \eqref{eq:distributional_delta_not_a_measure}) and formally obtain the `density' of the convolution $\nu_{\mathrm{in}} \ast D \delta_{0}$ as
\begin{align}
\left( \rho_{\mathrm{in}} \ast D\delta_{0}\right)(x)
    &= \left( \rho_{\mathrm{in}} \ast D\left( \lim_{t \to 0} \eta_{t} \right) \right)(x) \nonumber \\
    &= \lim_{t \to 0} \int_{\mathbb{R}} \rho_{\mathrm{in}}(x-y) \left(D\eta_{t}\right)(y)\ d\beta(y) \nonumber \\
    &= \lim_{t \to 0} \int_{\mathbb{R}} (D\rho_{\mathrm{in}})(x-y) \eta_{t}(y)\ d\beta(y) \nonumber \\
    &= (D\rho_{\mathrm{in}})(x).
\end{align}
Introducing the notation $D\nu_{\mathrm{in}}$ as defined in \eqref{def:diff_measure}, we thus obtain
\begin{align}\label{eq:diff_ups_output}
\nu_{\mathrm{out}}^{\prime}(0) = \mathbb{E}[x;\mu] \cdot (-D) \nu_{\mathrm{in}}.
\end{align}
The basic idea is that, while we have seen that the distributional derivative of the delta $D\delta_{0}$ does not allow itself to be  expressed by a complex measure, the distributional derivative $D\nu_{\mathrm{in}}$ of some probability measure $\nu_{\mathrm{in}}$ might belong to the space $\mathbf{M}_{\mathbb{C}}(\mathfrak{B})$ of complex measures%
\footnote{
As one may expect, the distributional derivative $D \nu$ of an arbitrary complex measure $\nu$ can be made well-defined by extending our framework into the theory of generalised functions. In general, the derivative derivative $D \nu$ is a distribution itself (as we have seen for the special case $\nu = \delta_{0}$), but not necessarily a complex measure anymore.
}.
If we could moreover find a condition for which the differentiability \eqref{eq:diff_ups_output} is valid with respect to the norm topology of the total variation ({\it i.e.}, strongly differentiable), we could develop a line of argument that is totally confined in the space $\mathbf{M}_{\mathbb{C}}(\mathfrak{B})$, without referring to the theory of distributions at all.

\paragraph{On the Main Results}
One finds below that the the above idea is indeed valid. To this end, we assume
\begin{itemize}
\item The probability measure $\mu$ has compact support.
\item The density of $\nu_{\mathrm{in}}$ belongs to the Schwartz space $d\nu_{\mathrm{in}}/d\beta \in \mathscr{S}(\mathbb{R})$.
\end{itemize}
Under the above two conditions, we demonstrate below that the map $t \mapsto \nu_{\mathrm{out}}(t)$ is in fact arbitrarily many times strongly differentiable, and that its higher-order derivatives read
\begin{equation}\label{eq:upsm_prob_diff_gen}
\nu_{\mathrm{out}}^{(n)}(t) = ((-D)^{n}\nu_{\mathrm{in}}) \ast (x^{n} \odot \mu)_{t}, \quad t \in \mathbb{R},\ n \in \mathbb{N}_{0},
\end{equation}
which in particular implies
\begin{equation}\label{eq:upsm_prob_diff_0}
\nu_{\mathrm{out}}^{(n)}(0) = \mathbb{E}[x^{n};\mu] \cdot (-D)^{n} \nu_{\mathrm{in}}, \quad n \in \mathbb{N}_{0}
\end{equation}
at the origin $t=0$. Here, $D^{n} \nu_{\mathrm{in}}$ denotes the signed measure defined in \eqref{def:diff_measure}, and the signed measure $x^{n} \odot \mu$ is defined in \eqref{def:Mass_mit_Dichte}. Note that our two conditions above, namely, the compactness of the support of $\mu = x^{0} \odot \mu$ and the density of $\nu_{\mathrm{in}} = (-D)^{0}\nu_{\mathrm{in}}$ belonging to the Schwartz space, are true not only for $n=0$, but for all $n \in N_{0}$. Note also that compactness of the support of $\mu$ guarantees the finiteness of all its higher-order moments $|\, \mathbb{E}[x^{n};\mu] \,| < \infty$, $n \in \mathbb{N}_{0}$. Applying \eqref{eq:upsm_prob_diff_0} to our physical situation by letting $\mu = \mu_{A}^{\phi}$ and $\nu_{\mathrm{in}} = \mu_{Q}^{\psi}$ would prove Proposition~\ref{prop:WUCM_II}.

\begin{proof}[Proof of our Main Result]
For demonstration, we provide a sketch of the proof by mathematical induction. One may readily confirm by definition that the above statement is trivially true for $n=0$. Now, assuming that the statement is true for $n \in N_{0}$,  we rewrite $\tilde{\nu}_{\mathrm{in}} := (-D)^{n}\nu_{\mathrm{in}}$, $\tilde{\mu} := x^{n} \odot \mu$ and $\tilde{\nu}_{\mathrm{out}}(t) := \tilde{\nu}_{\mathrm{in}} \ast \tilde{\mu}_{t}$ for better readability.
Now, recalling that the convolution algebra $L^{1}(\mathfrak{B})$ is an ideal in the measure algebra $\mathbf{M}_{\mathbb{C}}(\mathfrak{B})$, one finds that $\tilde{\nu}_{\mathrm{out}}(t)$ is absolutely continuous for all $t \in \mathbb{R}$ (in passing, one moreover finds that the density of $\tilde{\nu}_{\mathrm{out}}(t)$ is also a Schwartz function), and that its density $\tilde{\rho}_{\mathrm{out}}(t) := d\tilde{\nu}_{\mathrm{out}}(t)/d\beta$ is given by
\begin{equation}
\tilde{\rho}_{\mathrm{out}}(t)(x) = \int_{\mathbb{R}} \tilde{\rho}_{\mathrm{in}}(x - ty)\ d\tilde{\mu}(y), \quad t \in \mathbb{R},
\end{equation}
where $\tilde{\rho}_{\mathrm{in}}$ denotes the density of $\tilde{\nu}_{\mathrm{in}}$ (see \eqref{eq:density_ac_cmeas} for this result). 

In order to prove the strong differentiability of the map $t \mapsto \tilde{\nu}_{\mathrm{out}}(t)$, we work in the space of density functions. We start by demonstrating the point-wise differentiability of the map $t \mapsto \tilde{\rho}_{\mathrm{out}}(t)$, and to this end, we fix $t_{0}, x \in \mathbb{R}$ and observe
\begin{align}\label{eq:main_diff_out_prob_point-wise}
\left( \tilde{\rho}_{\mathrm{out}}^{\prime}(t_{0})\right) (x)
    &:= \lim_{t \to t_{0}} \frac{\tilde{\rho}_{\mathrm{out}}(t)(x) - \tilde{\rho}_{\mathrm{out}}(t_{0})(x)}{t - t_{0}} \nonumber \\
    &= \lim_{t \to 0} \int_{\mathbb{R}} \frac{\tilde{\rho}_{\mathrm{in}}(x - ty) - \tilde{\rho}_{\mathrm{in}}(x - t_{0}y)}{t - t_{0}} \ d\tilde{\mu}(y) \nonumber \\
    &= \int_{\mathbb{R}} y(-D\tilde{\rho}_{\mathrm{in}})(x - t_{0}y) \ d\tilde{\mu}(y),
\end{align}
where the exchange of the limit and integration in the second equality, while we shall omit any details of its proof, is essentially a consequence of the dominated convergence theorem.  Next, we return to its strong differentiability ({\it i.e.}, differentiability with respect to the $L^{1}$-norm). To this end, we assume $t_{0} < t$ without loss of generality and recall the mean-value theorem, which state that there exists a $t_{1} \in ]t_{0}, t[$ such that
\begin{equation}
\frac{\tilde{\rho}_{\mathrm{out}}(t)(x) - \tilde{\rho}_{\mathrm{out}}(t_{0})(x)}{t - t_{0}} = \left( \tilde{\rho}_{\mathrm{out}}^{\prime}(t_{1})\right) (x)
\end{equation}
holds. Then, one has
\begin{align}
&\left\|\frac{\tilde{\rho}_{\mathrm{out}}(t)(x) - \tilde{\rho}_{\mathrm{out}}(t_{0})(x)}{t - t_{0}} - \tilde{\rho}_{\mathrm{out}}^{\prime}(t_{0}) \right\|_{1} \nonumber \\
  &\qquad  = \int_{\mathbb{R}} \left| \int_{\mathbb{R}} y(-D\tilde{\rho}_{\mathrm{in}})(x - t_{1}y) - y(-D\tilde{\rho}_{\mathrm{in}})(x - t_{0}y) \ d\tilde{\mu}(y) \right|\ d\beta(x) \nonumber \\
    &\qquad  \leq \int_{\mathbb{R}} |y| \cdot \left\|\tau_{(-t_{1}y)}(D\tilde{\rho}_{\mathrm{in}})- \tau_{(-t_{0}y)}(D\tilde{\rho}_{\mathrm{in}})\right\|_{1}\ d\tilde{\mu}(y),
\end{align}
where the exchange of the order of integration in the last inequality is guaranteed to hold (Fubini's theorem), and the translation operator $\tau_{a}$ is defined in \eqref{eq:translation}. Compactness of the support of $\tilde{\mu}$ together with an analogous argument made in \eqref{ineq:approx_ident} implies that the r.~h.~s. of the above inequality tends to $0$ as $t \to 0$, which completes our proof for strong differentiability. We thus have by \eqref{eq:main_diff_out_prob_point-wise}
\begin{align}
\rho_{\mathrm{out}}^{(n+1)}(t)
    &= \tilde{\rho}_{\mathrm{out}}^{\prime}(0) \nonumber \\
    &= (-D\tilde{\rho}_{\mathrm{in}}) \ast (x \odot \tilde{\mu})_{t} \nonumber \\
    &= ((-D)^{n+1}\rho_{\mathrm{in}}) \ast (x^{n+1} \odot \mu)_{t}, \quad t \in \mathbb{R}
\end{align}
and
\begin{align}
\rho_{\mathrm{out}}^{(n+1)}(0)
    &= ((-D)^{n+1}\rho_{\mathrm{in}}) \ast (x^{n+1} \odot \mu)_{0} \nonumber \\
    &= \mathbb{E}[x^{n+1};\mu] \cdot (-D)^{n+1} \nu_{\mathrm{in}},
\end{align}
where we have used \eqref{def:measure_scaling} and $(x^{n+1} \odot \mu)(\mathbb{R}) = \mathbb{E}[x^{n+1};\mu]$ in the last equality. This completes our whole proof.
\end{proof}

\newpage

\newpage
\section{Conditioned Measurement I: In Terms of Conditional Expectations}\label{sec:ps_I}

We shall next embark on our study of the measurement scheme that we call the \emph{conditioned measurement} (CM) scheme.
As the name indicates, the CM scheme involves \emph{conditioning}, where one employs the measurement of another observable on the target system on top of the UM scheme studied earlier. The CM scheme can be understood as a natural generalisation of the \emph{post-selected measurement scheme}, which has recently been attracting much attention of several groups among the physics community. While the post-selected measurement scheme
itself has been practiced for quite a while, it has caught a renewed interest since Aharonov {\it et al.}~reintroduced it with the term \emph{weak measurement} which in particular applies to the post-selected measurement in the weak limit, along with the complex quantity termed \emph{weak value} purported to be measured by it. Two sections starting from here is devoted to the analysis on the CM scheme, and by following the same line as that of the former unconditioned counterpart, we start by examining the measurement scheme in terms of conditional expectations (Section~\ref{sec:ps_I}), and subsequently in terms of conditional probabilities (Section~\ref{sec:ps_II}).

\paragraph{Organisation of this Section}
The contents of this section is organised as follows. We first provide a concise summary of some of the necessary mathematical concepts that provides us the tools for conducting the analysis. We then make a brief review on the CM scheme from a relatively general framework, and make some comments on the technique of employing conditioning (or post-selection, as a special case) in precision measurements, whose alleged advantages has recently become the topic of intensive debate. We shall then investigate how one could reclaim the information of the configuration of the target system from the the measured outcomes, and to this end, we concentrate on the behaviour of the conditional expectation of the meter observable around the weak limit $g=0$ of the interaction parameter.  In parallel to the unconditional case, we call this procedure the \emph{weak conditioned measurement scheme} in this paper. We finally close this section by introducing the concept of \emph{conditional quasi-expectations} of a quantum observable given another (not necessarily simultaneously measurable) observable, as a generalisation to that of the standard conditional expectations, and examine some of their notable properties. 
\subsection{Reference Materials}

In this subsection, we shall briefly recall the necessary mathematical definitions and results regarding the formal mathematical description of conditioning.

\subsubsection{Conditioning}

The essence of the CM scheme lies in the conditioning of the outcomes of a measurement of an observable $X$ of the meter system $\mathcal{K}$ by that of an additional observable $B$ of the target system $\mathcal{H}$. The quantity of interest is then the \emph{conditional expectation} of $X$ given $B$, in contrast to the UM scheme described in Section~\ref{sec:ups_I}, where the quantity of interest was the mere (unconditional) expectation value of $X$.

\paragraph{Conditional Expectation given a Sub-$\sigma$-algebra}

Since one may find the general definition of conditional expectations to be rather involved, we start by some preliminary discussion in order to ease the introduction.
Let $(\mathbb{R}^{n}, \mathfrak{B}^{n}, \mu)$ be a probability space, and let $f: \mathbb{R}^{n} \to \mathbb{R}$ be $\mu$-integrable. Given a Borel set $B \in \mathfrak{B}^{n}$ with non-vanishing probability $\mu(B) \neq 0$, one defines the \emph{conditional expectation of $f$ given the measurable set $B \in \mathfrak{B}^{n}$} by the real number
\begin{equation}\label{def:cond_exp_set_elementary}
\mathbb{E}[f|B] := \frac{\int_{B}f(x)\ d\mu(x)}{\mu(B)}.
\end{equation}
Now, let $\mathbb{R}^{n} = \cup_{i=1}^{N} B_{i}$, $B_{i} \in \mathfrak{B}^{n}$ be a decomposition of $\mathbb{R}^{n}$ into finite numbers of mutually disjoint Borel sets, and let $\mathfrak{E} := \{B_{i}\}_{i=1,\dots,N}$ denote their collection. We then define
\begin{align}\label{def:sub_algebra_elementary}
\mathfrak{A}
    &:= \sigma(\mathfrak{E})
    = \left\{\bigcup_{i \in I} B_{i} : I \subset \{1, \dots, N\}\right\}
\end{align}
to be the sub-$\sigma$-algebra of $\mathfrak{B}^{n}$ generated by $\mathfrak{E}$. Assuming $\mu(B_{i}) \neq 0$ for all $i=1, \dots, N$, this gives rise to an $\mathfrak{A}$-$\mathfrak{B}$ measurable function
\begin{equation}\label{def:cond_exp_elementary}
\mathbb{E}[f|\mathfrak{A}](x) := \sum_{i = 1}^{N} \mathbb{E}[f|B_{i}] \cdot \chi_{B_{i}}(x),
\end{equation}
where each $\chi_{B_{i}}$ is the characteristic function of the subset $B_{i}$. Observing that each element $A \in \mathfrak{A}$ can be expressed by a union of elements of $\mathfrak{E}$, one has
\begin{align}
\int_{A}f(x)\ d\mu(x)
    &= \sum_{B_{i} \subset A} \mathbb{E}[f|B_{i}] \cdot \mu(B_{i}) \nonumber \\
    &= \int_{A} \mathbb{E}[f|\mathfrak{A}](x)\ d\mu|_{\mathfrak{A}}(x), \quad \forall A \in \mathfrak{A},
\end{align}
where $\mu|_{\mathfrak{A}}$ denotes the restriction of the probability measure $\mu$ on the sub-$\sigma$-algebra $\mathfrak{A}$.
Guided by this observation, the conditional expectation of an integrable function $f$ given a sub-$\sigma$-algebra $\mathfrak{A}$ is defined in the following manner:

\begin{definition*}[Conditional expectation given a sub-$\sigma$-algebra]
Let $(\mathbb{R}^{n}, \mathfrak{B}^{n}, \mu)$ be a probability space. For a sub-$\sigma$-algebra $\mathfrak{A} \subset \mathfrak{B}^{n}$ and a $\mu$-integrable function $f: \mathbb{R}^{n} \to \mathbb{R}$, the conditional expectation of $f$ given $\mathfrak{A}$, denoted as $\mathbb{E}[f|\mathfrak{A}]$, is defined as a $\mu|_{\mathfrak{A}}$-integrable function satisfying
\begin{equation}\label{def:cond_exp}
\int_{A} f(x)\ d\mu(x) = \int_{A} \mathbb{E}[f|\mathfrak{A}](x)\ d\mu|_{\mathfrak{A}}(x), \quad \forall A \in \mathfrak{A}.
\end{equation}
The conditional expectation $\mathbb{E}[f|\mathfrak{A}]$ exists, and is unique $\mu|_{\mathfrak{A}}$-a.e.
\end{definition*}
\noindent
To see the validity of the definition, first observe that the l.~h.~s. of \eqref{def:cond_exp} defines a complex measure $A \mapsto (f \odot \mu)(A)$, $A \in \mathfrak{A}$. Since $(f \odot \mu)|_{\mathfrak{A}} \ll \mu|_{\mathfrak{A}}$, the Radon-Nikod{\'y}m theorem leads to the existence and uniqueness $\mu|_{\mathfrak{A}}$-a.e. of the conditional expectation
\begin{equation}
\mathbb{E}[f|\mathfrak{A}] := \frac{d(f \odot \mu)|_{\mathfrak{A}}}{d\mu|_{\mathfrak{A}}},
\end{equation}
which is nothing but the Radon-Nikod{\'y}m derivative (density) of the restriction $(f \odot \mu)|_{\mathfrak{A}}$ with respect to the restriction $\mu|_{\mathfrak{A}}$. Note that the conditional expectation is defined as a \emph{function} (or more precisely, an equivalent class of functions) rather than a mere number. The elementary definition \eqref{def:cond_exp_elementary} 
mentioned earlier is in fact a special case of the above general definition, in which the sub-$\sigma$-algebra concerned is given by \eqref{def:sub_algebra_elementary}. The conditional expectation $\mathbb{E}[f|\mathfrak{A}]$ serves as the, so to speak, best approximation of the original function $f$ by measurable functions defined on the coarser%
\footnote{
Given two $\sigma$-algebras $\mathfrak{A} \subset \mathfrak{B}$, $\mathfrak{A}$ is said to be smaller or \emph{coarser} than $\mathfrak{B}$, and on the other hand, $\mathfrak{B}$ is said to be larger or \emph{finer} than $\mathfrak{A}$.
}
$\sigma$-algebra $\mathfrak{A} \subset \mathfrak{B}^{n}$.

\paragraph{Conditional Expectation given another Function}
We next recall the definition of the conditional expectation given another real measurable function. As above, we first provide an introductory argument. Let $(\mathbb{R}^{n}, \mathfrak{B}^{n}, \mu)$ be a probability space, and let $f: \mathbb{R}^{n} \to \mathbb{R}$ be $\mu$-integrable.  Given another measurable function $g: \mathbb{R}^{n} \to \mathbb{R}$, suppose that the probability of obtaining the outcome $y \in \mathbb{R}$ of $g$ is non-vanishing $\mu(g^{-1}(y)) \neq 0$. In a similar manner as before, one may define the conditional expectation of $f$ given the outcome $y$ of $g$ as
\begin{equation}\label{def:cond_exp_outcome}
\mathbb{E}[f|g = y] := \mathbb{E}[f|g^{-1}(y)] = \frac{\int_{g^{-1}(y)}f(x)\ d\mu(x)}{\mu(g^{-1}(y))},
\end{equation}
where we have just replaced $B = g^{-1}(y)$ in \eqref{def:cond_exp_set_elementary}.
It is now tempting to construct a function $y \mapsto \mathbb{E}[f|g = y]$ that maps each of the possible outcomes of $g$ to the corresponding conditional expectation. Assuming that the function $g$ only takes a finite number of distinct outcomes $\{y_{i}\}_{i=1, \dots, N}$, $y_{i} \in \mathbb{R}$, one accordingly obtains a decomposition $\mathbb{R}^{n} = \cup_{i=1}^{N} g^{-1}(y_{i})$ of $\mathbb{R}^{n}$ into a finite number of mutually disjoint Borel sets. 
Assuming moreover that $\mu(g^{-1}(y_{i})) \neq 0$ for all $i$, one obtains a well-defined measurable function
\begin{equation}
\mathbb{E}[f|g] : \mathbb{R} \to \mathbb{R},\quad y \mapsto \mathbb{E}[f|g = y],
\end{equation}
called the \emph{conditional expectation of $f$ given $g$}.

To see how this relates to the previous definition of the conditional expectation given a sub-$\sigma$-algebra, consider a general situation in which one is given a set $X$ (without a $\sigma$-algebra), a measurable space $(Y,\mathfrak{A})$ and a function $g: X \to Y$. The collection
\begin{equation}\label{def:initial_sigma_algebra}
\mathcal{I}(g) := g^{-1}(\mathfrak{A}) := \{g^{-1}(A) : A \in \mathfrak{A}\}
\end{equation}
makes itself into a $\sigma$-algebra, called the \emph{initial $\sigma$-algebra} on $X$ with respect to $g$, and it is the coarsest $\sigma$-algebra on $\mathbb{R}^{n}$ for which the map $g$ is measurable.
In the above situation, we take $(Y,\mathfrak{A}) = (\mathbb{R},\mathfrak{B}^{1})$ and define
\begin{equation}
\mathcal{I}(g) := g^{-1}(\mathfrak{B}^{1}) = \sigma\left( \mathfrak{E}\right),
\end{equation}
where we have let $\mathfrak{E} := \{g^{-1}(y_{i})\}_{i=1,\dots,N}$.
Now, since we have assumed that $\mu(g^{-1}(y_{i})) \neq 0$ for all $i$, the conditional expectation of $f$ given $\mathcal{I}(g)$ can be expressed as
\begin{equation}
\mathbb{E}[f|\mathcal{I}(g)] = \mathbb{E}[f|\sigma(\mathfrak{E})] = \sum_{i = 1}^{N} \mathbb{E}[f|g^{-1}(y_{i})] \cdot \chi_{g^{-1}(y_{i})},
\end{equation}
where the last equality is due to \eqref{def:cond_exp_elementary} by replacing $B_{i} = g^{-1}(y_{i})$. It is then fairly straightforward to see that the conditional expectations $\mathbb{E}[f|\mathcal{I}(g)]$, $\mathbb{E}[f|g]$ and the conditioning function $g$ are related to one another through the commutative diagram,
\begin{equation}\label{diagram:cond_exp}
\xymatrix{
(\mathbb{R}^{n}, \mathcal{I}(g))
	\ar[rd]_{\mathbb{E}[f|\mathcal{I}(g)]}
	\ar[r]^{g}
& (\mathbb{R}, \mathfrak{B}^{1})
	\ar[d]^{\mathbb{E}[f|g]}\\
& (\mathbb{R}, \mathfrak{B}^{1})
}
\end{equation}
where each of the functions is measurable.  In this sense, the function $\mathbb{E}[f|g]$ is understood to be nothing but the factorisation of $\mathbb{E}[f|\mathcal{I}(g)]$ by $g$.  The validity of such observation for the general case is guaranteed by the following Factorisation Theorem.
\begin{theorem*}[Factorisation Theorem]
Let $X$ be a non-empty set, and let $\mathcal{I}(g) := g^{-1}(\mathfrak{A})$ be the initial $\sigma$-algebra of a map $g:X \to (Y,\mathfrak{B})$. A function $h: (X,\mathcal{I}(g)) \to (\mathbb{R},\mathfrak{B}^{1})$ is measurable if and only if there exists a measurable function $\tilde{h}: (Y,\mathfrak{B}) \to (\mathbb{R},\mathfrak{B}^{1})$ that makes the diagram
\begin{equation}
\xymatrix{
(X, \mathcal{I}(g))
	\ar[rd]_{h}
	\ar[r]^{g}
& (Y, \mathfrak{B})
	\ar[d]^{\tilde{h}}\\
& (\mathbb{R}, \mathfrak{B}^{1})
}
\end{equation}
commute.
\end{theorem*}
\noindent
By letting $(Y,\mathfrak{B}) = (\mathbb{R},\mathfrak{B}^{1})$ and $h = \mathbb{E}[f|\mathcal{I}(g)]$, this guarantees the existence of the function $\mathbb{E}[f|g] := \tilde{h}$ that makes the desired diagram commute, even for the general case.

As for the integrability of the conditional expectation $\mathbb{E}[f|g]$, we first observe that the probability of obtaining the outcome of $g$ in a Borel set $B \in \mathfrak{B}^{1}$ is dictated by the probability measure 
\begin{equation}
g(\mu)(B) := \mu(g^{-1}(B)), \quad B \in \mathfrak{B}^{1},
\end{equation}
which is nothing but the image measure of $\mu$ with respect to $g$ (see \eqref{def:Bildmass} for its definition and properties). One thus sees by the formula
\begin{align}
\int_{\mathbb{R}} \mathbb{E}[f|g]\ dg(\mu)
    &= \sum_{i=1}^{N} \mathbb{E}[f|g](y_{i}) \cdot g(\mu)(\{y_{i}\}) \nonumber \\
    &= \sum_{i=1}^{N} \mathbb{E}[f|g = y_{i}] \cdot \mu(g^{-1}(y_{i})) \nonumber \\
    &= \sum_{i=1}^{N} \int_{g^{-1}(y_{i})}f(x)\ d\mu(x) \nonumber \\
    &= \int_{\mathbb{R}^{n}}f(x)\ d\mu(x),
\end{align}
that the function $\mathbb{E}[f|g]$ is $g(\mu)$-integrable, and its expectation value coincides with the expectation value of $f$ under $\mu$, which is what one naturally expects.

Guided by the above observation, the conditional expectation of an integrable function $f$ given another measurable function $g$ is defined in the following manner:
\begin{definition*}[Conditional expectation given a measurable function]
Let $(\mathbb{R}^{n}, \mathfrak{B}^{n}, \mu)$ be a probability space, and let $f : \mathbb{R}^{n} \to \mathbb{R}$ be $\mu$-integrable. The conditional expectation of $f$ given a measurable function $g : \mathbb{R}^{n} \to \mathbb{R}$, denoted as $\mathbb{E}[f|g]$, is defined as a $g(\mu)$-integrable function that makes the diagram
\begin{equation}\label{def:conditional_expectation}
\xymatrix{
\left(\mathbb{R}^{n}, \mathcal{I}(g), \mu|_{\mathcal{I}(g)}\right)
	\ar[rd]_{\mathbb{E}[f|\mathcal{I}(g)]}
	\ar[r]^{g}
& \left(\mathbb{R}, \mathfrak{B}^{1}, g(\mu)\right)
	\ar[d]^{\mathbb{E}[f|g]}\\
& \left(\mathbb{R}, \mathfrak{B}^{1}\right)
}
\end{equation}
commute. Its existence and uniqueness $g(\mu)$-a.e. is known to be guaranteed.
\end{definition*}
\noindent
Note that integrability of $\mathbb{E}[f|\mathcal{I}(g)]$ is due to the change of variables formula \eqref{eq:Transformationsformel_Bildmass} for image measures, and its uniqueness $g(\mu)$-a.e. is immediate by definition.
Based on the above definition, let $\mathbb{E}[f|g]$ be (a representative of) the conditional expectation of $f$ given $g$. We write
\begin{equation}
\mathbb{E}[f|g = y] := \mathbb{E}[f|g](y)
\end{equation}
to denote the \emph{conditional expectation of $f$ given the outcome $y$ of $g$}. Note that this definition is dependent on the choice of the representative and 
may admit ambiguity. Indeed, for the choice $y \in \mathbb{R}$ for which the probability of obtaining the outcome of $g$ in $\{y\}$ is vanishing: $g(\mu)(\{y\}) = \mu(g^{-1}(\{y\})) = 0$, one sees that $\mathbb{E}[f|g = y]$ is \emph{indefinite} and may take \emph{any} real number. 
As exemplified in here, the conditional expectation $\mathbb{E}[f|g]$ of $f$ given $g$ is appropriate to be viewed as an equivalent class of integrable functions, rather than a function alone.

\paragraph{Conditioning by Simultaneously Measurable Observables}

As in the previous section, we occasionally denote the Borel sets on $\mathbb{R}^{n}$ by $\Delta \in \mathfrak{B}^{n}$ in place of $B$ for better understanding and readability, especially in the context of quantum theory, where the confusion of the notation of $B$ with that of an operator may become a concern.
Let $A$ and $B$ be a pair of simultaneously measurable observables on a quantum system $\mathcal{H}$. We have seen that this yields a probability measure $\mu_{A,B}^{\phi}$  on $(\mathbb{R}^{2}, \mathfrak{B}^{2})$ ({\it cf.} \eqref{def:prob_measrue_A_simul}), which is interpreted as the joint-probability distribution describing the outcomes of a simultaneous measurement of $A$ and $B$ performed on the quantum system in the state $|\phi\rangle \in \mathcal{H}$.  Letting $f(a,b)=\pi_{A}(a,b) := a$ and $g(a,b) = \pi_{B}(a,b) := b$ describe the measurement outcomes of each of the observables $A$ and $B$,
we shall briefly see below how the previous discussions on conditioning fits in the context of quantum mechanics. 
For our purpose, assume $|\phi\rangle \in \mathrm{dom}(A)$ so that the projection $\pi_{A}(a,b) = a$ may be integrable
\begin{align}
\int_{\mathbb{R}^{2}} \pi_{A}(a,b)\ d\mu_{A,B}^{\phi}(a,b)
    &= \int_{\mathbb{R}^{2}} a\ d\mu_{A}^{\phi}(a) \nonumber \\
    &= \mathbb{E}[A;\phi],
\end{align}
with respect to the probability measure $\mu_{A,B}^{\phi}$. Observing that the image measure of $\mu_{A,B}^{\phi}$ with respect to the second projection
\begin{equation}
\pi_{B}\left(\mu_{A,B}^{\phi}\right)(\Delta_{B}) := \mu_{A,B}^{\phi}(\mathbb{R} \times \Delta_{B}) = \mu_{B}^{\phi}(\Delta), \quad \Delta_{B} \in \mathfrak{B}
\end{equation}
is nothing but the probability measure describing the outcome of $B$, we define the \emph{conditional expectation} $\mathbb{E}[A|B ; \phi]$ of an observable $A$ given $B$ on the state $|\phi\rangle $ as the (equivalence class of) $\mu_{B}^{\phi}$-integrable function(s)
\begin{align}\label{def:conditional_expectation_of_observables}
\mathbb{E}[A|B ; \phi] := \mathbb{E}[\pi_{A}|\pi_{B}],
\end{align}
where the r.~h.~s. is the conditional expectation of $\pi_{A}$ given $\pi_{B}$ under the probability measure $\mu_{A,B}^{\phi}$. Under the same assumption, we analogously define the conditional expectation of an observable $A$ given the outcome $b$ of an observable $B$ on the state $|\phi\rangle \in \mathrm{dom}(A)$ by
\begin{equation}
\mathbb{E}[A|B = b ; \phi] := \mathbb{E}[A|B ; \phi](b).
\end{equation}
We note again that the last definition incorporates some ambiguity, in which the number $\mathbb{E}[A|B = b; \phi]$ is not well-defined in the case where the probability that the measurement of $B$ yields the outcome $b$ is vanishing.

\subsection{Conditioned Measurement}

The CM scheme incorporates the measurements of \emph{two} observables, where the experimenter measures one local observable on the meter system and the other on the target system. In this paper, we generally define the CM scheme as the act of measuring the conditional expectation
\begin{equation}\label{def:post-selected_measurement}
\mathbb{E}[X|B ; \Psi^{g}] := \mathbb{E}[I \otimes X|B \otimes I ; \Psi^{g}]
\end{equation}
of an observable for the choice of either $X = Q$ or $X = P$ of the meter system given another observable $B$ of the target system.
Here, for better readability, we have made a little abuse of notation by writing $X$ instead of $I \otimes X$ and $B$ for $B \otimes I$. We emphasise again that the conditional expectation \eqref{def:post-selected_measurement} is defined as an \emph{equivalence class of functions} that are integrable with respect to the probability measure
\begin{equation}
\mu_{B}^{\phi^{g}} := \mu_{B \otimes I}^{\Psi^{g}},
\end{equation}
which describes the behaviour of the outcome of the measurement of the local observable $B$ on the target system.  Here, we have introduced the density matrix
\begin{equation}\label{def:targetstate}
\phi^{g} := \mathrm{Tr}_\mathcal{K}[|\Psi^{g}\rangle\langle\Psi^{g}|]
\end{equation}
on the target system defined in a parallel manner as in \eqref{def:meterstate}.
For its well-definedness, we note the following statement for reference.

\begin{proposition}[Well-definedness of the Conditional Expectation]\label{prop:def_CM_I}
In the context of the CM scheme, let
\begin{enumerate}
\item If $X = Y$: $|\phi\rangle \in \mathcal{H}$, $|\psi\rangle \in \mathrm{dom}(X)$
\item If $X \neq Y$: $|\phi\rangle \in \mathrm{dom}(A)$, $|\psi\rangle \in \mathrm{dom}(X)$
\end{enumerate}
be the choice of the initial states of the target and meter systems. Then, the conditional expectation $\mathbb{E}[X|B ; \Psi^{g}]$ is well-defined for all range of the interaction parameter $g \in \mathbb{R}$.
\end{proposition}
\begin{proof}
For demonstration, we shall only refer to Proposition~\ref{prop:UCM_I} that guarantees the integrability of the outcomes of the measurement of $X$ ({\it i.e.}, $|\,\mathbb{E}[I \otimes X; \Psi^{g}]\,| < \infty$) for all range of $g \in \mathbb{R}$, given the conditions assumed.
\end{proof}

\paragraph{Post-selected Measurement}
As a special subclass of this measurement scheme, we prepare the term \emph{post-selected measurement scheme} to refer to the case where the conditioning observable $B = |\phi^{\prime}\rangle\langle\phi^{\prime}|$ happens to be a projection on some one-dimensional subspace of $\mathcal{H}$ spanned by some normalised vector $|\phi^{\prime}\rangle \in \mathcal{H}$, and in such a case, the act of conditioning will be occasionally referred to as the \emph{post-selection}.  It is also a common practice found in various literatures to call the state $|\phi\rangle$ prepared prior to the measurement the \emph{initial} or the \emph{pre-selected} state, and the normalised vector $|\phi^{\prime}\rangle$ spanning the image of the one-dimensional projection $B = |\phi^{\prime}\rangle\langle\phi^{\prime}|$ the \emph{final} or the \emph{post-selected} state.

\subsubsection{Topic: `Amplification Technique' by Conditioning}\label{sec:psI_amplification}
It is widely known that, in general, the range of conditional expectation may exceed the (unconditional) expectation value, {\it i.e.}, for some clever choice of the conditioning observable $B$ and its outcome $b \in \mathbb{R}$, one has
\begin{equation}\label{eq:amplification_by_conditioning}
\big\vert \, \mathbb{E}\left[X; \Psi^{g} \right] \big\vert \leq \big\vert\, \mathbb{E}\left[X | B=b ; \Psi^{g} \right] \big\vert
\end{equation}
with non-vanishing probability.  Clearly, this property should prove itself useful in some certain situations.

While this property has occasionally been utilised in experiments, it has recently caught wide attention due to the reports on the success of application in precision measurements, including the experimental detection of the spin-Hall effect of light (SHEL) in 2008 \cite{Hosten_2008}, and the detection of an ultra-sensitive beam deflection in a Sagnac interferometer in 2009 \cite{Dixon_2009}. The experiments have effectively utilised the technique of conditioning (or post-selection) to yield an enhancement (or `amplification') of an extremely small beam displacement to the extent that it is large enough to overcome various technical imperfections (noise level), and eventually realising significant detection of such tiny effects.  In this context, 
this technique has often been referred to as the `weak value amplification' or as `Aharonov-Albert-Vaidman effect' of amplification \cite{Aharonov_1988}.

\paragraph{Review of the Recent theoretical Analyses}
Extensive theoretical analyses have been conducted in recent years  from various viewpoints
on the technical advantages of the technique of post-selection over the conventional unconditioned counterpart. 
Some of them addressed the question of signal amplification and its limit, where 
one asks the question as to what extent one can amplify the signal \cite{Koike_2011} and how one could achieve the optimisation \cite{Susa_2012}; the question of the existence of the limit of amplification will be addressed shortly in a more general framework.
As far as the authors are aware of, the first sound analytic result appeared around 2012 \cite{Nakamura_2012}, in which the limit to the amplification rate, as well as the signal-to-noise ratio has been explicitly presented. The computation was conducted for a special case where the observable $A$ fulfils the condition $A^2 = I$ and the meter wave functions were assumed to be of Gaussian states, which we shall also address in a relatively more general setting later in this section, and also in Appendix~\ref{sec:PSM}.

Others focused on the statistical loss which occurs due to the post-selection and examine the feasibility of improving the parameter estimation of the coupling constant $g$ by post-selection based on estimation theory (for a concise review on the topic form this point of view, see \cite{Knee_2014}).  The result is that the post-selection statistically deteriorates the quality of estimation, both in the case where ideal noiseless experiments can be performed \cite{Tanaka_2013}, and also in some case where certain types of fully-known or controllable noise are present \cite{Knee_2013_1,Knee_2013_2,Ferrie_2013}. 
In an attempt to address the question of how the post-selection technique, while being statistically inferior to the unconditioned case, could be advantageous in realistic experiments, the authors have conducted a theoretical analysis on post-selected measurement in the presence of some intractable `measurement uncertainty', a relatively modern concept in metrology to express \emph{unknown} or \emph{uncontrollable} source of technical imperfections \cite{Lee_2014}. It was then found that, while post-selection suffers from statistical deterioration, in certain cases the amplification effect becomes favourable in overcoming the unknown/uncontrollable  source of technical imperfections one could not completely eliminate through `noise hunting', which accordingly cannot be reduced from statistical reiteration.  This suggests that the post-selection technique should be understood as the \emph{practice of taking advantage of the trade-off relation} between the reduced contribution from intractable source of measurement uncertainty due to its signal amplification effect, and the statistical deterioration caused by the decrease in success probability.

\subsubsection{Topic: `Limit of Amplification' in Terms of Essential Suprema}
In what follows, we provide a somewhat general result regarding the question of `limit of amplification' by conditioning, which has been one of the hottest topics among the study of the technical advantages in employing conditioning in experiments. A typical way to address this problem is to ask oneself, to what extent one could enlarge the conditional expectation $\mathbb{E}[X|B ; \Psi^{g}]$  by choosing an appropriate conditioning observable $B$ and its outcome $b \in \sigma(B)$ with non-vanishing probability. By recalling the definition of essential supremum of a function \eqref{def:ess_sup}, one realises that the question is equivalent to asking to what extent one could make the essential supremum of the conditional expectation
\begin{equation}
\big\|\,\mathbb{E}[X|B ; \Psi^{g}]\,\big\|_{\infty}
\end{equation}
large by the choice of the conditioning observable $B$.

\paragraph{Preliminaries}
To prepare for our arguments, we first observe some basic facts regarding absolute continuity and essential suprema.
\begin{lemma}
Let $(X,\mathfrak{A},\mu)$ be a probability space, and let $\nu: \mathfrak{A} \to \mathbb{C}$ be a complex measure. Then, the following conditions are equivalent:
\begin{enumerate}
\item $\nu \ll \mu$.
\item $|\nu| \ll \mu$.
\item There exists a non-negative number $M \in [0,\infty]$ such that
\begin{equation}\label{eq:abs_cont_cond}
|\nu(A)| \leq |\nu|(A) \leq M \cdot \mu(A)
\end{equation}
holds for all $A \in \mathfrak{A}$.
\end{enumerate}
In such a cases, the Radon-Nikod{\'y}m derivative $d\nu/d\mu$ exists by the Radon-Nikod{\'y}m theorem, and its essential supremum $\|d\nu/d\mu\|_{\infty}$
gives the smallest of such $M$ that satisfies \eqref{eq:abs_cont_cond}.
\end{lemma}
\begin{proof}
For the equivalence of the condition $(i) \Leftrightarrow (ii)$, the reader is referred to any textbooks on measure and integration theory.  We already know from the Reference Material in Section~\ref{sec:ups_II_pre} that $|\nu(A)| \leq |\nu|(A)$, $A \in \mathfrak{A}$. The implication $(ii) \Rightarrow (iii)$ is then trivial by simply taking $M=\infty$. The converse $(iii) \Rightarrow (ii)$ is also immediate by the definition of absolute continuity.  Now that we have proved the equivalence of the three conditions, we move on to the demonstration of the final statement.  To this end, first observe the evaluation
\begin{align}
|\nu(A)|
    &= \left| \int_{A} \frac{d\nu}{d\mu}\ d\mu \right| \nonumber \\
    &\leq \int_{A} \left| \frac{d\nu}{d\mu} \right| d\mu \leq \left\|\frac{d\nu}{d\mu}\right\|_{\infty} \cdot \mu(A).
\end{align}
Combining this with the minimality of the variation $|\nu|$, one sees that the choice $M = \|d\nu/d\mu\|_{\infty}$ of the upper bound satisfies \eqref{eq:abs_cont_cond}. Now, suppose that there exists a non-negative number $0 \leq M < \|d\nu/d\mu\|_{\infty}$ satisfying \eqref{eq:abs_cont_cond}. Then, by definition of the essential supremum, there exists a measurable set $A$ satisfying $0 < \mu(A)$ and $M < |d\nu/d\mu||_{A}$ (just take $A := \{x \in X : M < |d\nu/d\mu|(x) \}$), hence
\begin{equation}
|\nu|(A) = \int_{A} \left| \frac{d\nu}{d\mu} \right|\ d\mu > M \cdot \mu(A),
\end{equation}
which contradicts the minimality of $|\nu|$.
\end{proof}
\noindent
As a corollary to this, the following observation is of special interest.
\begin{corollary}[Conditional Expectations and Essential Suprema]\label{cor:cond_exp_and_ess_sup}
Let $(X, \mathfrak{A}, \mu)$ be a probability space, $f: X \to \mathbb{R}$ be $\mu$-integrable, and $\mathfrak{B} \subset \mathfrak{A}$ be a sub-$\sigma$-algebra. Then the evaluation
\begin{equation}
\|\, \mathbb{E}[f|\mathfrak{B}]\,\|_{\infty} \leq \|f\|_{\infty}
\end{equation}
holds. As a direct consequence, if moreover a measurable function $g: X \to \mathbb{R}$ is given, the evaluation
\begin{equation}
\|\, \mathbb{E}[f|g]\,\|_{\infty} \leq \|f\|_{\infty}
\end{equation}
naturally holds.
\end{corollary}
\begin{proof}
First recall that the conditional expectation $\mathbb{E}[f|\mathfrak{B}]$ is nothing but the Radon-Nikod{\'y}m derivative of the complex measure $f \odot \mu$ with respect to the restriction $\mu|_{\mathfrak{B}}$. Letting $\nu := f \odot \mu$ and replacing $\mu$ by $\mu|_{\mathfrak{B}}$ in the above Lemma, one finds
\begin{equation}
|\nu(A)| = \left| \int_{A} f\ d\mu \right| \leq \|f\|_{\infty} \cdot \mu(A),
\end{equation}
hence
\begin{equation}
\|d\nu/d\mu\|_{\infty} = \|\, \mathbb{E}[f|\mathfrak{B}]\,\|_{\infty} \leq \|f\|_{\infty},
\end{equation}
which was to be demonstrated.
\end{proof}
\noindent
In casual language, this is to say that each value of the conditional expectation of $f$ never exceeds the maximum number that $f$ takes under a given probability measure, which is a result that should be intuitively clear. As a direct application of the result in the context of quantum measurement of a pair of simultaneously measurable observables $A$ and $B$, this reduces to the following.
\begin{corollary}\label{cor:cond_exp_obs_ess_sup}
Given a pair of strongly commuting self-adjoint operators $A$ and $B$ and a fixed state $|\phi\rangle \in \mathrm{dom}(A)$, the essential supremum of the conditional expectation of $A$ given $B$ is never greater than
\begin{equation}
\|\, \mathbb{E}[A|B;\phi] \,\|_{\infty} \leq \|A\|_{\infty}^{\phi},
\end{equation}
where $\|A \|_{\infty}^{\phi} := \|a\|_{\infty}$ denotes the essential supremum of the measurable function $a \mapsto a$ under the probability measure $\mu_{A}^{\phi}$ describing the behaviour of the outcome of the measurement of $A$ on the state $|\phi\rangle$. If $A$ happens to be bounded, its operator norm%
\footnote{For a bounded operator $X$, recall that the \emph{operator norm} of $X$ is defined by
\begin{equation}
\|X\| := \sup\{ \|X\phi\| : \|\phi\| = 1 \}.
\end{equation}
}
$\|A\|$ becomes the universal ({\it i.e.}, state independent) upper bound of $\|A \|_{\infty}^{\phi}$, hence
\begin{equation}
\|\, \mathbb{E}[A|B;\phi] \,\|_{\infty} \leq \|A \|_{\infty}^{\phi} \leq \| A\| < \infty
\end{equation}
holds for all $|\phi\rangle \in \mathcal{H}$.
\end{corollary}
\begin{proof}
The former part of the statement is immediate by Corollary~\ref{cor:cond_exp_and_ess_sup}. For the latter part, we first recall that the \emph{numerical range} of a self-adjoint operator $X$ is defined as
\begin{equation}
W(X) := \{ \langle \psi, X \psi \rangle : |\psi\rangle \in \mathrm{dom}(X), \|\psi\|^{2} = 1 \},
\end{equation}
which is nothing but the collection of all possible expectation values of $X$. Now, a direct application of the Cauchy-Schwarz inequality leads to
\begin{equation}
\left|\, \mathbb{E}[X;\phi] \,\right| \leq \|X\|, \quad \mathbb{E}[X;\phi] \in W(X),
\end{equation}
for bounded $X$, and by recalling the basic relation $\sigma(X) \subset \overline{W(X)}$, where the overline on $W(X)$ denotes its topological closure, one concludes
\begin{align}
\| X \|_{\infty}^{\phi}
    &\leq \sup\{|x| : x \in \sigma(X)\} \nonumber \\
    &\leq \sup\{|x| : x \in \overline{W(X)}\} \leq \|X\|,
\end{align}
which was to be demonstrated.
\end{proof}
\noindent
The latter part of the statement is to say that conditional expectations of a bounded observable has a universal upper bound given by its operator norm, which is also a result that should be intuitively clear.

\paragraph{On the `Limit of Amplification' by Conditional Measurement}
As a direct application of the above corollary to our problem, we obtain the main result of this passage.
\begin{proposition}[Amplification by Conditioning]\label{prop:limit_of_amplification}
Under the framework of the CM scheme, the essential supremum of the conditional expectation of $X$ given $B$ is never greater than that of the UM scheme of $X$
\begin{equation}\label{ineq:limit_of_amplification_01}
|\, \mathbb{E}[X|B = b;\Psi^{g}]\,| \leq \|\, \mathbb{E}[X|B;\Psi^{g}]\,\|_{\infty} \leq \|X\|_{\infty}^{\psi^{g}},
\end{equation}
where $\|X\|_{\infty}^{\psi^{g}}:= \|x\|_{\infty}$ denotes the essential supremum of $x$ under the probability measure $\mu_{X}^{\psi^{g}}$ describing the behaviour of the outcome of the local measurement $X$ on the meter system. In other words, $\|X\|_{\infty}^{\psi^{g}}$ gives the (conditioning-observable-independent) upper bound to the extent the conditional expectation can be `amplified' by means of conditioning%
\footnote{Recall the inherent subtlety when we use the expression $\mathbb{E}[X|B = b;\Psi^{g}]$. The left most inequality in \eqref{ineq:limit_of_amplification_01} should thus be understood to hold $\mu_{B}^{\phi^{g}}$-a.e.}.
\end{proposition}
\noindent
In physical terms, this is to say that the extent one may `amplify' the conditional expectation $\mathbb{E}[X|B ; \Psi^{g}]$ by means of changing the conditioning observable $B$ is predetermined by $\|X\|_{\infty}^{\psi^{g}}$. This is one general form to answer the question of the existence of the limit of `amplification' by conditioning.

As the next step, one might eventually be interested in seeking for the condition under which $\|X\|_{\infty}^{\psi^{g}}$ is bounded from above, even if we could freely choose the initial state $|\phi\rangle$ of the target system. This would create a universal upper bound of $\|\, \mathbb{E}[X|B;\Psi^{g}]\,\|_{\infty}$ that is indifferent to both the initial and final configurations of the target system ({\it i.e.}, the choice of the initial target state $|\phi\rangle$ and the conditioning observable $B$). As we have learned from the discussions above, this would typically be the case when there exists a subspace $U(g,\psi) \subset \mathcal{H} \otimes \mathcal{K}$, for fixed $g \in \mathbb{R}$ and $|\psi\rangle \in \mathrm{dom}(X)$, such that $|\Psi^{g}\rangle \in U(g,\psi)$ for all $|\phi\rangle \in \mathrm{dom}(A)$, and that the restriction of $I \otimes X$ on $U(g,\psi)$ is bounded.

\begin{proposition}[Limit of Amplification by Conditioning]\label{prop:limit_of_amplification_typ}\label{prop:lim_amp_example}
Under the framework of the CM scheme, let both the interaction parameter $g \in \mathbb{R}$ and the initial meter state $|\psi\rangle \in \mathrm{dom}(X)$ be fixed, and suppose that the target observable $A$ has a spectrum $\sigma(A) = \{a_{1}, \dots, a_{N}\}$, $N \in \mathbb{N}^{\times}$ of finite cardinality. Then, the following facts hold:
\begin{enumerate}
\item The density operator $\psi^{g}$ of the meter system \eqref{def:meterstate} can be written as a probabilistic mixture of a finite number of projection operators (pure states) supported on the finite-dimensional (at most $N$-dimensional) subspace
\begin{equation}\label{eq:fin_dim}
\mathcal{K}(g,\psi) := \mathrm{span}(\{ |e^{-iga_{1}Y}  \psi \rangle, \dots, |e^{-iga_{N}Y}  \psi \rangle\}),
\end{equation}
which is independent of the initial choice $|\phi\rangle \in \mathcal{H}$ of the target state.
\item The restriction $X|_{\mathcal{K}(g,\psi)}$ of the meter observable $X$ on the subspace \eqref{eq:fin_dim} is bounded, and thus its operator norm
\begin{equation}
\|\, \mathbb{E}[X|B;\Psi^{g}]\,\|_{\infty} \leq \left\|X|_{\mathcal{K}(g,\psi)}\right\| < \infty
\end{equation}
provides a finite universal upper bound to the conditional expectation that is independent of the configuration of the target system ({\it i.e.}, the choice of the initial state $|\phi\rangle \in \mathcal{H}$ and that of the conditioning observable $B$).
\end{enumerate}

\end{proposition}
\begin{proof}
Under the above condition, first observe that
\begin{align}
|\Psi^{g}\rangle
    &= \sum_{n = 1}^{N} \left( \Pi_{a_{n}} \otimes e^{-iga_{n}Y} \right) | \phi \otimes \psi \rangle \nonumber \\
    &= \sum_{n = 1}^{N} \left( \Pi_{a_{n}} |\phi\rangle \otimes |e^{-iga_{n}Y}\psi\rangle \right),
\end{align}
where we have used \eqref{eq:int_op_fin}. One readily finds from the above formula that the density operator
\begin{equation}
\psi^{g} = \mathrm{Tr}_{\mathcal{H}}\left[ |\Psi^{g}\rangle \langle\Psi^{g}| \right],
\end{equation}
defined as in \eqref{def:meterstate}, can indeed be written as a probabilistic mixture of a finite number of projection operators (pure states) supported on the subspace \eqref{eq:fin_dim}.
We then recall that any operator $X$ defined on a finite-dimensional Hilbert space are necessarily bounded, and thus observe that the current problem at hand reduces to the situation of Corollary~\ref{cor:cond_exp_obs_ess_sup}.
\end{proof}
\noindent
In physical terms, this is to say that there exists a finite limit $\left\|X|_{\mathcal{K}(g,\psi)}\right\| < \infty$ to the extent one may `amplify' the conditional expectation $\mathbb{E}[X|B ; \Psi^{g}]$ by means of only changing the configuration of the target system (namely, by changing either or both the conditioning observable $B$ and the initial state $|\phi\rangle \in \mathcal{H}$ of the target system). Specifically, the evaluation
\begin{equation}\label{eq:limit_of_amp_explicit}
\big|\, \mathbb{E}[X|B = b; \Psi^{g}] \,\big| \leq \left\|X|_{\mathcal{K}(g,\psi)}\right\| < \infty
\end{equation}
holds for all $b \in \sigma(B)$ up to a set of probability zero, and the upper bound $\left\|X|_{\mathcal{K}(g,\psi)}\right\|$ does not depend on the choice of $B$ nor $|\phi\rangle$. 
Naturally, if one could change either the interaction parameter $g$ or the initial state $|\psi\rangle$ of the meter system alongside, the above result is no more valid.

\subsection{Recovery of the Target Profile}

Parallel to the study of the UM scheme, we are now interested in the information of the target system which is to be  extracted from the CM scheme. 
Following the line of arguments for the UM scheme, we are specifically interested in investigating the local behaviour of the outcome of the CM scheme around $g=0$, {\it i.e.}, the \emph{weak conditioned measurement}, in which the target of our analysis is the map
\begin{equation}
g \mapsto \mathbb{E}[X | B; \Psi^{g}]
\end{equation}
from the interaction parameter $g$ to the conditional expectation of $X$ given $B$, which was in general defined as a map from the real line to an equivalent class of functions. To this end, we first conduct a preliminary observation.

\subsubsection{Preliminary Observation}\label{sec:ps_I_wps}

Since the definition of the conditional expectation is given in a rather abstract way, the conditional expectation \eqref{def:post-selected_measurement} in general does not admit an explicit expression by vectors and operators (in contrast to the UM case \eqref{prop:UCM_I_formula}, which always admits such an explicit expression). In view of this, it would be sometimes helpful if one could find a condition for which the conditional expectation \eqref{def:post-selected_measurement} of our interest may be explicitly written down. We first point out that this will be indeed the case given that the spectrum of the conditioning observable $B$ has finite cardinality.
Now, let
\begin{equation}\label{eq:spectr_decomp_B_fin}
B = \sum_{n=1}^{N} b_{n} \,\Pi_{b_{n}}
\end{equation}
be the spectral decomposition of $B$, where $\sigma(B) = \{b_{1}, \dots, b_{N}\}$ is any enumeration of its eigenvalues, and $\Pi_{b} := E_{B}(\{b\})$, $b \in \sigma(B)$ denotes the unique projection on the eigenspace associated to it. It is then fairly straightforward to see by definition that the conditional expectation of $X$ given $B$ is explicitly given by
\begin{align}\label{eq:cond_exp_B_fin}
&\mathbb{E}[X|B = b ; \Psi^{g}] \nonumber \\
   & \quad = \begin{cases}
    {\mathbb{E}\left[\Pi_{b} \otimes X; \Psi^{g} \right] / \left\|(\Pi_{b} \otimes I)\Psi^{g}\right\|^{2}}, & \quad (b \in \sigma(B),\ \left\|(\Pi_{b} \otimes I)\Psi^{g}\right\|^{2} \neq 0), \\
    \text{indefinite}, & \quad (\text{else}).
    \end{cases}
\end{align}
Here, recall that conditional expectations are defined as an equivalence class of functions, and hence its value for the outcome $b$ of the measurement of the observable $B$ such that the probability of observing it is vanishing, is indefinite by definition. 
The study of the weak CM scheme then reduces to the analysis of the map
\begin{equation}\label{def:wpsm_special}
g \mapsto \mathbb{E}[X | B = b; \Psi^{g}]
\end{equation}
for each $b \in \sigma(B)$ such that the probability of observing it is non-vanishing. Since this is a map from the real line to itself ({\it i.e.}, a function), it should be a much more familiar and straightforward object to deal with.

\paragraph{Objective of this Passage}
In what follows, we will be discussing the differentiability of the function \eqref{def:wpsm_special} at the point $g=0$.
To this end, first observe that the choice of $b \in \sigma(B)$ for which the probability of observing it is non-vanishing is dependent on $g$. Hence, for each $b \in \sigma(B)$, we must first guarantee its well-definedness, at least on some neighbourhood of $g=0$.
Fortunately, this is indeed the case for the choice $b \in \sigma(B)$ such that the probability of finding it on the initial state $|\phi\rangle$ of the target system $\mathbb{E}\left[ \Pi_{b} \otimes I; \Psi^{0} \right] = \|\Pi_{b} \phi \|^{2} \neq 0$ is non-vanishing, due to continuity of the function $g \mapsto \mathbb{E}\left[ \Pi_{b} \otimes I; \Psi^{g} \right]$. The main objective of this passage is to demonstrate the following statement.

\begin{proposition}[Differentiability of the Conditional Expectation: Preliminary]\label{prop:diff_cond_exp}
Suppose that the conditioning observable $B$ has spectrum of finite cardinality, and moreover let $|\phi\rangle \in \mathrm{dom}(A)$, $|\psi\rangle \in \mathcal{D} \subset \mathrm{dom}(X)$ (the subspace $\mathcal{D}$ is defined as in \eqref{eq:weyl_subspace}) be assumed, so that the conditional expectation $\mathbb{E}[X|B;\Psi^{g}]$ is well-defined for all range of $g \in \mathbb{R}$. Then for $b \in \sigma(B)$ such that $\|\Pi_{b} \phi \|^{2} \neq 0$, the conditional expectation $\mathbb{E}[X|B = b;\Psi^{g}]$ is well-defined on some neighbourhood of $g=0$. It is moreover differentiable with respect to $g$ at the origin, for which the differential coefficient reads
\begin{align}\label{eq:wpsm_diff}
&\left. \frac{d}{dg} \mathbb{E}[X|B = b; \Psi^{g}] \right|_{g=0} \nonumber \\
    & \qquad = 2\,\mathrm{Re} \left[\frac{\langle \phi, \Pi_{b}A\phi \rangle}{\|\Pi_{b}\phi\|^{2}}\right] \cdot \mathbb{CV}_{\mathrm{A}}[X,Y; \psi] +
   2\,\mathrm{Im}\left[\frac{\langle \phi, \Pi_{b}A\phi \rangle}{\|\Pi_{b}\phi\|^{2}}\right] \cdot \mathbb{CV}_{\mathrm{S}}[X,Y; \psi].
\end{align}
\end{proposition}
\noindent
Here, we have introduced the quantities,
\begin{align}\label{def:q_covariance}
\mathbb{CV}_{\mathrm{S}}[X,Y; \psi] &:= \mathbb{E}[\{X,Y\}/2; \psi] - \mathbb{E}[X; \psi]\mathbb{E}[Y; \psi], \\
\mathbb{CV}_{\mathrm{A}}[X,Y; \psi] &:= \mathbb{E}[[X,Y]/(2i); \psi],
\end{align}
occasionally called the symmetric and anti-symmetric (quantum) covariance%
\footnote{
Note that in the case where the two observables coincide $X=Y$, the symmetric quantum covariance reduces to the familiar variance,
\begin{equation}\label{def:variance}
\mathbb{CV}_{\mathrm{S}}[X,X; \psi] = \mathbb{V}[X;\psi] := \mathbb{E}[X^{2};\phi] - \mathbb{E}[X;\phi]^{2},
\end{equation}
which is reminiscent of the familiar result in classical probability theory, whereas the anti-symmetric covariance reduces to null $\mathbb{CV}_{\mathrm{A}}[X,X; \psi] = 0$.
}
of $X$ and $Y$ on the state $|\psi\rangle \in \mathcal{D}$, respectively, where $\{X,Y\} := XY + YX$ denotes the anti-commutator (not to be confused with the braces denoting sets).

\begin{proof}
Throughout the proof, we choose $b \in \sigma(B)$ such that $\|\Pi_{b} \phi \|^{2} \neq 0$. Then, it is fairly straightforward to see that the map
\begin{equation}\label{eq:cond_exp_val}
\mathbb{E}\left[X | B=b ; \Psi^{g} \right] = \frac{\mathbb{E}\left[ \Pi_{b} \otimes X; \Psi^{g} \right]}{\mathbb{E}\left[ \Pi_{b} \otimes I; \Psi^{g} \right]}, \quad g \in U_{0},
\end{equation}
is well-defined on some neighbourhood $U_{0}$ around the origin $g=0$. It then follows directly from the expression \eqref{eq:cond_exp_val} that the differentiability of both the numerator and the denominator of the r.~h.~s. gives a sufficient condition for the conditional expectation $\mathbb{E}[X|B = b; \Psi^{g}]$ to be differentiable. In order to simplify our notations, we assume in the following that all the vectors $|\phi\rangle$ and $|\psi\rangle$, respectively representing the initial quantum states of the target and the meter system, are normalised. Since the proof is rather lengthy, we divide it into several parts. 
\paragraph{Leibniz Rule}

To prepare for our arguments, we first recall some basic facts. Let $F, G: U \to \mathcal{H}$ be a map from an open subset $U \subset \mathbb{R}$ of the real line to a Hilbert space $\mathcal{H}$. If both maps $F$ and $G$ are strongly differentiable at $t_{0} \in U$, the inner product $t \mapsto \langle F(t), G(t) \rangle$ is differentiable at $t_{0} \in U$, and the derivative satisfies the Leibniz rule,
\begin{align}\label{thm:leibniz_rule}
\left. \frac{d}{dt} \langle F(t), G(t) \rangle \right|_{t=t_{0}}
    &:= \lim_{u \to 0} \frac{\langle F(u + t_{0}), G(u + t_{0}) \rangle - \langle F(t_{0}), G(t_{0}) \rangle}{u} \nonumber \\
    &= \lim_{u \to 0} \frac{\langle F(u + t_{0}) - F(t_{0}), G(u + t_{0}) \rangle + \langle F(t_{0}), G(u + t_{0}) - G(t_{0}) \rangle}{u} \nonumber \\
    &= \left\langle \frac{dF(t_{0})}{dt} , G(t_{0}) \right\rangle + \left\langle F(t_{0}), \frac{dG(t_{0})}{dt} \right\rangle.
\end{align}

\paragraph{Differentiability of the Numerator}
To prove the differentiability of the numerator of \eqref{eq:cond_exp_val} and obtain its derivative,
we first introduce two auxiliary maps $F(g) := |\Psi^{g}\rangle$ and $G_{X}(g) := (\Pi_{b} \otimes X) F(g)$, by which we rewrite the numerator
\begin{equation}\label{def:psm_numerator}
g \mapsto \mathbb{E}\left[ \Pi_{b} \otimes X; \Psi^{g} \right] = \langle F(g), G_{X}(g) \rangle
\end{equation}
in terms of their inner products. From the Leibniz rule, one sees that the desired result can be immediately obtained once the differentiability of both the maps $F(g)$ and $G_{X}(g)$ are proven and their derivatives are given.

As for the strong differentiability of the map $g \mapsto F(g)$, one readily finds by Stone's theorem on one-parameter unitary groups that the condition
\begin{equation}\label{eq:strong_diff_F}
|\phi\rangle \in \mathrm{dom}(A), \quad |\psi\rangle \in \mathcal{D} \subset \mathrm{dom}(Y)
\end{equation}
would suffice, in which case the derivative is given by
\begin{equation}\label{eq:diff_num_1}
\frac{dF(0)}{dg} = -i(A\otimes Y) |\phi\otimes\psi\rangle.
\end{equation}
As for the map $g \mapsto G_{X}(g)$, we first observe that it is written as
\begin{equation}
G_{X}(g) = (\Pi_{b} \otimes I)(I \otimes X) F(g).
\end{equation}
Due to the boundedness (continuity) of the operator $(\Pi_{b} \otimes I)$, strong differentiability of the vector-valued map
$
g \mapsto (I \otimes X) F(g)
$
would give a sufficient condition for $G_{X}(g)$ to be strongly differentiable, which one readily proves under the condition
\begin{equation}
|\phi\rangle \in \mathrm{dom}(A), \quad |\psi\rangle \in \mathcal{D} \subset  \mathrm{dom}(XY) \cap \mathrm{dom}(Y)
\end{equation}
by imitating the arguments we have made starting from \eqref{eq:diff_2_start} with the help of the relation \eqref{eq:weak_weyl_analogue01}. Now that the strong differentiability of both the maps $g \mapsto F(g)$, $G_{X}(g)$ are proven, one finds from the closedness of the self-adjoint operator $(\Pi_{b} \otimes X)$ that
\begin{align}\label{eq:diff_num_2}
\frac{dG_{X}(0)}{dg}
    &= (\Pi_{b} \otimes X) \frac{dF(0)}{dg} \nonumber \\
    &= -i(\Pi_{b} \otimes X)(A\otimes Y)|\phi\otimes\psi\rangle.
\end{align}

Given the results \eqref{eq:diff_num_1} and \eqref{eq:diff_num_2}, the Leibniz rule leads to the desired differentiability of the numerator \eqref{def:psm_numerator}, in which one computes its derivative as
\begin{align}\label{eq:diff_numerator}
 \left. \frac{d}{dg} \mathbb{E}\left[ \Pi_{b} \otimes X; \Psi^{g} \right] \right|_{g=0} 
    &= \left\langle \frac{dF(0)}{dg} , (\Pi_{b} \otimes X) F(0) \right\rangle + \left\langle F(0), (\Pi_{b} \otimes X) \frac{dF(0)}{dg} \right\rangle \nonumber \\
    &= 2\,\mathrm{Re}\left[ \left\langle F(0), (\Pi_{b} \otimes X) \frac{dF(0)}{dg} \right\rangle \right] \nonumber \\
    &= 2\,\mathrm{Re}\left[ -i\left\langle F(0) , (\Pi_{b} \otimes X)(A \otimes Y) F(0) \right\rangle \right] \nonumber \\
    &= 2\,\mathrm{Im}\left[ \langle \phi , \Pi_{b}A \phi \rangle \langle \psi, XY \psi \rangle \right] \nonumber \\
    &= 2\,\mathrm{Re}\left[ \langle \phi , \Pi_{b}A \phi \rangle \right] \cdot \mathbb{E}[[X,Y]/(2i);\psi] \nonumber \\
    &\qquad + 2\,\mathrm{Im}\left[ \langle \phi , \Pi_{b}A \phi \rangle \right] \cdot \mathbb{E}[\{X,Y\}/2; \psi], 
\end{align}
where we have used the operator equality
\begin{equation}
XY = \frac{\{X,Y\}}{2} + i\frac{[X,Y]}{2i}
\end{equation}
valid on the subspace $\mathcal{D}$.

\paragraph{Differentiability of the Denominator}
The proof for the differentiability of the denominator $\mathbb{E}\left[ \Pi_{b} \otimes I; \Psi^{g} \right]$ goes essentially the same as that for the numerator, where one readily proves its differentiability at $g=0$ under the condition $|\phi\rangle \in \mathrm{dom}(A)$, $|\psi\rangle \in \mathrm{dom}(Y)$, in which case the derivative reads
\begin{equation}\label{eq:diff_denominator}
\left. \frac{d}{dg} \mathbb{E}[ \Pi_{b} \otimes I; \Psi^{g} ] \right|_{g=0} =  2\,\mathrm{Im}\left[ \langle \phi , \Pi_{b}A \phi \rangle \right] \cdot \mathbb{E}[Y; \psi],
\end{equation}
by formally replacing $X$ with $I$ in \eqref{eq:diff_numerator}.

\paragraph{Final Result}
Combining the above two results \eqref{eq:diff_numerator} and \eqref{eq:diff_denominator}, one concludes that, given the choice $|\phi\rangle \in \mathrm{dom}(A)$ and $b \in \sigma(B)$ with $\|\Pi_{b}\phi\|^{2} \neq 0$ of the target configuration, and $|\psi\rangle \in \mathcal{D}$ for the meter system, the conditional expectation $\mathbb{E}[X|B = b; \Psi^{g}]$ is indeed differentiable at $g=0$.  Its derivative can then be evaluated based on the classical result of calculus (the quotient rule for derivative) as
\begin{align}
&\left. \frac{d}{dg} \mathbb{E}[X|B = b; \Psi^{g}] \right|_{g=0} \nonumber \\
    &\quad = \frac{\left. \frac{d}{dg} \mathbb{E}\left[ \Pi_{b} \otimes X; \Psi^{g} \right]\right|_{g=0} \cdot \mathbb{E}\left[ \Pi_{b} \otimes I; \Psi^{0} \right] - \mathbb{E}\left[ \Pi_{b} \otimes X; \Psi^{0} \right] \cdot \left. \frac{d}{dg} \mathbb{E}\left[\Pi_{b} \otimes I; \Psi^{g} \right]\right|_{g=0} }{\mathbb{E}\left[\Pi_{b} \otimes I; \Psi^{0} \right]^{2}} \nonumber \\
    &\quad = 2\,\mathrm{Re} \left[\frac{\langle \phi, \Pi_{b}A\phi \rangle}{\|\Pi_{b}\phi\|^{2}}\right] \cdot \mathbb{E}[[X,Y]/(2i); \psi] \nonumber \\
    & \qquad + 2\,\mathrm{Im}\left[\frac{\langle \phi, \Pi_{b}A\phi \rangle}{\|\Pi_{b}\phi\|^{2}}\right] \cdot \left( \mathbb{E}[\{X,Y\}/2; \psi] - \mathbb{E}[X; \psi]\mathbb{E}[Y; \psi] \right) \nonumber \\
    &\quad = 2\,\mathrm{Re} \left[\frac{\langle \phi, \Pi_{b}A\phi \rangle}{\|\Pi_{b}\phi\|^{2}}\right] \cdot \mathbb{CV}_{\mathrm{A}}[X,Y; \psi] + 2\,\mathrm{Im}\left[\frac{\langle \phi, \Pi_{b}A\phi \rangle}{\|\Pi_{b}\phi\|^{2}}\right] \cdot \mathbb{CV}_{\mathrm{S}}[X,Y; \psi].
\end{align}
We have thus verified our desired statement \eqref{eq:wpsm_diff}.
\end{proof}

\subsubsection{Conditional Quasi-expectations of Quantum Observables}\label{sec:psI_cond_quasi_exp}

Now that we have computed the derivative of the map \eqref{def:wpsm_special} for the special case, we are now interested in the case in which the conditioning observable $B$ is general, and wish to specify the limit of the formal expression
\begin{equation}
\lim_{g \to 0} \frac{\mathbb{E}[X|B; \Psi^{g}] - \mathbb{E}[X|B; \Psi^{0}]}{g}
\end{equation}
and the topology in which the convergence is meant. From the result of Proposition~\ref{prop:diff_cond_exp}, one might naturally conjecture that the limit is given by
\begin{equation}
2\,\mathrm{Re} f \cdot \mathbb{CV}_{\mathrm{A}}[X,Y; \psi] + 2\,\mathrm{Im}f \cdot \mathbb{CV}_{\mathrm{S}}[X,Y; \psi]
\end{equation}
with a `function' $f$ defined formally as
\begin{equation}\label{def:cond_quasi_exp_prelim}
f(b) :=
    \frac{\langle \phi, \Pi_{b}A\phi \rangle}{\|\Pi_{b}\phi\|^{2}}.
\end{equation}
In order to make this observation a precise mathematical statement, we first introduce a convenient concept.

\paragraph{Conditional Quasi-expectations}
Observing that in the case where $A$ and $B$ are simultaneously measurable, the function \eqref{def:cond_quasi_exp_prelim} is nothing but the conditional expectation of $A$ given $B$. In general, however, the target observable and the conditioning observable $B$ need not be simultaneously observable. We thus wish to define a quantum analogue of conditional expectations of an observable $A$ given another observable $B$, well-defined even for the pair that are not necessarily simultaneously measurable. To this end, we first fix a non-zero vector $|\phi\rangle \in \mathrm{dom}(A)$ and consider a complex measure
\begin{equation}\label{def:nu_temp}
\nu(\Delta) := \langle \phi, E_{B}(\Delta)A\phi \rangle/\|\phi\|^{2},
\quad \Delta \in \mathfrak{B},
\end{equation}
where $E_{B}$ is the unique spectral measure accompanying $B$.
Now, a direct application of the Cauchy-Schwarz inequality leads to
\begin{equation}
| \langle \phi, E_{B}(\Delta)A\phi \rangle | \leq \| E_{B}(\Delta)\phi\| \cdot \|A\phi\|,
\end{equation}
by which one finds the absolute continuity $\nu \ll \mu_{B}^{\phi}$, where $\mu_{B}^{\phi}(\Delta) := \| E_{B}(\Delta) \phi\|^{2}/\|\phi\|^{2}$ as usual. This allows us to define the Radon-Nikod{\'y}m derivative
\begin{equation}\label{def:cond_quas-exp}
\mathbb{E}[A|B;\phi] := d\nu/d\mu_{B}^{\phi}.
\end{equation}
By definition, it is the unique $\mu_{B}^{\phi}$-integrable (equivalence class of) function(s) that satisfies
\begin{equation}
\langle \phi, E_{B}(\Delta)A\phi \rangle / \|\phi\|^{2} = \int_{\Delta} \mathbb{E}[A|B = b;\phi]\ d\mu_{B}^{\phi}(b), \quad \Delta \in \mathfrak{B},
\end{equation}
and as such,
\begin{equation}
\mathbb{E}[A;\phi] = \int_{\mathbb{R}} \mathbb{E}[A|B = b;\phi]\ d\mu_{B}^{\phi}(b)
\end{equation}
holds in particular. Incidentally, when the state $|\phi_{a}\rangle \in \mathrm{dom}(A)$ happens to be an eigenvector of $A$ with the eigenvalue $a$, the map
\begin{equation}
\mathbb{E}[A|B;\phi_{a}] = a
\end{equation}
becomes a constant function independent of the choice of the conditioning observable $B$. The map $\mathbb{E}[A|B ;\phi]$ thus shares properties similar to the conditional expectations, and in the special case in which $A$ and $B$ happens to be simultaneously measurable, it actually reduces to the standard conditional expectation.  However, as one finds shortly below, it can be shown by reductio ad absurdum that the map $\mathbb{E}[A|B ;\phi]$ may not admit itself to be understood as a standard conditional expectation in the case where the pair of observables concerned does not admit coexistence.  These preliminary observations may tempt one to call the map \eqref{def:cond_quas-exp} a \emph{conditional quasi-expectation} of $A$ given $B$.

\paragraph{Arbitrariness to Conditional Quasi-expectations}

As one may immediately notice, there exists an arbitrariness to the way one may define conditional quasi-expectations. For example, one may just define the complex conjugate of the complex measure \eqref{def:nu_temp} as
\begin{equation}
\nu^{*}(\Delta) = \langle \phi, AE_{B}(\Delta)\phi \rangle/\|\phi\|^{2}
\end{equation}
and introduce the Radon-Nikod{\'y}m derivative as
\begin{equation}
\mathbb{E}^{*}[A|B;\phi] := d\nu^{*}/d\mu_{B}^{\phi} = \mathbb{E}[A|B;\phi]^{*}.
\end{equation}
One may conduct analogous reasoning to verify that the function $\mathbb{E}^{*}[A|B;\phi]$ also satisfies properties similar to the usual conditional expectations, and that both definitions coincide when the pair of $A$ and $B$ happens to be simultaneously measurable. One may even consider a complex linear combination of $\mathbb{E}[A|B;\phi]$ and its complex conjugate to define
\begin{align}\label{def:cond_quasi-exp_alpha}
\mathbb{E}^{\alpha}[A|B;\phi]
    &:= \frac{1+\alpha}{2} \cdot \mathbb{E}[A|B;\phi] + \frac{1-\alpha}{2} \cdot \mathbb{E}^{*}[A|B;\phi] \nonumber \\
    &= \mathrm{Re}\left[ \mathbb{E}[A|B;\phi] \right] + \alpha i\, \mathrm{Im}\left[ \mathbb{E}[A|B;\phi] \right] , \quad \alpha \in \mathbb{C},
\end{align}
for example, so that $\mathbb{E}^{1}[A|B;\phi] = \mathbb{E}[A|B;\phi]$ and $\mathbb{E}^{-1}[A|B;\phi] = \mathbb{E}^{*}[A|B;\phi]$. In fact, it reveals that there exists a multitude of potential candidates for possible definitions of such `conditional quasi-expectations', all sharing desirable properties mentioned earlier. We shall be returning to this problem in a more general framework of quasi-joint-probabilities of quantum observables in Section~\ref{sec:qp_qo}, but for our purpose and the scope of this paper, it suffices to concentrate only on the family \eqref{def:cond_quasi-exp_alpha} for definiteness, and we thus introduce:
\begin{definition*}[Conditional Quasi-expectation of $A$ given $B$]
Let $A$ and $B$ be self-adjoint operators on a Hilbert space $\mathcal{H}$, and let $E_{B}$ be the spectral measure of $B$. For a given state $|\phi\rangle \in \mathrm{dom}(A)$, we call the family of complex linear combinations of the Radon-Nikod{\'y}m derivatives \eqref{def:cond_quasi-exp_alpha}
the complex-parametrised family of conditional quasi-expectations of $A$ given $B$.  They are, by definition, a (family of) complex function(s) defined on the spectrum $\sigma(B)$.
\end{definition*}
\noindent
Note, by definition, that each member $\mathbb{E}^{\alpha}[A|B;\phi]$, $\alpha \in \mathbb{C}$, of the family of conditional quasi-expectations is integrable with respect to the probability measure $\mu_{B}^{\phi}$, and its total integration coincides with the expectation value $\mathbb{E}[A;\phi]$ of $A$.
If the conditioning observable $B$ happens to possess spectrum with finite cardinality, so that its spectral decomposition reads \eqref{eq:spectr_decomp_B_fin}, the conditional quasi-expectation admits an expression by operators and vectors as
\begin{equation}\label{eq:quasi-cond_exp_fin}
\mathbb{E}[A|B = b;\phi] =
    \begin{cases}
    \langle \phi, \Pi_{b}A\phi \rangle / \|\Pi_{b}\phi\|^{2}, \quad &(b \in \sigma(B),\ \|\Pi_{b}\phi\|^{2} \neq 0), \\
    \text{indefinite}, \quad &(\text{else})
    \end{cases}
\end{equation}
and
\begin{equation}
\mathbb{E}^{\alpha}[A|B;\phi] = \frac{1+\alpha}{2} \cdot \mathbb{E}[A|B;\phi] + \frac{1-\alpha}{2} \cdot \mathbb{E}[A|B;\phi]^{*}
\end{equation}
if explicitly written out.

\paragraph{Conditional Quasi-expectations, Two-state Values and the Weak Value}
Incidentally, if the conditioning observable happens to be a projection $B=|\phi^{\prime}\rangle\langle\phi^{\prime}|$ on a one-dimensional subspace of $\mathcal{H}$ spanned by a unit vector $|\phi^{\prime}\rangle$ ({\it i.e.}, a post-selection), the conditional quasi-expectation of $A$ given the outcome $B=1$ reads
\begin{align}\label{eq:qce_and_tsv}
\mathbb{E}^{\alpha}[A|B = 1;\phi]
    &= \frac{1+\alpha}{2} \cdot \frac{\langle \phi^{\prime}, A\phi \rangle}{\langle \phi^{\prime}, \phi \rangle} + \frac{1-\alpha}{2}  \cdot \frac{\langle \phi, A\phi^{\prime} \rangle}{\langle \phi, \phi^{\prime} \rangle},
\end{align}
given that the probability of finding the outcome $1$ of $B$ is non-vanishing $\mu_{B}^{\phi}(\{1\}) = |\langle \phi^{\prime}, \phi \rangle|^{2} \neq 0$.  Specifically for the choice $\alpha = 1$, this reduces to
\begin{equation}\label{def:weak_value}
\mathbb{E}[A|B = 1;\phi] = \frac{\langle \phi^{\prime}, A\phi \rangle}{\langle \phi^{\prime}, \phi \rangle} =: A_{w},
\end{equation}
The value $A_{w}$ is widely referred to as Aharonov's \emph{weak value} \cite{Aharonov_1964,Aharonov_1988} of $A$ for the pair of the pre-selected state $|\phi \rangle \in \mathrm{dom}(A)$ and the post-selected state $|\phi^{\prime} \rangle \in \mathcal{H}$. Historically, the weak value is said to have been originally introduced as a hypothetical value of an observable $A$ assigned to a quantum \emph{process} from the pre-selected to the post-selected state, generalising the common practice of solely assigning values to a single static \emph{state} in the standard framework of quantum mechanics. Following this philosophy, the value \eqref{eq:qce_and_tsv} termed the \emph{two-state value} \cite{Morita_2013} of $A$ under the respective selections of states was recently introduced in an attempt to generalise the idea of the weak value and to find out the possible form of a quantity of an observable specified by two quantum states. An application of the generalised Gleason's theorem revealed that, under certain desirable conditions, the most general form of the values of an observable $A$ that can be assigned to the two specification of the quantum states $|\phi\rangle \in \mathrm{dom}(A)$, $|\phi^{\prime}\rangle \in \mathcal{H}$ satisfying $\langle \phi^{\prime}, \phi\rangle \neq 0$ is given by \eqref{eq:qce_and_tsv} with a parameter $\alpha \in \mathbb{C}$ representing the ambiguity inherent to it.

\paragraph{Essential Supremum of Conditional Quasi-expectations}

While conditional quasi-expectations and the standard conditional expectations share various properties in common, the non-commutative nature of quantum observables results in some interesting distinctions between the two concepts.  In this paper, as an example, we shall focus on the remarkable difference in the behaviour of their essential suprema. Now, as one recalls from Corollary~\ref{cor:cond_exp_obs_ess_sup}, for a pair of \emph{simultaneously measurable} observables $A$ and $B$ and a fixed state $|\phi\rangle \in \mathrm{dom}(A)$, the essential supremum of the conditional expectation $\left\|\, \mathbb{E}[A|B;\phi] \,\right\|_{\infty}$ is never greater than the essential supremum $\|A\|_{\infty}^{\phi}$ of the measurable function $a \mapsto a$ under the probability measure $\mu_{A}^{\phi}$. If $A$ happens to be bounded, the operator norm $\|A\|$ gives the state independent universal upper bound to $\|A\|_{\infty}^{\phi}$, which in turn also naturally becomes an upper bound to the conditional expectation $\mathbb{E}[A|B;\phi]$.
However, in general, this property is no longer preserved when $A$ and $B$ fail to be simultaneously measurable. There are several possible ways to express this discrepancy, but for brevity, we formulate it in the following manner.

To this end, we first prepare a terminology. In this paper, we say that an observable $A$ on $\mathcal{H}$ is \emph{non-trivial} if $A$ is not a scalar multiple of the identity operator $tI$, $t \in \mathbb{R}$, or equivalently, if $A$ has a spectrum $\sigma(A)$ of cardinality not less than $2$. Note that the non-triviality of $A$ automatically implies $\mathrm{dim}(\mathcal{H}) \geq 2$, where $\mathrm{dim}(\mathcal{H})$ denotes the dimension of the Hilbert space $\mathcal{H}$. Since trivial operators strongly commute with any other self-adjoint operators, the function $\mathbb{E}^{\alpha}[A|B;\phi]$ always become an authentic conditional expectation, revealing itself to be a constant function always taking its unique eigenvalue $\mathbb{E}^{\alpha}[A|B;\phi] = t$, whose case is not interesting for our purpose. Hence, we shall from now on confine ourselves to the case where $A$ is non-trivial.

\begin{proposition}[Essential Supremum of Conditional Quasi-expectations]\label{prop:lim_amp_quasi_cond}
Let $A$ be a non-trivial observable, $|\phi\rangle \in \mathrm{dom}(A)$ a vector that is not an eigenvector of $A$, and let $\alpha \in \mathbb{C}$ be any choice of the ambiguity parameter of the conditional quasi-expectation.
Then, for any non-negative number $0 \leq M < \infty$, there exists a self-adjoint operator $B$ (not-necessarily simultaneously measurable with $A$) such that the essential supremum of the conditional quasi-expectation of $A$ given $B$ is not less than
\begin{equation}
M \leq \left\|\, \mathbb{E}^{\alpha}[A|B;\phi] \,\right\|_{\infty}.
\end{equation}
Specifically, one may always choose such conditioning observable $B = |\phi^{\prime}\rangle\langle\phi^{\prime}|$ to be a projection onto a one-dimensional subspace of $\mathcal{H}$ spanned by some unit vector $|\phi^{\prime}\rangle \in \mathcal{H}$.
\end{proposition}
\begin{proof}
It suffices to prove that, one may always adjust the choice of the conditioning observable $B = |\phi^{\prime}\rangle\langle\phi^{\prime}|$ so that the conditional quasi-expectation
\begin{equation}
\mathbb{E}^{\alpha}[A|B = 1;\phi] =
    \begin{cases}
         c, \quad (c \in \mathbb{C}), &(\alpha \neq 0) \\
         r, \quad (r \in \mathbb{R}), &(\alpha = 0)
    \end{cases}
\end{equation}
may take any complex number for the choice $\alpha \neq 0$, and any real number for the choice $\alpha = 0$, while maintaining the probability of observing it to be non-vanishing $\mu_{B}^{\phi}(\{1\}) > 0$.
The proof is a direct corollary of Proposition~\ref{prop:tsv} that follows immediately.
\end{proof}
\noindent
In particular, this result is to say that one may always choose a conditioning observable $B$ such that the essential supremum $\left\|\, \mathbb{E}[A|B;\phi] \,\right\|_{\infty}$ of the conditional quasi-expectation exceeds $\|A\|_{\infty}^{\phi}$, which is never possible for standard conditional expectations defined for a pair of simultaneously measurable observables. 
This `amplification of conditional quasi-expectations' is a noteworthy property of quantum mechanics, and the oft-discussed `amplification of weak values' could be understood as its special case.

\begin{proposition}[Range of the Two-state Value]\label{prop:tsv}
Let $A$ be a non-trivial observable on $\mathcal{H}$, and let $|\phi\rangle \in \mathrm{dom}(A)$ be a pre-selected state that is not an eigenvector of $A$. Then, the two-state value of $A$ under the pre-selected state $|\phi\rangle$ may take any complex number in the case $\alpha \neq 0$, and in turn any real number in the case $\alpha = 0$, given an appropriate choice of the post-selected state $|\phi^{\prime}\rangle \in \mathcal{H}$.
\end{proposition}
\begin{proof}
For simplicity, we only provide the proof of the statement for the specific choice $\alpha = 1$ of the ambiguity parameter without loss of generality.

Now, before we go into the main part of the proof, we first observe that, for a non-trivial self-adjoint operator $A$ and a normalised vector $|\phi\rangle \in \mathrm{dom}(A)$, there exists a normalised vector $|\chi\rangle \in \mathcal{H}$ orthogonal to $|\phi\rangle$ such that
\begin{equation}\label{eq:tsv_prop_lem}
A|\phi\rangle = \mathbb{E}[A;\phi] \cdot |\phi\rangle + \left\| \left( A - \mathbb{E}[A;\phi] \right) \phi \right\| \cdot |\chi\rangle
\end{equation}
holds%
\footnote{Note that in the case $|\phi\rangle \in \mathrm{dom}(A^{2})$, the equality \eqref{eq:tsv_prop_lem} is equivalent to
\begin{equation}
A|\phi\rangle = \mathbb{E}[A;\phi] \cdot |\phi\rangle + \sqrt{\mathbb{V}[A;\phi]} \cdot |\chi\rangle,
\end{equation}
since one has $\left\| \left( A - \mathbb{E}[A;\phi] \right) \phi \right\|^{2} = \langle \phi, \left( A - \mathbb{E}[A;\phi] \right)^{2} \phi \rangle = \mathbb{V}[A;\phi]$ with the variance defined as in \eqref{def:variance}.}.
To see this, we first consider the case
\begin{equation}
A|\phi\rangle = \mathbb{E}[A;\phi] \cdot |\phi\rangle,
\end{equation}
that is, when $|\phi\rangle$ is an eigenvector of $A$. Then, by choosing any normalised state $|\chi\rangle$ satisfying $\langle \chi, \phi \rangle = 0$ (the existence of such  $|\chi\rangle$ is guaranteed by the fact $\mathrm{dim}(\mathcal{H}) \geq 2$), one finds that the above equality is fulfilled. Next, suppose that $A|\phi\rangle \neq \mathbb{E}[A;\phi] \cdot |\phi\rangle$. Then, by defining
\begin{equation}
|\chi\rangle := \frac{\left( A - \mathbb{E}[A;\phi] \right) |\phi\rangle}{\left\| \left( A - \mathbb{E}[A;\phi] \right) \phi \right\|},
\end{equation}
one indeed learns that $\|\chi\| = 1$ and $\langle \chi, \phi \rangle = 0$ as stated.

Armed with this fact and by fixing such $|\chi\rangle$, we choose the post-selected state as
\begin{equation}
|\phi^{\prime}\rangle = \frac{1}{c^{*}} |\phi\rangle + |\chi\rangle
\end{equation}
with a free parameter $c \in \mathbb{C}^{\times}$.
One then finds
\begin{align}
\mathbb{E}^{1}[A|B = 1;\phi]
    &= \frac{\langle \phi^{\prime}, A\phi \rangle}{\langle \phi^{\prime}, \phi \rangle} \nonumber \\
    &= \mathbb{E}[A;\phi] \cdot \frac{\langle\phi^{\prime}, \phi\rangle}{\langle\phi^{\prime}, \phi\rangle} + \left\| \left( A - \mathbb{E}[A;\phi] \right) \phi \right\| \cdot \frac{\langle\phi^{\prime}, \chi\rangle}{\langle\phi^{\prime}, \phi\rangle} \nonumber \\
    &= \mathbb{E}[A;\phi] + c \left\| \left( A - \mathbb{E}[A;\phi] \right) \phi \right\|.
\end{align}
This shows that, for the choice of an initial state $|\phi\rangle \in \mathrm{dom}(A)$ that is not an eigenvector of $A$ (which is always possible due to the non-triviality of $A$), the weak value (hence, also the two-state value) may indeed take any complex number by adjusting the free parameter $c$ appropriately.
\end{proof}
\noindent
The difference between (standard) conditional expectations and conditional quasi-expectations in the behaviour of their essential suprema makes it clear that, conditional quasi-expectations are not conditional expectations in the classical sense. This provides an indirect proof for the fact that, in general, \emph{the `joint behaviour' of the outcomes of the pair of (generally non-commuting) quantum observables $A$ and $B$ does not allow itself to be described by probability spaces}.  This would be accounted for in depth in Section~\ref{sec:ps_II} and \ref{sec:qp_qo} shortly.

\subsubsection{Weak Conditioned Measurement}

Armed with our newly introduced concept of conditional quasi-expectations \eqref{def:cond_quasi-exp_alpha} of a quantum observable given another (not necessarily simultaneously measurable) quantum observable, we shall summarise our findings regarding the first-order local behaviour of the conditional expectation at the origin. Combining Proposition~\ref{prop:diff_cond_exp} and \eqref{eq:quasi-cond_exp_fin}, one is naturally tempted to conjecture that:
\begin{proposition}[Weak Conditioned Measurement]
Let $A$ and $B$ be self-adjoint operators defined on the target system $\mathcal{H}$, and let the respective initial states $|\phi\rangle \in \mathrm{dom}(A)$, $|\psi\rangle \in \mathcal{D}$ be fixed. Then, the conditional expectation $\mathbb{E}[X|B;\Psi^{g}]$ is well-defined for all range of $g \in \mathbb{R}$, and the limit converges to
\begin{align}\label{eq:wpsm_diff_02}
\left. \frac{d}{dg} \mathbb{E}[X|B; \Psi^{g}] \right|_{g=0} 
    &:= \lim_{g \to 0} \frac{\mathbb{E}[X|B; \Psi^{g}] - \mathbb{E}[X|B; \Psi^{0}]}{g} \nonumber \\
    &= 2\,\mathrm{Re} \left[ \mathbb{E}[A|B; \phi] \right] \cdot \mathbb{CV}_{\mathrm{A}}[X,Y; \psi] \nonumber \\
    & \qquad + 2\,\mathrm{Im}\left[ \mathbb{E}[A|B; \phi] \right] \cdot \mathbb{CV}_{\mathrm{S}}[X,Y; \psi]
\end{align}
point-wise $\mu_{B}^{\phi}$-almost everywhere.
\end{proposition}
\noindent
While we have explicitly proved the above statement only in the special case where $B$ has spectrum of finite cardinality, the same statement indeed holds for general $B$, although we do not go into the technical details for its demonstration. One may thus understand the process of the weak CM scheme as the practice of measuring (the real and imaginary parts of) the conditional quasi-expectation $\mathbb{E}[A|B; \phi]$ of the target system. This result is to be compared with the unconditioned counterpart, in which one may extract the standard (unconditional) expectation $\mathbb{E}[A; \phi]$ by means of the weak UM scheme from the first-order differential coefficient of the measurement outcomes.

\paragraph{Topic: Conditional Quasi-expectation as the Merkmal for Amplification}
Under the above conditions, Taylor's theorem states that one has the following first-order expansion of the conditional expectation
\begin{align}\label{eq:cexp_expansion}
\mathbb{E}[X|B; \Psi^{g}]
    &= \mathbb{E}[X;\psi] \nonumber \\
    &\quad + g \cdot \big( 2\,\mathrm{Re} \left[ \mathbb{E}[A|B; \phi] \right] \cdot \mathbb{CV}_{\mathrm{A}}[X,Y; \psi] + 2\,\mathrm{Im}\left[ \mathbb{E}[A|B; \phi] \right] \cdot \mathbb{CV}_{\mathrm{S}}[X,Y; \psi] \big) \nonumber \\
    & \qquad + o(g), 
\end{align}
where $o(g)$ (Landau symbol) denotes a member of the class of functions satisfying the asymptotic property
\begin{equation}
\lim_{g \to 0} \frac{o(g)}{|g|} = 0,
\end{equation}
and the equality \eqref{eq:cexp_expansion} is understood to hold $\mu_{B}^{\phi}$-almost everywhere.
The above fact purports that the conditional quasi-expectation $\mathbb{E}[A|B; \phi]$ gives the (best first-order) indicator on the degree of `amplification' of the conditional expectation $\mathbb{E}[X|B; \Psi^{g}]$ one may attain by means of choosing the conditioning observable $B$ on the target system. Colloquially speaking, if one hopes to gain large amplification effect by conditioning, the first place one should look for is its conditional quasi-expectation, and one may hopefully achieve it by adjusting the conditioning observable $B$ so that the conditional quasi-expectation $\mathbb{E}[A|B; \phi]$ becomes large enough. However, note here that while the conditional quasi-expectation (for non-trivial $A$, and in addition, for the choice of the initial state $|\phi\rangle \in \mathrm{dom}(A)$ that is not an eigenvector of $A$) admits arbitrary large amplification by a suitable choice of the conditioning observable $B$ (Proposition~\ref{prop:lim_amp_quasi_cond}), the classical conditional expectation $\mathbb{E}[X|B; \Psi^{g}]$ may have an upper bound depending on its configuration (Proposition~\ref{prop:limit_of_amplification}). This generally suggests that the discrepancies between the full-order behaviour of $\mathbb{E}[X|B; \Psi^{g}]$ and its first-order approximation becomes larger (in other words, the higher-order terms $o(g)$ becomes more significant) as one adjusts the choice of the conditioning observable $B$ so that the conditional quasi-expectation may become larger. As for the higher-order terms, although we shall omit details, we note that one may also prove higher-order differentiability of the conditional expectation by placing stricter conditions for the choice of both the initial states of the target system and the meter system, and subsequently compute higher-order derivatives through analogous procedure as demonstrated above. 

In order to confirm this observation with a concrete model, we have included in Appendix~\ref{sec:PSM} an analytic example where we compute the conditional expectation $\mathbb{E}[X|B; \Psi^{g}]$ for the special case in which the conditioning observable $B = |\phi^{\prime}\rangle\langle\phi^{\prime}|$ is a projection onto a one-dimensional subspace spanned by a unit vector $|\phi^{\prime}\rangle \in \mathcal{H}$ ({\it i.e.}, the post-selected measurement scheme), and moreover the target observable $A$ is dichotomic.  One shall indeed find the existence of the limit of `amplification' of the conditional expectation by the `weak value amplification', and various other general properties alongside that we found in the discussions throughout this section.

\newpage
\section{Conditioned Measurement II: In Terms of Conditional Probabilities}\label{sec:ps_II}

In Section~\ref{sec:ups_II}, we have elaborated the study of the UM scheme conducted in the preceding Section~\ref{sec:ups_I} in terms of probabilities. In this section, we follow the same line and intend to refine our analysis for the conditioned counterpart.

\paragraph{Preliminary Observations}
As one may recall, we have seen in Section~\ref{sec:ups_I} and Section~\ref{sec:ups_II} that, by means of the UM scheme, one could extract the information of the target system in both the form of the expectation value $\mathbb{E}[A;\phi]$ and the probability measure $\mu_{A}^{\phi}$, the former by looking at the expectation value of the meter observable $X$ conjugate to $Y$, whereas the latter by focusing at the probability measure of it, and they were obtained by either inspecting the strong region $g \to \pm \infty$ of the interaction or by probing its local behaviour at $g=0$, both in parallel manners.

Now, as for the conditioned case, while we have not looked into the strong region $g \to \pm \infty$ of the interaction parameter, our analysis on the local behaviour conducted in Section~\ref{sec:ps_I} revealed that the first-order derivative of the expectation value for the choice $X = Q, P$ both contain potions (real and imaginary parts) of the conditional quasi-expectation $\mathbb{E}[A|B;\phi]$.
By comparing this result to the unconditioned case, one may come to a na{\"i}ve conjecture that both the expectation value and the conditional quasi-expectation of an observable $A$ has some quality in common. Namely, since the CM scheme incorporate conditioning, one may speculate that the conditional quasi-expectation $\mathbb{E}[A|B;\phi]$ may be interpreted as some form of a `conditional average' with respect to an underlying `probability distribution' of some kind.

\paragraph{Quasi-joint-probability Distributions in Quantum Mechanics}
A quick observation on our previous result \eqref{eq:wpsm_diff} reveals that, the full description of the CM scheme must incorporate the information of the measurement outcomes of \emph{both} the choice $X=Q, P$ of the meter observables, which is in contrast to the unconditioned case where we may concentrate only on the analysis of the probability distribution describing the outcome of a \emph{single} observable $X$ that is conjugate to $Y$. 
In view of this, it would thus be natural to consider some form of a `joint-distribution' describing the measurement outcome of both the observables $Q$ and $P$.  However, as we have seen in Section~\ref{sec:sim_meas_obs}, and also from an indirect proof by observing the difference of conditional (quasi)-expectations in their behaviour regarding essential suprema that, by definition, only a pair of observables that are simultaneously measurable admits a description by joint-probability distributions in the classical sense and, unfortunately, the pair $\{Q, P\}$ of observables of our present interest does not fall into this category. 

On account of this, there have been various attempts to construct some alternative form of `joint-distributions' for pairs of (generally non-commuting) quantum observables that possess convenient or desirable properties in describing the behaviour of both their outcomes. 
The Wigner-Ville distribution (WD distribution) \cite{Wigner_1932,Ville_1948}, which purports to describe the `joint behaviour' of the otherwise incompatible pair of observables $\hat{x}$ and $\hat{p}$ on the normalised wave-function $\psi \in L^{2}(\mathbb{R})$, symbolically defined by
\begin{equation}\label{def:Winger_Distr}
W^{\psi}(x,p) := \frac{1}{\pi} \int_{\mathbb{R}} \psi^{*}(x + y) \psi(x - y)e^{2ipy}\ d\beta(y),
\end{equation}
and the Kirkwood-Dirac distribution (KD distribution) \cite{Kirkwood_1933, Dirac_1945}, which on the other hand allows itself to be defined for arbitrary pair of observables $A$ and $B$, symbolically defined by
\begin{equation}\label{def:Kirkwood-Dirac_Distr}
K_{A,B}^{\phi}(a,b) := \frac{\langle \phi, b\rangle\langle b, a \rangle\langle a, \phi \rangle}{\|\phi\|^{2}}
\end{equation}
with the symbolical decomposition $A = \int_{\mathbb{R}} a\ d|a\rangle\langle a|$, $B = \int_{\mathbb{R}} b\ d|b\rangle\langle b|$, are among the most well-known classical proposals. The former allows negative numbers to be assigned, whereas the latter even admits complex numbers. Despite their queerness, they both retain some properties that one finds common in the standard ({\it i.e.,} real and non-negative) joint-probability distributions, {\it e.g.}, that they both have total integration of unity, and that the marginals coincide with the probability distribution describing the behaviour of the remaining observable, and in this sense, they are occasionally referred to as \emph{quasi-joint-probability} (QJP) distributions of the specific pairs of observables.

\paragraph{Quasi-joint-probability Distributions and Conditional Quasi-expectations}

Now, as some may expect, conditional quasi-expectations are closely related to the notion of quasi-joint-probabilities in quantum mechanics. Indeed, a quick observation reveals that, given a symbolical spectral decomposition $B = \int_{\mathbb{R}} b\ d|b\rangle\langle b|$ of the conditioning observable, the complex-parametrised conditional quasi-expectation \eqref{def:cond_quasi-exp_alpha} for the choice $\alpha =1$ coincides with the, so to speak, `conditional average' of $A$ given the outcome $b$ of $B$ under the Kirkwood-Dirac distribution, as one finds under the formal computation
\begin{equation}
\mathbb{E}^{1}[A | B=b;\phi] = \frac{\langle \phi, b\rangle\langle b, A\phi \rangle}{|\langle b, \phi \rangle|^{2}} = \frac{\int_{\mathbb{R}} a\ K_{A,B}^{\phi}(a,b)d\beta(a)}{\int_{\mathbb{R}} K_{A,B}^{\phi}(a,b)d\beta(a)}.
\end{equation}
As for the Wigner-Ville distribution, pure realness of its values might lead one to think that this is in some form related to the parametrised conditional quasi-expectation for the choice $\alpha = 0$. Indeed, one confirms under the formal computation
\begin{align}
\mathbb{E}^{0}[\hat{p} | \hat{x}=x;\phi]
    &:= \frac{1}{2} \left[ \frac{\langle x, \hat{p} \psi \rangle}{\langle x, \psi \rangle} + \left(\frac{\langle x, \hat{p} \psi \rangle}{\langle x, \psi \rangle}\right)^{*} \right] \nonumber \\
    &= |\langle x, \psi \rangle|^{-2} \cdot \frac{\left[ \langle \psi, x\rangle\langle x, \hat{p} \psi \rangle + (\langle \psi, x\rangle\langle x, \hat{p} \psi \rangle)^{*}\right]}{2} \nonumber \\
    &= |\langle x, \psi \rangle|^{-2} \cdot \left. (-i)^{-1}\frac{d}{dy} \left( \frac{\langle \psi, e^{-iy\hat{p}} x\rangle\langle x, e^{-iy\hat{p}} \psi \rangle}{2} \right)\right|_{y=0} \nonumber \\
    &= |\langle x, \psi \rangle|^{-2} \cdot \left. (-i)^{-1} \frac{d}{dy}\left( \frac{\psi^{*}(x+y)\psi(x-y)}{2} \right) \right|_{y=0} \nonumber \\
    &= |\langle x, \psi \rangle|^{-2} \cdot \frac{1}{2\pi} \int_{\mathbb{R}} p \psi^{*}(x + y) \psi(x - y)e^{2ipy}\ d\beta(y) \nonumber \\
    &= \frac{\int_{\mathbb{R}} p\ W^{\psi}(x,p)d\beta(p)}{\int_{\mathbb{R}} W^{\psi}(x,p) d\beta(p)},
\end{align}
that the conditional quasi-expectation of $\hat{p}$  given $\hat{x}=x$ for the choice $\alpha = 0$ coincides with the, again so to speak, `conditional average' of the momentum $
\hat{p}$ given the outcome $x$ of the position $\hat{x}$ under the Wigner-Ville distribution.

\paragraph{Conditioned Measurement}
The above observation is instructive in guiding the direction of our analysis. Indeed, it would be natural to expect that the measurement of the meter system in view of QJP distributions of the pair of observables $\{Q,P\}$ would allow us to extract the information of the target system in the form that is `akin' to it, {\it i.e.}, one might hope to obtain a QJP distributions of the target system, of which `conditional average' coincides with the conditional quasi-expectations $\mathbb{E}^{\alpha}[A|B = b]$ of our interest.
Guided by this formal argument and heuristic observation, in this section, we shall be analysing the CM scheme in terms of quasi-probabilities, or more specifically, in terms of `conditional' quasi-probabilities.
Now, as our previous arguments (in particular, those developed in Section~\ref{sec:ups_II_wups}) indicate, analysis directly on the level of probabilities is better suited to be performed in the space of generalised functions, rather than density functions or measures, if one is to conduct it with decent mathematical rigour and generality. This becomes especially crucial when introducing `quasi-joint-probabilities' of a pair of (generally not necessarily simultaneously measurable) quantum observables, which is one of the main themes of this paper, and thus examined in depth in the next Section~\ref{sec:qp_qo}. However, since the present authors have judged the theory of generalised functions to be beyond the scope of this paper as a tool for analysis, we shall be working exclusively in the space of complex measures and density functions as usual. While this treatment comes with some unavoidable compromise on generality of the results and loss of transparency of the line of arguments, we hope that we may still convey the essence of the contents.

\paragraph{Conditioned Measurement in View of the WV Distributions}
In this section, the target of our interest for our measurement is the QJP distribution of the pair of observables $Q$ and $P$ on the meter, and we shall study how one may extract information of the configuration of the target system from this viewpoint. Now, as one may realise from the two concrete classical proposals given above (namely, WV distribution and KD distribution), there exist an indefiniteness/arbitrariness to the choice of such distributions, and by its very nature, one may equally conduct the analysis in view of any of one's own selection.
In this section, we shall be analysing the CM scheme exclusively in terms of the Wigner-Ville distribution.
The primary reason for our choice is merely based on its degree of familiarity in the physics community, and as mentioned above, the choice is essentially arbitrary. One may naturally conduct the same type of analysis in view of another type of quasi-probability distribution ({\it e.g.,} the Kirkwood-Dirac type) in a similar manner and obtain analogous results, or may treat them collectively from a more general viewpoint (more to this in Section~\ref{sec:qp_qo}).

\subsection{Reference Materials}

As usual, we first make a brief review on the basic concepts and facts that are used in our later discussion.

\subsubsection{Conditional Probabilities}

We first introduce some basic definitions and results on the topic of conditioning of probability measures and some intricacies inherent to it.
Let $(X, \mathfrak{A}, \mu)$ be a probability space, and let $B \in \mathfrak{A}$ such that $\mu(B) \neq 0$. For $A \in \mathfrak{A}$, we define the \emph{conditional probability of $A$ given $B$} by the number
\begin{equation}
\mu(A|B) := \frac{\mu(A \cap B)}{\mu(B)}, \quad A \in \mathfrak{A}.
\end{equation}
It is immediate that the map $\mu(\,\cdot\,|B)$ is itself a probability measure satisfying the relation
\begin{equation}\label{def:cond_prob_elementary}
\mu(A \cap B) = \mu(A|B) \cdot \mu(B), \quad A \in \mathfrak{A}.
\end{equation}

\paragraph{Conditional Probability given a Sub-$\sigma$-Algebra}
We now intend to generalise the elementary definition above to suit our further needs. In parallel to the manner we have done for conditional expectations in the previous section, let $\mathfrak{B} \subset \mathfrak{A}$ be a sub-$\sigma$-algebra, and for each measurable set $A \in \mathfrak{A}$, we define the \emph{conditional probability of $A$ given $\mathfrak{B}$} by
\begin{equation}\label{def:cond_prob}
\mu(A|\mathfrak{B}) := \mathbb{E}[\chi_{A} | \mathfrak{B}], \quad A \in \mathfrak{A},
\end{equation}
where $\chi_{A}$ is the characteristic function \eqref{def:characteristic_function} of $A$. For fixed $A \in \mathfrak{A}$, note that by definition, the conditional probability \eqref{def:cond_prob} is understood as an equivalence class of a \emph{family} of $\mu|_{\mathfrak{B}}$-integrable functions by identifying those that are indistinguishable under the given probability measure $\mu|_{\mathfrak{B}}$.
In the simplest case where $\mathfrak{B} = \sigma(B) = \{\emptyset, B, B^{c}, X\}$ given some $B \in \mathfrak{A}$, the conditional probability $\mu(A|\sigma(B))$ satisfies
\begin{align}
\mu(A \cap B)
    &= \int_{B} \chi_{A}\ d\mu \nonumber \\
    &= \int_{B} \mathbb{E}[\chi_{A} | \sigma(B)]\ d\mu|_{\sigma(B)}
    = \mu(A|\sigma(B)) \cdot \mu(B), \quad A \in \mathfrak{A}.
\end{align}
This clarifies the relation between the general definition \eqref{def:cond_prob} and the elementary definition \eqref{def:cond_prob_elementary}.

\paragraph{Conditional Probability given a Function}
Now, under the condition above, instead of being given a sub-$\sigma$-algebra, suppose that one is given a measurable function $g: X \to Y$ for conditioning. We thus define
\begin{equation}
\mu(A|g) := \mu(A|\mathcal{I}(g)), \quad A \in \mathfrak{A},
\end{equation}
to be the conditional probability of $A$ given $g$, where $\mathcal{I}(g)$ is the initial $\sigma$-algebra of $g$ (see \eqref{def:initial_sigma_algebra} for its definition), and also introduce 
\begin{equation}
\mu(A| g = y) := \mu(A|g)(y), \quad y \in \mathbb{R},
\end{equation}
of which notation involves subtlety regarding the choice of the representative, in parallel to the situation of conditional expectations we have seen earlier.

\paragraph{Conditional Probabilities as Equivalent Classes of Functions}

Given a probability space $(X, \mathfrak{A}, \mu)$ and a sub-$\sigma$-algebra $\mathfrak{B} \subset \mathfrak{A}$, the conditional probability $\mu(\,\cdot\, |\mathfrak{B})$ satisfies properties analogous to those of probability measures, namely
\begin{enumerate}
\item $\mu(\emptyset |\mathfrak{B}) = 0$, $\mu(X |\mathfrak{B}) = 1$,
\item $\mu( A |\mathfrak{B}) \geq 0, \quad A \in \mathfrak{A}$,
\item for any sequence $(A_{n})_{n \geq 1}$ of pairwise disjoint subsets of $X$, the equality
\begin{equation}
\mu\left(\left. \bigcup_{n=1}^{\infty} A_{n} \right| \mathfrak{B} \right) = \sum_{n=1}^{\infty} \mu(A_{n} | \mathfrak{B})
\end{equation}
holds.
\end{enumerate}
However, the key distinction to be noted between the usual probability measures is that, the above (in)equalities are guaranteed to hold \emph{almost everywhere}, since by definition, conditional probabilities are equivalent classes of functions. It is thus of natural interest whether we could raise the limitation by dropping `validity almost everywhere', which one may occasionally find troublesome.

\paragraph{Transition Kernels}
To this end, we first recall the definition of transition kernels. Let $(X, \mathfrak{A})$ and $(Y, \mathfrak{B})$ be measurable spaces. We say that a map $K : X \times \mathfrak{B} \to [0,\infty]$ that satisfies the conditions
\begin{enumerate}
\item the map $x \mapsto K(x, B)$ is $\mathfrak{A}$-measurable for every $B \in \mathfrak{B}$,
\item the map $B \mapsto K(x, B)$ is a measure on $(Y, \mathfrak{B})$ for every $x \in X$,
\end{enumerate}
a \emph{transition kernel} from $(X, \mathfrak{A})$ into $(Y, \mathfrak{B})$. A transition kernel is said to be \mbox{($\sigma$-)finite} if the map $B \mapsto K(x, B)$ is \mbox{($\sigma$-)finite} for all $x \in X$. If $K$ is normalised to unity $K(x, Y) = 1$ for all $x \in X$, we say that $K$ is a \emph{transition probability kernel}.
Given a $\sigma$-finite transition kernel $K : X \times \mathfrak{B} \to [0, \infty]$ from $(X, \mathfrak{A})$ into $(Y, \mathfrak{B})$ and a function $f \in \mathcal{M}^{+}(\mathfrak{B})$, the integral
\begin{equation}
(Kf)(x) := \int_{Y}  f(y) K(x,dy), \quad x \in X
\end{equation}
defines a function $Kf \in \mathcal{M}^{+}(\mathfrak{A})$. On the other hand, given a measure $\mu$ on $(X, \mathfrak{A})$, the integral
\begin{equation}
(\mu K)(B) := \int_{X} K(x,B)\ d\mu(x), \quad B \in \mathfrak{B}
\end{equation}
defines a measure $\mu K$ on $(Y,\mathfrak{B})$. Associative law is valid, which is to say that
\begin{align}
\mu(Kf)
    &:= \int_{X} \left( \int_{Y}  f(y) K(x,dy) \right) d\mu(x) \nonumber \\
    &= \int_{Y} f(y) \left( \int_{Y}  K(x,dy)\ d\mu(x) \right)
    =: (\mu K)f
\end{align}
holds. The following theorem is of much use.
\begin{theorem*}[Transition Kernels into Measures on Product Spaces]
Let $K : X \times \mathfrak{B} \to [0,\infty]$ be a $\sigma$-finite transition kernel from $(X, \mathfrak{A})$ into $(Y, \mathfrak{B})$, and let $\mu$ be a measure on $(X, \mathfrak{A})$. Then, there exists a measure $\pi$ on the product space $(X \times Y,\ \mathfrak{A} \otimes \mathfrak{B})$ that satisfies
\begin{align}
\int_{X \times Y} f(x,y)\ d\pi(x,y) :=  \int_{X} \int_{Y} f(x,y)  K(x,dy)\ d\mu(x)
\end{align}
for all $f \in \mathcal{M}^{+}(\mathfrak{A} \otimes \mathfrak{B})$. If, moreover, both $\mu$ and $K$ happens to be finite, then $\pi$ is the unique finite measure on the product space satisfying
\begin{equation}
\pi(A \times B) = \int_{A} K(x,B)\ d\mu(x), \quad A \in \mathfrak{A},\ B \in \mathfrak{B}.
\end{equation}
\end{theorem*}
\noindent
This provides us a convenient way to construct a measure on the product spaces given a transition kernel and a measure.

\paragraph{Conditional Probability Distributions}

We now return to our main line of arguments, and first introduce the definition of conditional probability measures.
\begin{definition*}[Conditional Probability Measure]
Let $(X, \mathfrak{A}, \mu)$ be a probability space, and let $\mathfrak{B} \subset \mathfrak{A}$ be a sub-$\sigma$-algebra. We call a transition probability kernel $K : X \times \mathfrak{B} \to [0, 1]$ a conditional probability measure (or a regular version) of the conditional probability $\mu(\,\cdot\, |\mathfrak{B})$ given $\mathfrak{B}$, if the map $x \mapsto K(x, A)$ happens to be a representative of $\mu( A |\mathfrak{B})$ for all $A \in \mathfrak{A}$, namely
\begin{equation}
K(\,\cdot\,, A) \in \big[\, \mu(A |\mathfrak{B}) \,\big], \quad A \in \mathfrak{A}
\end{equation}
holds, where the brackets around an element denote its equivalence class. If such a transition probability kernel exists, we customarily denote it with the same notation $\mu(\,\cdot\, |\mathfrak{B})$, and its images are in turn denoted as 
\begin{equation}
K(x,A) = \mu(A |\mathfrak{B})(x) = \mu_{x}(A), \quad x \in X, A \in \mathfrak{A}
\end{equation}
interchangeably, depending on the aesthetics of the formula in which it should appear.
\end{definition*}
\noindent
The presence of conditional probability measures allows us to readily make a connection between conditional expectations (defined previously in \eqref{def:cond_exp}) and averages with respect to conditional probabilities under consideration.
\begin{proposition*}[Conditional Expectations as Averages over Conditional Probability Measures]
Let $(X, \mathfrak{A}, \mu)$ be a probability space, $\mathfrak{B} \subset \mathfrak{A}$ be a sub-$\sigma$-algebra, and suppose that the conditional probability $\mu(\,\cdot\, |\mathfrak{B})$ has a conditional probability measure. Then, for every $\mu$-integrable function $f$, the map
\begin{equation}
x \mapsto \int_{X} f(x^{\prime})\ d\mu_{x}(x^{\prime}) \in \big[\, \mathbb{E}[f|\mathfrak{B}] \,\big]
\end{equation}
is a representative of the conditional expectation of $f$ given $\mathfrak{B}$.
\end{proposition*}
\noindent
We note that conditional probability measures do not necessarily exist for general measure spaces. However, fortunately for us, the case $(X, \mathfrak{A}) = (\mathbb{R}^{n}, \mathfrak{B}^{n})$ that we are interested in is known to always admit it.

\paragraph{Conditional Probability Distributions}
Given a probability space $(X, \mathfrak{A}, \mu)$ and a sub-$\sigma$-algebra $\mathfrak{B} \subset \mathfrak{A}$, suppose that a measurable map $f: (X,\mathfrak{A}) \to (X^{\prime},\mathfrak{A}^{\prime})$ is moreover given. In parallel to what we have seen for conditional expectations, this allows us to define an equivalence class of functions
\begin{equation}
\mu(f \in A^{\prime} | \mathfrak{B}) := \mu(f^{-1}(A^{\prime}) | \mathfrak{B})
\end{equation}
for all $A^{\prime} \in \mathfrak{A}^{\prime}$. Then, a transition probability kernel $K : X \times \mathfrak{A}^{\prime} \to [0,1]$ from $(X, \mathfrak{B})$ into $(X^{\prime},\mathfrak{A}^{\prime})$ satisfying
\begin{equation}
K(\,\cdot\, , A^{\prime}) \in \big[\, \mu(f \in A^{\prime} | \mathfrak{B}) \,\big], \quad A^{\prime} \in \mathfrak{A}^{\prime}
\end{equation}
is called a \emph{conditional probability distribution of $f$ given $\mathfrak{B}$}. Likewise, given another measurable map $g: (X,\mathfrak{A}) \to (Y^{\prime},\mathfrak{B}^{\prime})$, a transition probability kernel $K : X \times \mathfrak{A}^{\prime} \to [0,1]$ from $(X, \mathcal{I}(g))$ into $(X^{\prime},\mathfrak{A}^{\prime})$ satisfying
\begin{equation}
K(\,\cdot\, , A^{\prime}) \in \big[\, \mu(f \in A^{\prime} | \mathcal{I}(g)) \,\big], \quad A^{\prime} \in \mathfrak{A}^{\prime}
\end{equation}
is called a \emph{conditional probability distribution of $f$ given $g$}. Such transition probability kernels do not necessarily exist in general, but as above, the case $(X^{\prime},\mathfrak{A}^{\prime}) = (\mathbb{R}^{n}, \mathfrak{B}^{n})$ that we are interested in is known to always admit it.

\paragraph{Conditioning in Quantum Measurements}

Under the context of quantum measurements, let $A$ and $B$ be a pair of simultaneously measurable observables. Given a joint-probability distribution $\mu_{A,B}^{\phi}$ of $A$ and $B$ on some quantum state $|\phi\rangle \in \mathcal{H}$, we introduce
\begin{equation}
\mu_{A}^{\phi}(\Delta_{A} | B) := \mu_{A,B}^{\phi}(\pi_{A} \in \Delta_{A} | \pi_{B}), \quad \Delta_{A} \in \mathfrak{B}^{1},
\end{equation}
where $\pi_{A}(a,b) = a$ and $\pi_{B}(a,b) = b$ are measurable functions (projections) respectively representing the behaviour of the measurement outcomes of $A$ and $B$.
Accordingly, we define the \emph{conditional probability distribution of $A$ given $B$} to be a transition probability kernel $K : \mathbb{R} \times \mathfrak{B}^{1} \to [0,1]$  that satisfies
\begin{equation}
K(\,\cdot\,,\Delta_{A}) \in \big[\, \mu_{A}^{\phi}(\Delta_{A} | B) \,\big], \quad \Delta_{A} \in \mathfrak{B}^{1},
\end{equation}
which, as guaranteed above, is known to always exist. The values of the conditional probability distribution of $A$ given $B$ are in turn denoted interchangeably by
\begin{equation}
\mu_{A}^{\phi}(\Delta_{A} | B = b) = \mu_{B=b}^{\phi}(A \in \Delta_{A}) = \mu_{A}^{\phi}(\Delta_{A} | B)(b) := K(b, \Delta_{A}),
\end{equation}
depending on the context.

\subsubsection{Fourier Transformation}
We next recall the basic definitions and properties of the Fourier transformation.
For convenience, we first introduce the renormalised $n$-dimensional Lebesgue-Borel measure on $\mathbb{R}^{n}$ by
\begin{equation}\label{def:renormalised_LB_measure}
dm_{n} := (2\pi)^{-n/2} d\beta^{n}.
\end{equation}
Accordingly, in this section we employ the renormalised $L^{p}$-norm and the convolution defined by the renormalised Lebesgue-Borel measure,
\begin{gather}
\|f\|_{p} := \left( \int_{\mathbb{R}^{n}} |f(x)|^{p}\ dm_{n}(x)\right)^{1/p}, \quad f \in L^{p}(\mathbb{R}^{n}), \\
(f \ast g)(x) := \int_{\mathbb{R}^{n}} f(x-y)g(y)\ dm_{n}(y), \quad f, g \in L^{1}(\mathbb{R}^{n}).
\label{eq:convol}
\end{gather}
For brevity, we occasionally write $dm_{1} = dm$ whenever there is no risk for confusion.

Now, for a function $f \in L^{1}(\mathbb{R}^{n})$, recall that the functions $\hat{f}, \check{f} : \mathbb{R}^{n} \to \mathbb{C}$ defined by
\begin{align}
\hat{f}(q) &:= \int_{\mathbb{R}^{n}} e^{-i\langle q, x \rangle}f(x)\ dm_{n}(x), \\
\check{f}(q) &:= \int_{\mathbb{R}^{n}} e^{i\langle q, x \rangle}f(x)\ dm_{n}(x),
\end{align}
with the scalar product $\langle q, x \rangle := \sum_{k=1}^{n} q_{k}x_{k}$ of two real vectors in $\mathbb{R}^{n}$, are respectively called the \emph{Fourier transform} and the \emph{inverse Fourier transform} of $f$.
The $\mathbb{C}$-linear map $\mathscr{F}$ that maps $f$ to its Fourier transform $\hat{f}$ is called the \emph{Fourier transformation}.  It is known that the Fourier transformation is injective, {\it i.e.} $\hat{f} = \hat{g}$ implies $f = g$. For $f, g \in L^{1}(\mathbb{R}^{n})$, 
the following properties under the convolution \eqref{eq:convol}, scaling \eqref{def:function_scaling}, and translation \eqref{eq:translation},
\begin{align}
\widehat{(f \ast g)} &= \hat{f} \cdot \hat{g}, \\
\widehat{(f_{t})}(q) &= \hat{f}(tq), \quad t \neq 0,  \label{eq:Scaling_Fourier_translation} \\
\widehat{(\tau_{a}f)}(t) &= e^{i\langle a,x\rangle}\hat{f}(t),\quad a \in \mathbb{R}, \label{eq:Fourier_translation}
\end{align}
respectively, are basic. 
The Fourier transformation $\mathscr{F}$ plays particularly well on the subspace $\mathscr{S}(\mathbb{R}^{n}) \subset L^{1}(\mathbb{R}^{n})$, where it becomes a linear bijection of $\mathscr{S}(\mathbb{R}^{n})$ onto $\mathscr{S}(\mathbb{R}^{n})$, whose inverse is given by the inverse Fourier transformation (recall, on the other hand, that one does not necessarily have $\hat{f} \in L^{1}(\mathbb{R}^{n})$ for $f \in L^{1}(\mathbb{R}^{n})$ in general). One then has
\begin{align}
\partial^{\gamma}(\mathscr{F}f) &= (-i)^{|\gamma|}\mathscr{F}(x^{\gamma}f), \quad \gamma \in \mathbb{N}^{n}_{0}, \label{eq:fourier_diff_01} \\
\mathscr{F}(\partial^{\gamma}f) &= i^{|\gamma|}q^{\gamma}\mathscr{F}f, \quad \gamma \in \mathbb{N}^{n}_{0},\label{eq:fourier_diff_02}
\end{align}
for $f \in \mathscr{S}(\mathbb{R}^{n})$, where we have used the multi-index $\gamma := (\gamma_{1}, \dots, \gamma_{n}) \in \mathbb{N}^{n}_{0}$ as in \eqref{def:use_alpha} and
introduced the shorthand $|\gamma| := \gamma_{1} + \cdots + \gamma_{n}$.

\subsubsection{Wigner-Ville Distribution}

In order to make our line of arguments self-contained in the framework of density functions $L^{1}(\mathfrak{B})$, we assume $\psi \in L^{1}(\mathbb{R}) \cap L^{2}(\mathbb{R})$ throughout this passage.
Given such $\psi$, we define a complex function $\omega^{\psi} \in L^{1}(\mathbb{R}^{2}) \cap L^{2}(\mathbb{R}^{2})$ by
\begin{align}\label{def:func_V}
\tilde{\omega}^{\psi}(x,y) &:= \psi^{*}(x - y/2) \psi(x + y/2),
\end{align}
and evaluate its total integration as
\begin{align}\label{eq:V_tot_int}
\int_{\mathbb{R}^{2}} \tilde{\omega}^{\psi}(x,y)\ dm_{2}(x,y)
    &= \int_{\mathbb{R}^{2}} \psi^{*}(x - y/2) \psi(x + y/2)\ dm_{2}(x,y) \nonumber \\
    &= \left( \int_{\mathbb{R}} \psi^{*}(x) \ dm(x) \right) \left( \int_{\mathbb{R}} \psi(y) \ dm(y) \right) \nonumber \\
    &= \left| \int_{\mathbb{R}} \psi(x) \ dm(x) \right|^{2}.
\end{align}
Whenever the total integration \eqref{eq:V_tot_int} is non-vanishing, we introduce
\begin{equation}\label{eq:omega_norm}
\omega^{\psi}(x,y) := \frac{\tilde{\omega}^{\psi}(x,y)}{\int_{\mathbb{R}^{2}} \tilde{\omega}^{\psi}(x,y)\ dm_{2}(x,y)},
\end{equation}
to denote its normalisation.

On the other hand, if we consider the Fourier transform of $\tilde{\omega}^{\psi}(x,y)$ with respect to its second parameter $y$,
\begin{align}
\tilde{W}^{\psi}(x,p)
    &:= \int_{\mathbb{R}} e^{-ipy} \tilde\omega^{\psi}(x,y)\ dm(y) \nonumber \\
    &= \int_{\mathbb{R}} \psi^{*}(x + y/2) \psi(x - y/2)e^{ipy}\ dm(y),
\end{align}
we readily find that it is a real function,
\begin{align}
\left(\tilde{W}^{\psi}(x,p)\right)^{*}
    &= \tilde{W}^{\psi}(x,p),
\end{align}
whose marginals are given by
\begin{equation}
\begin{split}
\label{eq:wigner_merginals}
\int_{\mathbb{R}} \tilde{W}^{\psi}(x,p)\ dm(x) &= |\hat{\psi}(p)|^{2}, \\
\int_{\mathbb{R}} \tilde{W}^{\psi}(x,p)\ dm(p) &= |\psi(x)|^{2}.
\end{split}
\end{equation}
Applying Plancherel's theorem, one finds that $\tilde{W}^{\psi} \in L^{1}(m_{2})$, and thus its total integration reads
\begin{align}\label{eq:wigner_tot_int}
\int_{\mathbb{R}^{2}} \tilde{W}^{\psi}(x,p)\ dm_{2}(x,p) &= \|\psi\|_{2}^{2}.
\end{align}
If the total integration \eqref{eq:wigner_tot_int} is non-vanishing, which is equivalent to the condition $\psi \neq 0$, the real quasi-probability density function denoted by
\begin{equation}\label{def:Wigner-Ville_QPD}
W^{\psi}(x,p) := \frac{\tilde{W}^{\psi}(x,p)}{\int_{\mathbb{R}^{2}} \tilde{W}^{\psi}(x,p)\ dm_{2}(x, p)}
\end{equation}
is called the \emph{Wigner-Ville distribution} on $\psi$.  As we have seen in \eqref{eq:wigner_merginals},
the WV distribution
possesses useful properties for our analysis, namely, that its marginals yield the probability density function describing the behaviour of the measurement outcomes of the respective observables $\hat{x}$ and $\hat{p}$ on the state $\psi$, which is to say that
\begin{equation}
\begin{split}
\int_{\mathbb{R}^{2}} W^{\psi}(x,p)\, dm(x)
    &= \rho_{\hat{p}}^{\psi}(p), \\
\int_{\mathbb{R}^{2}} W^{\psi}(x,p)\, dm(p)
    &= \rho_{\hat{x}}^{\psi}(x),
\end{split}
\end{equation}
if explicitly written down. Thus, the choice $\psi \in L^{1}(\mathbb{R}) \cap L^{2}(\mathbb{R})$ defines a complex measure
\begin{equation}\label{def:qpm_for_WV}
\mu_{Q,P}^{\psi} := W^{\psi} \odot \beta^{2}
\end{equation}
on the measurable space $(\mathbb{R}, \mathfrak{B}^{2})$ that satisfies
\begin{equation}
\begin{split}\label{eq:WV_qualifies_as_QJP}
\mu_{Q,P}^{\psi}(\mathbb{R} \times \Delta)
    &= \mu_{P}^{\psi}(\Delta), \\
\mu_{Q,P}^{\psi}(\Delta \times \mathbb{R})
    &= \mu_{Q}^{\psi}(\Delta),
\end{split}
\quad \Delta \in \mathfrak{B}.
\end{equation}
For our later argument we note that, since the functions $\omega^{\psi}(x,y)$ and $W^{\psi}(x,y)$ are mapped to one another by Fourier transformation, they just represent the same contents seen from different viewpoints, and are thus essentially the same object.

\subsection{Conditioned Measurement}\label{sec:ps_II_ps}

We are now interested in simultaneously measuring the probability measure of $B$ on the target system and a QJP distribution of $Q$ and $P$ on the meter system. This should be possible since every local measurements can be simultaneously performed on separate systems, and this leads to an existence of a joint distribution of the probability measure of $B$ on one side, and a QJP distribution of $Q$ and $P$ on the other. Throughout this section, for definiteness, we exclusively treat the special case in which the meter state is described by the one-dimensional Schr{\"o}dinger representation of the CCR $\{ L^{2}(\mathbb{R}), \mathscr{S}(\mathbb{R}), \{\hat{x}, \hat{p}\}\}$ and choose $Y=\hat{p}$ without loss of generality.

\subsubsection{Conditioning over Quasi-probabilities}\label{sec:CM_II_cond_over_QP}
Since we are now dealing with complex measures, the definitions for conditioning must be suitably expanded accordingly. To this end, we first prepare a terminology:
\begin{definition*}[Quasi-probabilities]
Let $(X,\mathfrak{A})$ be a measurable space. We call a complex measure $\nu$ on $(X,\mathfrak{A})$ satisfying the normalisation condition $\nu(X) = 1$ a quasi-probability measure, and accordingly the triplet $(X,\mathfrak{A}, \nu)$, a quasi-probability space.
\end{definition*}
\noindent
If the underlying space is given by $(X,\mathfrak{A}) = (\mathbb{R}^{n}, \mathfrak{B}^{n})$, and the quasi-probability measure $\nu$ happens to be absolutely continuous, we call its density $d\nu/d\beta^{n} \in L^{1}(\mathbb{R}^{n})$, which is in general a complex function that has the total integration of unity, a \emph{quasi-probability density function}. According to the definition, note that the usual ({\it i.e.,} real and non-negative) probability measures and density functions are special members of the respective families of quasi-probability measures and density functions. In analogy to the standard probability spaces, given a quasi-probability space $(X,\mathfrak{A},\nu)$ and a $\nu$-integrable function $f$, we occasionally denote the total integration by
\begin{equation}
\mathbb{E}[f;\nu] := \int_{X} f\ d\nu,
\end{equation}
and call it the \emph{quasi-expectation value of $f$ under $\nu$}.

\paragraph{Quasi-joint-probabilities}
As a special subclass of quasi-probability measures, we say that a quasi-probability measure $\nu \in \mathbf{M}_{\mathbb{C}}(\mathfrak{B}^{n})$ \emph{qualifies as a QJP distribution} of the observables $A_{1}, \dots, A_{n}$ on the state $|\phi\rangle \in \mathcal{H}$, if it satisfies
\begin{equation}\label{def:qjpm}
\nu(\underbrace{\mathbb{K} \times \cdots \times \mathbb{K} \times \stackrel{k\text{-th}}{\Delta} \times \mathbb{K} \times \cdots \times \mathbb{K}}_{n}) = \mu_{A_{k}}^{\phi}(B), \quad \Delta \in \mathfrak{B}(\mathbb{K})
\end{equation}
for all $1 \leq k \leq n$. In parallel to it, we prepare the term \emph{QJP density function} for those $\nu$ that are absolutely continuous%
\footnote{Here, we occasionally admit complex parameters to describe outcomes of each observable $A_{k}$ for formal completeness. Accordingly, the r.~h.~s. of the above formula is understood as the probability measure induced by the two-dimensional spectral measure of $A_{k}$ seen as a normal operator ({\it cf.} spectral theorem for normal operators).}.
One confirms from \eqref{eq:WV_qualifies_as_QJP} that, for the choice of the quantum state $\psi \in L^{1}(\mathbb{R}) \cap L^{2}(\mathbb{R})$, the quasi-probability measure \eqref{def:qpm_for_WV} qualifies as a QJP distribution for the pair of observables $Q$ and $P$.  It should be intuitively straightforward to see by the formal arguments made in the introduction that the Kirkwood-Dirac distribution also qualifies as a QJP distribution of the pair of observables under consideration.

\paragraph{Conditional Quasi-expectations}
We next intend to introduce analogous definitions regarding conditioning on quasi-probability measure spaces $(X,\mathfrak{A},\nu)$. To this end, we make some very important remarks on the different properties between standard probability measures and quasi-probability measures. Recall that we have made extensive use of the Radon-Nikod{\'y}m theory for defining conditional expectations and conditional probabilities. In applying the theory, first note that positiveness of the measure $\mu$ is necessary in order for the Radon-Nikod{\'y}m derivative $d\nu/d\mu$ of some complex measure $\nu \ll \mu$ to be well-defined. Hence, conditioning by a sub-$\sigma$-algebra $\mathfrak{B} \subset \mathfrak{A}$ must be such that the restriction $\nu|_{\mathfrak{B}}$ becomes a measure. The second fact to notice is that, for a $\nu$-integrable function $f$, the complex measure on the sub-$\sigma$-algebra defined by
\begin{equation}
B \mapsto (f\odot\nu)(B) := \int_{B} f\ d\nu, \quad B \in \mathfrak{B}
\end{equation}
is not necessarily absolutely continuous with respect to the restriction $\nu|_{\mathfrak{B}}$, in contrast to that of positive measures. With these in mind, we hereby define:
\begin{definition*}[Conditional Quasi-expectation]
Let $(X,\mathfrak{A}, \nu)$ be a quasi-probability space, and $\mathfrak{B} \subset \mathfrak{A}$ a sub-$\sigma$-algebra such that the restriction $\nu|_{\mathfrak{B}}$ becomes a probability measure ({\it i.e.}, real and non-negative). For a $\nu$-integrable function $f$ such that $f \odot \nu \ll \nu|_{\mathfrak{B}}$, we define the conditional quasi-expectation of $f$ given $\mathfrak{B}$
\begin{equation}
\mathbb{E}[f|\mathfrak{B}] := \frac{d (f \odot \nu)}{d(\nu|_{\mathfrak{B}})}
\end{equation}
by the Radon-Nikod{\'y}m derivative of the complex measure $f \odot \nu$ with respect to the measure $\nu|_{\mathfrak{B}}$.
\end{definition*}
\noindent
Given another measurable function $g: X \to \mathbb{R}$ such that the above conditions are fulfilled for the initial $\sigma$-algebra $\mathfrak{B} = \mathcal{I}(g)$, we define $\mathbb{E}[f|g]$ and any other relevant notations such as $\mathbb{E}[f|g = y]$ {\it etc.} in an analogous manner to those defined for standard probability measures.

\paragraph{Conditional Quasi-probabilities}
We then intend introduce a complex analogue of conditional probabilities defined for quasi-probability measures.
\begin{definition*}[Quasi-Conditional Probabilities]
Let $(X,\mathfrak{A}, \nu)$ be a quasi-probability space, and $\mathfrak{B} \subset \mathfrak{A}$ a sub-$\sigma$-algebra such that the restriction $\nu|_{\mathfrak{B}}$ becomes a probability measure. For a measurable set $A \in \mathfrak{A}$, we define
\begin{equation}\label{def:cond_quasi-prob}
\nu(A|\mathfrak{B}) := \mathbb{E}[\chi_{A} | \mathfrak{B}], \quad A \in \mathfrak{A},
\end{equation}
to be the quasi-conditional probability of $A$ given $\mathfrak{B}$, whenever $\chi_{A} \odot \nu \ll \nu|_{\mathfrak{B}}$, where $\chi_{A}$ is the characteristic function of $A$.
Likewise, given a measurable function $f: (X,\mathfrak{A}) \to (X^{\prime},\mathfrak{A}^{\prime})$, we introduce
\begin{equation}
\nu(f \in A^{\prime} | \mathfrak{B}) := \nu(f^{-1}(A^{\prime}) | \mathfrak{B}), \quad A^{\prime} \in \mathfrak{A}^{\prime},
\end{equation}
whenever the r.~h.~s. is well-defined.
If, instead of being given a sub-$\sigma$-algebra, one is given a measurable function $g: X \to Y$ for conditioning, we define
\begin{equation}
\nu(A|g) := \nu(A|\mathcal{I}(g)), \quad A \in \mathfrak{A},
\end{equation}
where $\mathcal{I}(g)$ is the initial $\sigma$-algebra of $g$, whenever, as usual, the r.~h.~s. is well-defined.
\end{definition*}
\noindent
The conditional quasi-probability $\nu(\,\cdot\, | \mathfrak{B})$ satisfies properties analogous to those of quasi-probability measures, namely
\begin{enumerate}
\item $\nu(\emptyset |\mathfrak{B}) = 0$, $\nu(X |\mathfrak{B}) = 1$,
\item $\nu( A |\mathfrak{B}) \in \mathbb{C}, \quad A \in \mathfrak{A}$,
\item for any sequence $(A_{n})_{n \geq 1}$ of pairwise disjoint subsets of $X$, the equality
\begin{equation}
\nu\left(\left. \bigcup_{n=1}^{\infty} A_{n} \right| \mathfrak{B} \right) = \sum_{n=1}^{\infty} \nu(A_{n} | \mathfrak{B})
\end{equation}
holds,
\end{enumerate}
whenever every component above is well-defined.
In parallel to conditional probabilities, the validity of the (in)equalities above are significant only in the sense of $\nu|_{\mathfrak{B}}$-a.e.

\paragraph{Conditional Quasi-probability Measures}
We now expand the definition of transition kernels to fit into the theory of complex measures. Let $(X, \mathfrak{A})$ and $(Y, \mathfrak{B})$ be measurable spaces. We say that a map $K : X \times \mathfrak{B} \to \mathbb{C}$ that satisfies the conditions
\begin{enumerate}
\item the map $x \mapsto K(x, B)$ is $\mathfrak{A}$-measurable for every $B \in \mathfrak{B}$,
\item the map $B \mapsto K(x, B)$ is a complex measure on $(Y, \mathfrak{B})$ for every $x \in X$,
\end{enumerate}
a \emph{complex transition kernel} from $(X, \mathfrak{A})$ into $(Y, \mathfrak{B})$. If a complex transition kernel $K$ satisfies $K(x,Y) = 1$ for all $x \in X$, we call such $K$ a \emph{transition quasi-probability kernel}. The following analogous result is of use.
\begin{proposition}[Complex Transition Kernels into Complex Measures on Product Spaces]\label{prop:ctk_to_cm}
Let $K : X \times \mathfrak{B} \to \mathbb{C}$ be a complex transition kernel from $(X, \mathfrak{A})$ into $(Y, \mathfrak{B})$, and let $\mu$ be a measure on $(X, \mathfrak{A})$. Then, there exists a complex measure $\pi$ on the product space $(X \times Y,\ \mathfrak{A} \otimes \mathfrak{B})$ that satisfies
\begin{align}
\int_{X \times Y} f(x,y)\ d\pi(x,y) :=  \int_{X} \int_{Y} f(x,y)  K(x,dy)\ d\mu(x)
\end{align}
for all $f$, whenever the integration on the r.~h.~s. is well-defined. In particular, the complex measure $\pi$ satisfies
\begin{equation}
\pi(A \times B) = \int_{A} K(x,B)\ d\mu(x), \quad A \in \mathfrak{A},\ B \in \mathfrak{B}.
\end{equation}
\end{proposition}
\noindent
Armed with the above concepts, we thus introduce:
\begin{definition*}[Conditional Quasi-Probability Measure]
Let $(X, \mathfrak{A}, \nu)$ be a quasi-probability space, and let $\mathfrak{B} \subset \mathfrak{A}$ be a sub-$\sigma$-algebra such that the restriction $\nu|_{\mathfrak{B}}$ becomes a probability measure, and that $\nu(A|\mathfrak{B})$ is well-defined for all $A \in \mathfrak{A}$. We call a transition quasi-probability kernel $K : X \times \mathfrak{B} \to \mathbb{C}$ a conditional quasi-probability measure of the conditional quasi-probability $\nu(\,\cdot\, |\mathfrak{B})$, if the map $x \mapsto K(x, A)$ happens to be a representative of $\nu( A |\mathfrak{B})$ for all $A \in \mathfrak{A}$, namely
\begin{equation}
K(\,\cdot\,, A) \in \big[\, \nu(A |\mathfrak{B}) \,\big], \quad A \in \mathfrak{A}
\end{equation}
holds, where the brackets around an element denote its equivalence class. If such a transition quasi-probability kernel exists, we customarily denote it with the same notation $\nu(\,\cdot\, |\mathfrak{B})$, and its images are in turn interchangeably denoted by 
\begin{equation}
K(x,A) = \nu(A |\mathfrak{B})(x) = \nu_{x}(A), \quad x \in X, A \in \mathfrak{A},
\end{equation}
depending on the aesthetics of the formula in which it should appear.
\end{definition*}
\noindent
As above, such transition quasi-probability kernels do not exist in general, while the case $(X, \mathfrak{A}) = (\mathbb{R}^{n}, \mathfrak{B}^{n})$ is known to always admit it.
We then have:
\begin{proposition*}[Conditional Quasi-expectations as Averages over Conditional Quasi-probability Measures]
Let $(X, \mathfrak{A}, \nu)$ be a quasi-probability space, $\mathfrak{B} \subset \mathfrak{A}$ be a sub-$\sigma$-algebra such that the restriction $\nu|_{\mathfrak{B}}$ becomes a probability measure, and suppose that the conditional quasi-probability $\nu(\,\cdot\, |\mathfrak{B})$ has a conditional quasi-probability measure. Then, for every $\nu$-integrable function $f$, the map
\begin{equation}
x \mapsto \int_{X} f(x^{\prime})\ d\nu_{x}(x^{\prime}) \in \big[\, \mathbb{E}[f|\mathfrak{B}] \,\big]
\end{equation}
is a representative of the conditional quasi-expectation of $f$ given $\mathfrak{B}$.
\end{proposition*}

\paragraph{Conditional Probability Distributions}
On a quasi-probability space $(X, \mathfrak{A}, \nu)$, suppose that a measurable map $f: (X,\mathfrak{A}) \to (X^{\prime},\mathfrak{A}^{\prime})$ is moreover given. Choosing a sub-$\sigma$-algebra $\mathfrak{B} \subset \mathfrak{A}$ such that the restriction $\nu|_{\mathfrak{B}}$ is a measure, this allows us to define an equivalence class of functions
\begin{equation}
\nu(f \in A^{\prime} | \mathfrak{B}) := \nu(f^{-1}(A^{\prime}) | \mathfrak{B})
\end{equation}
for all $A^{\prime} \in \mathfrak{A}^{\prime}$, whenever they are well-defined. Then, a transition quasi-probability kernel $K : X \times \mathfrak{A}^{\prime} \to [0,1]$ from $(X, \mathfrak{B})$ into $(X^{\prime},\mathfrak{A}^{\prime})$ satisfying
\begin{equation}
K(\,\cdot\, , A^{\prime}) \in \big[\, \mu(f \in A^{\prime} | \mathfrak{B}) \,\big], \quad A^{\prime} \in \mathfrak{A}^{\prime}
\end{equation}
is called a \emph{conditional quasi-probability distribution of $f$ given $\mathfrak{B}$}. Likewise, given another measurable map $g: (X,\mathfrak{A}) \to (Y^{\prime},\mathfrak{B}^{\prime})$ such that the restriction of $\nu$ over its initial $\sigma$-algebra $\mathcal{I}(g)$ is a measure, a transition probability kernel $K : X \times \mathfrak{A}^{\prime} \to [0,1]$ from $(X, \mathcal{I}(g))$ into $(X^{\prime},\mathfrak{A}^{\prime})$ satisfying
\begin{equation}
K(\,\cdot\, , A^{\prime}) \in \big[\, \mu(f \in A^{\prime} | \mathcal{I}(g)) \,\big], \quad A^{\prime} \in \mathfrak{A}^{\prime}
\end{equation}
is called a \emph{conditional quasi-probability distribution of $f$ given $g$}.

\subsubsection{Conditioned Measurement via the WV Distributions}

Now that we have prepared the necessary concepts and results, we may embark on our analysis. By measuring $B$ locally on the target system on one side, and a specific QJP distribution of $Q$ and $P$ locally on the meter system on the other, we obtain a quasi-probability distribution that describes the joint behaviour of the target system and the meter system. If, by haps ({\it e.g.} by choosing the right initial state $|\psi\rangle \in \mathcal{K}$) the QJP distribution of $Q$ and $P$ on the meter admits representation by a complex measure, the total quasi-probability distribution of both the target and the meter system also admits representation by a complex measure.
We thus generally define the CM scheme as an act of measuring the conditional quasi-probability distribution of the `joint outcome' of $Q$ and $P$ of the meter system given the outcome of the conditioning observable $B$ on the target system.

\paragraph{WV Distribution}

To demonstrate our point with an example, we shall from now on exclusively concentrate on the \emph{Wigner-Ville distribution} for our choice of the QJP distribution of $Q$ and $P$ for definiteness.  Since the choice $\psi \in L^{1}(\mathbb{R}) \cap L^{2}(\mathbb{R})$ of the initial meter state allows the WV distribution to be described by quasi-probability measures on $(\mathbb{R}^{2}, \mathfrak{B}^{2})$, we assume such special choice throughout this passage in order to remain contained in the framework of measure and integration theory (so that we may not have to deal with the theory of generalised functions).
In this subsection, the CM scheme is studied in view of the WV distribution. We first start by transcribing the CM scheme, which was initially introduced in terms of vectors and operators on Hilbert spaces, into the description by quasi-probability density functions on $\mathbb{R}^{2}$.  It is then found that the transcription allows a much simpler expression in view of its Fourier transform (rather than the WV distribution itself), in which the description of the meter system after the interaction is given precisely by the convolution of the configuration of both the meter and the target system, quite analogous to the case of the UM scheme that we have previously seen. This allows us to extract the information of the target system either by means of deconvolution discussed earlier (specifically by constructing an approximate identity on the meter system), or by probing the behaviour of the distribution around the origin $g=0$. We shall then investigate the properties of the information of the target system we have just obtained, and find that this qualifies as a `conditional quasi-probability distribution of $A$ given $B$', of which the average has a connection to the conditional quasi-expectation of $A$ given $B$ introduced earlier.

\paragraph{Preliminary Observation}
As a preliminary observation, we start by assuming that the target observable $A$ has a spectrum consisting of a finite number of eigenvalues $\sigma(A) = \{a_{1}, \dots, a_{N}\}$ so that its spectral decomposition reads \eqref{eq:spect_decomp_fin}. For the ease of arguments, we further assume that the conditioning observable $B$ also has a spectrum consisting of a finite number of eigenvalues $\sigma(B) = \{b_{1}, \dots, b_{M}\}$, that every eigenvalue of $B$ is degenerate, {\it i.e.}, $\Pi_{b_{m}} = |b_{m}\rangle\langle b_{m}|$ for some normalised vectors $|b_{m}\rangle \in \mathcal{H}$ for all $1 \leq m \leq M$, and moreover that $\|\Pi_{b}\phi\|^{2} \neq 0$ for all $b \in \sigma(B)$. As for the state preparation, let $\psi \in L^{1}(\mathbb{R}) \cap L^{2}(\mathbb{R})$ be a wave-function of the meter system with normalisation $\|\psi\|_{2} = 1$ so that the WV distribution can be represented by a quasi-probability density function, and we also let the initial selection $|\phi\rangle \in \mathcal{H}$ of the target system be normalised $\|\phi\| = 1$. 

\paragraph{Computing the WV Distribution}
We are now interested in measuring the WV distribution of the meter system given the outcome of $B$ on the target system.
Since both the measurements are local measurements performed on the respective systems, this should be statistically equivalent to measuring the WV distribution for the meter state
\begin{equation}\label{def:mixed_state_given_b}
\psi_{B=b}^{g} := \mathrm{Tr}_{\mathcal{H}}\left[ \left|\Psi_{B=b}^{g}\right\rangle\left\langle\Psi_{B=b}^{g} \right| \right]
\end{equation}
for all $b \in \sigma(B)$, where
\begin{equation}\label{def:state_given_b}
\left|\Psi_{B=b}^{g}\right\rangle := \frac{(\Pi_{b} \otimes I) |\Psi^{g}\rangle}{\| (\Pi_{b} \otimes I) \Psi^{g} \|^{2}}
\end{equation}
is the, so-to-speak, `conditional' meter state%
\footnote{
Naturally, \eqref{def:state_given_b} and \eqref{def:mixed_state_given_b} are nothing but the state one would expect when the ideal measurement of $B$ yielded the outcome $b \in \sigma(B)$, if one adopted the standard von Neumann projection postulate.
} given the outcome $b$ of $B$.
Our analysis thus reduces to computing the WV distribution of the density operator $\psi_{B=b}^{g}$ for each of the outcomes $b \in \sigma(B)$.
In our case, in which we assume that the eigenvalues of $B$ are all degenerate, the density operator \eqref{def:mixed_state_given_b} in fact becomes a pure state, of which representation by wave-functions reads
\begin{align}\label{def:pure_state_given_b}
\psi_{B=b}^{g}(x)
    &= \sum_{n = 1}^{N} \frac{\langle b, \Pi_{a_{n}}\phi\rangle}{\langle b, \phi\rangle}  \left( e^{-iga_{n}\hat{p}} \psi \right)(x) \nonumber \\
    &= \sum_{n=1}^{N} \frac{\langle b, \Pi_{a_{n}}\phi\rangle}{\langle b, \phi\rangle} \psi(x - ga_{n}) \nonumber \\
    &= \sum_{n=1}^{N} \frac{\langle b, E_{A}(\{a_{n}\})\phi\rangle}{\langle b, \phi\rangle} \int_{\mathbb{R}} \psi(x - ga)\ d\delta_{a_{n}}(a) \nonumber \\
    &= \int_{\mathbb{R}} \psi(x - ga)\ d\nu_{b}(a), \quad g \in \mathbb{R},
\end{align}
where we have used \eqref{eq:interaction_finite} to obtain the first equality. Here, we have introduced an auxiliary quasi-probability measure
\begin{equation}\label{def:aux_qpm_b_a}
\nu_{b}^{\phantom{*}}(\Delta) := \frac{\langle b, E_{A}(\Delta) \phi\rangle}{\langle b, \phi\rangle}, \quad b \in \sigma(B),\ \Delta \in \mathfrak{B},
\end{equation}
defined by means of the spectral measure $E_{A}$ of $A$, the initial state $|\phi\rangle \in \mathcal{H}$ of the target system, and the outcome $b \in \sigma(B)$ of the conditioning observable, and have used a result analogous to \eqref{eq:prob_meas_fin_spec_obs} in the last equality. One then finds that the WV distribution of the meter wave-function \eqref{def:pure_state_given_b} reads
\begin{align}\label{eq:wigner_post}
&W^{\psi^{g}_{B=b}}(x,p) \nonumber \\
    &:= \int_{\mathbb{R}} \left(\psi_{B=b}^{g}(x + y/2)\right)^{*} \psi_{B=b}^{g}(x - y/2) e^{ipy}\ dm(y) \nonumber \\
    &= \int_{\mathbb{R}} \left( \int_{\mathbb{R}} \psi^{*}(x - ga_{1}^{\prime} + y/2)\ d\nu_{b}^{*}(a_{1}^{\prime}) \right) \left( \int_{\mathbb{R}} \psi(x - ga_{2}^{\prime} - y/2)\ d\nu_{b}^{\phantom{*}}(a_{2}^{\prime}) \right) e^{ipy}\ dm(y) \nonumber \\
    &= \int_{\mathbb{R}} \left( \int_{\mathbb{R}^{2}} \psi^{*}(x - ga_{1}^{\prime} + y/2) \psi(x - ga_{2}^{\prime} - y/2)\ d\left(\nu_{b}^{*} \otimes \nu_{b}^{\phantom{*}} \right)(a_{1}^{\prime},a_{2}^{\prime}) \right) e^{ipy}\ dm(y),
\end{align}
where we have introduced the product measure $\nu_{b}^{*} \otimes \nu_{b}^{\phantom{*}}$ of $\nu_{b}^{\phantom{*}}$ and its complex conjugate%
\footnote{
For a pair of complex measures $\mu$ and $\nu$, by observing that $\mu \ll |\mu|$ and $\nu \ll |\nu|$, we define the \emph{product complex measure} of $\mu$ and $\nu$ by
\begin{equation}
\mu \otimes \nu := \left( \frac{d\mu}{d|\mu|} \cdot \frac{d\nu}{d|\nu|} \right) \odot \left( |\mu| \otimes |\nu| \right).
\end{equation}
By definition, product complex measures share properties similar to those of product measures \eqref{def:product_measure}, and an analogue of Fubini's theorem holds. Product complex measures reduce to the usual product measures when both of the components happen to be finite  measures.
}
in the last equality.
In order to gain a better view of our findings, let us now change variables according to the linear transformation,
\begin{align}\label{eq:lin_trans_T_2i}
\left(
 \begin{array}{l}
a_{1} \\
a_{2}
 \end{array}
\right)
= T
\left(
 \begin{array}{l}
 a_{1}^{\prime} \\
 a_{2}^{\prime}
 \end{array}
\right),
\qquad 
T := \left(
    \begin{array}{cc}
    1/2 & 1/2 \\
    -1/2 & 1/2
    \end{array}
    \right).
\end{align}
Since $T \in \mathrm{GL}(2,\mathbb{R})$ belongs to the general linear group, for indeed $\det T = 1/2$, note that this transformation is invertible, {\it i.e.}, it is a linear automorphism. We then introduce the quasi-probability measure
\begin{align}\label{def:quasi_prob_A}
\mu_{A}^{\phi}(\Delta|B = b)
    &:= T\left(\nu_{b}^{*} \otimes \nu_{b}^{\phantom{*}} \right)(\Delta) \nonumber \\
    &:= \left( \nu_{b}^{*} \otimes \nu_{b}^{\phantom{*}} \right)(T^{-1}\Delta), \quad \Delta \in \mathfrak{B}^{2},\ b \in \sigma(B)
\end{align}
defined on the measurable space $(\mathbb{R}^{2}, \mathfrak{B}^{2})$ as the image measure ({\it cf}.~see \eqref{def:Bildmass} for the definition of image measures) of the product complex measure with respect to the automorphism $T$ (we shall be shortly returning to the properties of the quasi-probability measure \eqref{def:quasi_prob_A} and the righteousness of its notation). Then, due to the change of variables formula \eqref{eq:Transformationsformel_Bildmass}, one may rewrite our previous findings \eqref{eq:wigner_post} by letting $a_{1}^{\prime} = a_{1} - a_{2}$ and $a_{2}^{\prime} = a_{1} + a_{2}$ as
\begin{align}
&W^{\psi^{g}_{B=b}}(x,p) \nonumber \\
    &= \int_{\mathbb{R}} \left( \int_{\mathbb{R}^{2}} \psi^{*}(x - g(a_{1} - a_{2}) + y/2) \right. \nonumber \\
        &\qquad \qquad \left. \phantom{\int_{\mathbb{R}}} \times \psi(x - g(a_{1} + a_{2}) - y/2)\ d\mu_{A}^{\phi}(a_{1}, a_{2}|B= b) \right) e^{ipy}\ dm(y) \nonumber \\
    &= \int_{\mathbb{R}^{2}} \left( \int_{\mathbb{R}} \psi^{*}((x - ga_{1}) + (y + 2ga_{2})/2)  \right. \nonumber \\
        &\qquad \qquad \left. \phantom{\int_{\mathbb{R}}} \times \psi((x - ga_{1}) - (y + 2ga_{2})/2) e^{ipy}\ dm(y) \right) d\mu_{A}^{\phi}(a_{1},a_{2}|B= b) \nonumber \\
    &= \int_{\mathbb{R}^{2}} e^{-i2ga_{2}p} \left( \int_{\mathbb{R}} \psi^{*}((x - ga_{1}) + y/2) \right. \nonumber \\
        & \phantom{e^{-i2ga_{2}p}} \qquad \qquad \left. \phantom{\int_{\mathbb{R}}} \times \psi((x - ga_{1}) - y/2) e^{ipy}\ dm(y) \right) d\mu_{A}^{\phi}(a_{1},a_{2}|B= b)\nonumber \\
    &= \int_{\mathbb{R}^{2}} e^{-i2ga_{2}p} W^{\psi}(x - ga_{1},p)\ d\mu_{A}^{\phi}(a_{1},a_{2}|B= b),
\end{align}
where the change of the order of the integration in the second equality is guaranteed by the Fubini's theorem.
For later convenience, we introduce the complex number $a \in \mathbb{C}$ defined by $a := a_{1} + i a_{2}$ by identifying $\mathbb{C} \cong \mathbb{R}^{2}$ in a usual manner, and write
\begin{equation}\label{eq:psm_in_W}
W^{\psi^{g}_{B=b}}(x,p) = \int_{\mathbb{C}} e^{-i2ga_{2}p} W^{\psi}(x - ga_{1},p)\ d\mu_{A}^{\phi}(a|B= b), \quad g \in \mathbb{R}.
\end{equation}
To sum up, here we have learned how the CM scheme may be rewritten in terms of quasi-probability measures, in which the  WV distribution of the initial meter wave-function $\psi$ is acted upon by the quasi-probability measure \eqref{def:quasi_prob_A} of the target system to yield the final WV distribution of the meter wave-function $\psi_{B=b}^{g}$.

\paragraph{Changing the Viewpoint through Fourier Transformation}

One finds below that the transcription \eqref{eq:psm_in_W} of the CM scheme admits a much simpler expression when described in terms of the inverse Fourier transform \eqref{eq:omega_norm} of the WV distribution,
rather than the WV distribution itself. Introducing the (yet to be normalised) function $\tilde{\omega}^{\psi^{g}_{B=b}}(x,y)$ uniquely specified through the relation\begin{equation}\label{def:W_to_w_g}
W^{\psi^{g}_{B=b}}(x,p) = \int_{\mathbb{R}} e^{-ipy} \tilde{\omega}^{\psi^{g}_{B=b}}(x,y)\ dm(y), \quad g \in \mathbb{R}
\end{equation}
({\it cf.}, injectivity of the Fourier transformation), the goal of this small paragraph is to show that our finding \eqref{eq:psm_in_W} is equivalent to
\begin{align}\label{eq:psm_in_V_pre}
\tilde{\omega}^{\psi^{g}_{B=b}}(x,y)
    &= \int_{\mathbb{C}} \tilde{\omega}^{\psi}(x - ga_{1},y - 2ga_{2})\ d\mu_{A}^{\phi}(a|B = b), \quad g \in \mathbb{R},
\end{align}
which is essentially nothing but the convolution of the initial profile $\tilde{\omega}^{\psi}$ of the meter state by that of the two-dimensional quasi-probability measure $\Delta \mapsto \mu_{A}^{\phi}(\Delta|B = b)$ scaled by $g$. If, moreover, the total integration of $\tilde{\omega}^{\psi}$ happens to be non-vanishing, we may renormalise both sides of the above equality to obtain
\begin{equation}\label{eq:psm_in_V}
\omega^{\psi^{g}_{B=b}}(x,y)
    = \int_{\mathbb{C}} \omega^{\psi}(x - ga_{1},y - 2ga_{2})\ d\mu_{A}^{\phi}(a|B = b), \quad g \in \mathbb{R},
\end{equation}
for later use%
\footnote{
Here, note that we have used the general property of convolutions
\begin{equation}
\int_{\mathbb{R}^{n}} (f \ast g)\ d\beta^{n} = \int_{\mathbb{R}^{n}} f\ d\beta^{n} \cdot \int_{\mathbb{R}^{n}} g\ d\beta^{n},\quad f, g \in L^{1}(\mathbb{R}^{n})
\end{equation}
regarding integration.
}.
Observe here the analogy between the unconditioned case \eqref{eq:outcome_prob01}: in both cases, the profile of the `output' of the meter is given by the convolution of the profile of the `input' of the meter and that of the target system scaled by $g$.

To verify our statement, one may simply repeat the previous argument to obtain the result directly, but it is actually easier to demonstrate that the Fourier transforms of the two sides of the above equality coincide. Indeed, the Fourier transform of the l.~h.~s. is nothing but $W^{\psi^{g}_{B=b}}$, which is just the definition \eqref{def:W_to_w_g}. As for the r.~h.~s., one has
\begin{align}
&\int_{\mathbb{R}} e^{-ipy} \left( \int_{\mathbb{C}} \tilde{\omega}^{\psi}(x - ga_{1},y - 2ga_{2})\ d\mu_{A}^{\phi}(a|B = b) \right) dm(y) \nonumber \\
     &\qquad = \int_{\mathbb{C}} \left( \int_{\mathbb{R}} e^{-ipy}\,  \tilde{\omega}^{\psi}(x - ga_{1},y - 2ga_{2})\ dm(y)\right) d\mu_{A}^{\phi}(a|B = b) \nonumber \\
    &\qquad = \int_{\mathbb{C}} e^{-i2ga_{2}p}\, W^{\psi}(x - ga_{1},p)\ d\mu_{A}^{\phi}(a|B = b), 
\end{align}
where the exchange of the order of the integration (the first equality) is guaranteed by Fubini's theorem, and the last equality is due to \eqref{eq:Fourier_translation}. Combining the above two results and by observing \eqref{eq:psm_in_W}, the injectivity of the Fourier transformation leads to the desired statement.
We emphasise again that both \eqref{eq:psm_in_W} and \eqref{eq:psm_in_V} represent the same contents seen from different viewpoints.

\subsection{Recovery of the Target Profile}

We are now interested in how one may recover the profile $\Delta \mapsto \mu_{A}^{\phi}(\Delta|B = b)$ of the target system for each $b \in \sigma(B)$ through CM scheme. 
As one may expect, the procedure essentially goes analogously to that of the recovery of the probability measure $\mu_{A}^{\phi}$ in the case of the UM scheme demonstrated in Section~\ref{sec:ups_II_ups}. 
Recalling the techniques employed there, and by introducing the rescaling
\begin{equation}
\upsilon^{\psi}(x,y) := 2^{-1} \omega^{\psi}(x,2y),
\end{equation}
for the ease of discussion,
one may readily rewrite \eqref{eq:psm_in_V} into
\begin{equation}\label{eq:psm_outcome_prob01}
\upsilon^{\psi^{g}_{B=b}}
    = \upsilon^{\psi} \ast \left( \mu_{A}^{\phi}(\, \cdot \, | B=b) \right)_{g}, \quad g \in \mathbb{R},
\end{equation}
or equivalently
\begin{equation}\label{eq:psm_outcome_prob02}
\upsilon^{\psi^{g}_{B=b}}_{g^{-1}}
    = \upsilon^{\psi}_{g^{-1}} \ast \mu_{A}^{\phi}(\, \cdot \, | B=b), \quad g \in \mathbb{R}^{\times}
\end{equation}
where the subscript on the respective quasi-probability measures/density functions denotes the scaling \eqref{def:measure_scaling} and \eqref{def:function_scaling}, just as we have done for the case of the UM scheme (see \eqref{eq:outcome_prob01} and \eqref{eq:outcome_prob02}). In parallel to the case of the UM case, these two expressions \eqref{eq:psm_outcome_prob01} and \eqref{eq:psm_outcome_prob02} correspond to the manner in which one combines the interaction parameter \eqref{eq:combining_of_the_interaction}, where the former corresponds to the scaling of the target observable $A \to gA$, whereas the latter corresponds to the scaling of the pair of the meter observables $\{Q, P\} \to \{g^{-1}Q, gP\}$ ({\it cf.}  \eqref{eq:outcome_prob_mod_01} and \eqref{eq:outcome_prob_mod_02}).

\subsubsection{Strong Conditioned Measurement}

We now intend to recover the quasi-probability measure $\Delta \mapsto \mu_{A}^{\phi}(\Delta|B = b)$ by making use of the latter expression \eqref{eq:psm_outcome_prob02}. The idea and the procedure are essentially the same as those we have employed in the unconditional case, namely, we manipulate both the interaction parameter $g \in \mathbb{R}^{\times}$ and the initial meter state $\psi \in L^{1}(\mathbb{R}) \cap L^{2}(\mathbb{R})$ so that the scaling $\upsilon^{\psi}_{g^{-1}}$ of the inverse Fourier transform of the WV distribution tends towards the delta measure $\delta_{0}$ centred at the origin $0 \in \mathbb{R}^{2}$.

\paragraph{Recovery of the Conditional Quasi-joint-probability}
For the same reason discussed in Section~\ref{sec:Strong_Unconditioned_Measurement}, we assume throughout this passage:
\begin{itemize}
\item The target observable $A$ admits description by density functions.
\item The total integration of $\tilde{\omega}^{\psi}$ is non-vanishing.
\end{itemize}
The first condition guarantees that the quasi-probability measure $\Delta \mapsto \mu_{A}^{\phi}(\Delta|B = b)$, $\Delta \in \mathfrak{B}^{2}$ is absolutely continuous for all $b \in \sigma(B)$, of which density we shall write
\begin{equation}
\rho_{A}^{\phi}(\,\cdot\, |B = b) := \frac{d\mu_{A}^{\phi}(\,\cdot\,|B = b)}{d\beta^{2}}.
\end{equation}
The last condition is necessary in order to assure the well-definedness of $\omega^{\psi}$.
Then, one sees from an analogous argument that we have previously made in Section~\ref{sec:Strong_Unconditioned_Measurement} that, if one adjusts the pair of $g \in \mathbb{R}^{\times}$ and $\psi \in L^{1}(\mathbb{R}) \cap L^{2}(\mathbb{R})$ so that $\upsilon^{\psi}_{g^{-1}}$ makes itself an approximate identity in $L^{1}(\mathbb{R}^{2})$, one may let the product of the convolution ({\it i.e.}, the `outcome') converge towards the desired target
\begin{align}
\upsilon^{\psi^{g}_{B=b}}_{g^{-1}} \ \to \ \rho_{A}^{\phi}(\,\cdot\, |B = b)
\end{align}
with respect to the $L^{1}$-norm. A typical way to construct such an approximate identity is to start by preparing a compactly supported wave-function $\psi$, which automatically guarantees $\psi \in L^{1}(\mathbb{R}) \cap L^{2}(\mathbb{R})$, and to consider a family $\{\psi_{(h)}\}_{h > 0}$ of the initial meter state defined as in \eqref{eq:approx_ident_state_Sch}. One then finds
\begin{align}
\tilde{\omega}^{\psi_{(h)}}(x,y)
    &:= \psi^{*}_{(h)}(x - y/2) \psi_{(h)}(x + y/2) \nonumber \\
    &= |h|^{-1}\psi^{*}\left(\frac{x - y/2}{h}\right) \psi\left(\frac{x + y/2}{h}\right) \nonumber \\
    &= |h| \cdot \tilde{\omega}^{\psi}_{h}(x,y),
\end{align}
and hence by observing that the above equality has total integration of $|h| \int_{\mathbb{R}^{2}} \tilde{\omega}^{\psi} dm_{2}$, its normalisation becomes
\begin{equation}
\omega^{\psi_{(h)}}(x,y) = \omega^{\psi}_{h}(x,y).
\end{equation}
This should further lead to
\begin{equation}
\upsilon^{\psi_{(h)}}_{g^{-1}} = \upsilon^{\psi}_{hg^{-1}},
\end{equation}
when scaled by $g^{-1}$. With the initial $\omega^{\psi}$ (or equivalently $\upsilon^{\psi}$) being compactly supported, one then sees that this indeed makes an example of an approximate identity, and we may thus achieve our objective by either narrowing the wave-function $h \to 0$, by intensifying the interaction $g^{-1} \to 0$ ($g \to \pm \infty$) or by appropriately balancing both manoeuvres and letting $hg^{-1} \to 0$ altogether.

\subsubsection{Weak Conditioned Measurement}

We shall next investigate how the map
\begin{equation}\label{eq:parameter_to_omega}
g \mapsto  \upsilon^{\psi^{g}_{B=b}}(x,y)
\end{equation}
behaves locally around $g=0$, and discuss what information of the target configuration one might reveal through it. In parallel to the case of the UM scheme discussed in Section~\ref{sec:ups_II_wups}, one finds below that the information of the configuration of the target system is encoded into the differential coefficients of the above map at $g=0$, and that by knowing all the higher-order derivatives, one may fully recover the quasi-probability measure $\Delta \mapsto \mu_{A}^{\phi}(\Delta|B = b)$ of our interest.

\paragraph{Main Objective}

Throughout the following passage, we assume the following.
\begin{itemize}
\item The quasi-probability measure $\Delta \mapsto \mu_{A}^{\phi}(\Delta|B = b)$ has a compact support.
\item The total integration of $\tilde{\omega}^{\psi}$ is non-vanishing, and its normalisation belongs to the Schwartz space $\omega^{\psi} \in \mathscr{S}(\mathbb{R}^{2})$.
\end{itemize}
These requirements are imposed primarily for the same reason as we have previously discussed in analysing the weak UM scheme in Section~\ref{sec:ups_II_wups} (which, in short, is to say that we do not wish to get involved in the theory of generalised functions). A sufficient condition for the first and second assumptions would be to respectively require that the spectral measure $E_{A}$ be compactly supported,  and that $\psi \in \mathscr{S}(\mathbb{R})$.
Under such conditions, the main objective of this passage is to demonstrate the following Proposition:
\begin{proposition}[Weak Conditioned Measurement]\label{prop:WPSM}
Under the above conditions, the map \eqref{eq:parameter_to_omega} is arbitrarily many times strongly differentiable on all the real line $\mathbb{R}$, and its $n$th derivatives at the origin $g=0$ reads
\begin{align}
\left. \frac{d^{n}}{dg^{n}} \upsilon^{\psi^{g}_{B=b}} \right|_{g=0} =  \sum_{|\gamma| = n} \mathbb{E} \left[a^{\gamma} ; \mu_{A}^{\phi}(\,\cdot\,|B = b) \right] \cdot (- D)^{\gamma}\upsilon^{\psi}, \quad n \in \mathbb{N}_{0}.
\end{align}
Here, $\gamma = (\gamma_{1}, \gamma_{2}) \in \mathbb{N}_{0}^{2}$ is a multi-index introduced in \eqref{def:use_alpha}, and the `quasi-moments' under the quasi-probability measure $\mu_{A}^{\phi}(\,\cdot\,|B = b)$ is defined by
\begin{equation}\label{def:moments_qjpm}
\mathbb{E} \left[a^{\gamma} ; \mu_{A}^{\phi}(\,\cdot\,|B = b) \right] := \int_{\mathbb{C}} a_{1}^{\gamma_{1}}a_{2}^{\gamma_{2}}\ d\mu_{A}^{\phi}(a|B = b),
\end{equation}
in its explicit form, where we understand $a = a_{1} + i a_{2} \in \mathbb{C}$.
\end{proposition}
\begin{proof}
Since the assumptions and reasonings are essentially the same as those provided for the unconditioned counterpart, we shall avoid reiteration and provide a rough sketch of the proof. In order to avoid clumsiness of notation, we write $\upsilon := \upsilon^{\psi}$, $\upsilon[g] := \upsilon^{\psi^{g}_{B=b}}$ and $\mu := \mu_{A}^{\phi}(\,\cdot\,|B = b)$ for simplicity, and denote by $\upsilon^{(n)}[g]$ the $n$th derivative of the map $g \mapsto \upsilon[g]$.

We first prove that the $n$th derivative of $\upsilon[g]$ reads
\begin{equation}\label{eq:n_wcd}
\upsilon^{(n)}[g](x) = \sum_{|\gamma| = n} \int_{\mathbb{R}^{2}} (- D)^{\gamma}\upsilon(x - ga) a^{\gamma}\  d\mu(a).
\end{equation}
As above, we argue by mathematical induction. The case $n=0$ is trivial. Suppose that the statement is true for $n \in \mathbb{N}_{0}$. Then, one may compute its point-wise derivative as
\begin{align}
\frac{d}{dg} \upsilon^{(n)}[g](x)
    &= \sum_{|\gamma| = n} \int_{\mathbb{R}^{2}} \left( \frac{d}{dg} (- D)^{\gamma}\upsilon(x - ga) \right) a^{\gamma}\  d\mu(a) \nonumber \\
    &= \sum_{|\gamma| = n} \int_{\mathbb{R}^{2}} \left(  \sum_{i=1}^{2} a_{i} (-D)_{i}(- D)^{\gamma}\upsilon(x - ga) \right) a^{\gamma}\  d\mu(a) \nonumber \\
    &= \sum_{|\gamma| = n + 1} \int_{\mathbb{R}^{2}} (- D)^{\gamma}\upsilon(x - ga) a^{\gamma}\  d\mu(a),
\end{align}
and subsequently prove its strong differentiability by employing the same technique as above. This completes our first step of the proof.

Now, by taking $g=0$ of \eqref{eq:n_wcd}, we observe
\begin{align}
\upsilon^{(n)}[0](x)
    &= \sum_{|\gamma| = n} \int_{\mathbb{R}^{2}} (- D)^{\gamma}\upsilon(x - 0a) a^{\gamma}\  d\mu(a) \nonumber \\
    &= \sum_{|\gamma| = n} \int_{\mathbb{R}^{2}} a^{\gamma}\  d\mu(a) \cdot (- D)^{\gamma}\upsilon(x) \nonumber \\
    &= \sum_{|\gamma| = n} \mathbb{E}[a^{\gamma}; \nu] \cdot (- D)^{\gamma}\upsilon(x),
\end{align}
which completes our proof.
\end{proof}
\noindent
One then immediately obtains the following corollary by applying the Stone-Weierstra{\ss} approximation theorem and the Riesz-Markov-Kakutani representation theorem.
\begin{corollary}[Recovery of the Target Profile by Weak Conditioned Measurement]
The weak CM scheme ({\it i.e.}, the knowledge of all the `quasi-moments' \eqref{def:moments_qjpm}) allows us to uniquely specify the quasi-probability measure $\mu_{A}^{\phi}(\,\cdot\,|B = b)$ of our interest.
\end{corollary}
\noindent
Compare these results to those obtained in the case of the weak UM scheme described in Section~\ref{sec:ups_II_wups}.

\subsection{Profile of the Target System}

We have so far investigated how the CM scheme can be transcribed into the language of conditional quasi-probabilities, rather than in terms of mere conditional expectations. As a result, we found that the measurement outcome after the interaction incorporates two components: one being the profile of the meter system in the form of the WV distribution and the other being the that of the target system in the form of the quasi-probability measure $\mu_{A}^{\phi}(\, \cdot \, | B=b)$ defined in \eqref{def:quasi_prob_A}. Specifically, in view of the (scaled) inverse Fourier transform of the WV distribution, we found that the manner in which the two components interact with each other admits a simple description by convolution \eqref{eq:psm_in_V}, which is quite analogous to the unconditioned case. Based on our findings, we have thus analysed how one may recover the profile $\mu_{A}^{\phi}(\, \cdot \, | B=b)$ by means of both the strong and weak CM schemes, whose procedures are also quite analogous to the unconditioned counterpart. We are now interested in the properties of the quasi-probability measure $\mu_{A}^{\phi}(\, \cdot \, | B=b)$ we have obtained, which should be expected to convey some information of the target system.

\paragraph{Quasi-joint-probability Distribution of a Pair of Observables}

By means of either the strong or weak CM scheme, we have so far obtained the family of quasi-probability measures $\mu_{A}^{\phi}(\, \cdot \, | B=b)$ for all $b \in \sigma(B)$.
Allowing it to extend on the whole real line, one may construct a complex transition kernel by
\begin{equation}\label{def:quasi_trans_kern}
\mu_{A}^{\phi}(\Delta_{A} | B=b) :=
    \begin{cases}
        \mu_{A}^{\phi}(\Delta_{A} | B=b), & (b \in \sigma(B)) \\
        \text{indefinite}, & (b \notin \sigma(B))
    \end{cases}
    , \quad \Delta_{A} \in \mathfrak{B}(\mathbb{C})
\end{equation}
from the space $(\mathbb{R}, \mathfrak{B}^{1})$ of the measurement outcomes of $B$ into $(\mathbb{C}, \mathfrak{B}(\mathbb{C}))$. For definiteness, we assign to each $b \notin \sigma(B)$ any quasi-probability measure, so that \eqref{def:quasi_trans_kern} defines a \emph{transitional quasi-probability kernel} as a whole.
This allows us to construct a quasi-probability measure $\mu_{A,B}^{\phi}$ on the product space $(\mathbb{C} \times \mathbb{R}, \mathfrak{B}(\mathbb{C}) \otimes \mathfrak{B}^{1})$, by combining the transition quasi-probability kernel \eqref{def:quasi_trans_kern} and the probability measure $\mu_{B}^{\phi}$, that satisfies
\begin{equation}\label{def:quasi_joint_prob_meas}
\mu_{A,B}^{\phi}(\Delta_{A} \times \Delta_{B}) = \int_{\Delta_{B}} \mu_{A}^{\phi}(\Delta_{A} | B=b)\ d\mu_{B}^{\phi}(b), \quad \Delta_{A} \in \mathfrak{B}(\mathbb{C}),\ \Delta_{B} \in \mathfrak{B}^{1},
\end{equation}
whose existence is guaranteed by Proposition~\ref{prop:ctk_to_cm}.
The target of our analysis in this passage is the quasi-probability measure \eqref{def:quasi_joint_prob_meas}. As one may expect from the notation employed, we shall shortly see that this qualifies as a QJP distribution of the target observable $A$ and the conditioning observable $B$.
\begin{proposition}[Quasi-joint-probability Distribution]
Under the definitions above, the quasi-probability measure $\mu_{A,B}^{\phi}$ qualifies as a QJP distribution of $A$ and $B$ in the sense of \eqref{def:qjpm}, namely
\begin{align}
\begin{split}
\mu_{A,B}^{\phi}(\Delta_{A} \times \mathbb{R}) &= \mu_{A}^{\phi}(\Delta_{A}), \quad \Delta_{A} \in \mathfrak{B}(\mathbb{C}), \\
\mu_{A,B}^{\phi}(\mathbb{C} \times \Delta_{B}) &= \mu_{B}^{\phi}(\Delta_{B}), \quad \Delta_{B} \in \mathfrak{B}(\mathbb{R})
\end{split}
\end{align}
holds. Here, $\mu_{A}^{\phi}$ denotes the probability measure on $(\mathbb{C}, \mathfrak{B}(\mathbb{C}))$ generated by the two-dimensional spectral measure associated to $A$ understood as a normal operator, whereas $\mu_{B}^{\phi}$ denotes the probability measure on $(\mathbb{R}, \mathfrak{B})$ generated by the one-dimensional spectral measure associated to the self-adjoint operator $B$.
\end{proposition}
\begin{proof}
We start by demonstrating that the marginal of the quasi-probability measure $\mu_{A,B}^{\phi}$ of the first term coincides with the probability measure $\mu_{A}^{\phi}$ on $(\mathbb{C}, \mathfrak{B}(\mathbb{C}))$ generated by the spectral measure of $A$ (seen as a normal operator). To this end, we first observe
\begin{align}\label{eq:marginal_A}
\mu_{A,B}^{\phi}(\Delta_{A} \times \mathbb{R})
    &= \int_{\mathbb{R}} \mu_{A}^{\phi}(\Delta_{A} | B=b)\ d\mu_{B}^{\phi}(b) \nonumber \\
    &= \sum_{b \in \sigma(B)} \left( \nu_{b}^{*} \otimes \nu_{b}^{\phantom{*}} \right)(T^{-1}\Delta_{A}) \cdot |\langle b , \phi \rangle|^{2},
\end{align}
where we have used \eqref{def:quasi_prob_A} in the last equality.

Now, in order to proceed further, we then maintain that the measure
\begin{equation}
\Delta \mapsto \mu(\Delta) := \sum_{b \in \sigma(B)} \left( \nu_{b}^{*} \otimes \nu_{b}^{\phantom{*}} \right)(\Delta) \cdot |\langle b , \phi \rangle|^{2}, \quad \Delta \in \mathfrak{B}(\mathbb{C})
\end{equation}
is essentially the same object as the continuous $\mathbb{C}$-linear map defined by
\begin{equation}
I_{\mathrm{diag}}(f) := \int_{\mathbb{R}} f(a,a)\ d\mu_{A}^{\phi}(a), \quad f \in C_{0}(\mathbb{C})
\end{equation}
in the sense of the Riesz-Markov-Kakutani representation theorem.
The proof can be carried out in several ways, but for the sake of simplicity, we rather take an elementary approach. 
Observing that any two measures on a product space $(X \times Y,\, \mathfrak{A} \otimes \mathfrak{B})$ coincides with each other if they coincide on the subset $\mathfrak{A} \ast \mathfrak{B} \subset \mathfrak{A} \otimes \mathfrak{B}$ (see \eqref{def:product_set_of_sigma_algebras} for the definition), one proceeds as
\begin{align}
\int_{\mathbb{C}}\chi_{\Delta_{1}}(a_{1})\chi_{\Delta_{2}}(a_{2})\ d\mu(a)
    &= \mu(\Delta_{1} \times \Delta_{2})\nonumber \\
    &= \sum_{b \in \sigma(B)} \left( \nu_{b}^{*} \otimes \nu_{b}^{\phantom{*}} \right)(\Delta_{1} \times \Delta_{2}) \cdot |\langle b , \phi \rangle|^{2} \nonumber \\
    &= \sum_{b \in \sigma(B)} \langle \phi, E_{A}(\Delta_{1}) b\rangle \langle b, E_{A}(\Delta_{2}) \phi \rangle \nonumber \\
    &= \langle \phi, E_{A}(\Delta_{1})E_{A}(\Delta_{2}) \phi \rangle \nonumber \\
    &= \int_{\mathbb{R}} \chi_{\Delta_{1}}(a)\chi_{\Delta_{2}}(a)\ d\mu_{A}^{\phi}(a) \nonumber \\
    &= I_{\mathrm{diag}}(\chi_{\Delta_{1}}\chi_{\Delta_{2}}),
\end{align}
which proves
\begin{equation}
\mu \cong I_{\mathrm{diag}}.
\end{equation}

Armed with the findings, we return to our original problem \eqref{eq:marginal_A} and finally obtain
\begin{align}
\sum_{b \in \sigma(B)} \left( \nu_{b}^{*} \otimes \nu_{b}^{\phantom{*}} \right)(T^{-1}\Delta_{A}) \cdot |\langle b , \phi \rangle|^{2}
    &= I_{\mathrm{diag}}(\chi_{(T^{-1}\Delta_{A})}) \nonumber \\
    &= \int_{\mathbb{R}} \chi_{(T^{-1}\Delta_{A})}(a,a)\ d\mu_{A}^{\phi}(a) \nonumber \\
    &= \int_{\mathbb{R}} \chi_{\Delta_{A}}(a,0)\ d\mu_{A}^{\phi}(a) \nonumber \\
    &= \langle \phi, E_{A, 0}(\Delta_{A}) \phi \rangle \nonumber \\
    &= \mu_{A}^{\phi}(\Delta_{A}), \quad \Delta_{A} \in \mathfrak{B}(\mathbb{C}),
\end{align}
where $E_{A, 0}$ denotes the product spectral measure of the one-dimensional spectral measure $E_{A}$ of $A$ (as a self-adjoint operator) and that of the $0$ operator $E_{0}$ ({\it i.e.,} the `delta spectral measure' \eqref{def:delta_spectral_meas_0} centred at the origin), and the last equality is due to the observation that the two-dimensional spectral measure $\tilde{E}_{A}$ of $A$ as a normal operator coincides with the product spectral measure $\tilde{E}_{A} = E_{A, 0}$ introduced above. This completes our proof for the marginal of the first term.

It now remains to compute the marginal of $\mu_{A,B}^{\phi}$ of the second term, which one carries out as
\begin{align}
\mu_{A,B}^{\phi}(\mathbb{C} \times \Delta_{B})
    &= \int_{\mathbb{C} \times \mathbb{R}} \chi_{\Delta_{B}}(b)\ d\mu_{A,B}^{\phi}(a,b) \nonumber \\
    &= \int_{\mathbb{C} \times \mathbb{R}} \chi_{\Delta_{B}}(b) \mu_{A}^{\phi}(\mathbb{C} | B=b)\ d\mu_{B}^{\phi}(b) \nonumber \\
    &= \int_{\mathbb{C} \times \mathbb{R}} \chi_{\Delta_{B}}(b)\ d\mu_{B}^{\phi}(b) \nonumber \\
    &= \mu_{B}^{\phi}(\Delta_{B}), \quad \Delta_{B} \in \mathfrak{B},
\end{align}
where the second equality is due to the definition of $\mu_{A,B}^{\phi}$, and the third equality is due to the fact that $\mu_{A}^{\phi}(\mathbb{C} | B=b) = 1$ is normalised to unity ({\it i.e.}, a quasi-probability measure) for all $b \in \sigma(B)$.
\end{proof}
\noindent
As for the relation between the QJP distribution $\mu_{A,B}^{\phi}$ and the transition quasi-probability kernel $(b,\Delta_{A}) \mapsto \mu_{A}(\Delta_{A}|B=b)$, one immediately has the following corollary by construction.
\begin{corollary}
The transition quasi-probability kernel \eqref{def:quasi_trans_kern} is a conditional quasi-probability distribution of $A$ given $B$ under the QJP distribution $\mu_{A,B}^{\phi}$.
\end{corollary}

\paragraph{Conditional Quasi-expectation of $A$ given $B$}

It is now tempting to investigate how the `conditional average' of the QJP distribution $\mu_{A,B}^{\phi}$ relates to the conditional quasi-expectation $\mathbb{E}^{\alpha}[A|B;\phi]$ we have introduced earlier in \eqref{def:cond_quasi-exp_alpha}.
\begin{proposition}[Conditional Average of the Quasi-joint-probability Distribution]
Under the definitions above, the conditional average of $A$ given $B$ under the QJP distribution $\mu_{A,B}^{\phi}$ reads
\begin{align}\label{eq:cond_quas-exp_as_conditional_average}
\int_{\mathbb{C}} a\ d\mu_{A}^{\phi}(a | B=b)
    &= \mathbb{E}^{i}[A | B = b ;\phi],
\end{align}
where the r.~h.~s. is the member of the complex-parametrised sub-family of conditional quasi-expectations of $A$ given $B$ introduced in \eqref{def:cond_quasi-exp_alpha} for the purely imaginary choice $\alpha = i$ of the parameter.
\end{proposition}
\begin{proof}
For the demonstration, let $b \in \sigma(B)$.
One then has
\begin{align}
\int_{\mathbb{C}} a\ d\mu_{A}^{\phi}(a | B=b)
    &= \int_{\mathbb{R}^{2}} (a_{1} + i a_{2})\ dT\left( \nu_{b}^{*} \otimes \nu_{b}^{\phantom{*}} \right)(a_{1}, a_{2}) \nonumber \\
    &= \int_{\mathbb{C}} \frac{a + a^{\prime}}{2} + i\frac{a - a^{\prime}}{2} \ d\left( \nu_{b}^{*} \otimes \nu_{b}^{\phantom{*}} \right)(a^{\prime}, a) \nonumber \\
    &= \int_{\mathbb{C}} \frac{1 + i}{2} a + \frac{1 - i}{2} a^{\prime} \ d\left( \nu_{b}^{*} \otimes \nu_{b}^{\phantom{*}} \right)(a, a^{\prime}) \nonumber \\
    &=  \frac{1 + i}{2} \cdot \frac{\langle b, A \phi \rangle}{\langle b, \phi \rangle} + \frac{1 - i}{2} \cdot \left( \frac{\langle b, A \phi \rangle}{\langle b, \phi \rangle} \right)^{*} \nonumber \\
    &= \mathbb{E}^{i}[A | B = b ;\phi],
\end{align}
where the second equality is due to the change of variables formula \eqref{eq:Transformationsformel_Bildmass} for image measures. 
\end{proof}

\paragraph{Obtaining Conditional Quasi-probability Distribution by Conditioned Measurement}
We now realise that the CM scheme, in view of conditional quasi-probabilities, can be regarded as a \emph{method of obtaining conditional quasi-probability distributions of the target observable $A$ given the conditioning observable $B$}, and that it implies the existence of QJP distributions of a pair of (generally not necessarily simultaneously measurable) quantum observables lying underneath. Moreover, we have seen a connection between the concept of conditional quasi-expectations and the `conditional average' of the QJP distributions, which is reminiscent of the familiar relation between classical conditional expectations and conditional average of probability measures.
While we have conducted an analysis for the special case in which both $A$ and $B$ happen to possess spectra of finite cardinalities (and that $B$ is degenerate), we note that one may suitably generalise the results obtained here by introducing appropriate mathematical tools and some little more advanced mathematical languages.

\newpage
\section{Quasi-probabilities of Quantum Observables}\label{sec:qp_qo}

By studying the both the UM and CM schemes in depth throughout the preceding four sections, we have so far naturally arrived,  by a purely bottom-up construction, at the concept of \emph{quasi-joint-probability} (QJP) of an arbitrary pair of quantum observables. While such an operational way of demonstration has its own merit of being solid and down to earth, it has an apparent downside in that the line of argument lacks transparency and that the whole structure may become obscure on occasions. In this section, we will be conducting a top-down study on the topic 
as a complement to the analyses made in the preceding sections.

\paragraph{Organisation of this Section}
In this section, we first devote several pages to introducing some mathematical tools for our analysis as usual. We then propose a general prescription for the construction of QJP distributions of a given pair of quantum observables, and observe their basic properties. Since it is difficult to perform a general analysis on the whole class of all possible candidates of QJP distributions with full mathematical rigour due to the limited framework and tools available, for our demonstration we shall mostly concentrate on a special sub-family of such distributions parametrised by a single complex number, hopefully without loss of too much essence. We finally close this section by observing where the bottom-up line of discussion performed in Section~\ref{sec:ps_II} fits in this more general framework.

\subsection{Reference Materials}

As usual, we first prepare some necessary mathematical tools for reference. As a generalisation to those defined on integrable functions, we now introduce Fourier transforms of complex measures.

\subsubsection{Fourier Transform of Complex Measures}

Analogous to the manner in which we have defined Fourier transforms of elements of $L^{1}(\mathbb{R}^{n})$ (namely, the density functions), one may define Fourier transforms of complex measures. Given a complex measure $\mu \in \mathbf{M}_{\mathbb{C}}(\mathfrak{A})$ on a measurable space $(X, \mathfrak{A})$, we define the Fourier transform and the inverse Fourier transform of $\mu$, respectively by the functions
\begin{align}
\hat{\mu}(q) &:= \int_{X} e^{-i\langle q, x \rangle}\ d\mu(x), \label{def:FT_cm} \\
\check{\mu}(q) &:= \int_{X} e^{i\langle q, x \rangle}\ d\mu(x), \label{def:IFT_cm}
\end{align}
where $\langle q, x \rangle := \sum_{k=1}^{n} q_{k}x_{k}$ denotes the scalar product on $\mathbb{R}^{n}$ as usual. Note that the functions $\hat{\mu}$, $\check{\mu}$ are well-defined, for indeed $|\hat{\mu}(q)| \leq \int_{X} |e^{-i\langle q, x \rangle}|\, d|\mu|(x) = \|\mu\| < \infty$ for all $q \in \mathbb{R}^{n}$, where $|\mu|$ and $\|\mu\|$ are respectively the variation and the total variation of $\mu$ (a similar evaluation holds for $\check{\mu}$).

\paragraph{Basic Properties}
To see how this newly introduced definition of Fourier transforms relates to that of integrable functions introduced earlier, let $L^{1}(\mathfrak{B}^{n}) \subset \mathbf{M}_{\mathbb{C}}(\mathfrak{B}^{n})$ be the sub-algebra of absolutely continuous complex measures with respect to $m_{n}$, where $m_{n}$ denotes the renormalised $n$-dimensional Lebesgue-Borel measure on $(\mathbb{R}^{n},\mathfrak{B}^{n})$ defined in \eqref{def:renormalised_LB_measure}. Choosing $\mu \in L^{1}(\mathfrak{B}^{n})$ and letting $\rho := d\mu/dm_{n}$ be the Radon-Nikod{\'y}m derivative of $\mu$, one finds by a direct application of \eqref{eq:Mass_mit_Dichte_AC} that
\begin{align}\label{eq:FT_measure_vs_density}
\hat{\mu}(q)
    &:= \int_{\mathbb{R}^{n}} e^{-i\langle q, x \rangle}\ d\mu(x) \nonumber \\
    &= \int_{\mathbb{R}^{n}} e^{-i\langle q, x \rangle} \rho(x)\ dm_{n}(x)
    =: \hat{\rho}(q),
\end{align}
holds as expected. An analogous relation holds for the inverse Fourier transform as well. The $\mathbb{C}$-linear map $\mathscr{F}$ that maps a complex measure into its Fourier transform is called the \emph{Fourier transformation}. In parallel to that defined for integrable functions, the Fourier transformation on the measure algebra is injective, {\it i.e.}, $\hat{\mu} = \hat{\nu}$ implies $\mu = \nu$. For $\mu, \nu \in \mathbf{M}_{\mathbb{C}}(\mathfrak{B}^{n})$, the
properties
\begin{align}
\widehat{(\mu \ast \nu)} &= \hat{\mu} \cdot \hat{\nu}, \\
\widehat{(\mu_{t})}(q) &= \hat{\mu}(tq), \quad t \neq 0, \\
\widehat{(\tau_{a}\mu)}(t) &= e^{i\langle a,x\rangle}\hat{\mu}(t),\quad a \in \mathbb{R},
\end{align}
are basic, in which one sees how the Fourier transformation behaves under the convolution \eqref{def:convolution_measure}, scaling \eqref{def:measure_scaling}, and translation 
\begin{equation}\label{def:translation_measure}
(\tau_{a}\mu)(B) := \mu(B + a), \quad a \in \mathbb{R}^{n},
\end{equation}
respectively.

\paragraph{Linear Transformation}
Let $T$ be a linear operator on $\mathbb{R}^{n}$ ({\it i.e.}, an $n \times n$ real matrix), and let $\mu \in \mathbf{M}_{\mathbb{C}}(\mathfrak{B}^{n})$ be a complex measure. We then define the \emph{linear transform}
\begin{equation}\label{def:lin_trans_measure}
\mu_{T}(B) := \mu(T^{-1}B), \quad B \in \mathfrak{B}^{n},
\end{equation}
of $\mu$ with respect to $T$ by its image measure. By definition, one readily finds the validity of the product rule
\begin{equation}\label{eq:product_rule_image_measure_lin}
(\mu_{T})_{S} = \mu_{(ST)}
\end{equation}
for a pair of linear operators $S$ and $T$ on $\mathbb{R}^{n}$, and that
\begin{equation}
\int_{X} f(x)\ d\mu_{T}(x)
    = \int_{X} f(Tx)\ d\mu(x)
\end{equation}
by the change of variables formula \eqref{eq:Transformationsformel_Bildmass}, whenever the integration exists. Note that the familiar scaling $\mu_{t}$, $t \neq \mathbb{R}$ defined in \eqref{def:measure_scaling}, and the translation $\tau_{a}\mu$, $a \in \mathbb{R}^{n}$ defined in \eqref{def:translation_measure} are respectively special cases of the linear transform of $\mu$ with respect to $T = tI$ and $T = I - a$, where $I$ denotes the identity operator. In such a cases, note also that the linear operators involved are automorphisms, hence members of the general linear group $\mathrm{GL}(n;\mathbb{R})$.
In relation to the Fourier transformation, one finds
\begin{align}\label{eq:FT_and_Lin_Trans}
(\mathscr{F}\mu_{T})(q)
    &:= \int_{\mathbb{R}^{n}} e^{-i\langle q,x\rangle}\ d\mu_{T}(x) \nonumber \\
    &= \int_{\mathbb{R}^{n}} e^{-i\langle q,Tx\rangle}\ d\mu(x) \nonumber \\
    &= \int_{\mathbb{R}^{n}} e^{-i\langle T^{*}q,x\rangle}\ d\mu(x) \nonumber \\
    &= (\mathscr{F}\mu)(T^{*}q),
\end{align}
where $T^{*}$ denotes the adjoint (in this case, the transpose $T^{*} = T^{t}$) of the Matrix $T$.

\paragraph{Complex Conjugate}

We finally review how the Fourier transform behaves regarding the operation of taking the complex conjugate of a complex measure.  To this end, let $\mu$ be a complex measure on $(\mathbb{R}^{n}, \mathfrak{B}^{n})$, and define the complex conjugate of $\mu$ by $\mu^{*}(\Delta) := \mu(\Delta)^{*}$, $\Delta \in \mathfrak{B}^{n}$ in a natural manner. One then readily finds
\begin{align}\label{eq:FT_cc_Measure}
(\mathscr{F}\mu^{*})(q)
    &:= \int_{\mathbb{R}^{n}} e^{-i\langle q,x\rangle}\ d\mu^{*}(x) \nonumber \\
    &= \left( \int_{\mathbb{R}^{n}} e^{-i\langle -q,x\rangle}\ d\mu(x) \right)^{*} \nonumber \\
    &= (\mathscr{F}\mu)^{*}(-q) \nonumber \\
    &= (\mathscr{F}\mu)^{\dagger}(q),
\end{align}
where $f^{\dagger}(x) := f^{*}(-x)$ denotes the \emph{involution} of a function $f$.

\paragraph{Differentiation}

We finally make a brief note on the basic results regarding differentiability and derivatives of a Fourier transform of a complex measure at the origin.
\begin{lemma}\label{lem:FT_CM_diff_exp}
Let $\mu$ be a complex measure on $(\mathbb{R}^{n},\mathfrak{B}^{n})$, and let $\gamma \in \mathbb{N}_{0}^{n}$ be a multi-index. If the integration
\begin{equation}
\int_{\mathbb{R}^{n}} x^{\gamma^{\prime}}\ d\mu(x), 
\end{equation}
exists for all $0 \leq \gamma^{\prime} \leq \gamma$, then the derivative $D^{\gamma}\hat{\mu}$ of the Fourier transform of $\mu$ exists at the origin, in which case the derivative reads
\begin{equation}
\left(D^{\gamma}\hat{\mu}\right)(0) = (-i)^{|\gamma|}\int_{\mathbb{R}^{n}} x^{\gamma}\ d\mu(x).
\end{equation}
\end{lemma}
\begin{proof}
One readily computes
\begin{align}
\left(D^{\gamma}\hat{\mu}\right)(0)
    &:= \left. D^{\gamma} \left( \int_{\mathbb{R}^{n}} e^{-i\langle q, x\rangle}\ d\mu(x) \right) \right|_{q=0} \nonumber \\
    &= \int_{\mathbb{R}^{n}} \left( \left. D^{\gamma} e^{-i\langle q, x\rangle} \right|_{q=0} \right) d\mu(x) \nonumber \\
    &= (-i)^{|\gamma|}\int_{\mathbb{R}^{n}} x^{\gamma}\ d\mu(x),
\end{align}
where the second equality (exchange of the differentiation and integration) is a consequence of the dominated convergence theorem.
\end{proof}
\noindent
Compare this result to that for Schwartz functions \eqref{eq:fourier_diff_01}.

\subsection{Quasi-joint-probabilities of a Combination of Quantum Observables}

We now intend to provide a general prescription for defining a QJP distribution of a combination of generally not necessarily simultaneously measurable quantum observables.

\subsubsection{Preliminary Observations}
In this passage, we conduct some formal discussions on the topic of QJP distributions of a combination of quantum observables.  Since rigorous treatment requires advanced mathematical tools that is beyond the scope of this paper, we first conduct a formal and intuitive argument to obtain the essence of the idea.  
Now, before we embark on our main objective, we first recall a basic theorem regarding strong commutativity of $A$ and $B$ and that of their unitary operators.
\begin{theorem*}
Let $A$ and $B$ be self-adjoint. Then, the following conditions are equivalent.
\begin{enumerate}
\item The operators $A$ and $B$ strongly commute with each other.
\item The operators $e^{isA}$ and $e^{itB}$ commute with each other for all $s,t \in \mathbb{R}$, namely
\begin{equation}\label{eq:strong_commutativity_vs_FT}
e^{itA}e^{isB} = e^{isB}e^{itA}, \quad s,t \in \mathbb{R}
\end{equation}
holds.
\end{enumerate}
\end{theorem*}
\noindent
This familiar theorem builds the starting point of our discussion that follows.

\paragraph{Fourier Transform of Product Spectral Measures}

Recall that the joint behaviour of the outcomes of an ideal measurement of a pair of simultaneously measurable observables $A$ and $B$ is governed by the product spectral measure $E_{A,B}$ of their respective spectral measures $E_{A}$, $E_{B}$ introduced earlier in \eqref{eq:product_spectral_measure}.  An important observation here is to see that the `Fourier transform' of the product spectral measure $E_{A,B}$ is nothing but the product \eqref{eq:strong_commutativity_vs_FT} of the parametrised unitary operators
\begin{align}
(\mathscr{F}E_{A,B})(s,t) 
    &:= \int_{\mathbb{R}^{2}} e^{-i(\langle s, a\rangle + \langle t, b\rangle)}\ dE_{A,B}(a,b) \nonumber \\
    &= e^{-i\overline{(sA + tB)}} \nonumber \\
    &= \lim_{N \to \infty} ( e^{-isA/N} e^{-itB/N} )^{N} \nonumber \\
    &= e^{-isA}e^{-itB},
\end{align}
where the overline on the essentially self-adjoint operator $sA + tB$ denotes its unique self-adjoint extension as usual, and the second equality is due to the familiar Trotter formula.

\paragraph{Hashed Operators}
We now consider a pair of \emph{arbitrary} (not necessarily strongly-commuting) self-adjoint operators $A$ and $B$. Guided by the above observation, we formally introduce
\begin{align}\label{def:decent_mixture}
\#(s,t) :=\, \text{a `decent' mixture of the disintegrated components of } e^{-isA} \text{ and } e^{-itB}
\end{align}
for the pair of $A$ and $B$. Example of such mixtures of the disintegrated components of the unitary operators are given by:
\begin{align}\label{eq:mixture_examples}
\#(s,t) =
    \begin{cases}
        e^{-isA}e^{-itB}, \\
        e^{-itB}e^{-isA}, \\
        \Pi_{k = 1}^{N} e^{-isA/L_{k}}e^{-itB/M_{k}}, \qquad \left(\sum_{k=1}^{N}L_{k}^{-1} = 1,\ \sum_{k=1}^{N}M_{k}^{-1} = 1\right), \\
        \left( e^{-isA/N} e^{-itB/N} \right)^{N}, \\
        e^{-i\overline{(sA + tB)}} = \lim_{N \to \infty} \left( e^{-isA/N} e^{-itB/N} \right)^{N}, \\
        \textit{etc.},
    \end{cases}
\end{align}
or even any linear combinations of them.
The term `decent' is intended to express a mathematical condition as to what qualifies as a reasonable `mixture' to meet our purpose. However, we do not intend to discuss its precise mathematical definition here, for it is beyond the scope of this paper.   In this paper, the `parametrised family of operators' $\#(s,t)$ shall occasionally be referred to as \emph{hashed operators} of the unitary operators, in a rather casual manner.  Due to the commutativity of the unitary operators for a simultaneously measurable pair, the hashed operator $\#(s,t) = e^{-isA}e^{-itB}$ is always unique, while on the other hand, hashed operators admit variety for non-commutative pairs.
Now, given a hashed operator $\#$ of $A$ and $B$, we then introduce the collection of all parametrised operators of the form
\begin{equation}\label{def:QSM_FT}
\hat{\mathfrak{M}}_{A,B} := \left\{ \# : \# \text{ is a hashed operator of $A$ and $B$ defined as in \eqref{def:decent_mixture} } \right\}.
\end{equation}
As we have seen above, in the case in which $A$ and $B$ are simultaneously measurable, the above collection in fact consists of only one trivial element
\begin{equation}
\hat{\mathfrak{M}}_{A,B} = \left\{ e^{-isA}e^{-itB} \right\},
\end{equation}
due to the strong commutativity of the two operators. 
On the other hand, one readily observes that the cardinality of $\hat{\mathfrak{M}}_{A,B}$ is always greater than unity if the pair of observables $A$ and $B$ fails to strongly commute.
\begin{lemma}\label{lem:card_hash_op}
The cardinality of the collection $\hat{\mathfrak{M}}_{A,B}$ is equal to unity if and only if the observables $A$ and $B$ strongly commute with each other. Otherwise, the cardinality is always greater than unity.
\end{lemma}

\paragraph{Distributions generated by Hashed Operators}
We now consider the inverse Fourier transform of all the elements of the hashed operators $\hat{\mathfrak{M}}_{A,B}$, and thus formally introduce
\begin{equation}\label{def:QJSD}
\mathfrak{M}_{A,B} := \left\{\mathscr{F}^{-1} \# : \# \in \hat{\mathfrak{M}}_{A,B} \right\},
\end{equation}
without any consideration of the mathematical intricacies involved in its well-definedness.  By the injectivity of the Fourier transformation, one intuitively expects that the collection reduces to the single element
\begin{equation}
\mathfrak{M}_{A,B} = \left\{ E_{A,B} \right\}
\end{equation}
when the operators $A$ and $B$ strongly commute with each other, which should be nothing but the original product spectral measure of $A$ and $B$. On the other hand, Lemma~\ref{lem:card_hash_op} implies that the cardinality of the  collection $\mathfrak{M}_{A,B}$ is always greater than unity in the case where $A$ and $B$ are not simultaneously measurable.  Although being possibly highly non-unique, each element of the collection $\mathfrak{M}_{A,B}$ defined for non-commuting pairs retains similar properties to those of the standard product spectral measures.  Incidentally, choosing any element $\Pi \in \mathfrak{M}_{A,B}$ of the collection, the `total integration' reduces to the unit $I$, as one finds under the formal computation
\begin{align}\label{eq:tot_int_joint_spect_dist}
\int_{\mathbb{K}^{2}} \Pi(a,b)\ dm_{2}(a,b)
    &= \int_{\mathbb{K}^{2}} e^{-i (\langle 0,a\rangle + \langle 0,b\rangle)} \Pi(a,b)\ dm_{2}(a,b) \nonumber \\
    &= \left( \mathscr{F}\Pi \right)(0,0) \nonumber \\
    &= \#(0,0) = I,
\end{align}
where $\#$ is the hashed operator whose inverse Fourier transform is the element $\Pi = \mathscr{F}^{-1}\#$ of our choice.  As for the marginals, by formally introducing
\begin{equation}
\Pi_{B}(b) := \int_{\mathbb{K}} \Pi(a,b)\ dm(a),
\end{equation}
one observes under a formal computation that
\begin{align}
\left(\mathscr{F}\Pi_{B}\right) (t)
    &= \int_{\mathbb{K}} e^{-i \langle t,b\rangle} \left( \int_{\mathbb{K}} \Pi(a,b)\ dm(a) \right) dm(b) \nonumber \\
    &=  \int_{\mathbb{K}^{2}} e^{-i (\langle 0,a\rangle + \langle t,b\rangle)} \Pi(a,b)\ dm_{2}(a,b) \nonumber \\
    &= \#(0,t) \nonumber \\
    &= e^{-itB} \nonumber \\
    &= \left(\mathscr{F}E_{B}\right) (t).
\end{align}
The injectivity of the Fourier transformation $\mathscr{F}$ leads us to conclude that the marginal $\Pi_{B} =  E_{B}$ is essentially the same object as the original spectral measure governing the probabilistic behaviour of the outcomes of $B$.  By a parallel argument, one also finds that the marginal
\begin{equation}
\Pi_{A}(a) := \int_{\mathbb{K}} \Pi(a,b)\ dm(b)
\end{equation}
is nothing but $\Pi_{A} = E_{A}$.  These properties are naturally found common in product spectral measures defined for strongly commuting pairs of self-adjoint operators, although each $\Pi(a,b)$ is not necessarily a projection, or may not be even positive.  This  tempts us to introduce the term \emph{quasi-joint-spectral distributions} of a pair of observables, which can be understood as a generalisation of the concept of spectral measures or POVMs.
\begin{definition*}[Quasi-joint-spectral Distribution of a Pair of Quantum Observables]
Let $A$ and $B$ be self-adjoint operators on $\mathcal{H}$.  We call an element of $\mathfrak{M}_{A,B}$ a quasi-joint-spectral distribution of the pair of observables $A$ and $B$.  
The cardinality of the collection $\mathfrak{M}_{A,B}$ is equal to unity if and only if $A$ and $B$ strongly commute with each other.  Otherwise, it is always greater than unity.
\end{definition*}
\noindent
In the case where the observables $A$ and $B$ happen to strongly commute with each other, we specifically call the unique element of $\mathfrak{M}_{A,B}$ the \emph{joint-spectral distribution} of $A$ and $B$, which is nothing but the product spectral measure $E_{A,B}$ of the pair in standard terminology.
We note that the terminologies introduced above are non-standard, and are to be used only in this paper.

\subsubsection{Quasi-joint-probability Distributions}

Although the study on the precise definitions and properties of the family of quasi-joint-spectral distributions would be of mathematical interest in its own right, we shall refrain from going further due to the limited mathematical tools available.  Instead, we turn to a more elementary object to ease our discussion.

Now, given a quasi-joint-spectral distribution $\Pi \in \mathfrak{M}_{A,B}$ of $A$ and $B$, we fix a specific quantum state $|\phi\rangle \in \mathcal{H}$, and consider a distribution of the form formally defined by
\begin{align}\label{def:QJP_explicit}
p(a,b)
    &:= \frac{\langle \phi, \Pi(a,b) \phi \rangle}{\|\phi\|^{2}} \nonumber \\
    &= \frac{\langle \phi, (\mathscr{F}^{-1}\#)(a,b) \phi \rangle}{\|\phi\|^{2}} \nonumber \\
    &= \left( \mathscr{F}^{-1} \frac{\langle \phi, \#(\,\cdot\, , \,\cdot\,) \phi \rangle}{\|\phi\|^{2}} \right)(a,b),
\end{align}
where $\#$ is the hashed operator of which inverse Fourier transform $\Pi = \mathscr{F}^{-1}\#$ is the quasi-joint-spectral distribution under consideration.
Since the distribution $p$ is `scalar valued', it should be a much more feasible object to deal with than the `operator valued' distribution $\Pi$ introduced earlier.
We thus introduce the collection
\begin{equation}\label{def:QJP_FT}
\hat{\mathfrak{M}}_{A,B}^{\phi} := \left\{ \frac{\langle \phi, \#(s,t) \phi \rangle}{\|\phi\|^{2}} : \# \in \hat{\mathfrak{M}}_{A,B}, |\phi\rangle \in \mathcal{H} \right\}
\end{equation}
of all distributions generated by the hashed operators of the parametrised unitary operators give a fixed state, and in turn formally define
\begin{equation}\label{def:QJP}
\mathfrak{M}_{A,B}^{\phi} := \left\{\mathscr{F}^{-1} u : u \in \hat{\mathfrak{M}}_{A,B}^{\phi} \right\},
\end{equation}
by their inverse Fourier transforms. We thus summarise as:
\begin{definition*}[Quasi-joint-probability Distribution of a Pair of Quantum Observables]
Let $A$ and $B$ be self-adjoint operators on $\mathcal{H}$, and let $|\phi\rangle \in \mathcal{H}$.  We call an element of $\mathfrak{M}_{A,B}^{\phi}$ a quasi-joint-probability (QJP) distribution of the pair of observables $A$ and $B$ on $|\phi\rangle$.  The cardinality of the collection $\mathfrak{M}_{A,B}^{\phi}$ is equal to unity for every choice of the vector $|\phi\rangle$ if and only if the observables $A$ and $B$ strongly commute with each other. Otherwise, there exists a vector $|\phi\rangle$ such that the cardinality is greater than unity.
\end{definition*}
\noindent
In the case where the observables $A$ and $B$ happen to strongly commute with each other, we specifically call the unique element of $\mathfrak{M}_{A,B}^{\phi}$ the joint-probability distribution of $A$ and $B$  on $|\phi\rangle$, which is nothing but the probability measure $\mu_{A,B}^{\phi}$ of the pair introduced in \eqref{def:prob_measrue_A_simul}.
Given a hashed operator $\#$ of the parametrised unitary groups and a quantum state $|\phi\rangle \in \mathcal{H}$, we call an element $p \in \mathfrak{M}_{A,B}^{\phi}$ specified by
\begin{equation}
(\mathscr{F}p)(s,t) = \frac{\langle \phi, \#(s,t) \phi \rangle}{\|\phi\|^{2}},
\end{equation}
the QJP distribution generated by $\#$ and $|\phi\rangle$.
Our choice of the denomination of the elements of $p \in \mathfrak{M}_{A,B}^{\phi}$ is due to the fact that they retain similar properties to those of classical joint-probability distributions.  
Indeed, the `total integration' reduces to
\begin{align}
\int_{\mathbb{K}^{2}} p(a,b)\ dm_{2}(a,b)
    &= \int_{\mathbb{K}^{2}} e^{-i (\langle 0,a\rangle + \langle 0,b\rangle)} p(a,b)\ dm_{2}(a,b) \nonumber \\
    &= \left( \mathscr{F}p \right)(0,0) \nonumber \\
    &= \frac{\langle \phi, \#(0,0) \phi\rangle}{\|\phi\|^{2}} = 1,
\end{align}
where $\#$ is the hashed operator that, together with $|\phi\rangle$, generates $p$.  As for the marginals, by introducing the marginal distribution formally defined by
\begin{equation}
p_{B}(b) := \int_{\mathbb{K}} p(a,b)\ dm(a),
\end{equation}
one observes through a formal computation that
\begin{align}
\left(\mathscr{F}p_{B}\right) (t)
    &= \int_{\mathbb{K}} e^{-i \langle t,b\rangle} \left( \int_{\mathbb{K}} p(a,b)\ dm(a) \right) dm(b) \nonumber \\
    &=  \int_{\mathbb{K}^{2}} e^{-i (\langle 0,a\rangle + \langle t,b\rangle)} p(a,b)\ dm_{2}(a,b) \nonumber \\
    &= \frac{\langle \phi, \#(0,t) \phi\rangle}{\|\phi\|^{2}} \nonumber \\
    &= \frac{\langle \phi, e^{-itB} \phi\rangle}{\|\phi\|^{2}} \nonumber \\
    &= \left(\mathscr{F}\mu_{B}^{\phi}\right) (t).
\end{align}
The injectivity of the Fourier transformation $\mathscr{F}$ leads us to conclude that the distribution $p_{B}(b)$ is essentially the same object as the probability measure $\mu_{B}^{\phi}$ describing the probabilistic behaviour of the outcomes of $B$.  By a parallel argument, one also finds that the marginal
\begin{equation}
p_{A}(b) := \int_{\mathbb{K}} p(a,b)\ dm(b)
\end{equation}
is nothing but $p_{A} = \mu_{A}^{\phi}$.
Before we proceed further, we make notes on some mathematical intricacies involved in their definitions for the interested.

\paragraph{Mathematical Remarks}

One may notice some subtleties inherent to the definition of $\mathfrak{M}_{A,B}^{\phi}$. The first problem might be the domain of the definition of the inverse Fourier transformation: while the Fourier transform of a complex measure $\mu$ is a function, in regard that it does not necessarily lie in $\hat{\mu} \notin L^{1}(\mathbb{R}^{n})$, its inverse Fourier transform may not be well-defined, even in the case where $A$ and $B$ strongly commute with each other. This can be temporarily remedied by understanding the inverse Fourier transform of an element $u \in \hat{\mathfrak{M}}_{A,B}^{\phi}$ to be the unique complex measure $\mu$ such that $u = \hat{\mu}$ holds, which should be a reasonable treatment due to the injectivity of the Fourier transformation. This provides a sufficient cure in the case where the pair of observables strongly commutes.

On the other hand, another problem arises in the case in which the pair of self-adjoint operators fails to strongly commute: it might happen that, for some element $u \in \hat{\mathfrak{M}}_{A,B}^{\phi}$, there is no complex measure $\mu$ such that its Fourier transform coincides with $u = \hat{\mu}$.
A straightforward and more fundamental cure for this would be to expand our framework into that of generalised functions, specifically, by embedding the space of complex measures into that of tempered distributions.
Indeed, since the Fourier transformation is a bijection on the space of tempered distributions, by understanding that each of the elements of $\hat{\mathfrak{M}}_{A,B}^{\phi}$ to be a tempered distribution, its inverse Fourier transform itself always exists as a tempered distribution.
In consideration of this, since we do not wish to get involved with the theory of generalised functions, we shall be exclusively dealing with those elements $u \in \hat{\mathfrak{M}}_{A,B}^{\phi}$ for which there exists a complex measure $\mu$ satisfying $u = \hat{\mu}$, and understand the element $\mu := \mathscr{F}^{-1}u \in \mathfrak{M}_{A,B}^{\phi}$ to be the complex measure. To this end, we introduce:
\begin{definition*}[Representation by Quasi-probability Measures]
Under the above situation, let $p \in \mathfrak{M}_{A,B}^{\phi}$ be a QJP distribution of $A$ and $B$, and let $u \in \hat{\mathfrak{M}}_{A,B}^{\phi}$ be an element such that $p = \mathscr{F}^{-1}u$. We say that the QJP distribution $p$ admits representation by a quasi-probability measure, if there exists a quasi-probability measure $\mu$ on $\mathbb{K}^{2}$ such that $u = \hat{\mu}$ holds, and understand the QJP distribution $p = \mu$ to be the quasi-probability measure.
\end{definition*}
\noindent
A similar concern arises for the definition of quasi-joint-spectral distributions $\Pi = \mathscr{F}^{-1}\#$ defined as inverse Fourier transforms of hashed operators $\#$ of the unitary operators $e^{-isA}$ and $e^{-itB}$.  Parallel to the `scalar valued' case seen above, quasi-joint-spectral distributions $\Pi$ are better understood as an object generalising the concept of spectral measures (or POVMs), in the sense that, while spectral measures $E$ (or POVMs) yield probability measures $\langle \phi, E(\,\cdot\,)\phi\rangle/\|\phi\|^{2}$ when combined with a vector $|\phi\rangle$, quasi-joint-spectral distributions $\Pi$ yield generalised functions, symbolically denoted by $\langle \phi, \Pi(a,b)\phi\rangle/\|\phi\|^{2}$.  In this respect, quasi-joint-spectral distributions are to be understood as elements of the space of \emph{operator valued (tempered) distributions} (OVDs), which should serve as a generalisation to that of POVMs.

We also note that the methods introduced above in defining QJSDs admit a straightforward generalisation in defining them, not only for a \emph{pair} ($N=2$) of quantum observables as presented above, but also for \emph{arbitrary combinations} $(N \geq 2)$ of quantum observables, or even for \emph{arbitrary combinations of POVMs}.  Also, one may readily generalise the discussion for defining QJP distributions, not just for pure states as presented above by sandwiching the QJSPs by kets and bras, but also for mixed states by taking the trace of the product of QJSPs and density operators.

\subsection{Complex-parametrised Sub-families}

Since we have decided to confine ourselves in the framework of complex measures rather than that of generalised functions due to our restricted mathematical tools available, we would mostly refrain from treating the general cases, and shall concentrate on a special sub-families of QJP distributions of a pair of quantum observables $A$ and $B$.

\subsubsection{Additive Sub-family}\label{sec:additive_subfamily}

As a simple example of QJP distributions admitting representation by quasi-probability measures, we observe:
\begin{lemma}
Let $A$ and $B$ be self-adjoint operators on $\mathcal{H}$, and consider the hashed operator of either of the form
\begin{equation}
\#(s,t) =
    \begin{cases}
    e^{-itA}e^{-isB} \\
    e^{-isB}e^{-itA}
    \end{cases}, \quad s,t \in \mathbb{R}.
\end{equation}
Then, the QJP distributions generated by $\#$ and any choice of the vector $|\phi\rangle \in \mathcal{H}$ admit representation by quasi-probability measures.
\end{lemma}
\begin{proof}
We provide the proof for the first case without loss of generality. Observe that the complex measure
\begin{equation}\label{def:c_qpm_1}
\Delta \mapsto  \nu(\Delta,\Delta_{A}) := \frac{\langle \phi, E_{B}(\Delta)E_{A}(\Delta_{A})\phi \rangle}{\|\phi\|^{2}}, \quad \Delta \in \mathfrak{B}
\end{equation}
is absolutely continuous with respect to $\mu_{B}^{\phi}$ for all fixed $\Delta_{A} \in \mathfrak{B}$. This allows us to construct a transition quasi-probability kernel by taking the Radon-Nikod{\'y}m derivative of the above complex measure with respect to $\mu_{B}^{\phi}$. A direct application of Proposition~\ref{prop:ctk_to_cm} with $f(a,b) := e^{-i(as + bt)}$ then leads to the desired statement.
\end{proof}
\noindent
This inspires us to introduce the complex linear combinations of the above two distributions.
We hereby consider the hashed operators of the form
\begin{equation}\label{def:hash_add_cparam}
\#_{\mathrm{add}}^{\alpha}(s,t) := \frac{1 + \alpha}{2} e^{-itB}e^{-isA} + \frac{1 - \alpha}{2}e^{-isA}e^{-itB}, \quad s, t \in \mathbb{R},\ \alpha \in \mathbb{C},
\end{equation}
and observe that the QJP distributions induced by them naturally admit representation by quasi-probability measures.
\begin{corollary}
The QJP distributions generated by the hashed operators of the form \eqref{def:hash_add_cparam} and $|\phi\rangle \in \mathcal{H}$ admits representation by quasi-probability measures.
\end{corollary}
\noindent
In this paper, we call the above sub-family of QJP distributions the \emph{additive complex-parametrised sub-family} of QJP distributions of $A$ and $B$ on $|\phi\rangle$ (or simply, the additive sub-family, for short).

\subsubsection{Convolutive Sub-family}

One realises below that another class of QJP distributions parametrised by a complex number can be introduced. 
We hereby consider the hashed operators of the form
\begin{equation}\label{def:hash_conv_cparam}
\#_{\mathrm{cnv}}^{\alpha}(s,t) := e^{-i\langle s, (1 -\alpha)/2 \rangle A}e^{-itB}e^{-i\langle s, (1 + \alpha)/2 \rangle A}, \quad s \in \mathbb{C}, t \in \mathbb{R},\ \alpha \in \mathbb{C}
\end{equation}
where $\langle s, \alpha \rangle := s_{1} \alpha_{1} + s_{2} \alpha_{2}$ denotes the inner product of
\begin{equation}\label{eq:complex_as_vector}
s = s_{1} + i s_{2}, \quad \alpha = \alpha_{1} + i\alpha_{2},
\end{equation}
each of them understood as real vectors of $\mathbb{R}^{2} \cong \mathbb{C}$, and introduce the \emph{convolutive complex-parametrised sub-family} of QJP distributions of $A$ and $B$ on $|\phi\rangle$ (or simply, the convolutive sub-family, for short) by those elements of $\mathfrak{M}_{A,B}^{\phi}$ that are generated by the hashed operators of the form \eqref{def:hash_conv_cparam} and $|\phi\rangle$.

\paragraph{Linear Transformation}
It is of natural interest to find out the condition as to when an element of the convolutive sub-family admits representation by quasi-probability measures. Obviously, the choice $\alpha = \pm 1$ admits it, since they are also members of the additive sub-family introduced earlier.  As for the other choices of the complex parameter $\alpha \in \mathbb{C}$, we first introduce an auxiliary distribution defined by
\begin{equation}\label{def:FT_gen_dist}
\tilde{u}(s,t) := \frac{\langle \phi, e^{-is_{1} A}e^{-itB}e^{-i s_{2} A} \phi \rangle}{\|\phi\|^{2}}, \quad s \in \mathbb{C}, t \in \mathbb{R},
\end{equation}
where $s = s_{1} + is_{2} \in \mathbb{C}$, $s_{1}, s_{2} \in \mathbb{R}$ is defined as \eqref{eq:complex_as_vector}.  Once there exists a quasi-probability measure $\tilde{\mu}$ such that its Fourier transform coincides with $\mathscr{F}\tilde{\mu} = \tilde{u}$, one finds below that every member of the convolutive sub-family is a linear transform of the quasi-probability measure $\tilde{\mu}$, hence themselves admit representation by quasi-probability measures. To see this, we first introduce the matrix
\begin{align}\label{eq:lin_trans_T}
T_{\alpha} := \left(
    \begin{array}{cc}
     (1 - \alpha_{1})/2 & (1 + \alpha_{1})/2 \\
    -\alpha_{2}/2 & \alpha_{2}/2
    \end{array}
    \right),
\end{align}
defined for each complex number $\alpha = \alpha_{1} + i \alpha_{2} \in \mathbb{C}$, $\alpha_{1}, \alpha_{2} \in \mathbb{R}$.  The Fourier transform $\mathscr{F} \tilde{\mu}_{(T_{\alpha} \times I)}$ of the linear transform of the quasi-probability measure $\tilde{\mu}$ with respect to the operator
\begin{equation}\label{def:two_lin_trans}
(T_{\alpha} \times I)(a,b) := (T_{\alpha}a, b), \quad a \in \mathbb{C}, b \in \mathbb{R},
\end{equation}
reads
\begin{align}
\left( \mathscr{F} \tilde{\mu}_{(T_{\alpha} \times I)} \right)(s,t)
    &= \tilde{u}(T_{\alpha}^{*}s, t) \nonumber \\
    &= \frac{\langle \phi, e^{-i\langle s, (1 -\alpha)/2 \rangle A}e^{-itB}e^{-i\langle s, (1+\alpha)/2 \rangle A} \phi \rangle}{\|\phi\|^{2}},
\end{align}
where we have used \eqref{eq:FT_and_Lin_Trans} in the first equality.
We thus have:
\begin{lemma}[Transformation between Parameters]\label{lem:Trans_b_Param}
Let $A$ and $B$ be self-adjoint operators on $\mathcal{H}$, and let $|\phi\rangle \in \mathcal{H}$.
\begin{enumerate}
\item If the inverse Fourier transform of the auxiliary distribution \eqref{def:FT_gen_dist} admits representation by a quasi-probability measure, then all the members of the convolutive sub-family admit representation by quasi-probability measures.
\item Under the above situation, let $T_{\alpha}$ be the linear transformation defined for every choice of the complex parameter $\alpha \in \mathbb{C}$ as in \eqref{eq:lin_trans_T}, and let $\tilde{u}$ be the auxiliary distribution \eqref{def:FT_gen_dist} and $\tilde{\mu}$ be the quasi-probability measure such that $\mathscr{F}\tilde{\mu} = \tilde{u}$.  Then, every member of the convolutive sub-family can be described as the linear transform of $\tilde{\mu}$ as
\begin{equation}\label{eq:alpha_generation}
\mu^{\phi,\alpha}_{\mathrm{cnv}} = \tilde{\mu}_{(T_{\alpha} \times I)},
\end{equation}
where $\mu^{\phi,\alpha}_{\mathrm{cnv}}$ denotes the quasi-probability measure generated by the hashed operator of the form \eqref{def:hash_conv_cparam}, and $T_{\alpha} \times I$ is the operator defined as in \eqref{def:two_lin_trans}.
\end{enumerate}
\end{lemma}

\paragraph{Representation by Quasi-probability Measures}
Now, observing that the determinant of the linear transform $T_{\alpha}$ reads
\begin{equation}
\det T_{\alpha} = \mathrm{Im}\, \alpha/2,
\end{equation}
one finds that the transformation $\mu \mapsto \mu_{(T_{\alpha} \times I)}$ is invertible if and only if $\alpha \in \mathbb{C} \setminus \mathbb{R}$, for indeed $T_{\alpha} \in \mathrm{GL}(2;\mathbb{R}) \Leftrightarrow \alpha \in \mathbb{C}\setminus\mathbb{R}$.  The product rule \eqref{eq:product_rule_image_measure_lin} then reveals that, one may move from one member of the convolutive sub-family to another by a sequential application of the transformations as
\begin{equation}\label{scheme:trans_param}
\mu^{\alpha} \xrightarrow{\quad T_{\alpha}^{-1} \times I \quad} \tilde{\mu} \xrightarrow{\quad T_{\alpha^{\prime}} \times I \quad} \mu^{\alpha^{\prime}}
\end{equation}
for the choice $\alpha \in \mathbb{C}\setminus\mathbb{R}$ and $\alpha^{\prime} \in \mathbb{C}$.  Combining Lemma~\ref{lem:Trans_b_Param} with the above observation, one concludes:
\begin{corollary}[Representation by Quasi-probability Measures]\label{cor:rep_by_qpm_condition}
The following conditions are equivalent.
\begin{enumerate}
\item The inverse Fourier transform of the distribution $\tilde{u}$ defined in \eqref{def:FT_gen_dist} admits representation by quasi-probability measures.
\item A member of the convolutive sub-family for the choice of the parameter $\alpha \in \mathbb{C}\setminus\mathbb{R}$ admits representation by quasi-probability measures.
\item Every member of the convolutive sub-family admits representation by quasi-probability measures.
\end{enumerate}
\end{corollary}

\paragraph{Explicit Computation of the Members of the convolutive Sub-family}
We shall provide an explicit example of the case in which every member of the convolutive complex-parametrised sub-family admits representation by quasi-probability measures.   
\begin{proposition}\label{prop:construction_of_qjp}
Let $A$ and $B$ self-adjoint, and suppose that $B$ has spectrum $\sigma(B)$ of finite cardinality and that it is non-degenerate
\begin{equation}
B = \sum_{b \in \sigma(B)} b \cdot |b \rangle\langle b|.
\end{equation}
For a quantum state $|\phi\rangle \in \mathcal{H}$ such that the probability of finding the outcomes of $B$ is non-vanishing $|\langle b, \phi \rangle|^{2} \neq 0$  for all its eigenvalues $b \in \sigma(B)$, every member of the convolutive sub-family of QJP distributions admits representation by quasi-probability measures.
\end{proposition}
\begin{proof}
Corollary~\ref{cor:rep_by_qpm_condition} purports that it suffices to construct the quasi-probability measure $\tilde{\mu}$ that satisfies $\mathscr{F}\tilde{\mu} = \tilde{u}$, where $\tilde{u}$ is the auxiliary distribution \eqref{def:FT_gen_dist}.  
Now, under the above conditions, let $b \in \sigma(B)$ and $\Delta_{A} \in \mathfrak{B}^{1}$ be fixed, and introduce the Radon-Nikod{\'y}m derivative
\begin{align}
\nu_{b}^{\phantom{*}}(\Delta_{A})
    &:= (d\nu(\,\cdot\,,\Delta_{A}) / d\mu_{B}^{\phi})(b) \nonumber \\
    &= \frac{\langle b, E_{A}(\Delta_{A})\phi\rangle}{\langle b, \phi\rangle},
\end{align}
where the complex measure $\nu(\,\cdot\,,\Delta_{A})$ was defined in \eqref{def:c_qpm_1}.  For every fixed $b \in \sigma(B)$, this defines a quasi-probability measure $\Delta_{A} \mapsto \nu_{b}^{\phantom{*}}(\Delta_{A})$, which is in fact nothing but a slight generalisation of the quasi-probability measure \eqref{def:aux_qpm_b_a} previously introduced in Section~\ref{sec:ps_II}. Defining the product complex measure
\begin{equation}
K(b,\,\cdot\,) := \nu(b,\,\cdot\,)^{*} \otimes \nu(b,\,\cdot\, ),
\end{equation}
on the product space $\mathbb{R}^{2} \cong \mathbb{C}$ for each $b \in \sigma(B)$, we intend to extend the domain of the variable $b$ to the whole real line to make a transition quasi-probability kernel from $(\mathbb{R},\mathfrak{B})$ into $(\mathbb{C},\mathfrak{B}(\mathbb{C}))$ by defining, for example,
\begin{equation}
\tilde{K}(b,\Delta_{A}) :=
    \begin{cases}
    K(b,\Delta_{A}), & (b \in \sigma(B)) \\
    \delta_{0}(\Delta_{A}) , & (b \notin \sigma(B)),
    \end{cases}
\end{equation}
where $\delta_{0}$ is the delta measure centred at the origin (for the extension into $\mathbb{R} \setminus \sigma(B)$, we could have assigned any quasi-probability measure so that the extension makes a transition quasi-probability kernel as a whole).
Letting $\tilde{\mu}$ denote the quasi-probability measure on the product space $\mathbb{C}\times\mathbb{R}$ defined by $\tilde{K}$ and $\mu_{B}^{\phi}$ by means of
\begin{equation}
\int_{\mathbb{C} \times \mathbb{R}} f(a,b)\ d\tilde{\mu}(a,b) =  \int_{\mathbb{R}} \int_{\mathbb{C}} f(a,b)  \tilde{K}(b,da)\ d\mu_{B}^{\phi}(b),
\end{equation}
(see \eqref{prop:ctk_to_cm}), we maintain that $(\mathscr{F}\tilde{\mu})(s,t) = \langle \phi, e^{-is_{1} A}e^{-itB}e^{-i s_{2} A} \phi \rangle / \|\phi\|^{2}$.
To see this, just let $f(a,b) := e^{-i\langle s,a\rangle}e^{-itb}$ above and compute
\begin{align}
\int_{\mathbb{R}} \int_{\mathbb{C}} e^{-i\langle s,a\rangle}e^{-itb}  \tilde{K}(b,da)\ d\mu_{B}^{\phi}(b)
    &= \sum_{b \in \sigma(B)} \int_{\mathbb{C}} e^{-i\langle s,a\rangle}e^{-itb} K(b,da)\cdot \frac{|\langle b, \phi \rangle|^{2}}{\|\phi\|^{2}} \nonumber \\
    &= \sum_{b \in \sigma(B)} e^{-itb} \frac{\langle \phi, e^{-is_{1}A} b\rangle \langle b, e^{-is_{2}A}\phi\rangle}{\|\phi\|^{2}} \nonumber \\
    &= \frac{\langle \phi, e^{-is_{1} A}e^{-itB}e^{-i s_{2} A} \phi \rangle}{\|\phi\|^{2}},
\end{align}
which was to be demonstrated.  We have thus achieved a concrete construction of the quasi-probability measure, whose Fourier transform is the distribution $\tilde{u}$ defined in \eqref{def:FT_gen_dist}.
\end{proof}
\noindent
In passing, we note that, by comparing the transformation matrices \eqref{eq:lin_trans_T} and \eqref{eq:lin_trans_T_2i}, one finds that the quasi-probability measure $\mu_{A,B}^{\phi}$ obtained in the preceding Section~\ref{sec:ps_II} defined as in \eqref{def:quasi_joint_prob_meas} is nothing but the member of the convolutive complex-parametrised sub-family
\begin{equation}\label{eq:qjpm_2i_revealed}
\mu_{A,B}^{\phi} = \mu^{\phi,i}_{\mathrm{cnv}}
\end{equation}
for the purely imaginary choice $\alpha = i$ of the complex parameter.

\subsubsection{Qualification as Quasi-joint-probability Distributions}
Although we have provided a formal discussion to the problem, it yet remains to be confirmed by a rigorous treatment that every member of either the additive or the convolutive complex-parametrised sub-family of QJP distributions of a pair of quantum observables indeed qualifies as what its name indicates itself to be.  Without loss of generality, we only provide the demonstration for the convolutive sub-family, since the proof for the additive subfamily is essentially the same.
\begin{proposition}[Qualification as Quasi-joint-probability Distributions]
Let $A$ and $B$ be self-adjoint, $|\phi\rangle \in \mathcal{H}$, and suppose that there exists a quasi-probability measure $\mu$ on the product space $\mathbb{C} \times \mathbb{R}$ such that
\begin{equation}
\left( \mathscr{F} \mu \right)(s,t) = \frac{\langle \phi, \#_{\mathrm{cnv}}^{\alpha}(s,t) \phi \rangle}{\|\phi\|^{2}}
\end{equation}
holds for some $\alpha \in \mathbb{C}$. Then, $\mu$ qualifies as a QJP distribution of $A$ and $B$ on $|\phi\rangle$, in the sense that \eqref{def:qjpm} holds.
\end{proposition}
\begin{proof}
We first observe a general result regarding marginals of complex measures and Fourier transformations. Let $\mu$ be a complex measure on the product space $\mathbb{R}^{m} \times \mathbb{R}^{n}$, and define the marginal of $\mu$ by
\begin{align}
\mu_{2} : \Delta \mapsto \mu_{2}(\Delta) := \mu(\mathbb{R}^{m} \times \Delta), \quad \Delta \in \mathfrak{B}^{n},
\end{align}
which is itself a complex measure on $(\mathbb{R}^{n},\mathfrak{B}^{n})$. One then observes
\begin{align}
\hat{\mu}_{2}(p)
    &:= \int_{\mathbb{R}^{n}} e^{-i\langle p,y \rangle}\ d\mu_{2}(y) \nonumber \\
    &= \int_{\mathbb{R}^{(m+n)}} e^{-i\langle 0,x \rangle} e^{-i\langle p,y \rangle}\ d\mu(x,y) \nonumber \\
    &= \hat{\mu}(0,p),
\end{align}
where the second equality is due to the change of variables formula \eqref{eq:Transformationsformel_Bildmass} for the image measure $\mu_{2} = \pi_{2}(\mu)$, where $\pi_{2}(x,y) = y$, $x \in \mathbb{R}^{m}$, $y \in \mathbb{R}^{n}$ is the projection on the second variable.
Applying this fact to our situation as
\begin{equation}
\left( \mathscr{F} \mu \right)(0,t)
    := \frac{\langle \phi, e^{-itB} \phi \rangle}{\|\phi\|^{2}} = \left( \mathscr{F} \mu_{B}^{\phi} \right)(t),
\end{equation}
one readily finds
\begin{equation}
\mu(\mathbb{C} \times \Delta) = \mu_{B}^{\phi}(\Delta), \quad \Delta \in \mathfrak{B}^{1},
\end{equation}
by the injectivity of the Fourier transformation.
One may also demonstrate $\mu(\Delta \times \mathbb{R}) = \mu_{A}^{\phi}(\Delta)$, $\Delta \in \mathfrak{B}(\mathbb{K})$ by an analogous reasoning, which completes our proof.
\end{proof}

\subsubsection{Relation to other known Proposals}

We demonstrate below, in passing, that the complex-parametrised sub-families of the QJP distributions of a pair of quantum observables serve as generalisations to the other well known proposals of quasi-probability distributions.

\paragraph{Kirkwood-Dirac Distribution}

We first note that the Kirkwood-Dirac distribution, introduced in \eqref{def:Kirkwood-Dirac_Distr} in a formal manner, can be given a mathematically rigorous definition within our framework, and that it belongs to both the additive and convolutive sub-families of the QJP distributions for the choice $\alpha = 1$.

\begin{definition*}[Kirkwook-Dirac Quasi-joint-probability Distribution]
Let $A$ and $B$ be self-adjoint on $\mathcal{H}$, and let $|\phi\rangle \in \mathcal{H}$. We call the member of the additive/convolutive sub-family of the QJP distributions of the pair of observables $A$ and $B$ for the choice $\alpha = 1$, the Kirkwook-Dirac QJP distribution of the pair.
\end{definition*}
\noindent
To see how this definition can be justified, observe the following formal chain of expressions
\begin{align}
\int_{\mathbb{R}^{2}} e^{-i(as + bt)} K_{A,B}^{\phi}(a,b)\ dm_{2}(a,b)
    &= \int_{\mathbb{R}^{2}} e^{-i(sa + tb)} \frac{\langle \phi, b\rangle\langle b, a \rangle\langle a, \phi \rangle}{\|\phi\|^{2}}\ dm_{2}(a,b) \nonumber \\
    &= \frac{\langle \phi, e^{-itB}e^{-isA} \phi \rangle}{\|\phi\|^{2}},
\end{align}
where $K_{A,B}^{\phi}$ is the formal definition of the Kirkwood-Dirac distribution introduced in \eqref{def:Kirkwood-Dirac_Distr}. The injectivity of the Fourier transformation leads to the desired statement.

\paragraph{Wigner-Ville Distribution}

We next note that the Wigner-Ville distribution, introduced in \eqref{def:Wigner-Ville_QPD}, is also a special member of the convolutive sub-family of QJP distributions.
\begin{proposition}[Wigner-Ville Distribution]
Let $\{ L^{2}(\mathbb{R}), \mathscr{S}(\mathbb{R}), \{\hat{x}, \hat{p}\}\}$ denote the one-dimensional Schr{\"o}dinger representation of the CCR. Then, for the choice $\psi \in L^{1}(\mathbb{R}) \cap L^{2}(\mathbb{R})$ of the wave-function, the member of the convolutive sub-family of the QJP distributions of the canonically conjugate pair $\hat{p}$ and $\hat{x}$ admits representation by quasi-probability measures for the choice $\alpha = 0$, which we denote by $\mu_{\mathrm{cnv}}^{\psi,0}$. The quasi-probability measure $\mu_{\mathrm{cnv}}^{\psi,0}$ is absolutely continuous, and its Radon-Nikod{\'y}m derivative with respect to the renormalised two-dimensional Lebesgue-Borel measure reads
\begin{equation}
\left( d\mu_{\mathrm{cnv}}^{\psi,0} / dm_{2} \right)(p,x) := W^{\psi}(x,p),
\end{equation}
where the r.~h.~s. is the Wigner-Ville distribution introduced in \eqref{def:Wigner-Ville_QPD}.
\end{proposition}
\begin{proof}
Observe that the condition $\psi \in L^{1}(\mathbb{R}) \cap L^{2}(\mathbb{R})$ guarantees the integrability $W^{\psi} \in L^{1}(\mathbb{R}^{2})$ of the WV distribution, based on which we compute
\begin{align}
&\int_{\mathbb{R}^{2}} e^{-i(sp + tx)} W^{\psi}(x,p)\ dm_{2}(x,p) \nonumber \\
    &\qquad := \int_{\mathbb{R}^{2}} e^{-i(sp + tx)} \left( \int_{\mathbb{R}} \psi^{*}(x + y/2) \psi(x - y/2)e^{ipy}\ dm(y) \right)dm_{2}(x,p) \nonumber \\
    &\qquad = \int_{\mathbb{R}} e^{-itx} \psi^{*}(x + s/2) \psi(x - s/2)\ dm(x) \nonumber \\
    &\qquad = \int_{\mathbb{R}} e^{-itx} (e^{is\hat{p}/2}\psi)^{*}(x) (e^{-is\hat{p}/2}\psi)(x)\ dm(x) \nonumber \\
    &\qquad = \langle e^{is\hat{p}/2}\psi, e^{-it\hat{x}} e^{-is\hat{p}/2}\psi\rangle \nonumber \\
    &\qquad = \mathscr{F}\left( \mu_{\mathrm{cnv}}^{\psi,0} \right).
\end{align}
Combining \eqref{eq:FT_measure_vs_density} and the injectivity of the Fourier transformation, one arrives at the desired statement.
\end{proof}

\subsection{Some General Properties}

We next observe some general properties of QJP distributions.  We first provide some discussion regarding the operation of taking the complex conjugate, and subsequently seek for the condition for their realness.

\subsubsection{Complex Conjugate}

We are interested in the complex conjugate of QJP distributions of a pair of observables $A$ and $B$ on $|\phi\rangle$.  To this, let $\#$ be a hashed operator of the unitary operators $e^{-isA}$, $e^{-itB}$, and let $|\phi\rangle \in \mathcal{H}$ be such that the QJP distribution generated by them admits representation by a quasi-probability measure $\mu$. By applying \eqref{eq:FT_cc_Measure}, one readily finds that the Fourier transform of the complex conjugate $\mu^{*}$ reads
\begin{align}\label{eq:FT_cc}
(\mathscr{F}\mu^{*})(s,t)
    &= (\mathscr{F}\mu)^{*}(-s,-t) \nonumber \\
    &= \frac{\langle \#(-s,-t) \phi,\phi\rangle}{\|\phi\|^{2}} \nonumber \\
    &= \frac{\langle \phi, \#(-s,-t)^{*} \phi\rangle}{\|\phi\|^{2}}, \quad s, t \in \mathbb{K}
\end{align}
where $\#(s,t)^{*}$ denotes the `adjoint' of the hashed operator.  Since the `involution' $\#(-s,-t)^{*}$ is itself a hashed operator of $e^{-isA}$ and $e^{-itB}$, one concludes that the complex conjugate $\mu^{*}$ is again a QJP distribution of the pair of observables $A$ and $B$, and that it is precisely the distribution generated by the `involution' of the original hashed operator.
One also specifically finds that the sub-family of the QJP distributions $\mathfrak{M}_{A,B}^{\phi}$ that admit representations by quasi-probability measures is closed under the operation of taking the complex conjugate.

Parallel to this, by observing that the left most hand side of \eqref{eq:FT_cc} can be written as
\begin{align}
(\mathscr{F}\mu^{*})(s,t) = \frac{\langle \phi,(\mathscr{F}\Pi^{*})(s,t) \phi \rangle}{\|\phi\|^{2}}, \quad s, t \in \mathbb{K}
\end{align}
where $\Pi = \mathscr{F}^{-1}\#$ is the quasi-joint-spectral distribution, one concludes the validity of the equality
\begin{align}\label{eq:FT_cc_SpectDist}
(\mathscr{F}\Pi^{*})(s,t)
    &=  \#(-s,-t)^{*} \nonumber \\
    &= (\mathscr{F}\Pi)(-s,-t)^{*} \nonumber \\
    &=: (\mathscr{F}\Pi)^{\dagger}(s,t), \quad s, t \in \mathbb{K},
\end{align}
where $(\mathscr{F}\Pi)^{\dagger}$ denotes the `involution' (observe the analogy between \eqref{eq:FT_cc_Measure}).
This shows that the `adjoint' $\Pi^{*}$ of the quasi-joint-spectral distribution of $A$ and $B$ is again a quasi-joint-spectral distribution of the pair, and that it is precisely the inverse Fourier transform of the `involution' of the original hashed operator.

\paragraph{Complex-parametrised Sub-families}

Armed with our findings, one may explicit compute the complex conjugate of the elements of both the additive and the convolutive sub-families, and see that the sub-families are also closed under the operation of taking the complex conjugate.  Indeed, if we respectively introduce
\begin{equation}
\Pi_{\mathrm{add}}^{\alpha} := \mathscr{F}^{-1} \#_{\mathrm{add}}^{\alpha}, \qquad \Pi_{\mathrm{cnv}}^{\alpha} := \mathscr{F}^{-1} \#_{\mathrm{cnv}}^{\alpha}, 
\end{equation}
for the members of the additive and convolutive sub-families, by observing that the `involution' of the hashed operators read
\begin{align}
\#_{\mathrm{add}}^{\alpha}(-s,-t)^{*} &= \#_{\mathrm{add}}^{-\alpha^{*}}(s,t), \\
\#_{\mathrm{cnv}}^{\alpha}(-s,-t)^{*} &= \#_{\mathrm{cnv}}^{-\alpha}(s,t),
\end{align}
one finds
\begin{align}
(\Pi_{\mathrm{add}}^{\alpha})^{*} &= \Pi_{\mathrm{add}}^{-\alpha^{*}}, \\
(\Pi_{\mathrm{cnv}}^{\alpha})^{*} &= \Pi_{\mathrm{cnv}}^{-\alpha}.
\end{align}
This provide explicit formulae for the computation of the complex conjugate of the members of the sub-families, and one specifically finds from it that both the sub-families are closed under the operation of taking the complex conjugate as promised.
Now, fixing $|\phi\rangle \in \mathcal{H}$ of one's choice, one observes:
\begin{lemma}[Complex Conjugate: Additive Sub-family]
Let $\mu_{\mathrm{add}}^{\alpha}$ denote the QJP distribution of $A$ and $B$ generated by $\#_{\mathrm{add}}^{\alpha}$ and $|\phi\rangle$.  Then, its complex conjugate reads
\begin{equation}
\left(\mu_{\mathrm{add}}^{\phi,\alpha}\right)^{*} = \mu_{\mathrm{add}}^{\phi,-\alpha^{*}}.
\end{equation}
\end{lemma}
\begin{lemma}[Complex Conjugate: Convolutive Sub-family]
Suppose that the member of the convolutive sub-family admits representation by the quasi-probability measure $\mu_{\mathrm{cnv}}^{\phi,\alpha}$ for the choice $\alpha \in \mathbb{C}$.   Then, the member for the choice $-\alpha \in \mathbb{C}$ also admits representation by quasi-probability measures, and the equality
\begin{equation}
\left(\mu_{\mathrm{cnv}}^{\phi,\alpha}\right)^{*} = \mu_{\mathrm{cnv}}^{\phi,-\alpha}.
\end{equation}
holds.
\end{lemma}

\subsubsection{Realness of the QJP Distributions}
One may naturally be interested in the condition as to when the quasi-joint-spectral distribution $\Pi = \mathscr{F}^{-1}\#$ becomes `self-adjoint' so that the resulting QJP distribution, symbolically denoted by $p(a,b) = \langle \phi, \Pi(a,b) \phi \rangle/\|\phi\|^{2}$, is also `real' for any choice of the vector $|\phi\rangle \in \mathcal{H}$.  While the task of finding the explicit condition for which $\Pi(a,b)$ becomes `self-adjoint' seems at first non-trivial, the problem becomes significantly tractable if one considers its Fourier transform.  Indeed, combining \eqref{eq:FT_cc_SpectDist} with the injectivity of the Fourier transform, one concludes that $\Pi = \Pi^{*}$ is `self-adjoint' if and only if its Fourier transform (namely, the hashed operator) $\# =  \#^{\dagger}$ is a `self-involution'.  Examples of such `self-involutive' hashed operators are provided by
\begin{align}
\#(s,t) =
    \begin{cases}
        e^{-itB/2}e^{-isA}e^{-itB/2}, \\
        e^{-isA/2}e^{-itB}e^{-isA/2}, \\
        \frac{1}{2}\left( e^{-itB}e^{-isA} + e^{-isA}e^{-itB} \right), \\
        e^{-itB/L_{N}}e^{-isA/M_{N}} \cdots e^{-itB/L_{1}}e^{-isA/M_{1}}e^{-itB/L_{1}} \cdots e^{-isA/M_{N}}e^{-itB/L_{N}}, \\
        e^{-i\overline{(sA + tB)}} = \lim_{N \to \infty} \left( e^{-isA/N} e^{-itB/N} \right)^{N}, \\
        \textit{etc.},
    \end{cases}
\end{align}
where $\sum_{k=1}^{N}M_{k}^{-1} = 1$, $\sum_{k=1}^{N}L_{k}^{-1} = 1$ in the third example.  Colloquially speaking, hashed operators in which the disintegrated components of the unitary operators appear `symmetrically' provide straightforward examples.
As for our concrete examples, one finds:
\begin{corollary}[Condition for Realness]
A member of either the additive or convolutive sub-families of QJP distributions of $A$ and $B$ for the choice $\alpha = 0$ is always real.
\end{corollary}

\subsection{Conditioned Measurement Revisited}

We finally investigate how the CM scheme described in Section~\ref{sec:ps_II} fits into our general framework of quasi-joint-probabilities of quantum observables. What we see below is that the CM scheme is essentially \emph{a measurement scheme for measuring QJP distributions of an arbitrary pair of quantum observables}. As before, since the tools for the analysis of the most general cases are beyond the scope of this paper, we shall exclusively concentrate on the subfamily of quasi-joint-probabilities parametrised by a single complex number.
Without loss of generality, we only provide below a demonstration for the convolutive sub-family for simplicity.

\subsubsection{Short Introduction}

We now intend to construct a measurement scheme for obtaining the member of the convolutive sub-family of the QJP distributions for arbitrary choices of the parameter $\alpha \in \mathbb{C}$. 
As for the problem, let us first recall that the quasi-probability measure \eqref{def:quasi_joint_prob_meas} obtained in Section~\ref{sec:ps_II} was nothing but the member for the choice of the parameter $\alpha = i$ (see \eqref{eq:qjpm_2i_revealed} for the discussion). In fact, as we have seen before, once we know the member of the subfamily for the parameter $\alpha \in \mathbb{C}\setminus\mathbb{R}$, we may compute all other members of the complex parameters by sequentially applying linear transformations as depicted in \eqref{scheme:trans_param}.
Hence, the knowledge of the distribution for the choice $\alpha = i$, obtained by means of the CM scheme in view of the WV distribution, actually suffices for our purpose.
Even so, one might be interested in how one could measure the QJP distribution for some specific parameter in a more direct manner. This should also provide a much more transparent view of the measurement scheme described in Section~\ref{sec:ps_II} from a more general viewpoint, which may be beneficial in its own right.

\paragraph{Model and Assumption}
Throughout this subsection, we let $A$ denote an observable on the target system $\mathcal{H}$, and assume that the meter system $\mathcal{K}$ is described  by the one-dimensional Schr{\"o}dinger representation of the CCR $\{ L^{2}(\mathbb{R}), \mathscr{S}(\mathbb{R}), \{\hat{x}, \hat{p}\}\}$ for simplicity.
As usual, we prepare the two systems into their respective initial states $|\phi\rangle \in \mathcal{H}$, $|\psi\rangle \in \mathcal{K}$, and let them interact under the unitary operator $e^{-igA \otimes Y}$, $g \in \mathbb{R}$, for which we choose $Y=\hat{p}$ for definiteness, and let $|\Psi^{g}\rangle$ denote the state of the composite system after the interaction.
Since we intend to confine ourselves within the framework of complex measures, we place several conditions throughout this passage, so that, given a conditioning observable $B$ on the target system $\mathcal{H}$, all the members of the convolutive sub-family of the QJP distributions of $A$ and $B$ on $|\phi\rangle$ admits representation by quasi-probability measures.

\subsubsection{Conditioned Measurement Revisited}

In the previous section, the choice of the QJP distribution we intend to measure on the meter system was the WV distribution, which we found to be nothing but the member of our convolutive sub-family of the QJP distributions of the canonically conjugate pair of observables $A=\hat{p}$, $B=\hat{x}$ for the choice of the parameter $\alpha = 0$.  The result was that, one could obtain the member of the convolutive sub-family of the QJP distributions for arbitrary pairs of quantum observables for the choice $\alpha = i$. 
Motivated by this finding, it is then natural to conjecture that a different choice of the meter QJP distribution results in different choice of the target QJP distribution.

\paragraph{Meter QJP}
The starting point would be to find the equivalent object to $\omega^{\psi}$ for the other choices of the parameter $\alpha \in \mathbb{C}$. To this end, we first  assume $\psi \in L^{1}(\mathbb{R}) \cap L^{2}(\mathbb{R})$, and introduce the function
\begin{equation}
\tilde{\omega}^{\psi,\alpha_{1}}(x,y) := \psi^{*}\left(x - y (1 - \alpha_{1})/2\right) \psi\left(x + y (1 + \alpha_{1})/2 \right)
\end{equation}
and also its Fourier transform
\begin{equation}
\tilde{W}^{\psi,\alpha_{1}}(x,p) := \int_{\mathbb{R}} e^{-ipy} \tilde{\omega}^{\psi,\alpha_{1}}(x,y)\ dm(y),
\end{equation}
where we let $\alpha = \alpha_{1} + i\alpha_{2}$, $\alpha_{1}, \alpha_{2} \in \mathbb{R}$.
Needless to say, the function $\tilde{\omega}^{\psi,0} = \tilde{\omega}^{\psi}$ for the choice $\alpha_{1} = 0$ reduces to the original function introduced in \eqref{def:func_V}, and thus $\tilde{W}^{\psi,0} = \tilde{W}^{\psi}$ is nothing but the (yet-to-be-normalised) WV distribution.  By computing the Fourier transform
\begin{align}\label{eq:FT_W_w_alpha}
\left(\mathscr{F}\tilde{W}^{\psi,\alpha_{1}}\right)(q,y)
    &= \int_{\mathbb{R}^{2}} e^{-i(qx + yp)}\left( \int_{\mathbb{R}} e^{-ipy} \tilde{\omega}^{\psi,\alpha_{1}}(x,y)\ dm(y) \right)dm_{2}(x,p) \nonumber \\
    &= \int_{\mathbb{R}} e^{-iqx} \tilde{\omega}^{\psi,\alpha_{1}}(x,-y)\ dm(x) \nonumber \\
    &= \int_{\mathbb{R}} e^{-iqx} \psi^{*}\left(x + y (1 - \alpha_{1})/2\right) \psi\left(x - y (1 + \alpha_{1})/2 \right)\ dm(x) \nonumber \\
    &= \left\langle \psi, e^{-i\langle y, (1 -\alpha_{1})/2 \rangle \hat{p}} e^{-iq\hat{x}}e^{- i\langle y, (1 + \alpha_{1})/2 \rangle \hat{p}} \psi \right\rangle,
\end{align}
one concludes from the injectivity of the Fourier transformation that the normalisation
\begin{align}
W^{\psi,\alpha_{1}}
    &:= \tilde{W}^{\psi,\alpha_{1}} / \|\psi\|^{2} \nonumber \\
    &= \left( d\mu_{\mathrm{cnv}}^{\psi,\alpha_{1}} / dm_{2} \right)
\end{align}
is nothing but the Radon-Nikod{\'y}m derivative of the member of the convolutive sub-family of the QJP distributions of the canonically conjugate pair $A = \hat{p}$, $B = \hat{x}$ for the choice of the real parameter $\alpha_{1} \in \mathbb{R}$.

\paragraph{Rescaling}
For simplicity of the argument, we only treat the case for the choice $\alpha \in \mathbb{C} \setminus \mathbb{R}$, and for later convenience, we introduce the function
\begin{equation}
\tilde{\upsilon}^{\psi,\alpha}\left(x,y \right) := 
|2/\alpha_{2}|^{-1}\tilde{\omega}^{\psi,-\alpha_{1}}(x,(2/\alpha_{2})y), \quad (\alpha_{2} \neq 0)
\end{equation}
for a given choice of the parameter $\alpha \in \mathbb{C}\setminus\mathbb{R}$ (note the minus sign for the real part $\alpha_{1} := \mathrm{Re}\, \alpha$ in the definition).
Its Fourier transform then reads
\begin{align}\label{eq:FT_upsilon_alpha}
\int_{\mathbb{R}} e^{-iqx} \tilde{\upsilon}^{\psi,\alpha}(x,-y)\ dm(x)
    &= \left\langle \psi, e^{-i\langle (2/\alpha_{2})y, (1 +\alpha_{1})/2 \rangle \hat{p}} e^{-iq\hat{x}}e^{- i\langle (2/\alpha_{2})y, (1 - \alpha_{1})/2 \rangle \hat{p}} \psi \right\rangle \nonumber \\
    &= \left\langle \psi, e^{-i\langle y, (1 +\alpha_{1})/\alpha_{2} \rangle \hat{p}} e^{-iq\hat{x}}e^{- i\langle y, (1 - \alpha_{1})/\alpha_{2} \rangle \hat{p}} \psi \right\rangle,
\end{align}
where we have combined the second and the last equality of \eqref{eq:FT_W_w_alpha}, and applied the result \eqref{eq:Scaling_Fourier_translation}.

\paragraph{QJP of the `conditional' Meter State}
The next step is to compute the function $\tilde{\upsilon}^{\psi^{g}_{b},\alpha}$ for the `conditional' meter state $\psi^{g}_{b} := \psi^{g}_{B=b}$ introduced in \eqref{def:pure_state_given_b}.  What we find below is that, parallel to the findings in Section~\ref{sec:ps_II}, the resulting function $\tilde{\upsilon}^{\psi^{g}_{b},\alpha}$ is provided by the convolution of the initial profiles of both the meter and the target configurations. 
As above, we assume, for the ease of demonstration, that both the target and the conditioning observables $A$ and $B$ have spectra of finite cardinality, that $B$ is degenerate, and the probability of finding the outcomes of $B$ is non-vanishing $|\langle b, \phi \rangle|^{2} \neq 0$  for all its eigenvalues $b \in \sigma(B)$.  Since the essence of the demonstration is substantially the same as those provided in Section~\ref{sec:ps_II}, we proceed by sketching the proofs.

In computing the function of our interest, we first compute its Fourier transform to observe
\begin{align}
&\int_{\mathbb{R}} e^{-iqx} \tilde{\upsilon}^{\psi^{g}_{b},\alpha}(x,-y)\ dm(x) \nonumber \\
    &\quad = \left\langle \psi^{g}_{b}, e^{-i\langle y, (1 +\alpha_{1})/\alpha_{2} \rangle \hat{p}} e^{-iq\hat{x}}e^{- i\langle y, (1 - \alpha_{1})/\alpha_{2} \rangle \hat{p}} \psi^{g}_{b} \right\rangle \nonumber \\
    &\quad = \int_{\mathbb{R}^{2}} \left\langle e^{-igs_{1}\hat{p}} \psi, e^{-i\langle y, (1 +\alpha_{1})/\alpha_{2} \rangle \hat{p}} e^{-iq\hat{x}}e^{- i\langle y, (1 - \alpha_{1})/\alpha_{2} \rangle \hat{p}} e^{-igs_{2}\hat{p}}\psi \right\rangle d(\nu_{b}^{*} \otimes \nu_{b}^{\phantom{*}})(s_{1},s_{2}),
\end{align}
where we have used \eqref{eq:FT_upsilon_alpha} in the first equality, and where $\nu_{b}$ is the quasi-probability measure introduced in \eqref{def:aux_qpm_b_a}.
We next change variables of the above equality according to the linear transformation
\begin{align}
\left(
 \begin{array}{l}
a_{1} \\
a_{2}
 \end{array}
\right)
= T_{\alpha}
\left(
 \begin{array}{l}
 s_{1} \\
 s_{2}
 \end{array}
\right),
\qquad 
T_{\alpha} := \left(
    \begin{array}{cc}
     (1 - \alpha_{1})/2 & (1 + \alpha_{1})/2 \\
    -\alpha_{2}/2 & \alpha_{2}/2
    \end{array}
    \right)
\end{align}
by substituting
\begin{align}
 s_{1} = a_{1} - \frac{1+\alpha_{1}}{\alpha_{2}}a_{2}, \quad s_{2} = a_{1} + \frac{1-\alpha_{1}}{\alpha_{2}}a_{2},
\end{align}
to find
\begin{align}
&\int_{\mathbb{R}} e^{-iqx} \tilde{\upsilon}^{\psi^{g}_{b},\alpha}(x,-y)\ dm(x) \nonumber \\
    & \quad = \int_{\mathbb{R}^{2}} \left\langle \psi, e^{-i\langle y - ga_{2}, (1 +\alpha_{1})/\alpha_{2} \rangle \hat{p}} e^{-i(q\hat{x} - ga_{1}I)}e^{- i\langle y - ga_{2}, (1 - \alpha_{1})/\alpha_{2} \rangle \hat{p}} e^{-igs\hat{p}}\psi \right\rangle\ d\mu_{A}^{\phi,\alpha}(a|B=b) \nonumber \\
    & \quad = \int_{\mathbb{R}^{2}} \left( \int_{\mathbb{R}} e^{-iqx} \tilde{\upsilon}^{\psi,\alpha}(x - ga_{1},-(y - ga_{2}))\ dm(x) \right) d\mu_{A}^{\phi,\alpha}(a|B=b) \nonumber \\
    & \quad =  \int_{\mathbb{R}}  e^{-iqx} \left(\int_{\mathbb{R}^{2}} \tilde{\upsilon}^{\psi,\alpha}(x - ga_{1},-(y - ga_{2}))\  d\mu_{A}^{\phi,\alpha}(a|B=b) \right)dm(x),
\end{align}
where $\mu_{A}^{\phi,\alpha}(\Delta|B=b) := (\nu_{b}^{*} \otimes \nu_{b}^{\phantom{*}})(T_{\alpha}^{-1}\Delta)$, $\Delta \in \mathfrak{B}^{2}$ is the image measure, and we have combined \eqref{eq:FT_upsilon_alpha} with \eqref{eq:Fourier_translation} to obtain the second equality.
One thus concludes from the injectivity of the Fourier transformation that
\begin{equation}
\tilde{\upsilon}^{\psi^{g}_{b},\alpha}(x,y) = \int_{\mathbb{R}^{2}}  \tilde{\upsilon}^{\psi,\alpha}(x - ga_{1},y - ga_{2})\  d\mu_{A}^{\phi,\alpha}(a|B=b),
\end{equation}
as promised.

\paragraph{Recovery of the Target QJP}
As for the recovery of the target information $\Delta \mapsto \mu_{A}^{\phi,\alpha}(\Delta|B=b)$, $\Delta \in \mathfrak{B}(\mathbb{C})$, one may resort to the familiar techniques we have discussed so far in depth, namely, one may recover the full profile by either probing the strong or the weak region of the interaction parameter.  Once we obtained $\mu_{A}^{\phi,\alpha}(\,\cdot\,|B=b)$ for all $b \in \sigma(B)$, one may extend the domain of $b \in \sigma(B)$ to the whole real line $\mathbb{R}$ in a consistent manner, making it a transition quasi-probability kernel.  This allows us to construct the QJP $\mu_{A,B}^{\phi,\alpha}$ of the pair of $A$ and $B$ in a manner described in Proposition~\ref{prop:ctk_to_cm} that satisfies
\begin{equation}
\mu_{A,B}^{\phi,\alpha}(\Delta_{A} \times \Delta_{B}) = \int_{\Delta_{B}} \mu_{A}^{\phi,\alpha}(\Delta_{A} | B=b)\ d\mu_{B}^{\phi}(b), \quad \Delta_{A} \in \mathfrak{B}(\mathbb{C}),\ \Delta_{B} \in \mathfrak{B}^{1}.
\end{equation}
A close look on the proof of Proposition~\ref{prop:construction_of_qjp} leads one to conclude that the QJP obtained here
\begin{equation}
\mu_{A,B}^{\phi,\alpha} = \mu_{\mathrm{cnv}}^{\phi,\alpha},
\end{equation}
is in fact nothing but the member of the convolutive sub-family for the choice $\alpha \in \mathbb{C} \setminus \mathbb{R}$, and that $\mu_{A}^{\phi,\alpha}(\,\cdot\,|B=b)$ is the conditional quasi-probability distribution of $A$ given $B=b$.

\newpage
\section{Application: Interpretation of Aharonov's Weak Value}\label{sec:app}

As an application of the findings on the QJP distributions of quantum observables, we now focus on the geometric structure that the QJP distributions induce in the space of quantum observables.  Specifically, by drawing an analogy between the result of classical probability theory, we provide a geometric and statistical interpretation of Aharonov's weak value as `orthogonal projection' and `conditional average', respectively.

\subsection{Reference Materials}\label{sec:app_ref}

As usual, we start by preparing some necessary materials that become useful for our analysis.  The main objective of this subsection is to obtain a geometric understanding of conditional expectations in classical probability theory.

\subsubsection{$L^{2}$-Theory of Conditional Expectations}

\paragraph{$L^{p}$-spaces for finite Measures}
Let $\mu$ be a finite measure on a measurable space $(X,\mathfrak{A})$, {\it i.e.}, $\mu(X) < \infty$, and let $1 \leq p < q \leq \infty$. By defining $r > 0$ satisfying  $\frac{1}{r} = \frac{1}{p} - \frac{1}{q} $, a direct application of H{\"o}lder's inequality yields
\begin{equation}
\|f\|_{p} \leq \|f\|_{q} \cdot \|1\|_{r} = \|f\|_{q} \cdot |\mu(X)|^{1/r} < \infty
\end{equation}
for $f \in L^{q}(\mu)$.  The following Lemma is worth of special notice.
\begin{lemma}
Let $\mu$ be a finite measure on a measurable space $(X,\mathfrak{A})$. Then, for any $1 \leq p \leq q \leq \infty$, the relation
\begin{equation}
L^{q}(\mu) \subset L^{p}(\mu)
\end{equation}
holds.
\end{lemma}
\noindent
Specifically, for probability spaces, note that one has the evaluation $\|f\|_{p} \leq \|f\|_{q}$ for the choice of the parameters $1 \leq p \leq q \leq \infty$.

\paragraph{Conditional Expectations for square-integrable Functions}
Now, consider a probability space $(\mathbb{R}^{n}, \mathfrak{B}^{n}, \mu)$, and let $\mathfrak{A} \subset \mathfrak{B}^{n}$ be a sub-$\sigma$-algebra.  Since every square-integrable function $f \in L^{2}(\mu) \subset L^{1}(\mu)$ is integrable due to the above Lemma, its conditional expectation $\mathbb{E}[f|\mathfrak{A}] := d (f \odot \mu)|_{\mathfrak{A}} / d \mu|_{\mathfrak{A}} \in L^{1}(\mu|_\mathfrak{A})$ is well-defined, where $(f \odot \mu)|_{\mathfrak{A}}$ and $\mu|_\mathfrak{A}$ denotes the restriction of the respective (complex) measures on the sub-$\sigma$-algebra.  Now, observe that, for any square-integrable function $g \in L^{2}(\mu|_{\mathfrak{A}})$, the equality
\begin{align}\label{eq:cond_exp_orth_proj_adj}
\langle g, f \rangle 
    &:= \int_{\mathbb{R}^{n}} g^{*}(x) f(x)\ d\mu(x) \nonumber \\
    &= \int_{\mathbb{R}^{n}} g^{*}(x)\ d(f \odot \mu)(x) \nonumber \\
    &= \int_{\mathbb{R}^{n}} g^{*}(x) \frac{d (f \odot \mu)|_{\mathfrak{A}}}{d \mu|_{\mathfrak{A}}}(x)\ d \mu|_{\mathfrak{A}}(x) \nonumber \\
    &=: \langle g, \mathbb{E}[f|\mathfrak{A}] \rangle
\end{align}
holds by the definition of the Radon-Nikod{\'y}m derivative.  Specifically, note that this leads to the fact that the conditional expectation $\mathbb{E}[f|\mathfrak{A}] \in L^{2}(\mu|_\mathfrak{A})$ of a square-integrable function $f \in L^{2}(\mu)$ is again square-integrable.

\paragraph{Conditioning as Projection}
Another important observation to make from the above equality is that, the act of conditioning
\begin{equation}
\mathbb{E}[\,\cdot\,|\mathfrak{A}] : L^{2}(\mu) \to L^{2}(\mu|_{\mathfrak{A}}), \quad f \mapsto \mathbb{E}[f|\mathfrak{A}]
\end{equation}
that takes a $\mu$-square-integrable function to its conditional expectation, is an \emph{orthogonal projection}.  To see this, first observe that linearity $\mathbb{E}[af + bg| \mathfrak{A}] = a\mathbb{E}[f| \mathfrak{A}] + b\mathbb{E}[g| \mathfrak{A}]$, $f, g \in L^{2}(\mu)$, $a, b \in \mathbb{C}$ follows immediately by definition (naturally, equality is only valid $\mu|_{\mathfrak{A}}$-almost everywhere).  Now, since $L^{2}(\mu|_{\mathfrak{A}})$ is itself a complex Hilbert space, it is a topologically closed subspace of the larger complex Hilbert space $L^{2}(\mu)$.  By recalling that there is a one-to-one correspondence between closed subspaces and orthogonal projections in Hilbert spaces, let $P(\,\cdot\,|\mathfrak{A}) : L^{2}(\mu) \to L^{2}(\mu|_\mathfrak{A})$ denote the unique orthogonal projection associated with it.  By observing that
\begin{equation}
\langle g, f \rangle = \langle g, P(f |\mathfrak{A}) \rangle
\end{equation}
holds for all $g \in L^{2}(\mu|_\mathfrak{A})$ and $f \in L^{2}(\mu)$ by definition of orthogonal projections, one realises that the equality \eqref{eq:cond_exp_orth_proj_adj} combined with the non-degenerateness of inner products leads to
\begin{equation}\label{eq:cond_exp_orth_proj}
P(f |\mathfrak{A}) = \mathbb{E}[f|\mathfrak{A}], \quad f \in L^{2}(\mu).
\end{equation}
We summarise the results as follows.

\begin{proposition}(Orthogonal Projection and Conditional Expectation)
Let $\mu$ be a probability measure on $(\mathbb{R}^{n},\mathfrak{B}^{n})$, and let $\mathfrak{A} \subset \mathfrak{B}^{n}$ be a sub-$\sigma$-algebra.  Then, the unique orthogonal projection $P(\,\cdot\,|\mathfrak{A}) : L^{2}(\mu) \to L^{2}(\mu|_\mathfrak{A})$ associated with the subspace $L^{2}(\mu|_\mathfrak{A})$ is provided by the conditional expectation
\begin{equation}
P(f|\mathfrak{A}) = \mathbb{E}[f|\mathfrak{A}],
\end{equation}
where $f \in L^{2}(\mu)$.
\end{proposition}
\noindent
We here see how the geometric concept of orthogonality relates to the statistical concept of conditioning in $L^{2}$-spaces.

\subsubsection{Conditioning as Optimal Approximation}

The geometric property mentioned above leads to several important interpretation of conditional expectations.  One of the prominent characteristics of orthogonal projections is the validity of the \emph{Pythagorean identity}
\begin{equation}\label{eq:pythagorean_identity_classical}
\|f - g\|_{2}^{2} = \|f - \mathbb{E}[f|\mathfrak{A}]\|_{2}^{2} + \|\mathbb{E}[f|\mathfrak{A}] - g\|_{2}^{2},    \quad f \in L^{2}(\mu),\ g \in L^{2}(\mu|_{\mathfrak{A}}),
\end{equation}
where $\|\cdot\|_{2}$ denotes the standard $L^{2}$-norm introduced earlier in \eqref{def:Lp_norm}.  
An immediate consequence of the above Pythagorean identity is the following equality
\begin{equation}
\|f - \mathbb{E}[f|\mathfrak{A}]\|_{2} = \min_{g \in L^{2}(\mu|_{\mathfrak{A}})} \|f - g\|_{2},\quad f \in L^{2}(\mu),
\end{equation}
which states that the optimal $\mu|_{\mathfrak{A}}$-square-integrable function one can find in approximating a function $f$ is explicitly provided by the conditional expectation of $f$ given $\mathfrak{A}$, and the positive-definiteness of the $L^{2}$-norm shows that the optimum is unique $\mu|_{\mathfrak{A}}$-a.e.

\subsection{Statistical Interpretation of Geometric Structures}

Now that we have reviewed the geometric interpretation of conditional expectations in classical probability theory, we shall begin our main analysis.

\subsubsection{Preliminary Observation}

In classical probability theory, probability measures equip the space of square-integrable functions with a geometry, {\it i.e.,} an inner product defined by \eqref{def:L2-inner_product}, which we reiterate for the readers' convenience as
\begin{equation}\label{def:classical_geometry}
\langle g, f \rangle_{\mu} := \int g^{*}f\ d\mu
\end{equation}
(here, we have also explicitly written the probability measure $\mu$ under consideration for clarity).
In the context of classical physics in which observables are represented by functions, a probability measure defines quantities on a given pair of square-integrable classical observables interpreted as \emph{correlations} or \emph{covariances} between them.

\paragraph{`Correlations' in Quantum Theory}
In quantum mechanics, observables are represented by self-adjoint operators on Hilbert spaces, and the statistics of the system are in turn represented by vectors of Hilbert spaces.   In order to see how the two distinct frameworks of classical and quantum theory on correlations play together, first let $A$ and $B$ be a pair of simultaneously measurable bounded quantum observables on a Hilbert space $\mathcal{H}$, $|\psi\rangle \in \mathcal{H}$ a vector, and introduce
\begin{equation}\label{def:quantum_cor_sim}
\langle B, A \rangle_{\psi} := \frac{\langle B\psi, A\psi \rangle}{\|\psi\|^{2}}.
\end{equation}
One readily sees that, for the present case, the above geometry induced by the vector $|\psi\rangle$ is in accordance with the classical theory.  Indeed, the unique product spectral measure $E_{A,B}$ of $A$ and $B$ admits a unique representation of the observables given by
\begin{equation}
A = \int_{\mathbb{R}^{2}} A(a,b)\ dE_{A,B}(a,b), \quad B = \int_{\mathbb{R}^{2}} B(a,b)\ dE_{A,B}(a,b)
\end{equation}
with $A(a,b) = a$, $B(a,b) = b$, and moreover defines a joint-probability measure $\mu_{A,B}^{\psi}$ of the pair (see \eqref{def:prob_measrue_A_simul}) on the state $|\psi\rangle \in \mathcal{H}$.  It is then straightforward to see the validity of the equality
\begin{align}\label{eq:quantum_cor_sim_stat}
\langle B, A \rangle_{\psi}
    &= \int_{\mathbb{R}^{2}} B^{*}(a,b)A(a,b)\ d\mu_{A,B}^{\psi}(a,b) \nonumber \\
    &= \langle B, A \rangle_{\mu_{A,B}^{\psi}},
\end{align}
based on which one obtains a statistical interpretation of the geometry \eqref{def:quantum_cor_sim} as the correlation or covariance between the observables in the classical sense.

\paragraph{Non-commutative Case}
On the other hand, the problem is not so straightforward for the case where the pair $A$ and $B$ does not admit simultaneous measurability.  This is essentially to do with the lack of the unique product spectral measure of the pair.  As we have seen in the previous Section~\ref{sec:qp_qo}, the non-commutative analogues of product spectral measures are the quasi-joint-spectral distributions (QJSDs) defined as the inverse Fourier transforms of the hashed unitary groups \eqref{def:QJSD}.  The non-uniqueness of the QJSDs  for the non-commuting case generally leads to the non-uniqueness of the representation of operators and vectors by functions and quasi-probability distributions.  Specifically, given a choice of a QJSD $\Pi_{A,B}$ for a pair of generally non-commuting observables $A$ and $B$, one obtains a functional representation of operators as
\begin{equation}
A = \int_{\mathbb{R}^{2}} A(a,b)\ d\Pi_{A,B}(a,b), \quad B = \int_{\mathbb{R}^{2}} B(a,b)\ d\Pi_{A,B}(a,b)
\end{equation}
with $A(a,b) = a$, $B(a,b) = b$ (fortunately, as for this specific case, all representations coincide irrespective of the choice of the QJSD), and also a representation of quantum states $|\psi\rangle \in \mathcal{H}$ by QJP distributions $p_{A,B}^{\psi}(a,b)$ defined as in \eqref{def:QJP_explicit}.
Guided by a straightforward analogy, one realises that the quantity formally defined by
\begin{equation}\label{def:quasi_covariance_general}
\llangle B, A \rrangle_{\psi} := \int_{\mathbb{K}^{2}} B^{*}(a,b)A(a,b)\ p_{A,B}^{\psi}(a,b)\ dm_{2}(a,b)
\end{equation}
defines various different geometries between quantum observables dependent on the choice of the QJSDs.  This implies that, parallel to the classical case, QJSDs serves as a bridge that offers a `statistical' interpretation of the (non-unique) geometric structures that can be introduced in the space of quantum observables.

\subsubsection{Geometry induced by a specific Sub-family of QJSD}

The general treatment involving the entire class of QJSDs makes extensive use of the theory of generalised functions and its operator valued analogue (operator valued distributions), which may be far beyond the scope of this paper.  In this paper, mainly in order to confine our argument in the theory of complex measures and its operator valued analogue, we concentrate on the specific choice of the QJSD of a pair of quantum observables, namely, to the additive sub-family (introduced in Section~\ref{sec:additive_subfamily}) for the choice $-1 \leq \alpha \leq 1$ of the complex parameter, hopefully without essential loss of generality.  We shall moreover confine ourselves to bounded operators for simplicity, but the general treatment including unbounded operators is also possible without any essential alteration.

\paragraph{Sesquilinear Forms}
We are interested in introducing geometries in the space $L(\mathcal{H})$ of all bounded linear operators on a Hilbert space $\mathcal{H}$ given a fixed state $|\psi\rangle \in \mathcal{H}$. 
To this end, let $\alpha \in \mathbb{C}$ be a complex number, and define
\begin{equation}\label{def:alpha_sesquilinear_form}
\llangle Y, X \rrangle_{\psi,\alpha} :=  \frac{1 + \alpha}{2} \cdot \frac{\langle Y \psi, X \psi \rangle}{\|\psi\|^{2}} + \frac{1 - \alpha}{2} \cdot \frac{\langle X^{*} \psi, Y^{*} \psi \rangle}{\|\psi\|^{2}}, \quad X, Y \in L(H).
\end{equation}
As described earlier, this is just one possible straightforward way to extend the geometry \eqref{def:quantum_cor_sim} so that it can be defined even for non-commuting observables.  One  readily sees that this satisfies
\begin{enumerate}
\item $\llangle Y  + Y^{\prime}, X + X^{\prime} \rrangle_{\psi,\alpha} = \llangle Y, X \rrangle_{\psi,\alpha} + \llangle Y, X^{\prime} \rrangle_{\psi,\alpha} + \llangle Y^{\prime}, X \rrangle_{\psi,\alpha} + \llangle Y^{\prime}, X^{\prime} \rrangle_{\psi,\alpha}$,
\item $\llangle bY, aX \rrangle_{\psi,\alpha} = \overline{b}a\llangle Y, X \rrangle_{\psi,\alpha}$
\end{enumerate}
for any $X, X^{\prime}, Y, Y^{\prime} \in L(\mathcal{H})$ and $a, b \in \mathbb{C}$, hence it defines a \emph{sesquilinear form} on $L(\mathcal{H})$. By definition, one has
\begin{equation}
\llangle Y^{*}, X^{*} \rrangle_{\psi,\alpha} = \llangle X, Y \rrangle_{\psi,-\alpha}.
\end{equation}
By observing moreover that
\begin{equation}
\llangle Y, X \rrangle_{\psi,\alpha}^{*} = \llangle X, Y \rrangle_{\psi,\alpha^{*}}, \quad X, Y \in L(\mathcal{H})
\end{equation}
holds, the sesquilinear form \eqref{def:alpha_sesquilinear_form} is \emph{symmetric} (Hermitian) if the given parameter $\alpha \in \mathbb{R}$ is real. If one takes $X =Y$, this reads
\begin{equation}\label{def:eval}
\llangle X, X \rrangle_{\psi, \alpha} = \frac{1 + \alpha}{2} \cdot \frac{\|X\psi\|^{2}}{\|\psi\|^{2}} + \frac{1 - \alpha}{2} \cdot  \frac{\|X^{*}\psi\|^{2}}{\|\psi\|^{2}}, \quad X \in L(\mathcal{H}).
\end{equation}
Specifically for the choice $-1 \leq \alpha \leq 1$ of the parameter, note that the above quantity happens to be always \emph{positive}, hence becomes a \emph{semi-norm}, for which we introduce the notation
\begin{equation}
\|X\|_{\psi,\alpha} := \llangle X, X \rrangle_{\psi, \alpha}^{\frac{1}{2}}, \quad X \in L(\mathcal{H}),\ -1 \leq \alpha \leq 1.
\end{equation}
Note that
\begin{equation}
\|X\|_{\psi,\alpha}^{2} \leq \frac{1 + \alpha}{2} \cdot \|X\|^{2} + \frac{1 - \alpha}{2} \cdot \|X^{*}\|^{2} = \|X\|^{2},
\end{equation}
where $\|X\|$ denotes the usual operator norm of $X \in L(\mathcal{H})$, which shows that the semi-norm $\|\cdot\|_{\psi,\alpha}$ induces a topology on $L(\mathcal{H})$ coarser than that induced by the usual operator norm $\|\cdot\|$.

\paragraph{Statistical Interpretation}

In classical theory, the natural geometry \eqref{def:classical_geometry} introduced on the space of observables admits statistical interpretation as correlations by means of probability measures.  Parallel to this, we next intend to provide a statistical representation of the sesquilinear form \eqref{def:alpha_sesquilinear_form} for the quantum case, and this is achieved by means of QJP distributions.  As mentioned earlier, we let $\mu_{\mathrm{add}}^{\psi,\alpha}$ denote a quasi-probability measure satisfying
\begin{equation}\label{def:sec_app_additive}
(\mathscr{F}\mu_{\mathrm{add}}^{\psi,\alpha})(s,t) =  \frac{1 + \alpha}{2} \cdot \frac{\langle \psi, e^{-itB}e^{-isA}\psi\rangle}{\|\psi\|^{2}} + \frac{1 - \alpha}{2} \cdot \frac{\langle \psi, e^{-isA}e^{-itB}\psi\rangle}{\|\psi\|^{2}}, \quad \alpha \in \mathbb{C},
\end{equation}
which are namely members of the additive complex-parametrised sub-family of QJP distributions of $A$ and $B$ introduced in Section~\ref{sec:additive_subfamily}. 
One readily sees from a direct application of Lemma~\ref{lem:FT_CM_diff_exp} that the integration of polynomial functions reduces to
\begin{align}
\int_{\mathbb{R}^{2}} b^{m}a^{n}\ d\mu_{\mathrm{add}}^{\psi,\alpha}(a,b)
    &= \left. (i\partial_{t})^{m} (i\partial_{s})^{n} (\mathscr{F}\mu_{\mathrm{add}}^{\psi,\alpha})(s,t) \right|_{(s,t) = (0,0)} \nonumber \\
    &= \llangle B^{m},A^{n} \rrangle_{\psi,\alpha}.
\end{align}
One may also readily obtain a generalisation of this observation to continuous functions.
Indeed, according to the Stone-Weierstra{\ss} approximation theorem, since continuous functions $f, g$ defined on a compact space admit uniform approximations by polynomial functions as
\begin{align}
f(a) = \sum_{n=0}^{\infty} f_{n}a^{n}, \quad g(b) = \sum_{n=0}^{\infty} g_{n}b^{n}, 
\end{align}
one has
\begin{align}
\int_{\mathbb{R}^{2}} g^{*}(b)f(a)\ d\mu_{\mathrm{add}}^{\psi,\alpha}(a,b)
    &= \sum_{n,m=0}^{\infty} g_{m}^{*}f_{n}\int_{\mathbb{R}^{2}} b^{m}a^{n}\ d\mu_{\mathrm{add}}^{\psi,\alpha}(a,b) \nonumber \\
    &= \sum_{n,m=0}^{\infty} g_{m}^{*}f_{n} \llangle B^{m},A^{n} \rrangle_{\psi,\alpha} \nonumber \\
    &= \left\llangle \sum_{m}^{\infty}g_{m}B^{m}, \sum_{n}^{\infty}f_{n}A^{n} \right\rrangle_{\psi,\alpha} \nonumber \\
    &= \llangle g(B),f(A) \rrangle_{\psi,\alpha}.
\end{align}
This observation can be summarised as:
\begin{lemma}[Statistical Representation of Sesquilinear Forms]\label{lem:stat_rep_sesq_form}
Let $\mu_{\mathrm{add}}^{\psi,\alpha}$ be a member of the additive complex-parametrised sub-family of QJP distributions of the pair of observables $A, B \in L(\mathcal{H})$ defined as in \eqref{def:sec_app_additive}. Then, for any continuous functions $f$ and $g$ defined on the real line $\mathbb{R}$, the equality
\begin{equation}
\llangle g(B), f(A) \rrangle_{\psi,\alpha} = \int_{\mathbb{R}^{2}} g^{*}(b)f(a)\ d\mu_{\mathrm{add}}^{\psi,\alpha}(a,b)
\end{equation}
holds.
\end{lemma}
\noindent
We have thus obtained a convenient representation of the sesquilinear form by integration with respect to QJP distributions, which offers a `statistical' interpretation to the geometry as `correlations' of a pair of generally non-commuting quantum observables.

\paragraph{Topic: Quasi-covariances (Quantum Covariances)}

As a natural extension to the classical notion of covariances, we may introduce the term `quasi-covariance' (or `quantum-covariance') of a pair of quantum observables under a given QJP distribution $p_{A,B}^{\psi}(a,b)$ for the quantity formally defined by
\begin{align}\label{def:quasi_cov}
    \int_{\mathbb{K}^{2}} (a - \mathbb{E}[A;\psi])(b - \mathbb{E}[B;\psi]) p_{A,B}^{\psi}(a,b)\ dm_{2}(a,b),
\end{align}
whenever the integration exists.
Specifically, from an immediate application of the above Lemma, one may readily compute the quantum-covariances with respect to the additive subgroup $\mu_{\mathrm{add}}^{\psi,\alpha}(a,b)$ of the QJP distributions as
\begin{align}\label{eq:quantum_cov_alpha}
\mathbb{CV}^{\alpha}[A,B;\psi]
    &:= \int_{\mathbb{R}^{2}} (a - \mathbb{E}[A;\psi])(b - \mathbb{E}[B;\psi])\ d\mu_{\mathrm{add}}^{\psi,\alpha}(a,b) \nonumber \\
    &= \llangle A - \mathbb{E}[A;\psi], B - \mathbb{E}[B;\psi] \rrangle_{\psi,\alpha} \nonumber \\
    &= \mathbb{CV}_{\mathrm{S}}[A,B;\psi] + \alpha  i \,\mathbb{CV}_{\mathrm{A}}[A,B;\psi],
\end{align}
where $\mathbb{CV}_{\mathrm{S}}[A,B;\psi]$ and $\mathbb{CV}_{\mathrm{A}}[A,B;\psi]$ are respectively the symmetric and anti-symmetric quantum covariances introduced in \eqref{def:q_covariance}.

\subsubsection{The Hilbert Space of Bounded Operators given a fixed State}

In classical theory, observables were described by functions, whereas in quantum theory, observables become self-adjoint operators.  In order to conduct an analogue of the $L^{2}$-theory for quantum observables, it reveals for our purpose that it is convenient not just to deal with self-adjoint operators, but rather to consider the collection $L(\mathcal{H})$ of all bounded linear operators defined on the Hilbert space $\mathcal{H}$.

\paragraph{Identification}

In classical theory, recall that we made identification of observables that cannot be distinguished in view of the probability measure $\mu$ by introducing the equivalence relation $f \sim g \Leftrightarrow f = g \text{ $\mu$-a.e.}$, and slimmed down the space of functions into quotient spaces (see Section~\ref{sec:usp_I_MI}).
We intend to follow the same path for the quantum case, and
to this, we introduce the subspace
\begin{equation}
Z_{\psi}(H) := \{ X \in L(H) : X|\psi\rangle = X^{*}|\psi\rangle = 0 \}
\end{equation}
and define the $\mathbb{C}$-linear quotient space
\begin{equation}
L_{\psi}(\mathcal{H}) := L(H) / Z_{\psi}(H)
\end{equation}
by identifying those operators for which the action of both themselves and their adjoints are indistinguishable on the state $|\psi\rangle$.  In other words, this is to say that we identify two operators $X, Y \in L(\mathcal{H})$ by the equivalence relation
\begin{equation}
X \sim Y \Leftrightarrow
    X|\psi\rangle = Y|\psi\rangle \text{ and } X^{*}|\psi\rangle = Y^{*}|\psi\rangle.
\end{equation}
One readily sees that the sesquilinear form \eqref{def:alpha_sesquilinear_form} passes to the quotient, and we thus obtain a sesquilinear form
\begin{equation}\label{def:alpha_sesquilinear_form_quotient}
\llangle [Y]_{\psi}, [X]_{\psi} \rrangle_{\psi,\alpha} := \llangle Y, X \rrangle_{\psi,\alpha}, \quad [X]_{\psi}, [Y]_{\psi} \in L_{\psi}(\mathcal{H})
\end{equation}
on the quotient space $L_{\psi}(\mathcal{H})$. Whenever there is no risk of confusion, we shall mostly denote equivalence classes by their representatives for simplicity of notation.
Note also that the \emph{involution} $\ast : L(\mathcal{H}) \to L(\mathcal{H}),\  X \mapsto X^{*}$ that takes a bounded linear operator to its adjoint is also well-defined on the quotient space.

\paragraph{Hilbert Space of Operators}

We have already seen that the original sesquilinear form \eqref{def:alpha_sesquilinear_form} becomes positive and symmetric for the choice $-1 \leq \alpha \leq 1$ of the parameter.  Based on the identification above, the sesquilinear form \eqref{def:alpha_sesquilinear_form_quotient} on the quotient space $L_{\psi}(\mathcal{H})$ becomes \emph{positive definite}, which is to say that
\begin{enumerate}
\item $\llangle X, X \rrangle_{\psi,\alpha} \geq 0, \quad X \in L_{\psi}(\mathcal{H})$,
\item $\llangle X, X \rrangle_{\psi,\alpha} = 0 \quad  \Leftrightarrow \quad  X = 0$
\end{enumerate}
for the choice $-1 \leq \alpha \leq 1$. This makes \eqref{def:alpha_sesquilinear_form_quotient} an \emph{inner product} on $L_{\psi}(\mathcal{H})$ for $-1 \leq \alpha \leq 1$, allowing us to define the norm
\begin{equation}\label{def:norm_quotient}
\|X\|_{\psi,\alpha} := \llangle X, X \rrangle_{\psi, \alpha}^{\frac{1}{2}}, \quad X \in L_{\psi}(\mathcal{H}),\ -1 \leq \alpha \leq 1.
\end{equation}
One moreover proves by rudimentary technique that the space is in fact complete with respect to the norm. We thus have the following result.
\begin{proposition}[Hilbert Space of Operators]
For a fixed $|\psi\rangle \in \mathcal{H}$ and the choice $-1 \leq \alpha \leq 1$ of the parameter, the ordered pair $(L_{\psi}(\mathcal{H}), \llangle\, \cdot\, , \, \cdot\, \rrangle_{\psi, \alpha})$ defines a complex Hilbert Space.
\end{proposition}
\noindent
This convenient property greatly facilitates our further argument.  Hence, in what follows, we will be treating only those geometries associated with the specific choice $-1 \leq \alpha \leq 1$ of the complex parameter.

\paragraph{Sub-algebra generated by an Observable}

We next introduce an important subspace of $L(\mathcal{H})$.
Given a bounded self-adjoint operator $A \in L(\mathcal{H})$, we prepare a special symbol
\begin{equation}\label{def:subspace_generated_by_normal_operators}
\mathfrak{E}(A) := \left\{ f(A) : f \text{ is a continuous function on $\sigma(A)$} \right\}
\end{equation}
for the $\mathbb{C}$-linear subspace of $L(\mathcal{H})$ consisting of all operators defined by means of the functional calculus \eqref{def:func_calc}.  By definition, one proves that
\begin{equation}
\left\| f(A) \right\| = \| f \|_{\infty},
\end{equation}
where the l.~h.~s. is the operator norm of $f(A)$, and the r.~h.~s. is the supremum norm of the continuous function $f$ (note that the spectrum $\sigma(A)$ of a bounded self-adjoint operator $A$ is compact, hence any continuous function defined on the spectrum is necessarily bounded). Moreover, it is easy to see that $f(A)^{*} = f^{*}(A)$ holds, where the l.~h.~s. denotes the adjoint of the linear operator $f(A)$, whereas the r.~h.~s. denotes the operator induced by the complex conjugate $f^{*}$ of the original function $f$.  This implies that all the operators of the form $f(A)$ are normal%
\footnote{Recall that a bounded operator $N \in L(\mathcal{H})$ is \emph{normal} if and only if $\|N\psi\| = \|N^{*}\psi\|$ holds for all $|\psi\rangle \in \mathcal{H}$.},
and that the space $\mathfrak{E}(A)$ is closed under the operation of taking the adjoint. Moreover, one sees that any two operators $f(A), g(A) \in \mathfrak{E}(A)$ commute with each other $f(A)g(A) = (fg)(A) = g(A)f(A)$, and that the space $\mathfrak{E}(A)$ thus makes itself into a commutative $C^{*}$-algebra.  We call the space $\mathfrak{E}(A)$ the sub-algebra \emph{generated by $A$}.

\paragraph{Identification}

An immediate observation one makes is that the sesquilinear form \eqref{def:alpha_sesquilinear_form} is independent of the choice of the parameter $\alpha \in \mathbb{C}$ on the space $\mathfrak{E}(A)$.  Indeed, for any choice of a pair of continuous functions $f$, $g$ defined on the spectrum $\sigma(A)$, the equality
\begin{align}\label{def:sesquilinear_form_normal}
\llangle g(A), f(A) \rrangle_{\psi, \alpha}
    &= \frac{\langle g(A)\psi, f(A)\psi \rangle}{\|\psi\|^{2}} \nonumber \\
    &= \int_{\mathbb{R}} g^{*}(a)f(a)\ d\mu_{A}^{\psi}(a)
    = \langle g, f \rangle_{\mu_{A}^{\psi}}
\end{align}
holds, where the right-most hand side denotes the standard inner product introduced on the space of square-integrable complex functions.
Following the same line of arguments we have made in the previous discussion of this section, we next intend to identify those operators that are not distinguishable in view of a given state $|\psi\rangle \in \mathcal{H}$.  To this, we introduce the subspace
\begin{equation}
Z_{\psi}(A) :=
    \{ N \in \mathfrak{E}(A) : N|\psi\rangle = N^{*}|\psi\rangle = 0 \},
\end{equation}
and define the space
\begin{equation}\label{def:subspace_generated_by_normal_operators_quotient}
\mathfrak{E}_{\psi}(A) := \overline{\mathfrak{E}(A) / Z_{\psi}(A)}
\end{equation}
by identifying those normal operators for which the action of both themselves and their adjoints on the state $|\psi\rangle$ are indistinguishable.  Here, the overline on the quotient space $\mathfrak{E}(A) / Z_{\psi}(A) \subset L_{\psi}(\mathcal{H})$ denotes its topological closure with respect to the topology on the superset $L_{\psi}(\mathcal{H})$ induced by the norm $\|\cdot\|_{\psi,\alpha}$ \eqref{def:norm_quotient}.  Note here that, as a set, the closure $\mathfrak{E}_{\psi}(A)$ is independent of the choice of the parameter $-1 \leq \alpha \leq 1$, since all the norms $\|\cdot\|_{\psi,\alpha}$ coincide on the subspace $\mathfrak{E}(A) / Z_{\psi}(A)$.  Moreover, one may readily check that all the inner products $\llangle \,\cdot\, , \,\cdot\, \rrangle_{\psi,\alpha}$ also coincide for the pair of elements of the closure $\mathfrak{E}_{\psi}(A)$ for any choice of the parameter $-1 \leq \alpha \leq 1$.

 Now, since by definition the space $\mathfrak{E}_{\psi}(A)$ is a closed subspace of the complex Hilbert space,  it is itself a complex Hilbert space.  By denoting the restriction of the inner product as $\llangle \,\cdot\,, \,\cdot\, \rrangle_{\psi} := \llangle \,\cdot\,, \,\cdot\, \rrangle_{\psi,\alpha}|_{\mathfrak{E}_{\psi}(A)}$, which does not depend on the choice of the parameter as we have mentioned above, we have:
\begin{lemma}
For a fixed $|\psi\rangle \in \mathcal{H}$, the ordered pair $(\mathfrak{E}_{\psi}(A), \llangle\, \cdot\, , \, \cdot\, \rrangle_{\psi})$ defines a complex Hilbert space.
\end{lemma}
\noindent
The next Lemma is of special interest for our purpose.
\begin{lemma}
Let $A \in L(\mathcal{H})$ be self-adjoint, $|\psi\rangle \in \mathcal{H}$, and let $\mathfrak{E}_{\psi}(A)$ be defined as in \eqref{def:subspace_generated_by_normal_operators_quotient}.  Then, there exists a unique unitary operator $\Phi : L^{2}(\mu_{A}^{\psi}) \mapsto \mathfrak{E}_{\psi}(A)$ such that
\begin{equation}
\Phi(f) = [f(A)]_{\psi}
\end{equation}
holds for every continuous function $f$ on $\sigma(A)$, where the r.~h.~s. denotes the equivalent class of $f(A)$, which in turn is a bounded linear operator defined by means of the functional calculus \eqref{def:func_calc}.
\end{lemma}
\begin{proof}
We will construct the map $\Phi$ by continuous linear extension.  To this, first recall that, since $\sigma(A)$ is compact, the space 
$C(\sigma(A))$ of all continuous functions on $\sigma(A)$ is dense in $L^{2}(\mu_{A}^{\psi})$.  Since the map
\begin{equation}
\tilde{\Phi} : L^{2}(\mu_{A}^{\psi}) \supset C(\sigma(A)) \to \mathfrak{E}_{\psi}(A), \quad f \mapsto f(A),
\end{equation}
is an isometry $\|\tilde{\Phi}(f)\|_{\psi} := \| f(A) \psi \|= \|f\|_{2}$ from a dense subspace of a normed space to a Banach space, there exists a unique isometric extension $\Phi : L^{2}(\mu_{A}^{\psi}) \to \mathfrak{E}_{\psi}(A)$.  By construction, one may also prove the surjectivity of $\Phi$, hence $\Phi$ is unitary.
\end{proof}
\noindent
This is to say that the Hilbert spaces $\mathfrak{E}_{\psi}(A) \cong L^{2}(\mu_{A}^{\psi})$ are unitarily isomorphic, and that $\Phi$ gives an embedding of the space of square-integrable functions $L^{2}(\mu_{A}^{\psi})$ into the space of bounded operators on a Hilbert space.  For simplicity of notation, we occasionally denote the image of a square-integrable function $f \in L^{2}(\mu_{A}^{\psi})$ by $f(A) := \Phi(f)$.  A word of caution is to be made here for the possible confusion for the notation used.  Here, the notation $f(A) := \Phi(f)$ is meant to denote (a representative of) the equivalence class of bounded linear operators, whereas the notation $f(A)$ is usually used to represent (generally unbounded) linear operator defined by means of the functional calculus \eqref{def:func_calc}.  The relation between the two different notations can be understood in the following way.  For $f \in L^{2}(\mu_{A}^{\psi})$, let $T \in \Phi(f)$ be a representative of the equivalence class of bounded operators, and let $f(A)$ be a (generally unbounded) operator defined by means of the functional calculus \eqref{def:func_calc}, and note that $|\psi\rangle \in \mathrm{dom}(f(A))$ by definition.  Then, we have $T|\psi\rangle = f(A)|\psi\rangle$.

\subsection{Interpretation of Conditional Quasi-expectations}

Now that we have sufficiently prepared our tools, the most important among which is the embedding
\begin{equation}\label{eq:embedding_l2}
\Phi: L^{2}(\mu_{A}^{\psi}) \cong \mathfrak{E}_{\psi}(A) \subset L_{\psi}(\mathcal{H})
\end{equation}
of the $L^{2}$-space of functions into that of bounded linear operators on a Hilbert space, we next focus on orthogonal projections and `conditioning' with respect to QJP distributions.

\subsubsection{Geometric Interpretation of Conditional Quasi-expectations}

Recall that, with each closed subspace of a Hilbert space, a unique \emph{orthogonal projection} is associated.  In what follows, we are interested in the orthogonal projection of an observable $A$ onto the subspace $\mathfrak{E}_{\psi}(B)$ generated by another observable $B$, and see that this provides a geometric interpretation of the conditional quasi-expectation $\mathbb{E}^{\alpha}[A|B; \psi]$ introduced earlier in \eqref{def:cond_quasi-exp_alpha}.

To this, let $B \in L(\mathcal{H})$ be self-adjoint, $|\psi\rangle \in \mathcal{H}$, $-1 \leq \alpha \leq 1$, and let $\mathfrak{E}_{\psi}(B)$ be the space generated by $B$, which is a closed subspace of the Hilbert spaces $(L_{\psi}(\mathcal{H}), \llangle \,\cdot\,, \,\cdot\,\rrangle_{\psi,\alpha}$).  As a closed subspace of a Hilbert space, there exists a unique orthogonal projection
\begin{equation}\label{def:orth_proj}
P_{\alpha}(\ \cdot \ |B ; \psi): L_{\psi}(\mathcal{H}) \to \mathfrak{E}_{\psi}(B),\ X \mapsto P_{\alpha}(X|B;\psi)
\end{equation}
associated to $\mathfrak{E}_{\psi}(B)$.  Recalling the relation between orthogonal projections and conditional expectations in classical probability theory (see Section~\ref{sec:app_ref}), it is natural to conjecture that an analogous relation holds for the quantum case.  To this, let $A \in L(\mathcal{H})$ be self-adjoint, and consider the projection $P_{\alpha}(A |B ; \psi)$ of $A$ onto $\mathfrak{E}_{\psi}(B)$.  We have seen in Section~\ref{sec:ps_I} that the conditional quasi-expectations $\mathbb{E}^{\alpha}[A|B; \psi] \in L^{1}(\mu_{B}^{\psi})$ introduced in \eqref{def:cond_quasi-exp_alpha} serve as possible candidates of quantum analogues of conditional expectations that can even be defined for non-commuting pair of quantum observables.  Since, the observable $A$ we consider here is bounded, one may prove that the conditional quasi-expectation $\mathbb{E}^{\alpha}[A|B; \psi] \in L^{2}(\mu_{B}^{\psi})$ is in fact square-integrable.  By letting $\mathbb{E}^{\alpha}[A|B; \psi]$ denote both the square-integrable function and its image by the unitary map $\Phi : L^{2}(\mu_{B}^{\psi}) \to \mathfrak{E}_{\psi}(B)$, it is natural to conjecture the validity of the equality $P_{\alpha}(A|B;\psi) = \mathbb{E}^{\alpha}[A|B; \psi]$, where the l.~h.~s. is the orthogonal projection of $A$ onto the space $\mathfrak{E}_{\psi}(B)$ with respect to the (parameter-dependent) inner product $\llangle \,\cdot\,, \,\cdot\,\rrangle_{\psi,\alpha}$, whereas the r.~h.~s. denotes the image of the (parameter-dependent) conditional quasi-expectation of $A$ given $B$ by the unitary map $\Phi$.

\begin{proposition}[Orthogonal Projection and Conditional Quasi-Expectation]\label{prop:orth_proj_cond_qe}
Let $A, B \in L(\mathcal{H})$ be self-adjoint, $|\psi\rangle \in \mathcal{H}$ and $-1 \leq \alpha \leq 1$.  Then, the conditional quasi-expectation $\mathbb{E}^{\alpha}[A|B; \psi]$ introduced in \eqref{def:cond_quasi-exp_alpha} is $\mu_{B}^{\psi}$-square-integrable, which could thus be identified with the equivalence class of bounded operators
\begin{equation}
\mathbb{E}^{\alpha}[A|B; \psi] := \Phi(\mathbb{E}^{\alpha}[A|B; \psi])
\end{equation}
by means of the embedding $\Phi : L^{2}(\mu_{B}^{\psi}) \to \mathfrak{E}_{\psi}(B)$ defined in \eqref{eq:embedding_l2}.  Then, the orthogonal projection of $A$ onto the subspace $\mathfrak{E}_{\psi}(B)$ generated by $B$, defined in \eqref{def:orth_proj}, reads
\begin{equation}
P_{\alpha}(A|B;\psi) = \mathbb{E}^{\alpha}[A|B; \psi],
\end{equation}
which is to say that orthogonal projections are equivalent to conditional quasi-expectations.
\end{proposition}
\begin{proof}
Let $f \in L^{2}(\mu_{B}^{\psi})$, and let $f(B) := \Phi(f)$ denote the embedding of $f$ into the space $\mathfrak{E}_{\psi}(B)$.  By definition of orthogonal projections, one readily finds
\begin{equation}\label{eq:orth_proj_proof_1}
\llangle f(B), A \rrangle_{\psi,\alpha} = \llangle f(B), P_{\alpha}(A|B;\psi) \rrangle_{\psi}.
\end{equation}
On the other hand, observe that
\begin{align}
\langle f, \mathbb{E}^{\alpha}[A|B; \psi] \rangle_{\mu_{B}^{\psi}}
    &:= \int_{\mathbb{R}} f^{*}(b) \cdot \mathbb{E}^{\alpha}[A|B = b; \psi]\ d\mu_{B}^{\psi}(b) \nonumber \\
    &= \frac{1+\alpha}{2} \int_{\mathbb{R}} f^{*}(b)\ d\frac{\langle \psi,E_{B}(b)A\psi\rangle}{\|\psi\|^{2}} + \frac{1-\alpha}{2} \int_{\mathbb{R}} f^{*}(b)\ d\frac{\langle \psi, AE_{B}(b)\psi\rangle}{\|\psi\|^{2}} \nonumber \\
    &= \frac{1+\alpha}{2} \cdot \frac{\langle f(B)\psi, A\psi\rangle}{\|\psi\|^{2}} + \frac{1-\alpha}{2} \cdot \frac{\langle A^{*}\psi, f(B)^{*}\psi\rangle}{\|\psi\|^{2}} \nonumber \\
    &=: \llangle f(B), A \rrangle_{\psi,\alpha}
\end{align}
Combining this with the unitarity of the embedding $\Phi : L^{2}(\mu_{B}^{\psi}) \to \mathfrak{E}_{\psi}(B)$, one thus has
\begin{equation}\label{eq:orth_proj_proof_2}
\llangle f(B), A \rrangle_{\psi,\alpha} = \llangle f(B), \mathbb{E}^{\alpha}[A|B; \psi] \rrangle_{\psi}.
\end{equation}
The positive-definiteness of the inner product applied to the two results \eqref{eq:orth_proj_proof_1} and \eqref{eq:orth_proj_proof_2} proves our statement.
\end{proof}
\noindent
Just as we have seen for the classical case, this result provides a geometric interpretation of conditional quasi-expectations as orthogonal projections.  As a corollary, one has a geometric interpretation of Aharonov's weak value.
\begin{corollary}[Geometric Interpretation of Aharonov's Weak Value]\label{prop:orth_proj_wv}
Under the same conditions, let
\begin{equation}
A_{w}(B) := \Phi(A_{w})
\end{equation}
denote the embedding of the Aharonov's weak value $A_{w} : = \mathbb{E}^{1}[A|B; \psi]$ introduced in \eqref{def:weak_value}.  Then, the weak value 
\begin{equation}
A_{w}(B) = P_{1}(A|B;\psi) = \mathbb{E}^{1}[A|B; \psi]
\end{equation}
could be interpreted as the orthogonal projection of $A$ onto the subspace $\mathfrak{E}_{\psi}(B)$ generated by $B$.
\end{corollary}

\paragraph{Topic: Weak Value as Optimal Approximation}

As a direct consequence of Proposition~\ref{prop:orth_proj_wv} (specifically, Corollary~\ref{prop:orth_proj_wv}), we note an interesting result regarding conditional quasi-expectations (specifically, the weak value) and optimal approximation.  As orthogonal projections, observe that conditional quasi-expectations furnish the optimal proxy function for $A$ minimising the distance
\begin{equation}
\|A - \mathbb{E}^{\alpha}[A|B; \psi] \|_{\psi,\alpha} = \min_{f} \|A - f(B) \|_{\psi,\alpha}
\end{equation}
from an observable $A$ to the space of normal observables $\mathfrak{E}_{\psi}(B)$ generated by another observable $B$.  An equivalent expression to this is the equality
\begin{equation}
\|A - f(B) \|_{\psi,\alpha} = \|A - \mathbb{E}^{\alpha}[A|B; \psi] \|_{\psi,\alpha} + \|\mathbb{E}^{\alpha}[A|B; \psi] - f(B) \|_{\psi,\alpha},
\end{equation}
which is nothing but the `Pythagorean identity' valid for orthogonal projections%
\footnote{Specifically, observing that $\mathbb{E}^{\alpha}[A|B; \psi] = \mathrm{Re}A_{w}(B) + i \alpha \mathrm{Im}A_{w}(B)$, we have $\mathbb{E}^{\alpha}[A|B; \psi] = \mathrm{Re}A_{w}(B)$ for the choice $\alpha = 0$.  This gives
\begin{align}
\| \left( A - \mathrm{Re}A_{w}(B) \right) \psi \|
    &= \|A - \mathrm{Re}A_{w}(B) \|_{\psi,0} \nonumber \\
    &\leq \|A - f(B) \|_{\psi,0} \nonumber \\
    &= \| \left( A - f(B) \right) \psi \|, \quad f \in L^{2}(\mu_{B}^{\psi}),\ f \text{ is real},
\end{align}
as a special case.  This specific form is known by \cite{Hall_2001,Johansen_2004}, although proven from a different perspective than directly utilising the geometric observation made in this paper.}
in Hilbert spaces.
The interpretation of conditional quasi-expectations as orthogonal projections provide the core geometric observations why the weak value appears as the optimal choice for the proxy functions in the novel uncertainty relations for approximation/estimation \cite{Lee_2016}.
\begin{figure}
\includegraphics[width=160mm]{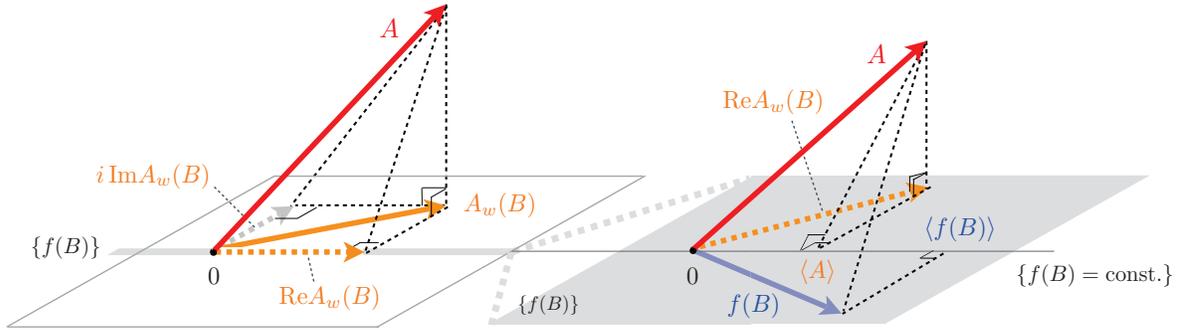}
\caption{Geometric relations among the operators involved for the choice $\alpha = 1$. The left illustrates how the operator $A$ is projected onto the subspace of normal operators $\mathfrak{E}_{\psi}(B)$ generated by $B$, with the center line representing the space of self-adjoint operators $\{ f(B)\}$.   The right elaborates the projection onto the space $\{ f(B)\}$, where
now the center line represents the space of constant functions $\{ f(B) = \hbox{const.} \}$ (more precisely, functions proportional to the identity operator $I$) including $f(B) = \langle A\rangle$.
}
\end{figure}

\subsubsection{Statistical Interpretation of Conditional Quasi-expectations}

In classical probability theory, conditional expectations not only admitted geometric interpretation as orthogonal projections, but also statistical interpretation as conditioned averages.  We next seek to provide a quantum analogue of this observation, namely, to provide a statistical interpretation of the conditional quasi-expectations as `conditional averages' with respect to QJP distributions.  To this, we first introduce a general term:

\begin{definition*}[Conditional Quasi-expectation of Quantum Observables]
Let $\mu \in \mathfrak{M}_{A,B}^{\psi}$ be a QJP distribution of a pair of quantum observables $A$ and $B$ on the state $|\psi\rangle \in \mathcal{H}$, such that is admits representation by quasi-probability measures, and suppose that the expectation value $\mathbb{E}[\pi_{A};\mu]$ exists.  Denoting the measurable functions representing the behaviour of the measurement outcomes of $A$ and $B$ by $\pi_{A}(a,b) = a$ and $\pi_{B}(a,b) = b$, respectively, we then define the conditional quasi-expectation of $A$ given $B$ under the QJP distribution $\mu$ by
\begin{equation}
\mathbb{E}[A|B;\mu] := \mathbb{E}[\pi_{A}|\pi_{B}; \mu],
\end{equation}
where the definition of the r.~h.~s. is given in Section~\ref{sec:CM_II_cond_over_QP}, whenever the Radon-Nikod{\'y}m derivatives concerned exist.
\end{definition*}
\noindent
We next see that the definition of conditional quasi-expectations agree with those introduced earlier \eqref{def:cond_quasi-exp_alpha}.
\begin{proposition}[Statistical Interpretation of Conditional Quasi-expectations]\label{prop:stat_int_cond_quasi_exp}
Let $A, B \in L(\mathcal{H})$ be self-adjoint, $|\psi\rangle \in \mathcal{H}$, and let $\mu_{\mathrm{add}}^{\psi,\alpha}$ be a member of the additive complex-parametrised sub-family of the QJP distributions of $A$ and $B$ for the choice $\alpha \in \mathbb{C}$ defined as in \eqref{def:sec_app_additive}.
Then, the conditional quasi-expectation $\mathbb{E}[A|B; \mu_{\mathrm{add}}^{\psi,\alpha}]$ of $A$ given $B$ under $\mu_{\mathrm{add}}^{\psi,\alpha}$ is well-defined, which reads
\begin{align}
\mathbb{E}^{\alpha}[A|B;\psi] = \mathbb{E}[A|B; \mu_{\mathrm{add}}^{\psi,\alpha}],
\end{align}
where the l.~h.~s. is the $\alpha$-parametrised conditional quasi-expectation introduced earlier in \eqref{def:cond_quasi-exp_alpha}.
\end{proposition}
\begin{proof}
We first demonstrate the well-definedness of $\mathbb{E}[A|B; \mu_{\mathrm{add}}^{\psi,\alpha}]$, and to this, let
\begin{align}
\nu(\Delta)
    &:= \int_{\pi_{B}^{-1}(\Delta)} \pi_{A}(a,b)\ d\mu_{\mathrm{add}}^{\psi,\alpha}(a,b) \nonumber \\
    &= \int_{\mathbb{R}\times\Delta} a\ d\mu_{\mathrm{add}}^{\psi,\alpha}(a,b) \nonumber \\
    &= \frac{1+\alpha}{2} \cdot \frac{\langle \psi, E_{B}(\Delta)A\psi\rangle}{\|\psi\|^{2}} + \frac{1-\alpha}{2} \cdot \frac{\langle \psi, AE_{B}(\Delta)\psi\rangle}{\|\psi\|^{2}}.
\end{align}
Since $\nu \ll \mu_{B}^{\psi}$, the Radon-Nikod{\'y}m derivative $\mathbb{E}[A|B; \mu_{\mathrm{add}}^{\psi,\alpha}] := d\nu/d\mu_{B}^{\psi}$ exists.  The validity of the equality $\mathbb{E}^{\alpha}[A|B;\psi] = \mathbb{E}[A|B; \mu_{\mathrm{add}}^{\psi,\alpha}]$ is immediate by definition.
\end{proof}
\noindent
This result provides a statistical interpretation of conditional quasi-expectations as `conditional averages' of an observable $A$ given another observable $B$ with respect to the QJP distributions concerned.  As a corollary, one also has a statistical interpretation of Aharonov's weak value.
\begin{corollary}[Statistical Interpretation of Aharonov's Weak Value]
Under the same conditions, let $A_{w} : = \mathbb{E}^{1}[A|B; \psi]$ denote the Aharonov's weak value introduced in \eqref{def:weak_value}.  Then, the weak value 
\begin{equation}
A_{w}(B) = \mathbb{E}[A|B; \mu_{\mathrm{add}}^{\psi,1}] = \mathbb{E}^{1}[A|B; \psi]
\end{equation}
admits interpretation as the `conditional averages' of an observable $A$ given another observable $B$ with respect to the additive subfamily of QJP distributions for the choice $\alpha = 1$.
\end{corollary}

\newpage
\section{Summary and Discussion}\label{sec:sc}

Now we present a recap of our results obtained in this paper before going over to our discussions. 

\subsection{Summary}
The underlying motivation for our study was to obtain a coherent understanding to the formalism of quasi-joint-probabilities (QJP) of quantum observables, and to find the interpretation of Aharonov's weak value within this framework.  The main body, starting from Section 2 to 7, was devoted to the discussion of three logical groups of mutually interrelated topics, namely (i) an heuristic construction of QJP distributions (Section~\ref{sec:ups_I} to \ref{sec:ps_II}), (ii) formal definition of QJP distributions (Section~\ref{sec:qp_qo}), and (iii) its application to the interpretation of the weak value (Section~\ref{sec:app}). 
Each of these sections will be summarised concisely below.

\subsubsection{QJP: Heuristic Construction}

Four sections starting from Section~\ref{sec:ups_I} to Section~\ref{sec:ps_II} were devoted to some careful analyses on the quantum measurement models that we called the unconditioned and the conditioned measurement (UM and CM) schemes.  By inspecting each of the measurement models in terms of statistical averages and raw probability distributions, we confirmed that one may obtain the desired information of the target system by either looking into the strong or weak regions of the intensity of the interaction parameter.  Specifically, we saw that the study on the CM scheme naturally lead us to the concept of quasi-joint-probability (QJP) distributions of generally non-commuting pair of observables.

\paragraph{Section~\ref{sec:ups_I} (UM I)}
Section~\ref{sec:ups_I} was devoted to a review on the UM scheme from a standard operator-centric approach.
The quantity of interest was the statistical average of the meter observable after the interaction, from which we confirmed the well known fact that the information of both the target observable $A$ and the target state $|\phi\rangle$ can be retrieved in the form of the expectation value $\mathbb{E}[A;\phi]$ of $A$.  The expectation value $\mathbb{E}[A;\phi]$ was shown to be obtained from the measurement outcome of the meter observable, either by probing the strong limit $g \to \pm \infty$ or the weak limit $g \to 0$ of the interaction parameter.

\paragraph{Section~\ref{sec:ups_II} (UM II)}
In Section~\ref{sec:ups_II}, we took a closer look at the UM scheme in the level of probabilities, where the quantity of interest was now not just the statistical average but the `raw' probability measure describing the probabilistic behaviour of the measurement outcomes of the meter observable.  We saw that the outcome of the meter observable after the interaction was given by a convolution of both the initial profiles of the target and the meter observables. As for the retrieval of the target information, we found that, in a parallel manner to the previous section, the full profile of the target observable  
can be reclaimed by either probing the strong or the weak limit of the interaction.

\paragraph{Section~\ref{sec:ps_I} (CM I)}
In Section~\ref{sec:ps_I}, we conducted an analysis of the conditioned measurement scheme in the operator level, where the quantity of interest became the \emph{conditional expectation} of the meter observable under another given conditioning observable $B$ of the target system.  Some relevant topics, including a review and comments on the recent theoretical analyses on the alleged technical advantages of employing conditioning for precision measurements were presented, along with a measure theoretic approach to the possible limit for `amplification' by conditioning, and a systematical method (with an example) to analytically evaluate the conditional expectation in the case where $A$ has a spectrum consisting of finite points.  As for the retrieval of the target information, we exclusively studied the behaviour of the meter outcome in the weak region of the interaction parameter, and observed that the obtained value can be understood as a quantum analogue of conditional expectations, which we termed \emph{conditional quasi-expectations}, of the target observable $A$ given the conditioning observable $B$, to which Aharonov's weak value belongs as a special case.    It was also revealed that there exists some qualitative difference on the properties between the classical conditional expectations and the quantum analogue discussed here.

\paragraph{Section~\ref{sec:ps_II} (CM II)}
In Section~\ref{sec:ps_II}, the study of the conditioned measurement scheme was given a probabilistic approach, where the quantity of interest now became the \emph{QJP distribution} of a pair of canonically conjugate observables on the meter system, conditioned by the outcome of the conditioning observable $B$ of the target system.  For definiteness, this was accomplished in view of the Wigner-Ville distribution, which was primarily chosen as a convenient realisation among the various candidates of the quasi-probability distributions of the canonically conjugate pair that may be naturally associated with the quantum state of the meter system.
It was then argued that, in parallel to the UM case, one can recover the information of the target system by examining either the strong or the weak region of the interaction parameter, and that the information obtained can be understood as a quantum analogue of conditional probabilities, which we termed \emph{conditional quasi-probabilities}, of the target observable $A$ given the conditioning observable $B$ on the initial state $|\phi\rangle$. 
We then found that the conditional quasi-probability shares similar properties with the classical counterpart, while it admits complex values unlike the classical one.  We subsequently confirmed that, given $B$, the `statistical average' of the conditional quasi-probability of $A$ coincides with the conditional quasi-expectation of $A$ obtained in the preceding section.  This is precisely the same as the relation between classical conditional probabilities and conditional expectations.  

\subsubsection{QJP: Formal Definition}

Inspired by the heuristic arguments from the bottom-up and operational analyses given in the preceding four sections, in Section~\ref{sec:qp_qo} we provided the top-down discussion on QJP distributions defined for arbitrary pairs of generally non-commutative quantum observables.

\paragraph{Section~\ref{sec:qp_qo} (QJP of Quantum Observables)}
Based on the results of the spectral theorem for self-adjoint/normal operators on Hilbert spaces and their Fourier transforms, we proposed a general prescription for defining distributions describing the `joint behaviour' of a pair of generally non-commuting quantum observables, which serves as a natural generalisation to that defined for a pair of simultaneously measurable observables.  We then observed that the QJP defined this way for a non-commutative pair of observables admits  arbitrariness, that is, there exists a multitude of candidates that all share in common certain desirable properties to be qualified as QJP.  
We subsequently concentrated on a special sub-family of the class of QJP distributions parametrised by a single complex number for the ease of further discussions, such that it includes both the Wigner-Ville type and the Kirkwood-Dirac type of QJP distributions which are among the most familiar examples considered in the literature.  We then summarised our results obtained up to Section~\ref{sec:ps_II} from a relatively aerial viewpoint gained here, and discussed where the heuristic arguments and observations in the foregoing sections find their places in this broader framework.

\subsubsection{Application}

As the final topic, we gave an example of application of our observations on QJP distributions of quantum observables.

\paragraph{Section~\ref{sec:app} (Application: Interpretation of the Weak Value)}

To discuss where the mathematical observations on QJP distributions may find their use, we studied on the quantum analogue of `correlations' (inner products) that can be defined even for a pair of non-commuting observables.  As is well known, due to the non-commutative nature of quantum observables, there is no unique way to introduce a `natural inner product' on the space of quantum observables.  We showed that the ambiguity of the possible geometries that can be introduced on the space corresponds precisely to the ambiguity of the definition of QJP distributions, and that the QJP distributions provides a convenient representation of the geometries in terms of `integration' (statistics).
We then concentrated on a special sub-family of all possible QJP distributions parametrised by a single complex number and observed that
the geometric concept of \emph{orthogonal projection} may be endowed with a statistical interpretation as \emph{conditioning}.  
This fact is analogous to the classical case, while the difference lying in the fact that, for the quantum case, there could be multiple orthogonal projections due to the non-uniqueness of the inner product.  
The main finding is that, Aharonov's weak value may be understood as a special realisation of the possible orthogonal projections of a quantum observable $A$ onto the space of all normal operators generated by another observable $B$, and at the same time, as a conditioning of $A$ when the outcomes of $B$ is given.  The former is a geometric interpretation of the weak value, while the latter is its statistical interpretation, but since QJP distributions tie them together, both interpretations are equivalent.

\subsection{Discussion}

Since the advent of quantum theory founded nearly a century ago,
non-commutativity of quantum observables has undoubtedly been one of the major sources of troubles we face when we try to interpret their measurement outcomes in a sensible manner.  
This has naturally led to various attempts of `quasi-classical' interpretation of quantum observables in terms of commuting quantities familiar to us in classical theory.  
Wigner, Weyl and Moyal were among the prominent figures who have made much contribution in this effort, bearing most notably the theory of Wigner-Weyl transform \cite{Weyl_1927} and Weyl-Groenewold-Moyal product \cite{Groenewold_1946, Moyal_1949}.  In particular, the theory of 
Wigner-Weyl transform provides an invertible mapping between functions defined on a phase space and operators on a Hilbert space, in which the mapping from functions to operators is called the Weyl transform, whereas the inverse is called the Wigner transform.   
It is notable in this respect that the Wigner-Ville distributions arise as the Wigner transform of density operators, and from this follows the fact that 
the expectation values of quantum observables can be expressed as the statistical average by integration of their Wigner transforms with respect to the Wigner-Ville distribution defined on the phase space.

Viewed from the broader context of these quasi-classical transforms, the mathematical methods developed in this paper may be understood as another functional analytic approach to this problem.  Recall that, in functional analysis, a map that assigns a ring of functions onto a \emph{commutative} sub-algebra of the algebra of quantum observables is known as the functional calculus, which in turn is known to be uniquely represented by a spectral measure.  The family of quasi-joint-spectral distributions (QJSDs) introduced in Section~\ref{sec:qp_qo} are non-commutative analogues of spectral measures, which induce maps that assign functions to generally \emph{non-commutative} sets of quantum observables.  Due to the possible non-commutativity of the chosen combination of observables, QJSDs are in general highly non-unique, and this leads to various candidates of quasi-classical transforms.  In fact, the Wigner-Weyl transform can be understood as a special case in this framework, namely, the quasi-classical transform corresponding to the member of our complex parametrised convolutive sub-family of QJSDs mentioned in the text for the particular choice $\alpha = 0$.  The method of `hashing' presented in this paper thus exemplifies a procedure for constructing a  broad class of candidates of quasi-classical transforms.

The method of quasi-classical transforms, to which the Wigner-Weyl transform belongs as a special case, not only offers a statistical interpretation of the behaviour of a combination of non-commuting quantum observables, but also sheds new light on the physical analysis in quantum mechanics pertaining to that process.  It should be obvious that one can draw an analogy to various concepts and results in classical probability theory when one considers the quantum counterparts obtained by this method, which allows for an intuitive treatment of the latter based on the geometric structure present in the probability theory.   Besides, transformation of Hilbert space operators into functions or quasi-probability distributions has its own technical merit in the mathematical analysis, since familiar results in measure and integration theory, including various convergence theorems, integral inequalities and representation theorems, are readily available.  

One of the direct applications taking advantage of these properties is the geometric/statistical interpretation of the weak value discussed in Section~\ref{sec:app}.  
There, we found that the weak value can be regarded as one of the possible quantum analogues of conditional expectations, which are indeed fundamental quantities in quantum mechanics as much as the standard conditional expectations are in classical probability theory.
This interpretation also leads to novel inequalities of uncertainty relations for approximation and estimation which are capable of treating both the position-momentum inequality and the time-energy inequality \cite{Lee_2016} within a unified framework.  

Finally, we wish to note that, in any conditioned quantum measurement such as the weak measurement, non-commutative observables must be dealt with in one way or another in the context of probability theory when one tries to make sense of the measurement outcome.
Given this, we expect that our method of quasi-classical transforms, which is established on a rigorous mathematical basis, may offer
a fundamental and practical scheme in which issues involving measurement results of non-commuting observables are analysed.

\newpage
\section*{Acknowledgment}
The authors appreciate Professor A. Hosoya and S. Tanimura for helpful discussions and insightful comments.
This work was supported in part by JSPS KAKENHI No.~25400423,  No.~26011506, and by the Center for the Promotion of Integrated Sciences (CPIS) of SOKENDAI.
\vfill\pagebreak

\bibliographystyle{unsrt}
\bibliography{wm}

\newpage

\appendix

\section{Post-selected Measurement}\label{sec:PSM}

Given that the conditioning observable $B$ has a spectrum of finite cardinality, we have seen in \eqref{eq:cond_exp_B_fin} that the conditional expectation $\mathbb{E}[X|B ; \Psi^{g}]$ admits an explicit expression, of which value reduces to
\begin{align}
\mathbb{E}[X|B = b ; \Psi^{g}]
    &= \frac{\mathbb{E}\left[\Pi_{b} \otimes X; \Psi^{g} \right]}{\left\|(\Pi_{b} \otimes I)\Psi^{g}\right\|^{2}} \nonumber \\
    &= \mathbb{E}[X|\Pi_{b} = 1 ; \Psi^{g}]
\end{align}
for the choice $b \in \sigma(B)$ such that the probability of observing it is non-vanishing. This is roughly to say that the description of a conditioning by a general observable $B$, hence the study of conditioned measurement scheme, essentially reduces to that given by a projection. Of course, this should be intuitively clear, since each self-adjoint operator admits a unique spectral decomposition.
As the extreme case, the choice of the conditioning observable
\begin{equation}\label{eq:cond_obs_1_dim}
B = |\phi^{\prime}\rangle\langle\phi^{\prime}|
\end{equation}
given by a projection on a one-dimensional subspace of $\mathcal{H}$ spanned by a unit vector $|\phi^{\prime}\rangle \in \mathcal{H}$ becomes of special interest for our study. The vast majority of literatures with similar interest to this paper is devoted to the study of this special type of conditional measurement, and the act of measuring the conditional expectation
\begin{equation}
\mathbb{E}[X| |\phi^{\prime}\rangle\langle\phi^{\prime}| = 1 ; \Psi^{g}]
\end{equation}
is mostly referred to as the `post-selected measurement' or the `weak measurement'. In such context, the unit vector $|\phi^{\prime}\rangle$ is occasionally called the \emph{final state}, denoted as $|\phi_{f}\rangle := |\phi^{\prime}\rangle$, in order to contrast it with the initial state denoted as $|\phi_{i}\rangle := |\phi\rangle$.

\subsection{Example: Analytic Model}

We are now interested in the construction of a model in which the conditional expectation can be analytically computed for all range of the interaction parameter $g \in \mathbb{R}$. To this end, we assume that the target observable $A$ has a spectrum of finite cardinality. 
One readily sees that the `conditional' composite state essentially reduces to the computation of the vector
\begin{equation}\label{eq:comp_stat_pure}
(\Pi_{f} \otimes I) |\Psi^{g}\rangle = \sum_{n=1}^{N} \langle \phi_{f}, \Pi_{a_{n}} \phi_{i} \rangle \cdot |\phi_{f} \otimes e^{-iga_{n}Y}\psi\rangle
\end{equation}
for the special choice of the conditioning observable $B=\Pi_{f} := |\phi_{f}\rangle\langle\phi_{f}|$ defined by some final state $|\phi_{f}\rangle \in \mathcal{H}$.  A careful observation reveals that the `conditional' meter state \eqref{def:mixed_state_given_b}, which is in general a mixed state for the general conditioning observable $B$, in fact becomes a pure state
\begin{equation}\label{eq:cond_met_stat_pure}
|\psi^{g}_{\Pi_{f} = 1}\rangle = \frac{\sum_{n=1}^{N} \langle \phi_{f}, \Pi_{a_{n}} \phi_{i} \rangle \cdot |e^{-iga_{n}Y}\psi\rangle}{\left\| \sum_{n=1}^{N} \langle \phi_{f}, \Pi_{a_{n}} \phi_{i} \rangle \cdot |e^{-iga_{n}Y}\psi\rangle \right\|}
\end{equation}
for the post-selected measurement case, in which the conditional expectation reads
\begin{align}\label{eq:psm_cond_exp_fin}
\mathbb{E}[X|\Pi_{f} = 1; \Psi^{g}]
    &= \mathbb{E}[X; \psi^{g}_{\Pi_{f} = 1}] \nonumber \\
    &= \frac{\sum_{m = 1}^{N}\sum_{n = 1}^{N} \langle \phi_{i}, \Pi_{a_{m}} \phi_{f} \rangle\langle \phi_{f}, \Pi_{a_{n}} \phi_{i} \rangle \langle e^{-iga_{m}Y} \psi, X e^{-iga_{n}Y} \psi\rangle}{\sum_{m = 1}^{N}\sum_{n = 1}^{N} \langle \phi_{i}, \Pi_{a_{m}} \phi_{f} \rangle\langle \phi_{f}, \Pi_{a_{n}} \phi_{i} \rangle \langle e^{-iga_{m}Y} \psi, e^{-iga_{n}Y} \psi\rangle},
\end{align}
whenever the denominator is non-vanishing, {\it i.e.}, when the `conditional' meter state is not a zero vector.
One thus learns that the computation of the conditional expectation essentially reduces to the that of the quantity of the form
\begin{equation}
\langle e^{-iga_{m}Y} \psi, Z e^{-iga_{n}Y} \psi\rangle, \quad 1 \leq m, n \leq N
\end{equation}
for the choice $Z = I, Q, P$.

\paragraph{Gaussian Example}

For our demonstration, we consider the simplest non-trivial model in which the target observable $A$ is dichotomic, that is, it has a discrete spectrum consisting of only two distinct eigenvalues $\{a_{1}, a_{2}\}$.  For concreteness, we now assume that the meter system is described by the Schr{\"o}dinger representation of the CCR $\{ L^{2}(\mathbb{R}), \mathscr{S}(\mathbb{R}), \{\hat{x}, \hat{p}\}\}$, and choose $Y = \hat{p}$ without loss of generality.
Despite its simplicity, this model should retain its usefulness in the sense that it covers the situations in recent experiments of weak measurement \cite{Hosten_2008,Dixon_2009}. We also note that the condition $A^{2} = I$, under which the previous works \cite{Wu_2011,Koike_2011,Nakamura_2012} performed a full order calculation, is in fact a special case ($\{a_{1}, a_{2}\} = \{ -1, 1\}$) of our setting.

Now, by recalling that the subspace $\mathscr{S}(\mathbb{R}) \subset L^{2}(\mathbb{R})$ is in particular invariant under the operations $\hat{x}$ and $\hat{p}$, hence $\mathscr{S}(\mathbb{R}) \subset \mathcal{D}$, we see from our previous general argument that for any choice of the initial meter state $\psi \in \mathscr{S}(\mathbb{R})$ and the pair of pre- and post-selections satisfying the non-orthogonality condition $|\langle\phi_{f} |\phi_{i}\rangle| \neq 0$, the conditional expectation \eqref{eq:cond_exp_val} should be well-defined on an appropriate neighbourhood of $g=0$.
For both definiteness and practicality, we shall choose the initial meter state $\psi \in \mathscr{S}(\mathbb{R})$ to be a real Gaussian wave function
\begin{equation}
\psi(x) := \pi^{-1/4}\exp\left(-\frac{x^{2}}{2}\right)
\end{equation}
centred at $x=0$ with normalisation $\|\psi\|_{2} = 1$.
In order to see how the choice of the parameter $g$ and that of the initial meter state affects the result of the measurement, we consider the family of states $\{ \psi_{(h)}\}_{h > 0}$ scaled from the Gaussian state defined by
\begin{equation}\label{eq:Gaussian_approximate_identity}
\psi_{(h)}(x) := (\pi h^{2})^{-1/4}\exp\left(-\frac{x^{2}}{2h^{2}}\right),
\end{equation}
where now the parameter $h$ specifies the `width' of the initial Gaussian profile of $\psi$ ({\it cf}. \eqref{eq:approx_ident_state}).
One then finds
\begin{align}\label{eq:gauss_computation}
\langle e^{-iga_{m}Y} \psi, Z e^{-iga_{n}Y} \psi\rangle
    &= \int_{\mathbb{R}} \psi^{*}_{(h)}(x-ga_{m}) Z \psi_{(h)}(x-ga_{n})\ d\beta(x) \nonumber \\
    &= 
 \begin{cases}
    \exp \left(-\frac{g^{2}}{h^{2}} \left( \frac{a_{m} - a_{n}}{2} \right)^{2} \right), & Z = I,  \\
    \left(g \cdot \frac{a_{m} + a_{n}}{2} \right) \exp \left(-\frac{g^{2}}{h^{2}} \left( \frac{a_{m} - a_{n}}{2} \right)^{2} \right), & Z = \hat{x}, \\
    \left( i\frac{g}{h^{2}} \cdot \frac{a_{m} - a_{n}}{2} \right) \exp \left(-\frac{g^{2}}{h^{2}} \left( \frac{a_{m} - a_{n}}{2} \right)^{2} \right), & Z = \hat{p}.
 \end{cases}
\end{align}

Given the spectral decomposition $A = a_{1}\Pi_{a_{1}} + a_{2}\Pi_{a_{2}}$ for our case, we introduce the shorthand
\begin{equation}
\Lambda_{m} := {{a_{1} + a_{2}}\over{2}},
\quad \Lambda_{r} := {{a_{2} - a_{1}}\over{2}},
\quad A_{w}^{0} := A_{w} - \Lambda_{m}
\end{equation}
for later convenience, which respectively represents the barycentre of the two eigenvalues, the half-width of the numerical range, and the `centralised' weak value of $A$ defined by
\begin{equation}
A_{w} := \frac{\langle \phi_{f}, A \phi_{i} \rangle}{\langle \phi_{f}, \phi_{i} \rangle}.
\end{equation}
One then finds through routine computation (see Appendix \ref{comp:Gaussian} for computational detail) the following results
\begin{align}
	\mathbb{E}\left[\hat{x} | \Pi_{f}=1 ; \Psi^{g} \right] &= g \cdot \frac{\,\mathrm{Re} [A_{w}^{0}] }{1 + a\left( 1 - e^{-g^{2}\Lambda_{r}^{2}/h^{2}} \right)} + g  \Lambda_{m}, \label{eq:Gauss_Shift_Q} \\
	\mathbb{E}\left[\hat{p} | \Pi_{f}=1 ; \Psi^{g} \right] &= \frac{g}{2h^{2}} \cdot \frac{ \,\mathrm{Im} [A_{w}^{0}] e^{-g^{2}\Lambda_{r}^{2}/h^{2}}}{1 + a \left( 1 - e^{-g^{2}\Lambda_{r}^{2}/h^{2}} \right)} \label{eq:Gauss_Shift_P},
\end{align}
where we have used the quantity
\begin{equation}
a := \frac{1}{2}  \left( \left| \frac{A_{w}^{0}}{\Lambda_{r}} \right|^{2} - 1 \right),
\end{equation}
which is to be understood as a parameter corresponding to the `amplification rate' of the `centralised' weak value $A_{w}^{0}$ of $A$ to the half-width of its numerical range $\Lambda_{r}$.
We mention again that the result of the previous works in which $A^{2}=1$ is assumed is indeed a special case of the above formulae: we just put $\Lambda_{m} = 0$, $\Lambda_{r} = 1$, $a = \frac{1}{2} (| A_{w} |^{2} - 1 )$ to reproduce it.

\paragraph{Some Observations}

While the general argument only assures that the shift of the conditional expectation values are well-defined on an appropriate neighbourhood $U_{0}$ of $g=0$ for a given non-orthogonal choice of pre- and post-selections, the above result shows that it is in fact well-defined on the whole real line (hence $U_{0} = \mathbb{R}$) for our case.  Moreover, we also find that the shifts are indeed bounded for any choice of the pair of states of the target system due to the presence of the term 
$| A_{w}^{0}|^{2}$ hidden in the quantity $a$ in the denominator.

As for the recovery of the weak value $A_{w}$, one realises that, since the present choice of the meter state implies $\psi \in \mathscr{S}(\mathbb{R}) \subset \mathcal{D}$, the general argument in the previous subsection guarantees the differentiability of the shift, and by noting that $\mathbb{CV}_{\mathrm{S}}[\hat{x},\hat{p}; \psi] = 0$ and $\mathbb{CV}[\hat{p},\hat{p}; \psi] = \mathbb{V}[\hat{p}; \psi] = (2h^{2})^{-1}$, one should have
\begin{equation}
\left. \frac{d}{dg} \mathbb{E}[X|\Pi_{f} =1; \Psi^{g}] \right|_{g=0} =
    \begin{cases}
        \,\mathrm{Re}[A_{w}], & X = \hat{x}, \\
        \,\mathrm{Im}[A_{w}] \cdot (2h^{2})^{-1}, & X = \hat{p},
    \end{cases}
\end{equation}
based on the result \eqref{eq:wpsm_diff}. Indeed, observing that $\mathbb{E}\left[X | \Pi_{f}=1 ; \Psi^{0} \right] = 0$ for both choices $X \in \{\hat{x}, \hat{p}\}$, one may directly verify this as
\begin{align}
\left. \frac{d}{dg} \mathbb{E}[\hat{x}|\Pi_{f} =1; \Psi^{g}] \right|_{g=0}
    &= \lim_{g \to 0} \frac{\mathbb{E}[\hat{x}|\Pi_{f} =1; \Psi^{g}]}{g} \nonumber \\
    &= \lim_{g \to 0} \frac{\,\mathrm{Re} [A_{w}^{0}] }{1 + a\left( 1 - e^{-g^{2}\Lambda_{r}^{2}/h^{2}} \right)} + \Lambda_{m} \nonumber \\
    &= \,\mathrm{Re}[A_{w}^{0}] + \Lambda_{m} \nonumber \\
    &= \,\mathrm{Re}[A_{w}]
\end{align}
and
\begin{align}
\left. \frac{d}{dg} \mathbb{E}[\hat{p}|\Pi_{f} =1; \Psi^{g}] \right|_{g=0}
    &= \lim_{g \to 0} \frac{\mathbb{E}[\hat{p}|\Pi_{f} =1; \Psi^{g}]}{g} \nonumber \\
    &= \lim_{g \to 0} \frac{1}{2h^{2}} \cdot \frac{ \,\mathrm{Im} [A_{w}^{0}] e^{-g^{2}\Lambda_{r}^{2}/h^{2}}}{1 + a \left( 1 - e^{-g^{2}\Lambda_{r}^{2}/h^{2}} \right)} \nonumber \\
    &= \,\mathrm{Im}[A_{w}] \cdot (2h^{2})^{-1}
\end{align}
as expected

Another observation worthy of note is that the scaled outputs $\mathbb{E}[\hat{x}|\Pi_{f} =1; \Psi^{g}]/g$ and $\mathbb{E}[\hat{p}|\Pi_{f} =1; \Psi^{g}]/(g/(2h^{2}))$ are dependent on the parameters $g$ and $h$ only through the combination $hg^{-1}$, and that both tend to their respective desired values
\begin{align}
\lim_{hg^{-1} \to \infty} \frac{\mathbb{E}[\hat{x}|\Pi_{f} =1; \Psi^{g}]}{g}
    &= \lim_{hg^{-1} \to \infty} \frac{\,\mathrm{Re} [A_{w}^{0}] }{1 + a\left( 1 - e^{-g^{2}\Lambda_{r}^{2}/h^{2}} \right)} + \Lambda_{m} = \,\mathrm{Re} [A_{w}],\\
\lim_{hg^{-1} \to \infty} \frac{\mathbb{E}[\hat{p}|\Pi_{f}=1; \Psi^{g}]}{g/(2h^{2})}
    &= \lim_{hg^{-1} \to \infty} \frac{ \,\mathrm{Im} [A_{w}^{0}] e^{-g^{2}\Lambda_{r}^{2}/h^{2}}}{1 + a \left( 1 - e^{-g^{2}\Lambda_{r}^{2}/h^{2}} \right)} = \,\mathrm{Im}[A_{w}]
\end{align}
by taking the limit of the combination $hg^{-1} \to \infty$.
Observe that the manner in which we take the limit of the combination $hg^{-1}$ to recover the desired information is the opposite between the unconditioned case (`strong'/`sharp' measurement) \eqref{eq:ups_g_and_h} and the post-selected case above.  Namely, here we may either take the interaction $g \to 0$ to the `weak' limit, broaden the wave-function $h \to \infty$ to the `unsharp' limit, or appropriately balancing the combination thereof and let $hg^{-1} \to \infty$ as a whole.

\subsection{Computation of the Gaussian Example}\label{comp:Gaussian}

For better readability, we write
\begin{equation}
c_{n} := \langle \phi_{f}, \Pi_{a_{n}} \phi_{i} \rangle, \quad n=1,2.
\end{equation}
Observing that
$\langle\phi_{f}, A \phi_{i}\rangle = a_{1}c_{1} + a_{2}c_{2}$ and $\langle\phi_{f} , \phi_{i}\rangle = c_{1} + c_{2}$, one has
\begin{align}
A_{r}
    &:= \frac{\langle \phi_{f}, A\phi_{i} \rangle}{\langle \phi_{f}, \phi_{i} \rangle} - \Lambda_{m} \nonumber \\
    &= \frac{a_{1}c_{1} + a_{2}c_{2}}{c_{1} + c_{2}}  - \Lambda_{m} \nonumber \\
    &= \frac{a_{1}|c_{1}|^{2} + a_{2}|c_{2}|^{2} + a_{1}c_{1}c^{*}_{2} + a_{2}c^{*}_{1}c_{2}}{|c_{1}|^{2} + |c_{2}|^{2} + 2\,\mathrm{Re}\left[c^{*}_{1}c_{2}\right]}  - \Lambda_{m} \nonumber \\
    &= \Lambda_{r} \cdot \frac{- |c_{1}|^{2} + |c_{2}|^{2} + 2i\,\mathrm{Im}[c^{*}_{1}c_{2}]}{|c_{1}|^{2} + |c_{2}|^{2} + 2\,\mathrm{Re}\left[c^{*}_{1}c_{2}\right]},
\end{align}
whereby one obtains
\begin{align}
	\,\mathrm{Re}[A_{r}] &= \Lambda_{r} \cdot \frac{- |c_{1}|^{2} + |c_{2}|^{2}}{|c_{1}|^{2} + |c_{2}|^{2} + 2\,\mathrm{Re}\left[c^{*}_{1}c_{2}\right]}, \\
	\,\mathrm{Im}[A_{r}] &= \Lambda_{r} \cdot \frac{2 \,\mathrm{Im}\left[c^{*}_{1}c_{2}\right] }{|c_{1}|^{2} + |c_{2}|^{2} + 2\,\mathrm{Re}\left[c^{*}_{1}c_{2}\right]}
\end{align}
and
\begin{align}
a &:= \frac{1}{2} \left(\frac{\vert A_{r}\vert^{2}}{\Lambda_{r}^{2}} - 1 \right) \nonumber \\
    &= \frac{1}{2} \left( \frac{(- |c_{1}|^{2} + |c_{2}|^{2})^{2} + 4(\,\mathrm{Im} \left[c^{*}_{1}c_{2}\right])^{2}}{(\lvert c_{1} \rvert^{2} + \lvert c_{2} \rvert^{2} + 2\,\mathrm{Re} \left[c^{*}_{1}c_{2}\right])^{2}} -1 \right) \nonumber \\
    &= \frac{1}{2} \left( \frac{(|c_{1}|^{2} + |c_{2}|^{2})^{2} -4|c^{*}_{1}c_{2}|^{2} + 4(\,\mathrm{Im} \left[c^{*}_{1}c_{2}\right])^{2}}{(\lvert c_{1} \rvert^{2} + \lvert c_{2} \rvert^{2} + 2\,\mathrm{Re} \left[c^{*}_{1}c_{2}\right])^{2}} -1 \right) \nonumber \\
    &= \frac{1}{2} \left( \frac{(|c_{1}|^{2} + |c_{2}|^{2})^{2} - 4(\,\mathrm{Re} \left[c^{*}_{1}c_{2}\right])^{2}}{(\lvert c_{1} \rvert^{2} + \lvert c_{2} \rvert^{2} + 2\,\mathrm{Re} \left[c^{*}_{1}c_{2}\right])^{2}} -1 \right) \nonumber \\
    &= - \frac{2\,\mathrm{Re} \left[c^{*}_{1}c_{2}\right]}{\lvert c_{1} \rvert^{2} + \lvert c_{2} \rvert^{2} + 2\,\mathrm{Re} \left[c^{*}_{1}c_{2}\right]}
\end{align}
in terms of $\{c_{n}\}$.
Then, based on \eqref{eq:psm_cond_exp_fin} and \eqref{eq:gauss_computation}, one has
\begin{align}
\langle \psi^{g}, I \psi^{g} \rangle
    &= |c_{1}|^{2} + |c_{2}|^{2} + 2\,\mathrm{Re}[c_{1}^{*}c_{2}] e^{-g^{2}\Lambda_{r}^{2}/d^{2}}, \\
\langle \psi^{g}, \hat{x} \psi^{g} \rangle
    &= |c_{1}|^{2} \cdot ga_{1} + |c_{2}|^{2} \cdot ga_{2} + 2\,\mathrm{Re}[c_{1}^{*}c_{2}] \cdot g\Lambda_{m} e^{-g^{2}\Lambda_{r}^{2}/d^{2}}, \\
\langle \psi^{g}, \hat{p} \psi^{g} \rangle
    &= 2\,\mathrm{Im}[c_{1}^{*}c_{2}] \cdot \frac{g}{d^{2}}\Lambda_{r} e^{-g^{2}\Lambda_{r}^{2}/d^{2}},
\end{align}
which in turn yields
\begin{align}
\mathbb{E}[\hat{x}; \psi^{g}]
    &= \frac{\langle \psi^{g}, \hat{x} \psi^{g} \rangle}{\langle \psi^{g}, I \psi^{g} \rangle} \nonumber \\
    &= g\frac{|c_{1}|^{2} \cdot a_{1} + |c_{2}|^{2} \cdot a_{2} + 2\,\mathrm{Re}[c_{1}^{*}c_{2}] \cdot \Lambda_{m} e^{-g^{2}\Lambda_{r}^{2}/d^{2}}}{|c_{1}|^{2} + |c_{2}|^{2} + 2\,\mathrm{Re}[c_{1}^{*}c_{2}] e^{-g^{2}\Lambda_{r}^{2}/d^{2}}} \nonumber \\
    &= g\frac{|c_{1}|^{2} \cdot a_{1} + |c_{2}|^{2} \cdot a_{2} - (|c_{1}|^{2} + |c_{2}|^{2}) \Lambda_{m} }{|c_{1}|^{2} + |c_{2}|^{2} + 2\,\mathrm{Re}[c_{1}^{*}c_{2}] e^{-g^{2}\Lambda_{r}^{2}/d^{2}}} + g\Lambda_{m} \nonumber \\
    &= g \frac{\Lambda_{r}(- |c_{1}|^{2} + |c_{2}|^{2})}{|c_{1}|^{2} + |c_{2}|^{2} + 2\,\mathrm{Re}[c_{1}^{*}c_{2}] e^{-g^{2}\Lambda_{r}^{2}/d^{2}}} + g\Lambda_{m} \nonumber \\
    &= g \frac{\Lambda_{r}(- |c_{1}|^{2} + |c_{2}|^{2})}{(|c_{1}|^{2} + |c_{2}|^{2} + 2\,\mathrm{Re}[c_{1}^{*}c_{2}]) - 2\,\mathrm{Re}[c_{1}^{*}c_{2}] (1- e^{-g^{2}\Lambda_{r}^{2}/d^{2}})}  + g\Lambda_{m} \nonumber \\
    &= g \cdot \frac{\,\mathrm{Re}[A_{r}]}{1 + a (1- e^{-g^{2}\Lambda_{r}^{2}/d^{2}})}  + g\Lambda_{m}
\end{align}
and
\begin{align}
\mathbb{E}[\hat{p}; \psi^{g}]
    &= \frac{\langle \psi^{g}, \hat{p} \psi^{g} \rangle}{\langle \psi^{g}, I \psi^{g} \rangle} \nonumber \\
    &= \frac{2\,\mathrm{Im}[c_{1}^{*}c_{2}] \cdot \frac{g}{d^{2}}\Lambda_{r} e^{-g^{2}\Lambda_{r}^{2}/d^{2}}}{|c_{1}|^{2} + |c_{2}|^{2} + 2\,\mathrm{Re}[c_{1}^{*}c_{2}] e^{-g^{2}\Lambda_{r}^{2}/d^{2}}} \nonumber \\
    &= \frac{g}{d^{2}} \cdot \frac{\,\mathrm{Im}[A_{r}] e^{-g^{2}\Lambda_{r}^{2}/d^{2}}}{1 + a (1- e^{-g^{2}\Lambda_{r}^{2}/d^{2}})}.
\end{align}

\end{document}